\PassOptionsToPackage{prologue,dvipsnames}{xcolor}
\documentclass[acmsmall,nonacm]{acmart}

\usepackage{ifthen}
\newif\ifappendix
\appendixtrue

\usepackage{preamble}

\newcommand{\mit}[1]{{\small#1}}

\begin{document}

\title{Logical Relations for Session-Typed Concurrency}

\author{Stephanie Balzer}
\authornotemark[1]
\affiliation{%
\institution{Carnegie Mellon University}
\country{USA}
}
\author{Farzaneh Derakhshan}
\authornote{The first two authors have equal contributions.}
\affiliation{%
\institution{Illinois Institute of Technology}
\country{USA}
}
\author{Robert Harper}
\affiliation{%
\institution{Carnegie Mellon University}
\country{USA}
}
\author{Yue Yao}
\affiliation{%
\institution{Carnegie Mellon University}
\country{USA}
}
\begin{abstract}

Program \emph{equivalence} is the heart of
reasoning about and proving properties of programs.
To assert noninterference, for example,
a program is shown to be equivalent to itself up to the confidentiality level of an observer.
A powerful enabler for such proofs are \emph{logical relations}, which,
guided by the type structure,
prescribe when two programs are indistinguishable.
Logical relations enjoy ample exploration in functional languages,
including languages with general recursion and a higher-order store---%
yet logical relations for session types only exist for terminating languages.
This paper scales logical relations to \emph{general recursive} session types.
It develops a logical relation for \emph{progress-sensitive equivalence}
for \emph{intuitionistic linear logic session types},
tackling the challenges non-termination and concurrency pose.
In particular, the relation only equates a diverging program with another diverging one and
accounts for nondeterminism of scheduling.
The logical relation has two distinguishing characteristics: it is
\text{(i)} indexed with an \emph{intuitionistic linear sequent},
validating \emph{cut reductions} and affording \emph{biorthogonal closure},
and \text{(ii)} bound by an \emph{observation index},
stratifying the logical relation in the presence of recursion.
Biorthogonal closure validates the logical relation,
proving that the induced equivalence is \emph{sound and complete} with regard to
closure of weak bisimilarity under parallel composition.
Soundness guarantees that the equivalence has enough discriminatory power,
completeness ensures that it is maximally permissive.
The logical relation is then put to test on the example of noninterference.

\end{abstract}



\maketitle

\section{Introduction}
Whether two programs are \emph{equivalent} is at the heart of many program verification problems,
such as proofs of parametricity, compiler correctness, and noninterference.
This paper develops a \emph{logical relation} to reason about program equivalence of \emph{session-typed processes} and
proves soundness and completeness of the relation via a \emph{biorthogonality} argument.

\paragraph{Message-passing concurrency.}

Rooted in process calculi \myCite{HoareBOOK1985,MilnerBook1980,MilnerBook1989},
message-passing concurrency enjoys widespread adoption, including languages such as Erlang, Go, and Rust.
A program in this setting amounts to a number of \emph{processes} connected by \emph{channels},
which compute by exchanging \emph{messages} along these channels.
Messages can even amount to channels themselves,
giving rise to so-called \emph{higher-order} channels,
as present in the $\pi$-calculus \myCite{MilnerBook1999,SangiorgiWalkerBook2001}.
This feature is equally empowering as daunting because it changes the process topology dynamically.
Originally untyped, the $\pi$-caluclus \myCite{MilnerBook1999} has gradually been enriched with types \myCite{SangiorgiWalkerBook2001}
to prescribe the kinds of messages that can be exchanged over a channel.
More advanced type systems
\myCite{KobayashiLICS1997,IgarashiPOPL2001,IgarashiARTICLE2004,KobayashiCONCUR2006}
additionally assert correctness properties, such as deadlock freedom and data race freedom.

To prescribe not only the types of exchanged messages but also the \emph{protocol} underlying the exchange,
\emph{session types} \myCite{HondaCONCUR1993,HondaESOP1998} were introduced.
Session types rely on a \emph{linear} treatment of channels to model the state transitions induced by a protocol.
This foundation manifests in a correspondence between linear logic and the session-typed $\pi$-calculus
\myCite{CairesCONCUR2010,WadlerICFP2012},
resulting in two families of session type languages:
\emph{intuitionistic linear logic session types (\illst)}
\myCite{CairesCONCUR2010,ToninhoESOP2013,ToninhoPhD2015,CairesARTICLE2016}
and \emph{classical linear logic session types (\cllst)}
\myCite{WadlerICFP2012,LindleyMorrisESOP2015,LindleyMorrisICFP2016,KokkePOPL2019}.
Due to their logical foundation well-typed  \illst/\cllst processes
not only are protocol-compliant (\aka preservation),
but also free of data races and deadlocks (\aka progress).

This paper studies logical relations for \illst-typed processes.
\illst reject linear negation and distinguish the \emph{provider} from the \emph{client} side of a channel.
Run-time configurations of processes thus give raise to a \emph{rooted tree},
with a \emph{child} node as the provider and the \emph{parent} node as the client.

\paragraph{Program equivalence.}

Today's predominant techniques for reasoning about program equivalence of stateful programs are
Kripke logical relations (KLRs) \myCite{PittsStarkHOOTS1998,AhmedPOPL2009,DreyerICFP2010,NeisJFP2011,HurDreyerPOPL2011,ThamsborgBirkedalICFP2011}
and bisimulations \myCite{KoutavasWandPOPL2006,SumiiPierceARTICLE2007,StovringLassenPOPL2007,SangiorgiLICS2007}.
KLRs tend to be used for sequential, ML-like languages,
bisimulations for concurrent process calculi.
KLRs are phrased by structural induction on types,
bisimulations by coinduction.
Hur et al. \myCite{HurPOPL2012} note that KLRs and bisimulations have mutually disjoint strengths and weaknesses.
Whereas KLRS naturally support higher-order features, but struggle with recursive types,
bisimulations naturally support recursion, but struggle with higher-order features.
To support recursive types KLRs employ step-indexing \myCite{AppelMcAllesterTOPLAS2001,AhmedPhD2004,AhmedESOP2006},
complicating proofs with step arithmetic \myCite{BentonHurDagstuhl2010} and
challenging transitivity of the logical relation \myCite{HurPOPL2012}.

This paper scales logical relations to session-typed concurrency in the presence of \emph{general recursive types},
contributing a \emph{recursive session logical relation (\slr)} for intuitionistic linear logic session types (\illst).
In contrast to KLRs, which index the logical relation with the type of the considered expression,
\slrs index the logical relation with an \emph{intuitionistic linear sequent},
denoting the types of the free channels along which the considered process configuration exchanges messages with the outside world.
The usual term interpretation clause
$$(e_1; e_2) \in  \mc{E}\llbracket \tau \rrbracket \; \m{iff} \; \dots$$
for expressions \tm{e_1} and \tm{e_2} and a type \tm{\tau} in KLRs
thus takes the form
$$(\mc{D}_1; \mc{D}_2) \in  \mc{E}\llbracket \Delta \Vdash A \rrbracket \; \m{iff} \; \dots$$
in \slrs,
for configurations \tm{\mc{D}_1} and \tm{\mc{D}_1} with the types of free channels denoted by the sequent \tm{\Delta \Vdash A}.
This generalization is necessitated by the underlying computational model (message-passing concurrency)
and the goal to accommodate various program verification problems, including noninterference.
The use of a sequent as an index makes explicit the \emph{duality} of session types,
allowing a type to be associated with two interpretations:
as a \emph{provider} of a session of that type and a \emph{client}.
For example, as a provider of a session \tm{\& \{\mathit{left}{:}A, \mathit{right}{:} B \}},
two related configurations may \emph{assume} to receive related messages,
\ie either they both receive \tm{\mathit{left}} or \tm{\mathit{right}}.
As a client of such a session, conversely,
related configurations must \emph{assert} to either both send \tm{\mathit{left}} or \tm{\mathit{right}}.
The latter condition, crucial for establishing noninterference,
is lost when only indexing the logical relation with one type (\ie the type of the provider).
Duality is also central to enforcing a \emph{resource} semantics for channels.
For example, as a provider of a session \tm{\chanout{A}{B}}, promising to send a channel of type \tm{A},
related configurations must give up access to the sent channel,
whereas a client of such a session will become its owner.
Thanks to the validity of semantic cut (\Cref{lem:compositionality}),
the usual, single-type-index interpretation falls out for free as a special case.

To accommodate general recursive types, \slrs use an \emph{observation index} to stratify the logical relation.
An observation index \tm{m} is associated with the free channels in
\tm{(\mc{D}_1; \mc{D}_2) \in  \mc{E}\llbracket \Delta \Vdash A \rrbracket^m}
and bounds the number of messages that \tm{\mc{D}_1} and \tm{\mc{D}_2} exchange along those channels.
The index is thus associated with an \emph{externally observable} event,
symmetrically bounding both configurations.
In contrast KLRs employ a step index
\myCite{AppelMcAllesterTOPLAS2001,AhmedESOP2006,DreyerLICS2009},
which is tied to the internal computation or unfolding steps of one of the two terms, making the relation asymmetric.

We establish several metatheoretic properties of \slrs.
In particular, we show that \slrs entail an \emph{equivalence relation},
whose proof of transitivity (\Cref{lem:transitivity-eq})
benefits from the extensionality of the observation index.
Given the possibility of divergence,
the equivalence is \emph{progress-sensitive} \myCite{HeidinSabelfeldMartkoberdorf2011}
in that it equates a divergent program only with another diverging one.
We also show that \slrs are \emph{sound and complete} (\Cref{thm:top-top-closure}) with regard to
closure of weak \emph{bisimilarity} under parallel composition,
using a \emph{biorthogonality} argument \`{a} la Pitts \myCite{PittsMSCS2000}.
The \emph{duality} of session types establishes a natural connection to biorthogonality.
To put the \slr to test, we employ it to verify \emph{progress-sensitive noninterference}.

\paragraph{Contributions.}

In summary, this paper makes the following contributions:

\begin{itemize}

\item \emph{Recursive session logical relation (\slr)},
a binary logical relation for \emph{progress-sensitive equivalence} of \illst processes with general recursion,
featuring an \emph{intuitionistic linear sequent} and \emph{observation index},
to capture duality and to support general recursion without compromising extensionality, \respb.

\item Closure of the \slr under parallel composition, \aka \emph{semantic cut} (\Cref{lem:compositionality}).

\item \slr induces an \emph{equivalence relation}, including \emph{transitivity} (\Cref{lem:transitivity-eq}).

\item Soundness and completeness of the \slr with regard to \emph{weak asynchronous bisimilarity},
\aka \emph{adequacy} (\Cref{lem:soundness-completeness}).

\item Biorthogonal (\tm{\top\top}-) closure of the induced equivalence relation (\Cref{thm:top-top-closure}),
guaranteeing that the equivalence has enough discriminatory power (soundness),
while being maximally permissive (completeness).

\item Small case study applying the \slr to \emph{progress-sensitive noninterference}.

\end{itemize}

\ifappendix
\paragraph{Appendix.}
The complete formalization and all the proofs are available in the Appendix.
\fi

\section{Background}
\label{sec:background}
This section familiarizes with intuitionistic linear logic session types (\illst).

\subsubsection{\illst type system.}
\label{subsec:type-system}

Our development is based on a variant of \illst languages
\myCite{CairesCONCUR2010,ToninhoESOP2013,ToninhoPhD2015,CairesARTICLE2016}
that supports general recursive types \myCite{ToninhoESOP2013,BalzerICFP2017}.
We refer to this language as \lang.

\lang's connectives are drawn from intuitionistic linear logic and obey the following grammar
$$A, B ::= \intchoice{A} \mid \extchoice{A} \mid \chanout{A}{B} \mid \chanin{A}{B} \mid 1 \mid Y,$$
where \tm{L} ranges over finite, non-empty sets of labels denoted by \tm{\ell} and \tm{k},
the primitive values in \lang.
Type variable \tm{Y} is a fixed point whose definition \tm{Y=A} is collected in a global signature \tm{\Sigma}.
General recursive types are supported through this definition mechanism.
They must be \emph{contractive} \myCite{GayHoleARTICLE2005}, demanding a message exchange before recurring,
and \emph{equi-recursive} \myCite{CraryPLDI1999},
avoiding explicit (un)fold messages and relating types up to their unfolding.

Process terms are typed using the \emph{intuitionistic sequent}
$$\Omega \vdash_{\Sigma} P::  x{:}A$$
to be read as
\textit{``process \mit{$P$} provides a session of type \mit{$A$} along channel variable \mit{$x$},
given the typing of sessions offered along channel variables in \mit{$\Omega$} and
given the type and process definitions in \mit{$\Sigma$}''}.
\mit{$\Omega$} is a \emph{linear} context, consisting of assumptions \mit{$y_i{:}B_i$},
indicating for each channel variable \mit{$y_i$} its session type \mit{$B_i$}.
The signature \mit{$\Sigma$} contains global type and process definitions
and facilitates recursive type and process definitions.
Bound channel variables are substituted with channels that are created at run-time upon process spawning.
When the distinction is clear from the context we refer to either as "channels" for brevity.

\illst rejects linear negation, allowing the singleton right-hand side of the sequent to be interpreted as the type of the \emph{providing} process,
which, conversely, is the \emph{client} of the sessions in \mit{$\Omega$}.
Channels can thus be typed with the type of the providing process,
rather than typing the two channel endpoints dually,
as done in classical linear logic session types (\cllst)
\myCite{WadlerICFP2012,LindleyMorrisESOP2015,LindleyMorrisICFP2016}.
To express the \emph{duality} in behavior the \illst type system is given as \emph{sequent calculus}, comprising both a \emph{right} and a \emph{left} rule per connective.
Right rules define a communication from the point of view of the provider, left rules from the point of view of the client.

\Cref{fig:structural-type-system} summarizes the typing rules.
Cut reduction in the sequent calculus invites a computational, bottom-up reading of the rules,
where the type of the conclusion denotes the protocol state of the provider before the message exchange
and the type of the premise the protocol state after the message exchange.
The polarity of type moreover determines the direction of communication: for positive connectives, the provider sends and the client receives, for negative connectives, the provider receives and the client sends.
We review each rule in turn next.

\begin{figure}
\centering
\begin{small}
\input{figs/structural-type-system}
\end{small}
\caption{Process term typing rules and signature checking rules of \lang.}
\label{fig:structural-type-system}
\end{figure}

The additive connectives \tm{\intchoice{A}} and \tm{\extchoice{A}}
denote \emph{labelled choices} of sessions \tm{A_{\ell}},
where the labels \tm{\ell} range over the finite, non-empty set \tm{L}.
The two connectives differ in who sends a label and thus makes a choice.
For \tm{\intchoice{A}}, the provider chooses,
for \tm{\extchoice{A}}, the client chooses,
giving the connectives the names \emph{internal choice} and \emph{external choice}, \respb.
The duality of behavior of a provider and client is conveyed by the process terms of the right and left rules, \respb.
For example, a provider of an internal choice (rule \tm{\oplus R}) executes the process term \tm{x.k; P} to send the label \tm{k} along its providing channel \tm{x},
after which it will continue with \tm{P} of type \tm{A_k}.
Conversely, a client of an internal choice (rule \tm{\oplus L}) executes the process term \tm{\mathbf{case}\, x(\ell \Rightarrow Q_\ell)_{\ell \in L}} to receive any label \tm{k \in L} along \tm{x},
after which it will continue with the chosen branch \tm{Q_k} of type \tm{A_k}.

The multiplicative connectives \tm{\chanout{A}{B}} (tensor) and \tm{\chanin{A}{B}} (lolli) allow channels themselves to be sent over channels,
changing the process topology at run-time.
These connectives endow \lang with \emph{higher-order channels},
making channels first-class values.
Again, the two connectives differ in who sends a channel:
for \tm{\chanout{A}{B}}, the provider sends, for \tm{\chanin{A}{B}}, the client sends.
Due to linearity, absence of weakening specifically, the sender loses access to the sent channel in its continuation;
as can be seen in the premises of rules \tm{\otimes R} and \tm{\multimap L}. 
We use \tm{\leftarrow} to denote variable binding.
For example, the process term \tm{z\leftarrow \mathbf{recv}\, x; P} in the conclusion of rule \tm{\otimes L} binds the received channel to \tm{y}, with scope \tm{P}.

The unit of \tm{\otimes}, \tm{1}, allows a process to terminate.
Here, the provider sends a close message,
\tm{\cls{x}}, along its providing channel \tm{x} (rule \tm{1R}),
awaited by the client.
Due to absence of weakening, rule \tm{1R} demands that the typing context \tm{\Omega} be empty, ensuring that no channels become orphans.
Consequently, the client continuation \tm{Q} loses access to \tm{x} after receipt of the close message (rule \tm{1L}).

Rule \tm{\msc{Spawn}} types a process that spawns another process using the process definition \tm{X}.
The invocation \mit{$x \leftarrow X[\gamma] \leftarrow \Omega_1$} binds the newly spawned process to \tm{x}
and gives it the argument channels \tm{\Omega_1}.
After the invocation, the process continues with executing \mit{$Q$}, now with \mit{$x{:}A$} in its context.
For the invocation to succeed, a corresponding process definition \mit{$\Omega'_1 \vdash X=P :: x'{:}A$} must exist in the signature \mit{$\Sig$}.
The definition associates a name \tm{X} with the process term \tm{P} and indicates the names and types of argument channels \tm{\Omega'_1} and
the name and type of the providing channel \tm{x'{:}A}.
To account for the difference in variable names,
the programmer must provide a corresponding mapping
\tm{\Omega_1, x{:}A\Vdash \gamma :: \Omega'_1, x'{:}A}
as part of the invocation,
assigning to each variable in the definition \tm{\Omega'_1, x'{:}A}
a corresponding variable of the same type in the invocation \tm{\Omega_1, x{:}A}.

Rule \tm{\msc{d:fwd}} allows a process to delegate requests to its singleton child $z$ by passing any messages received along $y$ to $z$ and vice versa.
The rule uses a forwarder definition
\tm{z'{:}C \vdash F_Y = \fwder{C,y' \leftarrow z'} :: y'{:}C},
for the type definition \tm{Y=C},
both to be present in the signature \tm{\Sigma}.
A corresponding forwarder \mit{$\fwder{C,y' \leftarrow z'}$} of type \mit{$C$} is defined
by structural induction on the type of every user-provided type definition,
amounting to an \emph{identity expansion}, and can be automatically generated.
In combination with rule $\msc{d:fwd}$, we get a coinductive definition
\ifappendix (see \Cref{apx:def:fwder} in the Appendix).\xspace \fi
The programmer can use the syntactic sugar \mit{$y \leftarrow z$} instead, which can be expanded into the corresponding forwarder.

Rules $\msc{TVar}_R$ and $\msc{TVar}_L$ allow us to unfold type definitions.  The rules insist that a corresponding definition exists in the global signature $\typdef{Y}{A} \in \Sig$.

Rules $\Sig_1$, $\Sig_2$, and $\Sig_3$ type check the signature itself.
Rule $\Sig_2$ type checks type definitions, requiring the type be contractive (\mit{$\Vdash_{\Sigma} A\; \mb{contr}$})
and thus not a type variable itself \myCite{GayHoleARTICLE2005}.
Rule $\Sig_3$ type checks the body \mit{$P$} of a process definition \mit{$X$},
whose existence is assumed when spawning the process in rule $\msc{Spawn}$.

We illustrate the typing rules on a simple PIN-based authentication example, shown in \Cref{fig:simple-verfier}.
\Cref{fig:simple-verfier} introduces two recursive type definitions,
$\m{pin}$ and $\m{ver}$,
as well as two process definitions,
$\m{Verifier}$ and $\m{Pin}$.
A verifier expects to receive a PIN channel from its client,
which it validates.
If validation is successful, the verifier sends the message $\mi{succ}$, otherwise the message $\mi{fail}$ to the client.
In either case, the verifier also returns the PIN channel.
The PIN itself is encoded as an internal choice of labels $\mathit{tok_i}$, with one being the correct security token.

\begin{figure}
\centering
\input{figs/simple-verifier.tex}
\caption{Example: PIN-based authentication.}
\label{fig:simple-verfier}
\end{figure}

\subsubsection{\illst configuration typing}
\label{subsec:config-typing}

At runtime, \lang programs become configurations of processes, forming \emph{rooted trees} and
obeying the following grammar:
\begin{center}
\begin{minipage}[t]{1\textwidth}
\begin{tabbing}
$\mathcal{C}, \mathcal{D}, \mathcal{F}, \mathcal{T}$ \=$::=$ \=$\mathbf{proc}(x_\alpha, \hat{\delta}(P)) \; \mathcal{C} \; \mid \; \mathbf{msg}(M) \; \mathcal{C} \; \mid \; \cdot$ \\[1pt]
$M$ \> $::=$ \> $x_\alpha.k \; \mid \; \mathbf{send}\, z_\beta\,x_\alpha \; \mid \; \mb{close}\, y_\alpha$
\end{tabbing}
\end{minipage}
\end{center}
Configurations consist of a set of processes of the form
\mit{$\mathbf{proc}(y_\alpha, \hat{\delta}(P))$} and
a set of messages of the form
\mit{$\mathbf{msg}(M)$}.
Every spawn results in a process \mit{$\mathbf{proc}(y_\alpha, \hat{\delta}(P))$},
every send in a message \mit{$\mathbf{msg}(M)$}, making \lang's dynamics asynchronous.
The metavariable \mit{$y_\alpha$} in \mit{$\mathbf{proc}(y_\alpha, \hat{\delta}(P))$}
denotes the process' providing channel and the metavariable \mit{$P$} the process' source code.
Run-time channels \mit{$y_\alpha$} can be distinguished from  channel variables \mit{$y$}
by their generation subscript \mit{$\alpha$},
whose meaning we clarify in \Cref{subsec:dynamics}.
The substitution \mit{$\delta$} maps variables to channels;
\mit{$\hat{\delta}(P)$} yields the term with all free channel variables substituted by channels.

The configuration typing rules of \lang are given in \Cref{fig:config-typing}, using the judgment
$\Delta_0 \Vdash \mathcal{C}:: \Delta$.
We allow configurations to have \emph{free channels},
which amount to \mit{$\Delta_0$} and \mit{$\Delta$} for the given judgment.
The channels in \mit{$\Delta$} comprise the channels for which there exist providing processes in the configuration \mit{$\mathcal{C}$},
the channels in \mit{$\Delta_0$} comprise the channels for which there exist client processes in the configuration \mit{$\mathcal{C}$}.
It may be surprising that we refer to these channels as free,
given that they occur in \mit{$C$}.
Would linearity not imply that they occur ``exactly once'', and thus must be bound?
It is more fruitful to think of channels as shared resources,
like memory locations,
\emph{shared} among processes.
\illst then ensures that at most two processes, a provider and a client,
share a channel at any point in time, where each "owns" one endpoint of (\ie a reference to) the channel.
The careful accounting of ``resources'' is visible in the configuration typing rules.
For example, rule $\mathbf{proc}$ insists that for all the channel endpoints used by process \mit{$\mathbf{proc}(x_\alpha, \hat{\delta}(P))$}
there exists either a provider in the sub-configuration \mit{$\mathcal{C}$}, making the channel bound,
or the endpoint is free and occurs in \mit{$\Delta'_0$}.

\begin{figure}
\centering
\begin{small}
\input{figs/config-typing}
\end{small}
\caption{Configuration typing rules of \lang.}
\label{fig:config-typing}
\end{figure}

\subsubsection{\illst asynchronous dynamics}
\label{subsec:dynamics}

\lang's asynchronous dynamics is given in \Cref{fig:dynamics}.
It is phrased as multiset rewriting rules \myCite{CervesatoARTICLE2009}.
Multiset rewriting rules express the dynamics as state transitions between configurations $C \ostep C'$ and
are \emph{local} in that they only mention the parts of a configuration they rewrite.

Being \emph{asynchronous}, the dynamics implements sends by spawning off a message \tm{\mathbf{msg}(M)} that carries the sent message \mit{$M$}.
In rule $\otimes_{\m{snd}}$, for example,
the provider \mit{$\mb{proc}(y_\alpha,\mathbf{send}\,x_\beta\,y_\alpha; P)$}
generates the message \mit{$\mb{msg}( \mathbf{send}\,x_\beta\,y_\alpha)$},
indicating that the channel \mit{$x_\beta$} is sent over channel \mit{$y_\alpha$}.
In rule $\otimes_{\m{rcv}}$, the client then consumes the message when receiving the channel. 
To preserve the invariant stipulated by typing that at most two processes
share a channel at any point in time,
a new providing channel must be generated for the continuation of the provider.
This new channel must be linked to the old providing channel, otherwise proper sequencing of messages would be violated.

An elegant way to achieve proper sequencing are \emph{channel generations}.
Rather than creating an entirely fresh channel name for each send,
we keep the name of a channel constant but increment its generation, provided as a subscript to the channel name.
For example, the provider \mit{$\mb{proc}(y_\alpha,\mathbf{send}\,x_\beta\,y_\alpha; P)$} in rule $\otimes_{\m{snd}}$ steps to its continuation
\mit{$\mb{proc}(y_{\alpha+1}, [y_{\alpha+1}/y_{\alpha}]P)$},
creating a new generation \mit{$\alpha + 1$} of its providing channel \mit{$y_\alpha$}.
The receiving client process
\mit{$\mb{proc}(u_\eta,w\leftarrow \mathbf{recv}\,y_\alpha; P)$}
will then increment the channel generation of the carrier channel \mit{$y_\alpha$} upon receipt in its continuation
\mit{$\mb{proc}(u_{\eta}, [x_\beta/w][y_{\alpha+1}/y_{\alpha}] P)$} (see rule $\otimes_{\m{rcv}}$).
Like bound variables, bound channels are subject to \mit{$\alpha$}-variance
and capture-avoiding substitution, as usual \myCite{SangiorgiWalkerBook2001},
with the understanding that only the name \mit{$y$} of a channel \mit{$y_\alpha$}
can be renamed, but not its generation.

\begin{figure}
\centering
\begin{small}
\input{figs/dynamics}
\end{small}
\caption{Asynchronous dynamics of \lang.}
\label{fig:dynamics}
\end{figure}

\section{Recursive session logical relation (\slr)}
\label{sec:rslr}
This section develops a logical relation for \illst with general recursive types (\Cref{subsec:rslr}) and
shows that the logical relation is closed under parallel composition (\Cref{subsec:compositionality}),
amounts to an equivalence relation (\Cref{subsec:eq}),
and adequate, \ie a weak asynchronous bisimilarity (\Cref{subsec:adequacy}).
The logical relation is shown to have
enough discriminatory power (soundness), while being maximally permissive (completeness)
via a biorthogonality argument (\Cref{thm:top-top-closure}).

\subsection{Logical relation}
\label{subsec:rslr}
Logical relations have been developed for session types
\myCite{PerezESOP2012,CairesESOP2013,PerezARTICLE2014,DeYoungFSCD2020,DerakhshanLICS2021,RochaCairesESOP2023,DerakhshanECOOP2024,VanDenHeuvelECOOP2024},
both in unary and binary forms, but \emph{without} considering general recursion.
We contribute a logical relation for \illst with support of \emph{general recursive types}.
In devising the resulting \emph{recursive session logical relation (\slr)}, we had to tackle the following challenges:

\begin{itemize}

\item \emph{Well-foundedness:} Recursive types mandate use of a measure like step-indexing or later modalities
\myCite{AppelMcAllesterTOPLAS2001,AhmedESOP2006,DreyerLICS2009} to keep the logical relation well-founded.
Our logical relation makes use of a more \emph{extensional} notion, an \emph{observation index},
which bounds the number of messages exchanged with the outside world along the free channels of a configuration.
The observation index facilitates proofs of various results shown in this section.

\item \emph{Non-termination:} In the presence of general recursive types, divergence is a possible outcome.
Our logical relation is \emph{progress-sensitive} \myCite{HeidinSabelfeldMartkoberdorf2011} with regard to divergence;
it equates a divergent program only with another diverging one.

\item \emph{Nondeterminism:} Although \illst are confluent, processes run \emph{concurrently}.
As a result, the logical relation has to account for the relatedness of messages that may not be sent in the same order due to nondeterministic scheduling.

\end{itemize}

\subsubsection{A logical relation for message-passing concurrency.}
\label{subsubsec:schema}
It may be helpful to pause a moment and ask what a logical relation for session-typed processes should amount to.
Logical relations for functional languages
relate terms by defining their equality at their type when evaluated to values.
This can be phrased in terms of two mutually recursive relations, a value interpretation
\mit{$\mc{V} \llbracket \tau \rrbracket$} and
a term interpretation
\mit{$\mc{E} \llbracket \tau \rrbracket$},
defined by structural induction on the type \mit{$\tau$}.
The former defines equality of values by inducting over a type, with values at ground type forming the base cases, and the latter steps the terms until they reach a value, demanding that these must be related by the value interpretation.
In imperative languages, logical relations are enriched with Kripke possible worlds
to model dynamically evolving shapes of storage
\myCite{PittsStarkHOOTS1998,AhmedPOPL2009,DreyerICFP2010,NeisJFP2011,HurDreyerPOPL2011,ThamsborgBirkedalICFP2011}.

To answer our question,
it is helpful to remind ourselves that logical relations are a means to prove observational equivalence
\myCite{MorrisPhD1968},
which equates two programs, if the same observations can be made about them.
In a pure functional setting, the observables are the values to which the terms evaluate.
In an imperative setting, the observables can additionally comprise the values of locations in the store.
The observables in a message-passing concurrent setting
are the messages that a configuration of processes exchanges with the outside world.

A logical relation for session-typed processes thus relates pairs of process configurations \mit{$(\mc{D}_1; \mc{D}_2)$},
such that 
\mit{$\Delta \Vdash \mc{D}_1:: x_\alpha{:}A$} and
\mit{$\Delta \Vdash \mc{D}_2:: x_\alpha{:}A$},
demanding that the same messages can be observed along the
free channels \mit{$\Delta$} and \mit{$x_\alpha{:}A$}.
\Cref{fig:rslr} shows the resulting logical relation for well-typed process configurations in \lang,
which is indexed with an \emph{intuitionistic linear sequent} \mit{$\Delta \Vdash A$},
comprising the free channels.
We find it convenient to distinguish a value interpretation \mit{$\mc{V}$}
from a term interpretation \mit{$\mc{E}$}.
As expected, the term interpretation steps the configurations until they reach a ``value''.
But what does it mean for a configuration to be a value in a message-passing concurrent setting?
It means that a configuration is \emph{ready to} \emph{send} or \emph{receive} along any of
the free channels in \mit{$\Delta \cup \{x_\alpha{:}A\}$}.
For any such channel,
the term interpretation invokes the value interpretation,
which demands that the exchanged messages be related,
and then invokes the term interpretation for the configurations resulting after the message exchange.
The mutually recursive relations are defined by clauses of the form
$$(\mc{D}_1; \mc{D}_2) \in \mc{V}\llbracket \Delta \Vdash A \rrbracket \;\;\;\; \m{iff} \;\;\;\;
    \dots (\mc{D}_1'; \mc{D}_2') \in \mc{E}\llbracket \Delta' \Vdash A' \rrbracket.$$
In a terminating setting, the relations are defined
\emph{multiset induction over the structure of types of the free channels} \myCite{DerakhshanLICS2021}
such that \mit{$\Delta' \Vdash A' < \Delta \Vdash A$}.
To support general recursive types, we will augment the relations with an observation index,
as discussed shortly,
yielding an inductive definition.

\begin{figure}
\centering
\begin{small}
\input{figs/logical-relation}
\end{small}
\caption{Recursive session logical relation (\slr) for \illst. (Clauses (1)-(10) are only defined for well-typed configurations, a condition we elided for concision.)}
\label{fig:rslr}
\end{figure}

The sequent \mit{$\Delta \Vdash A$} singles out the typing $A$ of the providing free channel of the configurations $\mc{D}_1$ and $\mc{D}_2$.
If the providing free channel is not observable, we write \mit{$\_$}.
We use the metavariable \mit{$K$} to stand for either \mit{$x_\alpha{:}A$} or simply \mit{$\_$}.
The ability to mark some free channels as unobservable becomes important when we apply the logical relation to noninterference in \Cref{sec:psni}.

The transition from the value to the term interpretation amounts to an \emph{observation} of a message exchange along a free channel in \mit{$\Delta \Vdash K$}.
The exchange will advance the protocol state of the concerned process,
and as a result transition the sequent to \mit{$\Delta' \Vdash K'$},
describing the post-state of the configurations \mit{$(\mc{D}_1'; \mc{D}_2')$} after the message exchange.
The logical relation comprises two value interpretation clauses for each connective,
one for the communications occurring on the \emph{right} along~\mit{$K$}, and
one for those occurring on the \emph{left} along~\mit{$\Delta$}.

\subsubsection{Assume and assert---a pas de deux.}
\label{subsubsec:assume-assert}

A key characteristic of logical relations is \emph{extensionality}.
In a functional context this means that relatedness has to be shown for values of types in positive positions,
but values of types in negative positions can be assumed to be related.
This idea translates equally to a message-passing concurrent setting,
embodied by the logical relation shown in \Cref{fig:rslr} as: 

\begin{itemize}
\item \emph{Positive} types: \emph{assert} sending of related messages when communicating on the \emph{right};
\emph{assume} receipt of related messages when communicating on the \emph{left}.

\item \emph{Negative} types: \emph{assume} receipt of related messages when communicating on the \emph{right};
\emph{assert} sending of related messages when communicating on the \emph{left}.
\end{itemize}

At base types (\mit{$1$}, \mit{$\intchoicesymb$}, and \mit{$\extchoicesymb$}),
relatedness of messages means that the same messages are being exchanged.
For example, clause (2) in \Cref{fig:rslr} for \mit{$\oplus$}-right asserts
existence of messages \mit{$\mathbf{msg}(y_\alpha.k_1)$} and \mit{$\mathbf{msg}(y_\alpha.k_2)$}
in \mit{$\mc{D}_1$} and \mit{$\mc{D}_2$}, \resp such that \mit{$k_1=k_2$}.
Conversely, clause (3) in \Cref{fig:rslr} for \mit{$\&$}-right assumes
receipt of messages \mit{$\mathbf{msg}(y_\alpha.k_1)$} and \mit{$\mathbf{msg}(y_\alpha.k_2)$}
to be added to \mit{$\mc{D}_1$} and \mit{$\mc{D}_2$} in the post-states, \resp
such that \mit{$k_1=k_2$}, for arbitrary \mit{$k_1$} and \mit{$k_2$}.

For higher-order types (\mit{$\chanoutsymb$} and \mit{$\chaninsymb$})
relatedness means that future observations to be made along
the exchanged channels must be related.
For example, clause (10) in \Cref{fig:rslr} for \mit{$\multimap$}-left asserts
existence of a message \mit{$\mathbf{msg}( \mb{send}x_\beta\,y_\alpha)$}
and of subtrees \mit{$\mathcal{T}_1$} and \mit{$\mathcal{T}_2$}
in \mit{$\mc{D}_1$} and \mit{$\mc{D}_2$}, \respb.
The subtrees \mit{$\mathcal{T}_1$} and \mit{$\mathcal{T}_2$} are rooted at the sent channel $x_\beta$
and will be transferred (and thus lost) together with the sent channel.
The clause comprises two invocations of the term relation,
\mit{$({\mathcal{T}_1};{\mathcal{T}_2}) \in  \mc{E}\llbracket \Delta' \Vdash x_\beta{:}A \rrbracket$},
asserting that future observations to be made along the sent channel \mit{$x_\beta$} are related, and
\mit{$({\mc{D}''_1}; {\mc{D}''_2} ) \in  \mc{E}\llbracket \Delta'', y_{\alpha+1}{:}B \Vdash K \rrbracket$},
asserting that the continuations \mit{$\mc{D}''_1$} and \mit{${\mc{D}''_2}$} are related.
\footnote{We will momentarily explain the superscript and subscripts of these invocations.}
The channel endpoint \mit{$x_\beta$} is existentially quantified and occurs either bound or free in both \mit{$\mc{D}_1$} and \mit{$\mc{D}_2$}.
Conversely, clause (5) in \Cref{fig:rslr} for \mit{$\multimap$}-right assumes
receipt of a message \mit{$\mathbf{msg}( \mb{send}x_\beta\,y_\alpha)$}
that carries a universally quantified free channel \mit{$x_\beta$},
different from the free channels available in the pre-state
(\mit{$\forall x_\beta \not\in \mathit{dom}(\Delta,y_\alpha{:}A \multimap B)$}).
The sequent is enlarged with the received channel \mit{$x_\beta$}
for the invocation of the term relation
\mit{$( {\mc{D}_1\mathbf{msg}( \mb{send}\,x_\beta\,y_\alpha)};
{\mc{D}_2\mathbf{msg}(\mb{send}\,x_\beta\,y_\alpha)} ) \in \mc{E}\llbracket \Delta, x_\beta{:}A\Vdash y_{\alpha+1}{:}B\rrbracket$}.

The \emph{duality} of \illst,
established by ``cutting'' the left and right rule of a connective,
translates equally to our logical relation,
allowing us to compose related configurations in parallel,
such that the assumptions made by one pair of configurations are discharged by the other.
We prove this result in \Cref{subsec:compositionality} and exploit it connect our development to biorthogonality in \Cref{subsec:biorthogonality}.

\subsubsection{Syntactic well-typedness.}

The value and term interpretations of our logical relation relate pairs of process configurations only if both are syntactically well-typed.  
The predicate \mit{$(\mathcal{D}_1; \mathcal{D}_2) \in \m{Tree}(\Delta \Vdash K)$}
stipulates this requirement and
is defined as \mit{$\mathcal{D}_1\in \m{Tree}(\Delta \Vdash K)$} and
\mit{$\mathcal{D}_2\in \m{Tree}(\Delta \Vdash K)$},
with \mit{$\mathcal{D}_i\in \m{Tree}(\Delta \Vdash K)$} defined as
\mit{$\Delta  \Vdash \mathcal{D}_i :: K$}
\ifappendix (see \Cref{apx:def:well-typed-config} in the Appendix).\xspace \fi
For concision, we elide appeals to the \tm{\m{Tree}()} predicate in the value interpretation (clauses (1)-(10) in~\Cref{fig:rslr}).
Syntactic typing ensures that configurations form \emph{rooted trees},
a property our proof of semantic cut (\Cref{lem:compositionality} in \Cref{subsec:compositionality}) relies upon.

\subsubsection{Well-foundedness.}
\label{subsubsec:well-foundedness}

In the presence of general recursive types,
it is no longer sound to define the logical relation by multiset induction over the structure of types in its sequent
\mit{$\Delta \Vdash K$}.
To restore well-foundedness a measure like step-indexing \myCite{AppelMcAllesterTOPLAS2001,AhmedESOP2006} can be employed.
Step-indexing is typically tied to the number of type unfolding or computation steps by which one of the two programs is bound.
A more natural and symmetric measure for our setting is the number of \emph{observations}
that can be made along the free channels \mit{$\Delta \Vdash K$} of the logical relation.
We thus bound the value and term interpretation of our logical relation by this number and
attach the resulting \emph{observation index} as the superscript \mit{$m$}
to the sequent \mit{$\Delta \Vdash K$},
yielding clauses of the form
$$(\mc{D}_1; \mc{D}_2) \in \mc{V}\llbracket \Delta \Vdash K \rrbracket^{m+1} \,\,\, \m{iff} \,\,\,
    \dots (\mc{D}_1'; \mc{D}_2') \in \mc{E}\llbracket \Delta' \Vdash K' \rrbracket^{m}$$
The observation index gets decremented whenever the value interpretation ``observes'' a message exchange.
If \mit{$m = 0$}, well-typed configurations are trivially related,
as expressed by clause (12) in \Cref{fig:rslr}.

\subsubsection{Non-termination and nondeterminism.}
\label{subsubsec:nondeterminism-non-termination}

It is now time to take a closer look at the definition of the term interpretation of our logical relation.
Its definition is challenged by the possibility of divergence and nondeterminism of scheduling.
For the former, a progress-sensitive statement of equivalence demands that the relation
relates a divergent program only with another diverging one.
For the latter, the term interpretation has to account for the possibility that
two configurations may not \emph{simultaneously} be ready to send or receive
along a free channel in \mit{$\Delta \Vdash K$}.
To address these challenges, the term interpretation is phrased as follows (we are repeating clause (11) in \Cref{fig:rslr}):
$$\begin{array}{@{}lcl@{}}
({\mc{D}_1}; {\mc{D}_2})\in\mc{E}\llbracket \Delta \Vdash K \rrbracket^{m+1} \,\,\,\m{iff} &\\[2pt]
({\mc{D}_1}; {\mc{D}_2})\in \m{Tree}(\Delta \Vdash K)\,\m{and} \, \forall\, \Upsilon_1, \Theta_1, \mc{D}'_1.
\m{if}\; {\mc{D}_1}\;\mapsto^{{*}_{\Upsilon_1; \Theta_1}}\; {\mc{D}'_1}\, \m{then}\\[2pt]
\,\exists \Upsilon_2, \mc{D}'_2 \m{\,such \, that\,}\,{\mc{D}_2}\mapsto^{{*}_{ \Upsilon_2}}{\mc{D}'_2}\,\m{and}\, \Upsilon_1 \subseteq \Upsilon_2\; \m{and} \\[2pt]
 \forall\, y_\alpha \in \mb{Out}(\Delta \Vdash K).\, \mb{if}\, y_\alpha \in \Upsilon_1.\, \mb{then}\, ({\mc{D}'_1}; {\mc{D}'_2})\in \mc{V}\llbracket  \Delta \Vdash K \rrbracket_{\cdot;y_\alpha}^{m+1}\,\\[2pt]
\m{and} \, \forall\, y_\alpha \in \mb{In}(\Delta \Vdash K). \mb{if}\, y_\alpha \in \Theta_1.\, \mb{then}\, ({\mc{D}'_1}; {\mc{D}'_2})\in \mc{V}\llbracket  \Delta \Vdash K \rrbracket_{y_\alpha;\cdot}^{m+1}
\end{array}$$
The term interpretation uses the transition \mit{${\mc{D}_1}\mapsto^{{*}_{\Upsilon_1; \Theta_1}}{\mc{D}'_1}$},
amounting to the iterated application of the rewriting rules defined in \Cref{fig:dynamics}.
The star expresses that zero to multiple steps can be taken.
The superscript \mit{$\Upsilon_1$} denotes the set of channels in \mit{$\Delta \Vdash K$}
for which there exist messages in \mit{$\mc{D}'_1$} to be sent along these channels.
The superscript \mit{$\Theta_1$} denotes set of channels in \mit{$\Delta \Vdash K$}
for which there exist processes in \mit{$\mc{D}'_1$} waiting to receive along these channels.

To ensure progress-sensitivity, the term interpretation asserts,
that, whenever \mit{$\mc{D}_1$} can step, so can \mit{$\mc{D}_2$}.
This correspondence is expressed by the condition \mit{$\Upsilon_1 \subseteq \Upsilon_2$},
ensuring that the messages ready to be sent to the outside world in \mit{$\mc{D}'_2$}
are at least the ones ready to be sent in \mit{$\mc{D}'_1$}.
An analogous condition for incoming messages is omitted due to asynchrony of the dynamics
where the receipt of a message is not observable.
The existentially quantified \mit{$\mc{D}'_2$} is justified by nondeterminism of scheduling,
allowing the term interpretation to choose an appropriate way of stepping \mit{$\mc{D}_2$} to catch up with \mit{$\mc{D}'_1$}.

The set $\mb{Out}(\Delta \Vdash K)$ includes all channels whose types in $\Delta, K$ indicate a send. This includes channels in $\Upsilon_1$ for which the messages are ready to be sent. Similarly, the set  $\mb{In}(\Delta \Vdash K)$ includes all channels whose types in $\Delta, K$ indicate a receive, including those in $\Theta_1$ for which the processes are ready to receive.

To resolve the issue of simultaneity, the term interpretation makes use of
two \emph{focus} channels ranging over the sets \mit{$\Upsilon_1$} and \mit{$\Theta_1$}.
These are added as a subscript \mit{$\_; \_$} to the value interpretation,
where \mit{$\cdot; y_\alpha$} indicates that \mit{$y_\alpha \in \Upsilon_1$} and
\mit{$y_\alpha; \cdot$} that \mit{$y_\alpha \in \Theta_1$}.
The term interpretation invokes the value interpretation
for all the channels in \mit{$\Upsilon_1$} and \mit{$\Theta_1$},
thus ensuring that any messages ready to be sent in \mit{$\mc{D}'_1$} and \mit{$\mc{D}'_2$}
and any processes waiting to receive in \mit{$\mc{D}'_1$} will be ``observed'' by the value interpretation.
Non-productive configurations are trivially accommodated by the term interpretation,
by allowing the sets \mit{$\Upsilon_1$} and \mit{$\Theta_1$} to be empty.

\subsubsection{Example.}

We revisit the $\m{Verifier}$ example (\Cref{fig:simple-verfier}) to showcase the design of our logical relation.
Consider the configuration~\mit{$\mc{D}=\mathbf{proc}(x_0,[x_0/x]\m{Verifier})$}, in which the channel $x_0$ instantiates the channel variable $x$ in the $\m{Verifier}$ process. 
We outline the first few steps required to establish that $\mc{D}$ is self-related at observation index $3$, i.e., \mit{$({\mc{D}}; {\mc{D}}) \in \mc{E}\llbracket \cdot \Vdash x_0 {:} \m{ver} \rrbracket^{3}$}.
By~\Cref{fig:dynamics}, we observe that $\mc{D}$ cannot take any internal steps, that no outgoing message exists in the configuration, and that the only channel on which $\mc{D}$ is ready to receive is $x_0$. Hence, the only possible transition for~$\mc{D}$ is ${\mc{D}}\;\mapsto^{{0}_{\emptyset; \{x_0\}}}\; {\mc{D}}$. 
It follows that the only value relation invoked by this term interpretation is \mit{$({\mc{D}}; {\mc{D}}) \in \mc{V}\llbracket \cdot \Vdash x_0 {:} \m{ver} \rrbracket^{3}_{x_0;\cdot}$}, which we can rewrite as
\mit{$({\mc{D}}; {\mc{D}}) \in \mc{V}\llbracket \cdot \Vdash x_0 {:} \m{pin} \multimap \oplus \{\mathit{succ}{:}\m{pin} \otimes \m{ver}, \mathit{fail}{:} \m{pin} \otimes \m{ver}\} \rrbracket^{3}_{x_0;\cdot}$} by unfolding the type $\m{ver}$; the latter now matches clause (5) in~\Cref{fig:rslr}.
To establish this value relation, it remains to show that for an arbitrary $u_\beta \neq x_\alpha$, we can derive \mit{$({\mc{D}\mathbf{msg}(\mb{send}\,u_\beta\,x_0)}; {\mc{D}\mathbf{msg}(\mb{send}\,u_\beta\,x_0)}) \in \mc{E}\llbracket u_\beta{:}\m{pin} \Vdash x_0 {:} \oplus \{\mathit{succ}{:}\m{pin} \otimes \m{ver}, \mathit{fail}{:} \m{pin} \otimes \m{ver}\} \rrbracket^{2}$}, which again invokes the term interpretation via clause (11). 
The remaining steps follow similarly.

A reader may wonder why the logical relation seems oblivious to recursive type unfoldings.
This is because the relation is indexed by the number of external observations and is therefore invariant under such unfoldings.
In particular, for the example above, 
\mit{$({\mc{D}}; {\mc{D}}) \in \mc{E}\llbracket \Delta \Vdash x_0 {:} \m{ver} \rrbracket^{3}$}
and \mit{$ ({\mc{D}}; {\mc{D}}) \in \mc{E}\llbracket \Delta \Vdash x_0{:} \m{pin} \multimap \oplus \{\mathit{succ}{:}\m{pin} \otimes \m{ver}, \mathit{fail}{:} \m{pin} \otimes \m{ver}\} \rrbracket^{3}$},
are equivalent:
the value interpretation with focus on channel \mit{$x_0$} is only triggered when the process is waiting to receive a message along \mit{$x_0$}.
After consumption of the message, the recursive invocation of the term interpretation in clause (5) will then happen at index \mit{$2$}.

\subsection{Semantic cut: closure under parallel composition}
\label{subsec:compositionality}

The use of an observation index rather than a step index by the \slr is fundamental
in proving various metatheoretic properties in the remainder of this section.
It allows us to keep the observation index \emph{constant} in the term interpretation,
while stepping the configurations internally,
permitting disagreement in the number of internal messages exchanged.
As a result, internal steps cannot depend on the observation index,
facilitating proof of the following lemma:

\begin{lemma} [Moving existential over universal quantifier]\label{lem:moving-existential}\strut

\noindent If 
\[\begin{array}{l}
    {\color{red}\forall m.}\, \forall\, \Upsilon_1, \Theta_1, \mathcal{D}'_1.\, \, \m{if}\,  \mc{D}_1 \mapsto^{*_{ \Upsilon_1; \Theta_1}} \mc{D}'_1,\, \m{then}\\
    \, \exists \Upsilon_2, \mc{D}'_2 \,\m{such\, that} \, \,\mc{D}_2\mapsto^{{*}_{ \Upsilon_2}} \mc{D}'_2\, \m{and}\, \Upsilon_1 \subseteq \Upsilon_2 \, \m{and}\\ 
     \; \; \forall\, x_\alpha \in \mb{Out}(\Delta \Vdash K).\, \mb{if}\, x_\alpha \in \Upsilon_1.\, \mb{then}\\
     \,\,\,\,\,\qquad\qquad\qquad \quad (\mc{D}'_1; \mc{D}'_2)\in \mc{V}\llbracket  \Delta \Vdash K \rrbracket_{\cdot;x_\alpha}^{m+1}\,\m{and}\\ 
     \; \; \forall\, x_\alpha \in \mb{In}(\Delta \Vdash K).\mb{if}\, x_\alpha \in \Theta_1.\,\mb{then}\\
     \,\,\,\,\,\qquad\qquad\qquad \quad (\mc{D}'_1; \mc{D}'_2)\in \mc{V}\llbracket  \Delta \Vdash K \rrbracket_{x_\alpha;\cdot}^{m+1},
    \end{array}\]
then 
\[\begin{array}{l}
    \forall\, \Upsilon_1, \Theta_1, \mathcal{D}'_1. \m{if}\,  \mc{D}_1 \mapsto^{*_{ \Upsilon_1; \Theta_1}} \mc{D}'_1,\, \m{then}\\
    \, \exists \Upsilon_2,\mc{D}'_2 \,\m{such\, that} \, \,\mc{D}_2\mapsto^{{*}_{ \Upsilon_2}} \mc{D}'_2\, \m{and}\, \Upsilon_1 \subseteq \Upsilon_2 \, \m{and}\\ 
    \; \; \forall\, x_\alpha \in \mb{Out}(\Delta \Vdash K).\, \mb{if}\, x_\alpha \in \Upsilon_1.\, \mb{then}\\
    \,\,\,\,\,\qquad\qquad\qquad \quad {\color{red}\forall \,m.}\, (\mc{D}'_1; \mc{D}'_2)\in \mc{V}\llbracket  \Delta \Vdash K \rrbracket_{\cdot;x_\alpha}^{m+1}\,\m{and}\\ 
    \; \; \forall\, x_\alpha \in \mb{In}(\Delta \Vdash K). \mb{if}\, x_\alpha \in \Theta_1.\,\mb{then}\\
    \,\,\,\,\,\qquad\qquad\qquad \quad{\color{red}\forall m.}\, (\mc{D}'_1; \mc{D}'_2)\in \mc{V}\llbracket  \Delta \Vdash K \rrbracket_{x_\alpha;\cdot}^{m+1}.
\end{array}\]
\end{lemma}
\begin{proof}
\ifappendix
See proof of \Cref{apx:lem:moving-existential} in the Appendix.
\fi
The proof relies on confluence \ifappendix (\Cref{apx:lem:confluence} in the Appendix) \else \fi
and backwards \ifappendix closure (\Cref{apx:lem:backsecondrun} in the Appendix) \fi
\end{proof}

With this lemma in hand, we can now prove that the logical relation is closed under parallel composition.

\begin{lemma}[Semantic cut]\label{lem:compositionality}
$\forall m.\, (\mc{D}_1;\mc{D}_2) \in \mc{E}\llbracket \Delta, u_\alpha{:}T \Vdash K \rrbracket^m$ iff for all $\mc{T}_1$ and $\mc{T}_2$ if $\forall m.\, (\mc{T}_1;\mc{T}_2) \in \mc{E}\llbracket \Delta'\Vdash  u_\alpha{:}T \rrbracket^m$ then $\forall k.\, (\mc{T}_1 \mc{D}_1;\mc{T}_2\mc{D}_2) \in \mc{E}\llbracket \Delta', \Delta\Vdash K \rrbracket^k$.
\end{lemma}
\begin{proof}
\ifappendix
See proof of \Cref{apx:lem:compositionality} in the Appendix.
\fi
The proof relies on \Cref{lem:moving-existential}.
\end{proof}

The semantic cut lemma is analogous to the syntactic cut-elimination theorem, with $T$ serving as the analog of the cut formula.
Consider the left-to-right direction. 
For internal steps in \(\mathcal{T}_1\mathcal{D}_1\) invoked by the term interpretation that do not involve the channel \(u_\alpha : T\), we can perform corresponding internal steps in \(\mathcal{T}_1\) or \(\mathcal{D}_1\). Using the fact that \(\mathcal{T}_1\) and \(\mathcal{D}_1\) are related to \(\mathcal{T}_2\) and \(\mathcal{D}_2\), resp., we can construct the next step(s) of \(\mathcal{T}_2\mathcal{D}_2\) as required by the term interpretation.
For internal steps in \(\mathcal{T}_1\mathcal{D}_1\) along the channel \(u_\alpha : T\), we use the fact that \(\mathcal{T}_1\) is ready to send a message along \(u_\alpha\) and \(\mathcal{D}_1\) is ready to receive along \(u_\alpha\) (or vice versa). 
Here, we take advantage of the duality in the value interpretation for the provider and client of each connective to show that the assertions made by the sender are sufficient to establish the assumptions of the receiver.
With this, we invoke the definition of value interpretations to transition to a new term interpretation involving subformula(s) of type \(T\) in the sequent and apply the inductive hypothesis.
This step is analogous to internal cut reductions in the cut-elimination proof.
For any communication in \(\mathcal{T}_1\mathcal{D}_1\) along its external channel, we use the fact that either \(\mathcal{T}_1\) or \(\mathcal{D}_1\) can perform the same communication. 
Using their respective value interpretations and the inductive hypothesis, we establish the value interpretation for \(\mathcal{T}_1\mathcal{D}_1\) and \(\mathcal{T}_2\mathcal{D}_2\).
This step corresponds to external cut reductions in the cut-elimination proof.

\subsection{Logical equivalence}
\label{subsec:eq}

To prove that the \slr induces an equivalence relation, we introduce 
the notation \mit{$\Delta \Vdash \mc{D}_1 :: K \equiv \Delta \Vdash \mc{D}_2 :: K $},
defined below, where
\mit{$(\mc{C}_1; \mc{C}_2)\in \m{Forest}(\Delta' \Vdash \Delta)$} stands for
\mit{$\Delta'  \Vdash \mc{C}_1 :: \Delta$} and \mit{$\Delta'  \Vdash \mc{C}_2 :: \Delta$}:

\begin{definition}[Logical equivalence]\label{def:eq}  \strut
\begin{itemize}
  \item  We define the relation $\Delta \Vdash \mc{D}_1 :: K \equiv \Delta \Vdash \mc{D}_2 :: K $  as $(\mc{D}_1; \mc{D}_2)\in \m{Tree}(\Delta \Vdash K)$ and $\forall m.\, (\mc{D}_1; \mc{D}_2)\in \mc{E}\llbracket \Delta \Vdash K \rrbracket^m$ and $\forall m. (\mc{D}_2; \mc{D}_1) \in \mc{E}\llbracket \Delta \Vdash K \rrbracket^m$.  
  \item  We define the relation $ \Delta \dashv \mc{C}_1[\,\,] \mc{F}_1 \dashv K \equiv \Delta \dashv \mc{C}_2 [\,\,]\mc{F}_2 \dashv K$  as (i) \,$\exists K'$ such that $K \Vdash \mc{F}_1 ::K' \equiv K \Vdash \mc{F}_2:: K'$ and (ii) \,$\exists \Delta'$ such that $(\mc{C}_1; \mc{C}_2)\in \m{Forest}(\Delta' \Vdash \Delta)$ and $\forall \mc{T}_1\in \mc{C}_1,\,\m{and}\, \forall \mc{T}_2 \in \mc{C}_2$ with $(\mc{T}_1;\mc{T}_2)\in \m{Tree}(\Delta'_1\Vdash K'')$, we have 
  $\Delta'_1 \Vdash \mc{T}_1 ::K'' \equiv \Delta'_1 \Vdash \mc{T}_2:: K''$.
\end{itemize}
\end{definition}

Next, we show that the relation \mit{$\equiv$} is reflexive, symmetric, and transitive.

\begin{lemma}[Reflexivity]\label{lem:reflexivity-eq}
  For all configurations $\Delta \Vdash \mathcal{D}:: x_\alpha {:}T$, we have 
  {
  \((\Delta \Vdash\mathcal{D}::x_\alpha {:}T)  \equiv (\Delta \Vdash \mathcal{D}:: x_\alpha {:}T).\)
  }
\end{lemma}
\begin{proof}
\ifappendix
See proof of \Cref{apx:lem:reflexivity-eq} in the Appendix.
\fi
\end{proof}
  
\begin{lemma}[Symmetry]\label{lem:symmetry-eq}
  For all configurations $\mc{D}_1$ and $\mc{D}_2$, we have 
  {\small
  $(\Delta \Vdash \mathcal{D}_1:: x_\alpha {:}T)  \equiv (\Delta \Vdash \mathcal{D}_2:: x_\alpha {:}T),$
  iff 
  $(\Delta \Vdash \mathcal{D}_2:: x_\alpha {:}T)  \equiv (\Delta \Vdash \mathcal{D}_1:: x_\alpha {:}T),$
  }
\end{lemma}
\begin{proof}
  The proof is straightforward by~\Cref{def:eq}.
\end{proof}
  
  \begin{lemma}[Transitivity]\label{lem:transitivity-eq}
  For all  configurations $\mc{D}_1$, $\mc{D}_2$, and $\mc{D}_3$, we have 
  {\small
   \[
  \begin{array}{l}
    \m{if}\,(\Delta \Vdash \mathcal{D}_1:: x_\alpha {:}T)  \equiv (\Delta \Vdash \mathcal{D}_2:: x_\alpha {:}T),\, \m{and}\\
    \,\,\,\,\, (\Delta \Vdash \mathcal{D}_2:: x_\alpha {:}T)  \equiv (\Delta \Vdash \mathcal{D}_3:: x_\alpha {:}T)\\
    \m{then}\,(\Delta \Vdash \mathcal{D}_1:: x_\alpha {:}T)  \equiv (\Delta \Vdash \mathcal{D}_3:: x_\alpha {:}T).
  \end{array}  
  \]
  }
  \end{lemma}
\begin{proof}
\ifappendix
See proof of \Cref{apx:lem:transitivity-eq} in the Appendix.
\fi
\end{proof}

\subsection{Adequacy}
\label{subsec:adequacy}

To show that the \slr is adequate, we prove that configurations related by the logical relation are bisimilar and vice versa.
To facilitate this proof, we first give the definition of weak \emph{asynchronous bisimilarity} \myCite{SangiorgiWalkerBook2001}
(\Cref{def:asynchronousbisim}).
The definition relies on a standard labeled transition system displayed in \Cref{fig:lts}, which rephrases our dynamics in~\Cref{fig:dynamics}.
We define \emph{weak} transitions as

\begin{enumerate}
    \item[(1)] $\xRightarrow{}$ is the reflexive and transitive closure of $\xrightarrow{\tau}$,
    \item[(2)] $\xRightarrow{\alpha}$ is $\xRightarrow{} \xrightarrow{\alpha}$.
  \end{enumerate}

\begin{definition}[Weak asynchronous bisimilarity]\label{def:asynchronousbisim}
Asynchronous bisimilarity, written {$\mc{D}_1 \approx_{\m{a}} \mc{D}_2$}, is the largest symmetric relation such that whenever {$\mc{D}_1 \approx_{\m{a}} \mc{D}_2$}, we have
\begin{itemize}[labelindent=0em,labelsep=1 pt,leftmargin=*]
  \item{$(\tau-\mathtt{step})$} if  { $\mc{D}_1 \xrightarrow{\tau} \mc{D}'_1$} then {$\exists \mc{D}'_2. \, \mc{D}_2 \xRightarrow{\tau} \mc{D}'_2$} and  {$\mc{D}'_1 \approx_{\m{a}} \mc{D}'_2,$}
  \item{$(\m{output})$} if $\mc{D}_1 \xrightarrow{\overline{x_\alpha}\,q} \mc{D}'_1$ then $\exists \mc{D}'_2. \, \mc{D}_2 \xRightarrow{\overline{x_\alpha}\,q} \mc{D}'_2$ and {$\mc{D}'_1 {\approx_{\m{a}}} \mc{D}'_2.$}
  \item{$(\m{left\, input})$} for all $q\not \in \m{fn}(\mc{D}_1)$, if $\mc{D}_1 \xrightarrow{{\mathbf{L}\, x_\alpha\,q}} \mc{D}'_1$ then  $\exists \mc{D}'_2. \, \mc{D}_2 \xRightarrow{\tau} \mc{D}'_2$ and {$\mc{D}'_1 \approx_{\m{a}} \mathbf{msg}(x_\alpha.q) \mc{D}'_2,$}
\item{$(\m{right\, input})$} for all $q\not \in \m{fn}(\mc{D}_1)$, if $\mc{D}_1 \xrightarrow{{\mathbf{R}\, x_\alpha\,q}} \mc{D}'_1$ then  $\exists \mc{D}'_2. \, \mc{D}_2 \xRightarrow{\tau} \mc{D}'_2$ and   {$\mc{D}'_1 \approx_{\m{a}} \mc{D}'_2\mathbf{msg}(x_\alpha.q).$}
\end{itemize}
where $\mb{msg}(x_\alpha.q)$ is defined as $\mb{msg}(\mathbf{close}\, x_\alpha)$ if $q=\mathbf{close}$, $\mb{msg}(x_\alpha.k)$ if $q=k$, and $\mb{msg}(\mathbf{send}\,z_\delta\, x_\alpha)$ if $q=z_\delta$. Here, $\m{fn}(\mc{D}_1)$ is the set of free channels in the configuration $\mc{D}_1$.
\end{definition}

\begin{figure}
\centering
\begin{small}
\input{figs/lts.tex}
\end{small}
\caption{Labeled transition system for \lang.}
\label{fig:lts}
\end{figure}

Our adequacy theorem (\Cref{lem:soundness-completeness}) shows that \slrs are sound and complete with regard to asynchronous bisimilarity.

\begin{lemma}[Adequacy]\label{lem:soundness-completeness}
For all $(\mc{D}_1; \mc{D}_2) \in \m{Tree}(\Delta \Vdash K)$, we have $\mc{D}_1 \approx_{\m{a}} \mc{D}_2$ iff $\forall m. (\mc{D}_1; \mc{D}_2) \in \mc{E}\llbracket \Delta \Vdash K\rrbracket^m$ and $\forall m. (\mc{D}_2; \mc{D}_1) \in \mc{E}\llbracket \Delta \Vdash K\rrbracket^m$.
\end{lemma}
\begin{proof}
\ifappendix
The result is a corollary of \Cref{apx:lem:soundness-completeness} in the Appendix.
\fi
\end{proof}

\subsection{Biorthogonality}
\label{subsec:biorthogonality}

Our main result, biorthogonal closure of \mit{$\equiv$} (\Cref{thm:top-top-closure}),
relies upon various definitions, including an \emph{orthogonality} operation $\top$ \myCite{PittsMSCS2000,MelliesVouillonLICS2005,BentonICFP2009} (\Cref{def:top}), which we give next.
We follow the development of Pitts \myCite{PittsMSCS2000}.

We first define the program relations \mit{$\mathit{PRel}(\Delta \Vdash K)$} and
environment relations \mit{$\mathit{ERel}(\Delta \Vdash K)$} (\Cref{def:program-env-relations}),
reliant upon a session-typed environment (\Cref{def:session-typed-env}):

\begin{definition}[Session-typed environment]\label{def:session-typed-env}
    A session-typed environment with the sequent $\Delta \Vdash K$, is of the form $\mc{C}[\,\,]\mc{F}$, such that for some $\Delta'$ and $K'$, we have $\Delta' \Vdash \mc{C}:: \Delta$ and $K \Vdash \mc{F}:: K'$.
\end{definition}

\begin{definition}[Program- and environment- relations]\label{def:program-env-relations}
A session program-relation is a binary relation between session-typed programs, i.e., open configurations of the form $\Delta \Vdash \mc{D}:: K$. Given the sequent $\Delta \Vdash K$, we write $\mathit{PRel}(\Delta \Vdash K)$ for the set of all program relations that relate programs $\Delta \Vdash \mc{D}:: K$.

A session environment-relation is a binary relation between session-typed environments. Given the sequent $\Delta \Vdash K$, we write $\mathit{ERel}(\Delta \Vdash K)$ for the set of all environment relations that relate environments $\mc{C}[\,\,]\mc{F}$ with the sequent $\Delta \Vdash K$.
\end{definition}

We can now define the orthogonality operation \mit{$\top$} \myCite{PittsMSCS2000,MelliesVouillonLICS2005}.

\begin{definition}[The $(\_)^\top$ operation]\label{def:top}
    Given sequent $\Delta \Vdash K$ and program relation $r \in \mathit{PRel}(\Delta \Vdash K)$, we define $r^\top \in \mathit{ERel}(\Delta \Vdash K)$ by
    \[\begin{array}{ll}
      (\mc{C}_1 [\,\,]\mc{F}_1, \mc{C}_2 [\,\,] \mc{F}_2) \in r^\top  \,\, \m{iff}\\ \forall (\mc{D}_1, \mc{D}_2) \in r. (\mc{C}_1\mc{D}_1\mc{F}_1 \approx_{\m{a}} \mc{C}_2\mc{D}_2\mc{F}_2),
    \end{array}\]
    and given $s\, {\in} \mathit{ERel}(\Delta \Vdash K)$, we define $s^\top {\in} \mathit{PRel}(\Delta \Vdash K)$ by
    \[\begin{array}{l}
(\mc{D}_1, \mc{D}_2) \in s^\top  \,\, \m{iff}\\ \forall (\mc{C}_1[\,\,]\mc{F}_1, \mc{C}_2[\,\,]\mc{F}_2) \in s. (\mc{C}_1\mc{D}_1\mc{F}_1 \approx_{\m{a}} \mc{C}_2\mc{D}_2\mc{F}_2).
    \end{array}\]

   \noindent By definition, the \tm{\top} operation is inflationary and idempotent \myCite{PittsMSCS2000}.
\end{definition}

Finally, we state our main result that the logical equivalence relation \mit{$\equiv$} induced by our \slr (\Cref{def:eq})
is sound and complete with regard to closure of weak bisimilarity under parallel composition.

\begin{theorem}[$\top\top$-closure]\label{thm:top-top-closure}
Consider $(\mc{D}_1, \mc{D}_2) {\in} \m{Tree}(\Delta \Vdash K)$, we have
    {\small\[\begin{array}{l}
    (\Delta \Vdash \mc{D}_1:: K) \equiv (\Delta \Vdash \mc{D}_2:: K)\,\,
    \mathbf{iff} \\[4pt]
    \forall \mathcal{C}_1, \mathcal{C}_2, \mathcal{F}_1, \mathcal{F}_2.\,\m{if}\,
    (\Delta \dashv \mc{C}_1[\,\,]\mc{F}_1\dashv  K) \equiv (\Delta \dashv \mc{C}_2[\,\,]\mc{F}_2 \dashv K)\\
    \;\;\;\quad\qquad\qquad\quad \m{then}
    \, {{\mathcal{C}_1}{\mathcal{D}_1} {\mathcal{F}_1}}\, \approx_{\m{a}} \,\mathcal{C}_2{\mathcal{D}_2} {\mathcal{F}_2}.
  \end{array}\]}
\end{theorem}
\begin{proof}
\ifappendix
See proof of \Cref{apx:thm:biorthogonal} in the Appendix.
\fi
\end{proof}

The main purpose of \Cref{thm:top-top-closure} is to validate the logical equivalence induced by our logical relation, ensuring that it has enough discriminatory power (soundness) while being maximally permissive (completeness). To do that, \Cref{thm:top-top-closure} uses biorthogonal closure. 
As detailed in \myCite{PittsMSCS2000}, to prove that a relation $r$ is $\top\top$-closed, i.e., \mit{$r = r^{\top \top}$}, it is enough to show that for some witness relation \mit{$s$}, we have \mit{$r =s^\top$}.
\Cref{thm:top-top-closure} instantiates the program-relation \mit{$r$} and environment-relation \mit{$s$} (\Cref{def:top}) with the logical equivalences (\Cref{def:eq}) induced by our logical relation:
it defines $(\mc{D}_1, \mc{D}_2) \in r$ as the program relation {\small$
    (\Delta \Vdash \mc{D}_1:: K) \equiv (\Delta \Vdash \mc{D}_2:: K)$},
and shows that if we define the witness relation \mit{$(\mc{C}_1[\,\,]\mc{F}_1, \mc{C}_2[\,\,]\mc{F}_2)\in s$} as the environment relation \mit{$(\Delta \dashv \mc{C}_1[\,\,]\mc{F}_1\dashv  K) \equiv (\Delta \dashv \mc{C}_2[\,\,]\mc{F}_2 \dashv K)$}, then we can prove \mit{$r=s^\top$}. 
In particular, the theorem proves that related programs cannot be discriminated by related environments \mit{$(r \subseteq s^\top)$} and that programs that cannot be discriminated by related environments are related \mit{$(s^\top \subseteq r)$}. 
\Cref{thm:top-top-closure} is stated in terms of the definitions of the relations  \mit{$r$} and \mit{$s$}, and~\Cref{def:top}, i.e., the left side is equivalent to \mit{$(\mc{D}_1, \mc{D}_2) \in r$}, by the definition of \mit{$r$}, and the right side is equivalent to  
\mit{$(\mc{D}_1, \mc{D}_2) \in s^\top$} by~\Cref{def:top} and \mit{$s$}.

\Cref{thm:top-top-closure} thus establishes that our logical relation is compositional and extensional (behavioral). 
  It shows that our logical relation includes exactly the programs whose compositions with related environments are observationally equivalent.
  Soundness ensures that related environments are not very strong and do not have too much discriminatory power with respect to the observable behavior -- related environments cannot distinguish between the related programs. 
  Completeness guarantees that related environments are not very weak and have enough discriminatory power -- if related environments cannot distinguish between two programs, the programs are indeed related.
  This aligns with the assert-assume pas de deux described in \Cref{subsubsec:assume-assert}, ensuring that the assumptions made by related programs are exactly those asserted by the related environments and the assertions of the programs are exactly those assumed by the environments.

\section{Noninterference}
\label{sec:psni}

This section applies the \slr to noninterference,
phrased as an equivalence up to the secrecy level \mit{$\xi$} of an observer, \mit{$\equiv^{\Psi_0}_{\xi}$},
given a security lattice \mit{$\Psi_0$} (\Cref{subsec:ifc-upto-eq}).
The section then illustrates the up-to equivalence \mit{$\equiv^{\Psi_0}_{\xi}$} on an example (\Cref{subsec:case-study})
and concludes with its metatheoretic properties (\Cref{subsec:up-to-equiv-metatheory}),
including $\top \top$-closure.

\subsection{Attacker model}
\label{subsubsec:attacker-model}
 
Our attacker model is parametric in the secrecy level \mit{$\xi \in \Psi_0$} of an attacker and a program \mit{$P$},
given prior annotation of the free channels of \mit{$P$} with secrecy levels \mit{$c \in \Psi_0$}.
It assumes

\begin{itemize}
\item that the attacker knows the source code of \mit{$P$};
\item that the attacker can only observe the messages sent along the free channels of \mit{$P$} with secrecy level \mit{$c \sqsubseteq \xi$};
\item that an attacker cannot measure the passing of time;
\item that an attacker is oblivious of channel names;
\item a nondeterministic scheduler.
\end{itemize}

\subsection{Up-to equivalence}
\label{subsec:ifc-upto-eq}

To prove noninterference, we introduce an \emph{equivalence relation up to} the secrecy level \mit{$\xi$} of an observer,
\mit{$\equiv^{\Psi_0}_{\xi}$}, defined in \Cref{def:noninterference-sec} below.
This relation is reminiscent of logical equivalence \mit{$\equiv$} (\Cref{def:eq}),
but expects free channels to be annotated with secrecy levels and
allows ignoring those channels with secrecy level \mit{$\sqsubseteq \xi$}.

We first define two auxiliary notions: \textit{(i)}
downward projections on context of free channels \mit{$\Gamma \Downarrow \xi$}
and on the providing free channel \mit{$x_\alpha{:}T[c] \Downarrow \xi$}
(\Cref{def:typing-proj}), as well as
\textit{(ii)} the predicates \mit{$\mb{H\text{-}Provider}^\xi$} and \mit{$\mb{H\text{-}Client}^\xi$}
(\Cref{def:highprovider}).
Since channels are now annotated with secrecy levels,
we use the metavariables \mit{$\G$} and \mit{$K^s$}
to range over the linear typing context and providing channel, \respb.

\begin{definition}[Typing context projections]\label{def:typing-proj}
    Downward projection on security linear contexts \mit{$\Gamma$} and providing channels \mit{$K^s$} is defined as follows:
  {\small \mit{ \[\begin{array}{lclc}
    \Gamma, x_\alpha{:}T[c] \Downarrow \xi& \defeq& \Gamma \Downarrow \xi,  x_\alpha{:}T[c] & \m{if}\; c\sqsubseteq \xi  \\
         \Gamma, x_\alpha{:}T[c] \Downarrow \xi& \defeq& \Gamma \Downarrow \xi & \m{if}\; c \not \sqsubseteq \xi  \\
        \cdot \Downarrow \xi& \defeq& \cdot &  \\
         x_\alpha{:}T[c] \Downarrow \xi& \defeq&   x_\alpha{:}T[c] & \m{if}\; c\sqsubseteq \xi  \\
        x_\alpha{:}T[c] \Downarrow \xi& \defeq& \_{:}1[\top] & \m{if}\; c\not\sqsubseteq \xi \\
    \end{array}\]}}
\end{definition} 
The projections are used to keep only the free channels that are observable to an attacker.
We close the nonobservable channels off by composing the configuration with
high-confidentiality clients (\mit{$\mb{H\text{-}Client}^\xi$}) and providers (\mit{$\mb{H\text{-}Provider}^\xi$}),
which may differ for each configuration.
These high-confidentiality clients and providers are connected to the configurations through nonobservable channels,
making the messages they exchange with the configurations nonobservable to the attacker as well.

\begin{definition}[High provider and High client]\label{def:highprovider}
      {\small \mit{\[\begin{array}{lll}
         \cdot \,\,\in \mb{H\text{-}Provider}^\xi(\cdot) \\[5pt]
         \mc{B} \in \mb{H\text{-}Provider}^{\xi}(\Gamma, x_\alpha{:}A[c])\, \,\qquad  \m{iff} &\\
         \;\;\;\mb{either}\, c \not \sqsubseteq \xi  \, \m{and}\,
         \mc{B} = \mc{B}' \mc{T} \,\, \m{and}\,\, \mc{B}' \in \mb{H\text{-}Provider}^\xi(\Gamma)\,\,\m{and}\\ \qquad \qquad \qquad \qquad   \mc{T} \in \m{Tree}(\cdot \Vdash x_\alpha{:}A), \\
         \quad \mb{or}\,\, c \sqsubseteq \xi  \, \m{and}\,
        \mc{B} \in \mb{H\text{-}Provider}^\xi(\Gamma)\\[10pt]
     
         \mc{T} \in \mb{H\text{-}Client}^\xi( x_\alpha{:}A[c])\, \,\qquad  \m{iff} \\
         \;\;\;\mb{either}\, c \not \sqsubseteq \xi  \, \m{and}\, \mc{T} \in \m{Tree}(x_\alpha{:}A \Vdash \_:1), \mb{or}\\
         \;\;\;\mb{or}\,\,  c \sqsubseteq \xi  \, \m{and}\, \mc{B}= \cdot
       \end{array}\]}}
\end{definition}

Finally, we can state the relation \mit{$\equiv^{\Psi_0}_{\xi}$}.
Like \mit{$\equiv$} (\Cref{def:eq}), \mit{$\equiv^{\Psi_0}_{\xi}$} is phrased in terms of the term interpretation of the \slr,
but restricts observations to those free channels that are observable.
The remaining free channels will be composed in parallel with arbitrary \illst-typed configurations.

\begin{definition}[Logical equivalence up to observer level]\label{def:noninterference-sec}
    We define the relation \mit{$(\Gamma_1 \Vdash \mc{D}_1:: x_\alpha {:}A_1[{c_1}]) \equiv^{\Psi_0}_{\xi} (\Gamma_2 \Vdash \mc{D}_2:: y_\beta {:}A_2[{c_2}])$} as
$$\begin{array}{l}    
      {\mc{D}_1} \in \m{Tree}(\erasure{\Gamma_1} \Vdash x_\alpha {:}A_1)\,\, \m{and}\,\,{\mc{D}_2} \in \m{Tree}(\erasure{\Gamma_2} \Vdash y_\beta {:}A_2)\,\, \m{and}\\
      \Gamma_1 {\Downarrow} \xi= \Gamma_2 {\Downarrow} \xi=\Gamma\,\,\m{and}\,\, x_\alpha {:}A_1[{c_1}]{\Downarrow} \xi= y_\beta {:}A_2[{c_2}] {\Downarrow} \xi= K^s\,\,\m{and}\\
      \forall \, \mc{B}_1 \in \mb{H\text{-}Provider}^\xi(\Gamma_1).\, \forall \, \mc{B}_2 \in \mb{H\text{-}Provider}^\xi(\Gamma_2).\\ \,\forall \mc{T}_1 \in \mb{H\text{-}Client}^\xi(x_\alpha {:}A_1[{c_1}]). \,\forall \mc{T}_2 \in \mb{H\text{-}Client}^\xi(y_\beta {:}A_2[{c_2}]).\\
\forall\,m.\,(\mc{B}_1\mc{D}_1\mc{T}_1, \mc{B}_2\mc{D}_2\mc{T}_2) \in \mc{E}\llbracket \erasure{\Gamma} \Vdash \erasure{K^s} \rrbracket^{m},\,
    \m{and} \\[1pt]
    \forall\,m.\,(\mc{B}_2\mc{D}_2\mc{T}_2, \mc{B}_1\mc{D}_1\mc{T}_1) \in \mc{E}\llbracket \erasure{\Gamma} \Vdash \erasure{K^s} \rrbracket^{m}.
    \end{array}$$
\end{definition}

\Cref{def:noninterference-sec} makes use of an erasure operation \mit{$\erasure{\_}$} to drop the secrecy annotations,
to match the sequent of the logical relation whose free channels lack secrecy annotations.
As a result, \mit{$\equiv^{\Psi_0}_{\xi}$} can be used for configurations that type check using an information flow control (IFC) type system as well as 
``plain-vanilla'' \illst-typed configurations,
facilitating \emph{semantic typing}
\myCite{LoefARTICLE1982,ConstableBook1986,TimanyJACM2024}
for progress-sensitive noninterference (PSNI).

\subsection{Case study: progress-sensitive noninterference (PSNI)}
\label{subsec:case-study}

To illustrate \mit{$\equiv^{\Psi_0}_{\xi}$}, we instantiate it on a simplistic example,
in the spirit of an insecure version of the \mit{$\m{Verifier}$} encountered earlier (\Cref{fig:simple-verfier}), shown below.
Assuming that our security lattice ($\Psi_0$) is
{$
\mb{guest} \sqsubseteq \mb{alice}
$},
and that the attacker level is $\mb{guest}$,
the below process leaks to an attacker $y$ whether authentication was successful or not by either sending the label \mit{$\mathit{s}$} or \mit{$\mathit{f}$}, \respb.
\begin{center}
\begin{minipage}[t]{1\textwidth}
\begin{tabbing}
$\m{pin}$ $=$ \= $\& \{$\= $\mathit{tok_1}{:}\m{1}, \mathit{tok_2}{:} \m{1}\}\quad \quad$
$\m{bit} = \& \{\mathit{zero}{:}\m{1}, \mathit{one}{:} \m{1}\}$ \\[5pt]
$y{:}\m{bit}[\mb{guest}] \vdash \m{X} :: x{:}\m{pin}[\mb{alice}]=$
\=$(\mb{case} \, x \, ($\=$\mathit{tok}_1 \Rightarrow y.\mathit{zero};  \mb{wait}\, y; \,\mb{close}\,x$ \\
\>\>$\mid \mathit{tok}_{2} \Rightarrow y.\mathit{one}; \mb{wait}\, y; \,\mb{close}\,x))$    
\end{tabbing}
\end{minipage}
\end{center}
To prove that this process is secure, we would have to show that for any substitution that instantiates free variables $x$ and $y$ with free channels $x_\alpha$ and $y_\beta$
$$\begin{array}{l}
    y_\beta:\m{bit}[\mb{guest}]\Vdash \mb{proc}(x_\alpha, \m{X}) ::x_\alpha{:}\m{pin}[\mb{alice}]\,\,
     \equiv^{\Psi_0}_{\mb{guest}}\\ y_\beta:\m{bit}[\mb{guest}]\Vdash \mb{proc}(x_\alpha, \m{X}) ::x_\alpha{:}\m{pin}[\mb{alice}]
\end{array}$$
We compose the process with the following high-secrecy clients \mit{$\mathcal{T}_1$} and \mit{$\mathcal{T}_2$},
sending different tokens along the high-secrecy channel \mit{$x_\alpha{:}\m{pin}[\mb{alice}]$},
$$\begin{array}{l}
    \mathcal{T}_1= \mathbf{proc}(z_\eta, x_\alpha.\mathit{tok}_1; \mb{wait}\,x_\alpha; \mb{close}\,z_\eta)\\
    \mathcal{T}_2= \mathbf{proc}(z_\eta, x_\alpha.\mathit{tok}_2; \mb{wait}\,x_\alpha; \mb{close}\,z_\eta),
\end{array}$$
leaving us left to show that for all $m$:
\begin{center}
\begin{minipage}[t]{1\textwidth}
\begin{tabbing}
$($\=$\color{Brown}\mb{proc}(x_\alpha, \m{X})\, \mathbf{proc}(z_\eta, x_\alpha.\mathit{tok}_1; \mb{wait}\,x_\alpha; \mb{close}\,z_\eta);$ \\
\> $\color{Blue}\mb{proc}(x_\alpha, \m{X})\, \mathbf{proc}(z_\eta, x_\alpha.\mathit{tok}_2; \mb{wait}\,x_\alpha; \mb{close}\,z_\eta))
\in \mathcal{E}\llbracket y_\beta{:}\m{bit} \vdash \_{:}1 \rrbracket^m$
\end{tabbing}
\end{minipage}
\end{center}
Obviously this does not hold true, confirming that process \mit{$\m{X}$} is insecure.
Consider a variant of \mit{$\m{X}$} that sends the bit $\mathit{zero}$ in either case.
This variant would be accepted by our logical relation, as it should because it is secure.

Alternatively, we could employ an \emph{information flow control (IFC)} type system
\myCite{VolpanoARTICLE1996,SmithVolpanoPOPL1998,SabelfeldIEE2003}
to verify whether process \mit{$\m{X}$} is secure.
Such type systems label observables (\eg output, locations, channels) with secrecy levels
drawn from a given security lattice \mit{$\Psi_0$}
to prevent ``flows from high to low''.
Next, we sketch such an IFC type system as a refinement of \lang.
The full refinement type system is included in
\ifappendix
\Cref{apx:sec:ifc}.
\fi
It is based on existing work \myCite{DerakhshanLICS2021,DerakhshanECOOP2024},
and thus we only summarize its main characteristics and
discuss the typing rules relevant for the example.

To type process terms, we use the judgment
$$\Psi;\G \vdash_{\Sigma} P@c ::  x:A[d].$$
The new elements, compared to \lang's judgment (see \Cref{subsec:type-system}),
is the security lattice \mit{$\Psi$} and the possible worlds annotations
\mit{$@c$} and \mit{$[d]$} \myCite{BalzerESOP2019}, to denote a process' running secrecy and maximal secrecy, \respb.
\mit{$\G$} is a linear typing context, like \mit{$\Delta$},
but where types now have maximal secrecy annotations.
The idea is that a process' \emph{maximal secrecy} indicates the maximal level of secret information the process may ever obtain,
whereas a process' \emph{running secrecy} denotes the highest level of secret information the process has obtained so far.

To constrain the propagation of information, the following invariants are presupposed
on the typing judgment \mit{$\Psi;\G \vdash_{\Sigma} P@c ::  x:A[d]$}

\begin{enumerate}[label=(\roman*)]

\item \mit{$\forall y{:}B[d'] \in \G .\,  \Psi \Vdash d' \sqsubseteq d$}

\item \mit{$\Psi \Vdash c \sqsubseteq d$}

\end{enumerate}

\noindent ensuring that the maximal secrecy of a child node is capped by the maximal secrecy of its parent and
that the running secrecy of a process is less than or equal to its maximal secrecy, \respb.
Additionally, the type system ensures that

\begin{enumerate}
\item no channel of high secrecy is sent along one with lower secrecy; and
\item a process does not send any message along a low-secrecy channel after receipt of high secrecy information.
\end{enumerate}

\noindent The first condition prevents direct flows and is enforced by requiring that
the maximal secrecy level of a carrier channel and the channels sent over it match.
The second condition is met by ensuring that a process' running secrecy is a sound approximation
of the level of secret information a process has obtained so far.
To this end, the type system increases the running secrecy of a process upon receipt
to \textit{at least} the maximal secrecy of the sending process and, correspondingly,
guards sends by making sure that the running secrecy of the sending process is \textit{at most} the maximal secrecy of the receiving process.

This working can be seen in the typing rules for external choice:
\begin{small}
\begin{mathpar}
\inferrule*[right=$\& R$]
{\Psi; \G \vdash_{\Sigma} Q_k@c::  y{:}A_{k}[c] \\
\forall k \in L}
{\Psi;\G \vdash_{\Sigma} (\mathbf{case}\, y^c(\ell \Rightarrow Q_\ell)_{\ell \in I})@d_1::  y{:}\&\{\ell:A_{\ell}\}_{\ell \in L}[c]}

\inferrule*[right=$\& L$]
{\Psi \Vdash d_1 \sqsubseteq c \\
\Psi;  \G, x{:}A_{k}[c] \vdash_{\Sigma}  P@d_1::  y{:}C[c'] \\
k \in L}
{\Psi; \G, x{:}\&\{\ell:A_{\ell}\}_{\ell \in I}[c] \vdash_{\Sigma} (x^c.k; P)@d_1::  y{:}C[c']}
\end{mathpar}
\end{small}

\noindent The right rule increases the provider's running secrecy from \tm{d_1} to its maximal secrecy \tm{c}
after receipt of the label \tm{k}.
The left rule guards the send by the premise \tm{\Psi \Vdash d_1 \sqsubseteq c},
demanding that the provider's running secrecy \tm{d_1} is less than or equal to the maximal secrecy \tm{c} of the recipient.
It is precisely this guard that prevents process \tm{\m{X}} to type check using the refinement IFC type system.
This guard also prevents the variant of \mit{$\m{X}$}, which sends the bit $\mathit{zero}$ in either case,
to type check,
which is accepted by our logical relation.
This outcome is expected because type systems have the benefit of automatic verification,
foregoing completeness.
Our logical relation thus amounts to a \emph{semantic} logical relation
\myCite{LoefARTICLE1982,ConstableBook1986,TimanyJACM2024}
for PSNI, as it only requires configurations to be \illst-typed, but not IFC-typed.
The fact that IFC-typed processes are non-interfering is proved
as the \emph{fundamental theorem} of the logical relation,
\ifappendix
see \Cref{apx:thm:ftlr} in \Cref{apx:sec:noninterference}.
\fi

\subsection{Metatheoretic properties of up-to equivalence}
\label{subsec:up-to-equiv-metatheory}

Unsurprisingly, \mit{$\equiv^{\Psi_0}_{\xi}$} has analogous properties to \mit{$\equiv$},
except for reflexivity,
which only holds for IFC-typed configurations,
as demonstrated by our \mit{$\m{X}$} process in \Cref{subsec:case-study}.

\ifappendix

\subsubsection{Semantic cut: closure under parallel composition}
\label{subsubsec:ifc-compositionality}

\begin{lemma}[Semantic cut]\label{lem:compositionality-sec}
 If 
 \begin{enumerate}
 \item  $(\Gamma'_1 \Vdash \mc{C}_1:: \Gamma) \equiv^\xi_{\Psi_0} (\Gamma'_2 \Vdash \mc{C}_2:: \Gamma)\,\, \m{with}\,\, \Gamma=\Gamma_1 \Downarrow_\xi= \Gamma_2\Downarrow_\xi\,\,\m{and}$
 \item $(\Gamma_1 \Vdash \mc{D}_1:: x_\alpha{:}A_1[c_1]) \equiv^\xi_{\Psi_0} (\Gamma_2 \Vdash \mc{D}_2:: y_\beta{:}A_2[c_2])\,\m{and}$
 \item $(K^s \Vdash \mc{F}_1::K^s_1) \equiv^\xi_{\Psi_0} (K^s \Vdash \mc{F}_2:: K^s_2)\\\,\,\m{with}\, K^s=  x_\alpha{:}A_1[c_1] \Downarrow \xi = y_\beta{:}A_2[c_2] \Downarrow \xi$
 \end{enumerate}
 then
  {\small\(\,(\Gamma'_1, \Gamma^h_1 \Vdash \mc{C}_1\mc{D}_1\mc{F}_1:: K^s_1) \equiv^\xi_{\Psi_0} (\Gamma'_2, \Gamma^h_2 \Vdash \mc{C}_2\mc{D}_2\mc{F}_2:: K^s_2),\)}
 where \mit{$\Gamma^h_1$} is the set of all channels \mit{$w_\eta{:}C[d] \in \Gamma_1$} with \mit{$d \not \sqsubseteq \xi$} and \mit{$\Gamma^h_2$} is the set of all channels \mit{$w_\eta{:}C[d] \in \Gamma_2$} with \mit{$d \not \sqsubseteq \xi$}.
\end{lemma}
\begin{proof}
\ifappendix
See proof of \Cref{apx:lem:compositionality-eq} in the Appendix.
\else
\refproof
\fi
\end{proof}

\subsubsection{Partial equivalence relation (PER)}
\label{subsec:per}

    \begin{lemma}[Symmetry]\label{lem:symmetry-sec}
    For all security levels \mit{$\xi$} and configurations \mit{$\mc{D}_1$} and \mit{$\mc{D}_2$}, we have 
    {\small
    $(\Gamma_1 \Vdash \mathcal{D}_1:: x_\alpha {:}T_1[c_1])  \equiv^{\Psi_0}_{\xi} (\Gamma_2 \Vdash \mathcal{D}_2:: y_\beta {:}T_2[c_2]),$
    iff 
    $(\Gamma_2 \Vdash \mathcal{D}_2:: y_\beta {:}T_2[c_2]) \equiv^{\Psi_0}_{\xi} (\Gamma_1 \Vdash \mathcal{D}_1:: x_\alpha {:}T_1[c_1]).$
    }
    \end{lemma}
    \begin{proof}
    \ifappendix
See proof of \Cref{apx:lem:symmetry} in the Appendix.
\else
\refproof
\fi
    \end{proof}
    
    \begin{lemma}[Transitivity]\label{lem:trans-sec}
    For all security levels \mit{$\xi$}, and configurations \mit{$\mc{D}_1$}, \mit{$\mc{D}_2$}, and \mit{$\mc{D}_3$}, we have 
    \mit{
     \[
    \begin{array}{l}
      \m{if}\,(\Gamma_1 \Vdash \mathcal{D}_1:: x_\alpha {:}T_1[c_1])  \equiv^{\Psi_0}_{\xi} (\Gamma_2 \Vdash \mathcal{D}_2:: y_\beta {:}T_2[c_2])\, \m{and}\\\, (\Gamma_2 \Vdash \mathcal{D}_2:: y_\beta {:}T_2[c_2])  \equiv^{\Psi_0}_{\xi} (\Gamma_3 \Vdash \mathcal{D}_3:: z_\eta {:}T_3[c_3])\\
      \m{then}\,(\Gamma_1 \Vdash \mathcal{D}_1:: x_\alpha {:}T_1[c_1]) \equiv^{\Psi_0}_{\xi} (\Gamma_3 \Vdash \mathcal{D}_3:: z_\eta {:}T_3[c_3]).
    \end{array}  
    \]
    }
    \end{lemma}
\begin{proof}
\ifappendix
See proof of \Cref{apx:cor:transitivity} in the Appendix.
\else
\refproof
\fi
\end{proof}

\else
\fi

\ifappendix
\subsubsection{Adequacy}
\label{subsubsec:ifc-adequacy}
\else
\fi

\begin{definition}[Bisimulation with secrecy annotated sequent]\label{def:simequiv-sec}
    For \mit{$\,\mc{D}_1\in \m{Tree}(\erasure{\Gamma_1} \Vdash x_\alpha {:}A_1)$, $\mc{D}_2\in \m{Tree}(\erasure{\Gamma_2} \Vdash y_\beta {:}A_2)\,$} we define \mit{$\Gamma_1 \Vdash \mc{D}_1 :: x_\alpha {:}A_1[c_1] \approx_{\m{a}}^{\xi} \Gamma_2 \Vdash \mc{D}_2:: y_\beta {:}A_2[c_2]$} as
  %
  \mit{\[\begin{array}{c}
    \Gamma_1 \Downarrow \xi= \Gamma_2,  \Downarrow \xi\,\,\,\m{and}\,\,\, y_\beta {:}A_2[c_2] \Downarrow \xi =  x_\alpha {:}A_1[c_1] \Downarrow \xi\,\, \,\m{and}\, \\
    \forall \mc{B}_1 \in \m{H\text{-}Provider}^\xi(\Gamma_1). \forall \mc{B}_2\in \m{H\text{-}Provider}^\xi(\Gamma_2).\\
    \forall  \mc{T}_1 \in \m{H\text{-}CLient}^\xi(x_\alpha{:}A_1[c_1]).\, \forall \mc{T}_2 \in \m{H\text{-}Client}^\xi(y_\beta{:}A_2[c_2]).\\ \mc{B}_1\mc{D}_1\mc{T}_1 \approx_{\m{a}} \mc{B}_2\mc{D}_2\mc{T}_2.
  \end{array}
  \]}
  \end{definition}

  \begin{lemma}[Adequacy]\label{lem:soundness-completeness-sec}
    For all \mit{$\mc{D}_1\in \m{Tree}(\erasure{\Gamma_1} \Vdash x_\alpha {:}A_1)$ and $\mc{D}_2\in \m{Tree}(\erasure{\Gamma_2} \Vdash y_\beta {:}A_2)$}, we have
   \mit{$(\Gamma_1 \Vdash \mc{D}_1:: x_\alpha {:}A_1[{c_1}]) \equiv^{\Psi_0}_{\xi} (\Gamma_2 \Vdash \mc{D}_2:: y_\beta {:}A_2[{c_2}])$} iff \mit{$(\Gamma_1 \Vdash \mc{D}_1:: x_\alpha {:}A_1[{c_1}]) \approx_{\m{a}}^{\xi} (\Gamma_2 \Vdash \mc{D}_2:: y_\beta {:}A_2[{c_2}])$}.
    \end{lemma}
    \begin{proof}
\ifappendix
See proof of \Cref{apx:cor:eq-soundness-completeness} in the Appendix.
\fi
\end{proof}

\ifappendix
\subsubsection{Biorthogonality}
\label{subsubsec:ifc-biorthogonality}
\else
\fi

\begin{theorem}[\mit{$\top\top$}-closure]\label{thm:toptop-sec}
    Consider \mit{$\mc{D}_1 \in \m{Tree}(\erasure{\Gamma_1} \Vdash \erasure{x_\alpha{:}A_1[c_1]})$} and \mit{$\mc{D}_2 \in \m{Tree}(\erasure{\Gamma_2} \Vdash \erasure{y_\beta{:}A_2[c_2]})$} and a given observer level \mit{$\xi \in \Psi_0$}. We have
   {\small{\[\begin{array}{l}
      (\Gamma_1 \Vdash \mc{D}_1:: x_\alpha{:}A_1[c_1]) \equiv^\xi_{\Psi_0} (\Gamma_2 \Vdash \mc{D}_2:: y_\beta{:}A_2[c_2])\,\,
      \mathbf{iff} \\[4pt]
      \forall \mathcal{C}_1, \mathcal{C}_2, \mathcal{F}_1, \mathcal{F}_2.\,\\
      \m{if}\,(\Gamma_1 \dashv \mc{C}_1[\,\,]\mc{F}_1\dashv  x_\alpha{:}A_1[c_1]) \equiv^\xi_{\Psi_0} (\Gamma_2 \dashv \mc{C}_2[\,\,]\mc{F}_2 \dashv y_\beta{:}A_2[c_2])\,\,\\
      \m{then}\,\, \forall \mc{B}_1 \in \m{H\text{-}Provider}^\xi(\Gamma_1). \forall \mc{B}_2\in \m{H\text{-}Provider}^\xi(\Gamma_2).\\\forall \mc{T}_1 \in \m{H\text{-}CLient}^\xi(x_\alpha{:}A_1[c_1]).\,\forall \mc{T}_2 \in \m{H\text{-}Client}^\xi(y_\beta{:}A_2[c_2]).\\{\mc{B}_1{\mathcal{C}_1}{\mathcal{D}_1} {\mathcal{F}_1}\mc{T}_1}\, \approx_{\m{a}} \,\mc{B}_2\mathcal{C}_2{\mathcal{D}_2} {\mathcal{F}_2}\mc{T}_2
    \end{array}\]}}
    \end{theorem}
\begin{proof}
\ifappendix
See proof of \Cref{apx:thm:biorthogonal-sc} in the Appendix.  The proof relies on closure under parallel composition and soundness/completeness.
\fi
\end{proof}

Completeness ensures that the logical relation relates all secure programs,
while IFC type systems, by nature, reject infinitely many secure programs. 
Our logical relation thus becomes generally applicable.
For example, we may conceive a more permissive IFC refinement type system that accepts more secure programs.
All we would have to do in this case
is prove the fundamental theorem for that new refinement type system,
guaranteeing that well-typed programs are self-related by our logical relation and thus non-interfering.

We conclude this section by noting that all the results presented here are corollaries of the results in Section~\ref{sec:rslr}.

\section{Related work and discussion}
\label{sec:related}

\paragraph{Logical relations for session types.}
The application of logical relations to session types has focused predominantly on unary logical relations
for proving termination \myCite{PerezESOP2012,PerezARTICLE2014,DeYoungFSCD2020,RochaCairesESOP2023},
except for a binary logical relation for parametricity \myCite{CairesESOP2013} and
noninterference \myCite{DerakhshanLICS2021,DerakhshanECOOP2024,VanDenHeuvelECOOP2024}.
Our \slr shares its foundation in linear logic with this line of work,
but contributes significantly through its support of general recursive types.

Caires et al.~\myCite{CairesESOP2013} show parametricity for a terminating language and restrict observations to the single offering channel of a configuration, \ie observations are only made at the root.
Thanks to the validity of semantic cut, a single-type-indexed logical relation, such as the authors', falls out as a special case of our \slr.  
However, more fine-grained program equivalences, such as noninterference,
demand distinguishing the two roles a process configuration may assume.
To the best of our understanding, our account of duality ultimately also enabled biorthogonal closure
to prove the logical equivalence sound and complete.

Among prior work on logical relations for \illst, the work by Derakhshan et al. \myCite{DerakhshanLICS2021} is most closely related to ours.
The authors contribute an IFC session type system and develop a logical relation to show noninterference.
The refinement type system that we use in our case study (\Cref{subsec:case-study}) is based on that IFC system,
extended with secrecy-polymorphic processes \myCite{DerakhshanECOOP2024}.
Besides the authors' confinement to a terminating language,
sidestepping the intricacies of non-termination and nondeterminism,
the authors' metatheoretic results are significantly more limited than ours.
In particular, the logical relation is dependent on IFC typing,
and thus only amounts to a \emph{syntactic} logical relation for noninterference.
Moreover, their development lacks metatheoretic results about the entailed program equivalence,
\ie that it is sound with respect to bisimilarity.
There are several nuances in the design of our logical relation that enable us to show such results in a more general setting, which also allows for non-termination:
Unlike~\myCite{DerakhshanLICS2021}, which only considers closed configurations, \ie it cannot relate two programs without their clients and providers, our logical relation is defined over open configurations and thus allows for connection to biorthogonality.
Furthermore, our term interpretation carefully uses the universal and existential quantifiers and thus enables the proof of transitivity for the logical relation, which is not proven in~\myCite{DerakhshanLICS2021}.
Our logical relation also handles non-termination via focus channels and the observation index. 
These features enable defining the logical relation for equivalence in a general context, proving that it is both sound and complete with respect to asynchronous bisimilarity by establishing a connection to biorthogonality, and then instantiating it in various settings, such as noninterference.

Also noteworthy is the work by Rocha and Caires \myCite{RochaCairesESOP2023,RochaPhD2022}.
The authors contribute CLASS, a calculus based on classical linear logic session types
enriched with memory cells, acting as mutexes to share affine data between processes.
Unlike more liberal support of sharing \myCite{BalzerICFP2017,BalzerCONCUR2018,BalzerESOP2019},
the authors adopt the propositions-as-types paradigm to ensure deadlock freedom and strong normalization.
These guarantees are realized by stratifying resources into nested hierarchies.
Interesting in relation to our work, is CLASS' metatheory:
to prove strong normalization, the authors develop a unary logical relation,
defined to be $\top\top$-closed reliant on Girard's orthogonality \myCite{GirardARTICLE1987}.
Besides being binary to express program equivalence and admitting open configurations,
our logical relation is not defined to be $\top\top$-closed,
but the logical equivalence induced by our relation shown to be $\top\top$-closed (\Cref{thm:top-top-closure}).
As discussed in \Cref{subsec:biorthogonality},
the proof of admissibility of $\top\top$ closure serves a validation of the equivalence,
guaranteeing that its discriminatory power is strong enough (soundness)
while being maximally permissive (completeness).

\paragraph{Relationship to logical relations for stateful languages.}
Logical relations have been scaled to accommodate state using Kripke logical relations (KLRs)~\myCite{PittsStarkHOOTS1998}.
KLRs are indexed by a \emph{possible world $W$}, providing a semantic model of the heap.
Invariants can then be imposed that must be preserved by any future worlds $W'$.
KRLs can be combined with step indexing~\myCite{AppelMcAllesterTOPLAS2001,AhmedESOP2006}
to address circularity arising from higher-order stores~\myCite{AhmedPOPL2009,DreyerICFP2010,DreyerJFP2012} and to express state transition systems.
For example, Gregersen et al. \myCite{OddershedePOPL2021} use the KLR supported in Iris \myCite{JungJFP2018}
to prove noninterference for a functional language with a higher-order store, recursive types, and impredicative polymorphism.
The authors also adopt semantic typing \myCite{LoefARTICLE1982,ConstableBook1986,TimanyJACM2024},
admitting syntactically ill-typed programs, if shown to inhabit the logical relation.
Besides the difference in language, the authors consider termination-\emph{insensitive} noninterference,
whereas we consider progress-\emph{sensitive} noninterference.
Like KLRs, our \slr is situated in a stateful setting because channels, like locations,
are subject to concurrent mutation.
However, our \slr is rooted in \illst, which guarantees race freedom, stratifies the store (the configuration of processes),
and prescribes state transitions.
When squinting one's eyes, the sequent \mit{$\Delta \Vdash K$} of the logical relation seems
reminiscent of a possible world, since it provides the semantic typing of the free channels.
However, linearity would not comport well with the usual monotonicity requirement of logical relations.
We would like to investigate this connection as future work.

\section{Concluding remarks}
\label{sec:conclusions}
We have contributed a recursive session logical relation (\slr)
for progress-sensitive equivalence of programs.
The \slr is rooted in intuitionistic linear logic session types,
inheriting from it a strong foundation in linear logic,
ensuring that channel endpoints are treated as resources.
A novel aspect of the \slr is its use of an \emph{observation index}
to keep the logical relation well-defined in the presence of general recursive types.
This shift, from a step index / unfolding index associated with \emph{internal} computation steps
to an index only associated with \emph{external} observations,
facilitates statement and proof of metatheoretic properties of the logical relation:
closure under parallel composition,
soundness and completeness with regard to weak asynchronous bisimilarity,
and biorthogonal closure of the logical equivalence relation induced by the \slr.
A biorthogonality argument arises naturally from the duality inherent to session types.
In future work, we would like to explore these connections more deeply,
possibly connecting to Kripke logical relations and game semantics.

\begin{acks}

This material is based upon work supported by the
\grantsponsor{AFOSR}{Air Force Office of Scientific Research}{https://www.afrl.af.mil/AFOSR/}
under Award No. \grantnum{AFOSR}{FA9550-21-1-0385} (Tristan Nguyen, program manager)
and upon work supported by the
\grantsponsor{NSF}{National Science Foundation}{https://www.nsf.gov/}
under Grant No. \grantnum{NSF}{2442461}.
Any opinions, findings, and conclusions or recommendations expressed in this material are those of the author(s) and do not necessarily reflect the views of
the U.S. Department of Defense or
the National Science Foundation.
\end{acks}

\bibliographystyle{ACM-Reference-Format}
\bibliography{ref}

\clearpage

\appendix

\section*{Appendices}

In the appendix, we also present an IFC refinement type system, along with proofs of progress, preservation, and noninterference.
The results related to the IFC type system can be found in \Cref{apx:sec:lattice}, \Cref{apx:sec:ifc}, \Cref{apx:sec:progress-preservation:ifc}, \Cref{apx:sec:noninterference:quasi-and-relevant}, and \Cref{apx:sec:noninterference:ftlr}.

\section{Abstract syntax}
\label{apx:abstract-syntax}
The abstract syntax of $\lang$ is given below.
Lines without a left-hand side are separated by $\alt$ from their preceding line.

\begin{small}
\renewcommand{\tabcolsep}{2mm}
\begin{longtable}{@{}llll@{}}
\\
\textbf{Sort} &
\phantom{$\defined$} &
\textbf{Abstract Form} &
\textbf{Remarks} \\
Metavariable & $\defined$
& $\Psi_0$ & concrete security lattice $\langle \mathcal{L}, \mc{E}_0, \sqcup, \sqcap \rangle$ \\
& & $\iota, \eta \in \mathcal{L}$ & set of concrete security levels \\
& & $\xi \in \mathcal{L}$ & concrete security level of observer \\
& & $E_0 \in \mathcal{E}_0$ & set of relations $E_0$ of form $\iota \sqsubseteq \iota'$ \\
& & $\Psi$ & security theory $\langle \mathcal{V}, \mathcal{E}, \sqcup, \sqcap \rangle$ \\
& & $\psi, \omega \in \mathcal{V}$ & set of security variables \\ 
& & $c, d, e, f \in \mathcal{S}$ & security terms of $\Psi$ \\
& & $p \in \mathcal{S} \times \mathcal{S}$ & pair of security terms of $\Psi$ \\
& & $E \in \mathcal{E}$ & set of relations $E$ of form $c \sqsubseteq d$ \\
& & $x_\alpha, x_\beta, x_\gamma, x_\delta$ &  channel \\
& & $x, y, z, u, v, w$ & channel variable \\
& & $j, k, \ell \in I, L$ & set of labels \\
& & $\Delta, \Lambda$ & linear typing contexts(channels) \\
& & $\Omega$ & linear typing contexts (variables)\\
& & $\Gamma$ & linear security typing contexts (channels)\\
& & $\Xi$ & linear security typing contexts(variables) \\
& & $K$ & linear typing context singleton(channel) \\
& & $K^s$ & linear security typing context singleton(channel) \\
& & $\gamma_{\m{sec}}, \hat{\gamma_{\m{sec}}}, \delta_{\m{sec}}, \hat{\delta_{\m{sec}}}$ & order-preserving substitution of security elements \\
& & $\gamma, \hat{\gamma}, \delta, \hat{\delta}$ & channel variable substitution \\
& & $ \mc{A}, \mc{B}, \mc{C}, \mc{D}, \mc{T}$ &  process configuration in $\lang$\\
& & $ \mathbb{A}, \mathbb{B}, \mathbb{C}, \mathbb{D},  \mathbb{T}$ &  process configuration in $\reflang$\\
Type $A, B, C, T$ & $\defined$
& $\intchoice{A}$ & internal choice, at least one label \\
& & $\extchoice{A}$ & external choice, at least one label \\
& & $\chanout{A}{B}$ & channel output \\
& & $\chanin{A}{B}$ & channel input \\
& & $\one$ & termination \\
& & $Y$ & type variable \\
Definition $X, Y$ & $\defined$
& $\procdefb{\Delta}{X}{P}{\psi_0}{\omega_0}{x}{A}{\psi}{\omega}$ & process definition \\
& & $Y = A$ & type definition \\
Process $P, Q$ & $\defined$
& $\labsnd{x}{k}{P}$ & label output \\
& & $\labrcv{x}{P}$ & label input \\
& & $\chnsnd{y}{x}{P}$ & channel output \\
& & $\chnrcv{y}{x}{P}$ & channel input \\
& & $\cls{x}$ & terminate process \\
& & $\wait{x}{;Q}$ & wait for process to terminate \\
& & $\spwn{x}{X}{\Delta}{Q}$ & spawn \\
& & $\fwd{x}{y}$ & forward $x$ to $y$ \\
& & $F_Y$ & forwarder process variable for $Y$ \\
Messages $M,N$ & $\defined$
& $x.k$ & label output \\
& & $\mb{send}\, y\, x$ & channel output \\
& & $\cls{x}$ & terminate process 
\end{longtable}
\end{small}

\section{\illst type system and asynchronous dynamics}
\label{apx:sec:type-system}
This section defines the rules to type \lang programs, both for source code as well as for run-time configurations.

\subsection{Process term typing}

\Cref{apx:fig:structural-type-system} summarizes the typing rules for process terms in \lang.

\begin{figure}
\centering
\begin{small}
\input{figs/structural-type-system}
\end{small}
\caption{Process term typing rules and signature checking rules of \lang.}
\label{apx:fig:structural-type-system}
\end{figure}

Rule $\msc{d:fwd}$ in \Cref{apx:fig:structural-type-system} relies on a forwarder for each user-defined type definition,
amounting to an identity expansion automatically generated as defined in \Cref{apx:def:fwder}.

\begin{definition}\label{apx:def:fwder}
    For all $Y = A \in \Sigma$, we extend $\Sigma$ by adding the following definition to the signature $\Sigma$:
\[\begin{array}{ll}
    x':A \vdash F_Y = \fwder{A,y' \leftarrow x'}:: y': A  & \qquad
    Y = A \in \Sigma
\end{array}\]
Here $F_Y$ is a specific process variable assigned to the forwarder process for type variable $Y$, and $\fwder{A, y\leftarrow x}$ is defined by induction on the structure of $A$ as a function from type $A$ to process terms as follows:
\begin{mathpar}
\begin{array}{lll}
\fwder{\oplus\{\ell{:}A_\ell\}_{\ell \in L}, y \leftarrow x}&:=& \mathbf{case} \, x( \ell \Rightarrow y.\ell; \fwder{A_\ell, y \leftarrow x})_{\ell \in L}\\[6pt]
\fwder{ \&\{\ell{:}A_\ell\}_{\ell \in L},  y \leftarrow x}&:=&\mathbf{case} \, y( \ell \Rightarrow x.\ell; \fwder{A_\ell, y \leftarrow x})_{\ell \in L}\\[6pt]
\fwder{A \otimes B,y \leftarrow x}&:=&w \leftarrow \mathbf{recv}\, x; \mathbf{send} \, w\,y; \fwder{B,y \leftarrow x}\\[6pt]
\fwder{A \multimap B,y \leftarrow x}&:=& w \leftarrow \mathbf{recv}\, y; \mathbf{send} \, w\,x; \fwder{B,y \leftarrow x}\\[6pt]
\fwder{1,y \leftarrow x}&:=& \mathbf{wait}\,x; \mathbf{close}\, y\\[6pt]
\fwder{Y; y \leftarrow x}&:=&F_Y[x\mapsto x', y \mapsto y'] \qquad Y =A \in \Sigma\\[6pt]
\end{array}
\end{mathpar}
\end{definition}

Forwarders are well-typed by construction, as shown next:

\begin{lemma}
   Given the extended signature $\Sigma$, for all type $A$, there is a derivation for $ x:A \vdash_\Sigma  \fwder{A, y \leftarrow x}:: y: A$.
\end{lemma}
\begin{proof}
The proof is by induction on the structure of type $A$. In a base case, where $A$ is a type variable $Y$, i.e., $A=Y$ for $Y = C \in \Sigma$, the proof is straightforward by the $\msc{d:fwd}$ rule since  $x:C\vdash F_Y := \fwder{C,  y \leftarrow x}:: y: A \in \Sigma$.
The proof of other cases is straightforward.

\end{proof}

\subsection{Configuration typing}

\Cref{apx:fig:config-typing} summarizes the typing rules for process configurations and messages in \lang.
Since $\Sigma$ is fixed, we may drop it in a configuration typing judgment and a process typing judgment, respectively, for brevity.

We introduce the judgments
$(\mathcal{D}_1; \mathcal{D}_2) \in \m{Tree}(\Delta \Vdash K)$,
$\mathcal{D} \in \m{Tree}(\Delta \Vdash K)$, and
${\mathcal{B}}\in \m{Forest}(\Delta \Vdash \Delta')$
to denote well-typed configurations, defined as follows:

\begin{figure}
\centering
\begin{small}
\input{figs/config-typing}
\end{small}
\caption{Configuration typing rules of \lang.}
\label{apx:fig:config-typing}
\end{figure}

\begin{definition}[Well-typed configuration]\label{apx:def:well-typed-config}
    Well-typed configuration(s) are defined in terms of the judgments $(\mathcal{D}_1; \mathcal{D}_2) \in \m{Tree}(\Delta \Vdash K)$ and $\mathcal{D} \in \m{Tree}(\Delta \Vdash K)$,
    where $K$ is either of the form $x{:}A$ or $\_{:}1$
    and $\_$ stands for an arbitrary channel name along which no observations are made.
    \begin{itemize}

    \item  We define ${\mathcal{D}}\in \m{Tree}(\Delta \Vdash K)$ as \[  \Delta  \Vdash \mathcal{D} :: K \]
    \item  We define $(\mathcal{D}_1; \mathcal{D}_2) \in \m{Tree}(\Delta \Vdash K)$ as $\mathcal{D}_1\in \m{Tree}(\Delta \Vdash K)$ and $\mathcal{D}_2\in \m{Tree}(\Delta \Vdash K)$   
    \item We define ${\mathcal{B}}\in \m{Forest}(\Delta \Vdash \Delta')$ as \[  \Delta  \Vdash \mathcal{B} :: \Delta' \]
    \item  We define the notation $\mc{T} \in \mc{B}$ meaning $\mc{T}$ is a particular tree in a forest of trees $\mc{B}$: For $\mc{B}\in \m{Forest}(\Delta \Vdash \Delta')$, we write $\mc{T} \in \mc{B}$ iff $\mc{B}= \mc{B}' \mc{T}$, and $\mc{B}'\in \m{Forest}(\Delta_1 \Vdash \Delta'_1)$ and $\mc{T}'\in \m{Tree}(\Delta_2 \Vdash K)$ with $\Delta=\Delta_1, \Delta_2$ and $\Delta'= \Delta'_1, K$.
\end{itemize}
    
    \defend
\end{definition}

\subsection{Asynchronous dynamics}

\Cref{apx:fig:dynamics} gives the asynchronous dynamics of \lang as multiset rewriting rules \myCite{CervesatoARTICLE2009}.

\begin{figure}
\centering
\begin{small}
\input{figs/dynamics}
\end{small}
\caption{Asynchronous dynamics of \lang.}
\label{apx:fig:dynamics}
\end{figure}

\section{Concrete security lattice and security theories}
\label{apx:sec:lattice}

Process configurations and terms are typed relative to a concrete security lattice and a security theory.
A \emph{concrete} lattice $\Psi_0$ is defined globally for an application and consists of concrete security levels $\iota$.
Our running example lattice
\[\mb{guest} \sqsubseteq \mb{alice} \sqsubseteq \mb{bank} \qquad \mb{guest} \sqsubseteq \mb{bob} \sqsubseteq \mb{bank}\]
\noindent is an example of a concrete security lattice.
Polymorphic process definitions make use of a \emph{security theory} $\Psi$, ranging over security variables $\psi$ and
concrete security levels $\iota$ from the given concrete security lattice $\Psi_0$.
At run-time, all security variables occurring in polymorphic processes will be replaced with concrete security levels.

\begin{definition}[Concrete security lattice and security theory]\label{apx:def:lattice}

Let $\Psi_0 \defined \langle \mathcal{L}, \mc{E}_0, \sqcup \rangle$ be a concrete 01-bounded join semi-lattice with a partial order $\mc{E}_0$
over concrete security levels $\iota, \eta \in \mathcal{L}$ such that $E_0 \in \mc{E}_0$ is of the form $\iota \sqsubseteq \iota'$.
\noindent We define a security theory $\Psi \defined \langle \mathcal{V}, \mathcal{E}\rangle$
that augments $\Psi_0$ with security variables $\psi$ and relations $\mc{E}$ over them such that
\begin{itemize}
\item $\psi \in \mathcal{V}$ is a set of security variables

\item $c, d \in \mc{S}$ is a set of security terms of the security theory $\Psi$, defined by the grammar
\[c, d:= c \sqcup d \mid \iota \mid \psi \]
\item $E \in \mathcal{E}$ is a set of relations over the elements $c, d$ of the security theory where $E = c \sqsubseteq d$.

\end{itemize}
For convenience, we define the projection $\mc{E}(\Psi)$ to extract the relations $\mc{E}$ of $\Psi$.
We write
\[\latcstr{E}\]
if $E$ is consequence of the lattice theory based on the relations.
The judgment is defined in \Cref{apx:fig:latcstr-def}.
\end{definition}

\begin{figure}
\centering
\begin{small}
\input{figs/latcstr-def.tex}
\end{small}
\caption{Inductive definition of $\latcstr{E}$.}
\label{apx:fig:latcstr-def}
\end{figure}

Process definitions and process spawning rely on an order-preserving substitution, defined as follows:

\begin{definition}[Order-preserving substitution]\label{apx:def:subst}

Let $\gamma \in \mc{V} \rightarrow \mc{S}$ be a total function that maps security variables to security terms.
Its lifting $\hat{\gamma}$ to other syntactic objects, such as security terms $c$, process terms $P$, typing contexts $\Delta$, and security lattices $\Psi$,
is defined by structural induction over the syntactic object, replacing simultaneously all variable occurrences $\psi_i$ in the object with $\gamma(\psi_i)$.
\noindent The function $\gamma$ and its lifting $\hat{\gamma}$ must be order-preserving, ensuring that if $\Psi' \Vdash E$, then $\hat{\gamma}(\Psi') \Vdash \hat{\gamma}(E)$.
To link a spawner and a spawnee, we define an order-preserving substitution $\Psi \Vdash \gamma : \Psi'$
\[\Psi \Vdash \gamma : \Psi' \defined \text{if}\; \Psi' \Vdash E, \text{then}\; \Psi \Vdash \hat{\gamma}(E)\]
\noindent Composition $\gamma \circ \gamma'$ of two substitutions $\gamma$ and $\gamma'$ is defined as usual.
We observe that if both $\gamma$ and $\gamma'$ are order-preserving so is their composition.

\end{definition}

The type system requires every spawner to provide an order-preserving substitution $\Psi \Vdash \gamma : \Psi'$ for the security variables of the spawnee (rule \msc{Spawn}).
As a result, the security theory $\Psi'$ of the spawnee must be satisfied by the security theory $\Psi$ of the spawner,
\ie $\text{if}\; \Psi' \Vdash E, \text{then}\; \Psi \Vdash \hat{\gamma}(E)$.
Moreover, if the spawner provides a substitution $\delta$ for $\Psi_0$, \ie $\Psi_0 \Vdash \delta : \Psi$,
so does the spawnee, \ie $\Psi_0 \Vdash \gamma \circ \delta : \Psi'$.

\section{IFC refinement type system}
\label{apx:sec:ifc}
This section develops an Information Flow Control (IFC) refinement type system for \illst, yielding the language \reflang.
Deviating from the main text, we use here the metavariable $\Xi$ instead of $\G$ to denote the security typing context.

\subsection{Process term typing}

\Cref{apx:fig:ifc-type-system} summarizes the typing rules for process terms in \reflang.

\begin{figure}
\centering
\begin{small}
\input{figs/ifc_type_system}
\end{small}
\caption{Process term typing rules and signature checking rules of \reflang.}
\label{apx:fig:ifc-type-system}
\end{figure}

Rule $\msc{Spawn}$ relies on two substitutions, $\gamma_{\m{var}}$ that provides a substitution for channel variables to match them up with the definitions in the signature, and $\gamma_{\m{sec}}$ which is an order-preserving substitution $\substmap{\Psi}{\gamma_{\m{sec}}}{\Psi'}$,
guaranteeing that the security terms provided by the spawner comply with the order expected among those terms by the spawnee.
Rule $\msc{Spawn}$ moreover establishes the invariants necessary to prevent information leakage for the newly spawned process,
by the premise $\latcstr{\substmapap{\gamma}{\psi}\sqsubseteq d}$, and
allows the newly spawned process to start at least at the spawner's running secrecy $d_0$, by the premise $d_0 \sqsubseteq \latcstr{\substmapap{\gamma}{\psi}}$.

Rule $\msc{d:fwd}$ in \Cref{apx:fig:ifc-type-system} relies on a forwarder for each user-defined type definition,
amounting to an identity expansion automatically generated as defined in \Cref{apx:def:ifc-fwder}.

\begin{definition}\label{apx:def:ifc-fwder}
    For all $Y = A \in \Sigma$, we extend $\Sigma$ by adding the following IFC-typed forwarder definition to the signature $\Sigma$:
\[\begin{array}{ll}
    \psi = \psi; x:A[\psi] \vdash F_A = \fwder{A,y^\psi \leftarrow x^\psi} @ \psi :: y: A[\psi]  & \qquad
    Y = A \in \Sigma
\end{array}\]
Here $F_Y$ is a specific process variable assigned to the forwarder process for type variable $Y$, and $\fwder{A, y \leftarrow x}$ is defined similar to \Cref{apx:def:fwder} as
\[\begin{array}{lll}
    \fwder{\oplus\{\ell{:}A_\ell\}_{\ell \in L}, y \leftarrow x}&:=& \mathbf{case} \, x( \ell \Rightarrow y.\ell; \fwder{A_\ell, y^c \leftarrow x^c})_{\ell \in L}\\[6pt]
    \fwder{ \&\{\ell{:}A_\ell\}_{\ell \in L},  y^c \leftarrow x^c}&:=&\mathbf{case} \, y( \ell \Rightarrow x.\ell; \fwder{A_\ell, y^c \leftarrow x^c})_{\ell \in L}\\[6pt]
    \fwder{A \otimes B,y^c \leftarrow x^c}&:=&w \leftarrow \mathbf{recv}\, x; \mathbf{send} \, w\,y; \fwder{B,y^c \leftarrow x^c}\\[6pt]
    \fwder{A \multimap B,y^c \leftarrow x^c}&:=& w \leftarrow \mathbf{recv}\, y; \mathbf{send} \, w\,x; \fwder{B,y^c \leftarrow x^c}\\[6pt]
    \fwder{1,y^c \leftarrow x^c}&:=& \mathbf{wait}\,x; \mathbf{close}\, y\\[6pt]
    \fwder{Y; y^c \leftarrow x^c}&:=&F_Y[(x'\mapsto x, y' \mapsto y), (\psi \mapsto c)] \qquad Y =A \in \Sigma\\[6pt]
    \end{array}\]
\end{definition}

Forwarders are well-typed by construction, as shown next:

\begin{lemma}
   Given the extended signature $\Sigma$, for all type $A$,  there are derivations for
   \begin{itemize}
   \item[(i)] $\Psi; x:A[\psi] \vdash_\Sigma  \fwder{A, y^\psi \leftarrow x^\psi} @\psi:: y: A[\psi]$, and
   \item[(ii)]  $\Psi; x:A[\psi] \vdash_\Sigma  \fwder{A, y^\psi \leftarrow x^\psi} @\psi':: y: A[\psi]$, when $\Psi \Vdash \psi' \sqsubseteq \psi$ and $A \neq Y$ for a type variable $Y$.
   \end{itemize}
\end{lemma}
\begin{proof}
    \begin{itemize}
        \item[(i)] The proof is by induction on the structure of type $A$. In a base case, where $A$ is a type variable $Y$, i.e., $A=Y$ for $Y = C \in \Sigma$, the proof is straightforward by the $\msc{d:fwd}$ rule since  $\psi=\psi;x:C[\psi]\vdash F_Y := \fwder{C,  y^\psi \leftarrow x^\psi} @ \psi:: y: A[\psi] \in \Sigma$, and the substitutions $\gamma_{\m{sec}}$ and $\gamma_{\m{var}}$ enforce the premises of the rule.
        The proof of other cases is straightforward.
        \item[(ii)] The proof is by case analysis on the structure of type $A$. In all cases, by the way we defined the process terms, after the very first application of a rule (the first communication which is always a receive), the running secrecy increases to $\psi$, and we can apply the previous item. Note that by the tree invariant every judgment in our typing derivation satisfies the condition $\Psi \Vdash \psi' \sqsubseteq \psi$.
    \end{itemize}
\end{proof}

\subsection{Configuration typing}

\Cref{apx:fig:ifc-config-typing} summarizes the typing rules for process configurations and messages in \reflang.
Since $\Psi_0$ and $\Sigma$ are fixed, we may drop $\Psi_0$ and $\Sigma$ in a configuration typing judgment and a process typing judgment, respectively, for brevity.

\begin{figure}
\centering
\begin{small}
\input{figs/ifc-config-typing}
\end{small}
\caption{Configuration typing rules of \reflang.}
\label{apx:fig:ifc-config-typing}
\end{figure}

\subsection{Security erasure}

We define erasure of typing contexts, signatures, and configurations to state noninterference as a logical equivalence between \lang configurations.
To this end, \Cref{apx:thm:ifc-types-is-session-typed} proves that well-typed \reflang configurations are well-typed \lang configurations.

\begin{definition}
    Define security-erasure $\erasure{\Sig}$ of the signature $\Sig$ as
    \[
    \begin{array}{lcl}
        \erasure{\cdot}&= & \cdot\\
        \erasure{Y=A, \Sigma'} &= & Y=A, \erasure{\Sigma'}\\
        \erasure{ \Psi; \Xi \vdash X=P @ \psi_0 :: x{:}A[\psi],\, \Sigma'} & = &  \erasure{\Xi} \vdash X= P :: x{:}A\, ,\erasure{\Sigma'}
    \end{array}
    \]
\end{definition}

\begin{definition}
Define security-erasures $\erasure{\Xi}$,  $\erasure{x{:}A[c]}$, $\erasure{\Gamma}$, and $\erasure{K^s}$ of the security linear variable context $\Xi$, security channel variable, security linear channel context $\Gamma$, and security channel singleton $K^s$, respectively, as
\[\begin{array}{lclc}
    \erasure{\Xi, x{:}A[c]} & \defeq& \erasure{\Xi},  x{:}A \\
    \erasure{x{:}A[c]} & \defeq&   x{:}A   \\[4pt]
    \erasure{\Gamma, x_\alpha{:}A[c]} & \defeq& \erasure{\Gamma},  x_\alpha{:}A \\
         \erasure{\cdot}& \defeq& \cdot &  \\[4pt]
         \erasure{x_\alpha{:}A[c]} & \defeq&   x_\alpha{:}A   \\
        \erasure{\_{:}1[\top]}& \defeq& \_{:}1\\
    \end{array}\]
\end{definition}

\begin{definition}
    Define a security-erasure $\erasure{\mathbb{C}}$ of the configuration $\mathbb{C}$ as
    \[
    \begin{array}{lcl}
        \erasure{\cdot}&= & \cdot\\
        \erasure{\mathbb{C}, \mathbf{proc}(x[d], P@d_1)} &= &  \erasure{\mathbb{C}}, \mathbf{proc}(x, P)\\
        \erasure{\mathbb{C}, \mathbf{msg}(M)} & = & \erasure{\mathbb{C}}, \mathbf{msg}(M)
    \end{array}
    \]
    \end{definition}

\begin{theorem}
\label{apx:thm:ifc-types-is-session-typed}
    Every IFC well-typed configuration $\Psi_0; \Gamma \Vdash \mathbb{C} :: K^s$  is session-typed $\erasure{\mathbb{C}} \in \m{Tree}(\erasure{\Gamma} \Vdash \erasure{K^s})$.
\end{theorem}
\begin{proof}
    By induction on the derivation of $\Psi_0; \Gamma \Vdash \mathbb{C} :: K^s$.
\end{proof}

\section{Progress and preservation}
\label{apx:sec:progress-preservation}
This section proves type safety, first for \lang and then for \reflang.

\subsection{Session-typed processes}
This section proves type safety of session-typed processes in \lang.
We first start with defining the notion of a poised configuration and proofs of necessary lemmas.
\subsubsection{Poised configuration and configuration permutation}\label{apx:sec:progress-preservation:defs}

Progress relies on the notion of a \emph{poised} configuration.
A configuration is poised if it is empty or cannot take any internal steps but wants to engage in a message exchange along any of its free channels.

\begin{definition}[Poised Configuration]\label{apx:def:poised}
A configuration $\Delta_1, \Delta_2 \Vdash \mc{C}_1, \mc{C}_2:: \Lambda, w{:}A'$ is poised iff either $\mc{C}_1, \mc{C}_2$ is empty or $\Delta_1 \Vdash \mc{C}_1:: \Lambda$ is poised and $ \Delta_2 \Vdash \mc{C}_2::  w{:}A'$ is poised. The configuration $\Delta_2 \Vdash \mc{C}_2 :: w{:}A'$ is poised iff it cannot take any steps and at least one of the following conditions hold:
\begin{enumerate}
\item\label{apx:def:poised:emp} $\mc{C}_2$ is an empty configuration.
\item\label{apx:def:poised:left-snd} $\mathcal{C}_2= \mc{C}'_2\, \mb{msg}(M)\, \mc{C}''_2$ such that $\mb{msg}(M)$ is a negative message along $y_\alpha \in \Delta_2$, i.e. $y_\alpha{:}\&\{\ell{:}A_\ell\}_{\ell \in L}\Vdash \mb{msg}(M):: y_{\alpha+1}{:}A_k$ or $y_\alpha{:}A\multimap B, z_\beta{:}A \Vdash \mb{msg}(M):: y_{\alpha+1}{:}B$, and both sub-configurations $\mc{C}_2'$ and $\mc{C}_2''$ are poised. 
\item\label{apx:def:poised:left-rcv}  $\mathcal{C}_2= \mc{C}'_2\, \mb{proc}(x,P)\, \mc{C}_2''$ such that $\mb{proc}(x, P)$ attempts to receive along a channel $y_\alpha {\in} \Delta_2$, and both sub-configurations $\mc{C}_2'$ and $\mc{C}_2''$ are poised.
\item\label{apx:def:poised:right-snd} $\mathcal{C}_2= \mc{C}_2'\, \mb{msg}(P)$ such that $\mb{msg}(M)$ is a positive message sent along $w_\beta{:}A'$, i.e. $ w_{\beta+1}{:}A_k \Vdash \mb{msg}(M)::w_{\beta}{:}\oplus\{\ell{:}A_\ell\}_{\ell \in L}$ or $w_{\beta+1}{:}B, z_{\gamma}{:}A \Vdash \mb{msg}(M):: w_{\beta}{:}A\otimes B$,or $\cdot \Vdash \mb{msg}(M):: w{:}1$,
and sub-configuration $\mc{C}'_2$ is poised.
\item\label{apx:def:poised:right-rcv} $\mathcal{C}_2= \mc{C}_2'\, \mb{proc}(w, P)$ such that $\mb{proc}(w, P)$ attempts to receive along $w{:}A'$,
and sub-configuration $\mc{C}'_2$ is poised.
\end{enumerate}

\defend
\end{definition}

The dynamics (see \Cref{apx:fig:dynamics}) is expressed in terms of multiset rewriting rules,
which update a configuration locally, without regard for the remaining configuration.
As a result, the updated configuration may not necessarily be well-typed, according to the rules in the configuration typing.
For example, stepping the configuration
\[\mc{C}_1 \mc{T}\; \mb{proc}(y_\alpha,(\mathbf{send}\,x_\beta\,y_\alpha; P))\; \mc{C}_2 \ostep\]
using rule $\otimes_{\m{snd}}$ (see \Cref{apx:fig:dynamics}), yields the configuration
\[\mc{C}_1 \mc{T}\; \mb{proc}(y_{\alpha+1}, ([y_{\alpha+1}/y_{\alpha}]P))\; \mb{msg}( \mathbf{send}\,x_\beta\,y_\alpha)\; \mc{C}_2\]
For well-typedness, the subtree $\mc{T}$ rooted at the message $\mb{msg}( \mathbf{send}\,x_\beta\,y_\alpha)$
would have to be moved left to the message.
Our proofs account for this possibility and only require that the dynamics yield a valid permutation of a well-typed configuration,
assuming that the pre-state is a valid permutation of well-typed configuration as well.
We define the notion of a valid permutation next.
A valid permutation may rearrange the order of processes and messages in a configuration as long as parent-child relationships are preserved.

\begin{definition}[Valid configuration permutation]\label{apx:def:permutation}

For a well-typed configuration $\Psi_0; \Delta \Vdash \mc{C} :: \Delta'$,
a valid permutation $\mc{P}(\mc{C})$ can be derived by simultaneously changing the position of a process or message in $\mc{C}$,
yielding the permutation $\mc{C}'$, \ie $\mc{P}(\mc{C}) = \mc{C}'$,
as long as the following conditions are met:
\begin{enumerate}
\item For a process $\mathbf{proc}(z_\alpha, P)$ in $\mc{C}'$ such that $\mc{C}' = \mc{C}_1 \mathbf{proc}(z_\alpha, P) \mc{C}_2$
and for all $y_\beta \notin \m{dom}(\Delta)$ that $P$ is using,
there must exist either a process $\mathbf{proc}(y_\beta, \_)$ or a message $\mathbf{msg}(\_\langle y_\beta \rangle)$ in $\mc{C}_1$.

\item For a positive message $\mathbf{msg}(M\langle z_\alpha \rangle)$ in $\mc{C}'$
such that $\mc{C}' = \mc{C}_1 \mathbf{msg}(M\langle z_\alpha \rangle) \mc{C}_2$,
if $z_\alpha \notin \m{dom}(\Delta)$,
there must exist either a process $\mathbf{proc}(z_\alpha, \_)$
or a message $\mathbf{msg}(\_\langle z_\alpha \rangle)$ in $\mc{C}_1$.
Moreover, for all $y \notin \m{dom}(\Delta)$ that $P$ is using,
there must exist either a process $\mathbf{proc}(y, \_)$ or a message $\mathbf{msg}(y.\_)$ in $\mc{C}_1$.

\item For a negative message $\mathbf{msg}(P\langle v_\beta \rangle)$ in $\mc{C}'$
such that $\mc{C}' = \mc{C}_1 \mathbf{msg}(P\langle v_\beta \rangle) \mc{C}_2$,
if $v_\beta \notin \m{dom}(\Delta)$,
there must exist either a process $\mathbf{proc}(v_\beta, \_)$
or a message $\mathbf{msg}(\_\langle w \rangle)$ in $\mc{C}_1$.
Moreover, for all $y_\alpha \notin \m{dom}(\Delta)$ that $P$ is using,
there must exist either a process $\mathbf{proc}(y_\alpha, \_)$ or a message $\mathbf{msg}(y_\alpha.\_)$ in $\mc{C}_1$.
\end{enumerate}

\defend
\end{definition}

\subsubsection{Lemmas}\label{apx:sec:progress-preservation:lemmas}

\begin{lemma}[Term variable substitution]\label{apx:lem:term-var-subst}
The following substitutions are type-preserving and thus admissible:
\begin{enumerate}
\item\label{apx:lem:term-var-subst:right}
If $ \Delta \vdash_{\Sigma} P :: x:A$
then, for any fresh $y:A$, we have $ \Delta \vdash_{\Sigma} \subst{y}{x}{P} :: y:A$.
\item \label{apx:lem:term-var-subst:left}
If $\Delta, y : B \vdash_{\Sigma} P :: x:A$
then, for any fresh $z:B$, we have $\Delta, z : B  \vdash_{\Sigma} \subst{z}{y}{P} :: x:A$.
\end{enumerate}
\end{lemma}

\begin{proof}
The proof is by induction on the process term typing rules.
\end{proof}

The next lemma allows us to break up an open forest into two sub-forests, as illustrated in \Cref{apx:fig:lemma_d3}.

\begin{figure}
\centering
\includegraphics[scale=0.4]{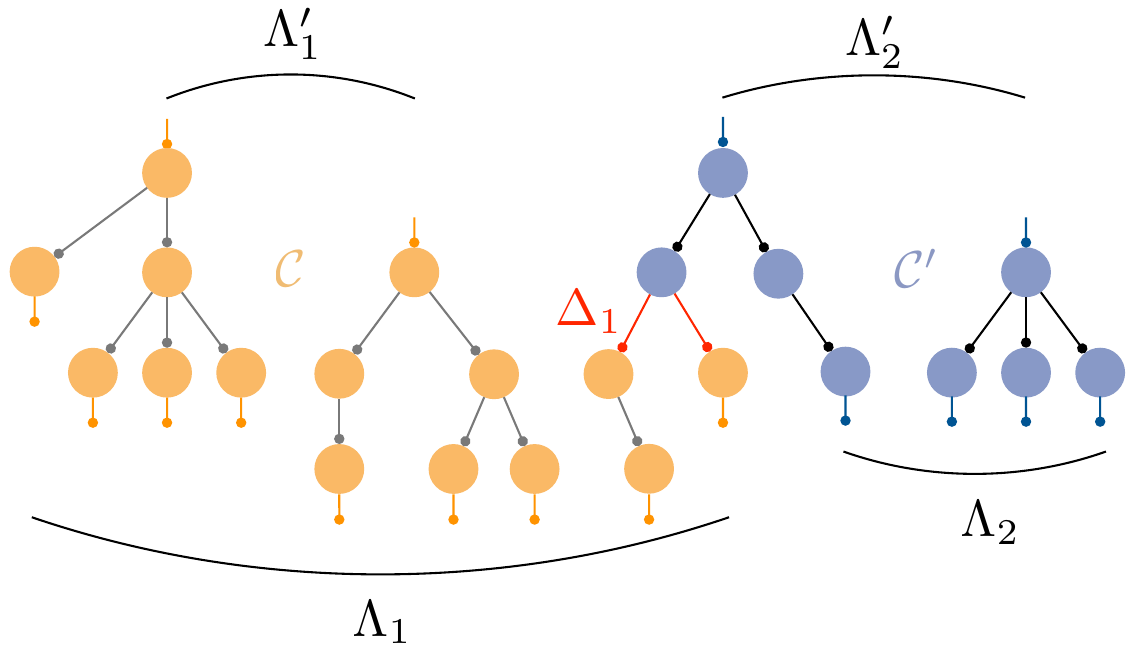}
\caption{Schematic illustration of \Cref{apx:lem:inv}, allowing us to break up an open forest into two open sub-forests.}
\label{apx:fig:lemma_d3}
\end{figure}

\begin{lemma}[Making two forests out of one]\label{apx:lem:inv}
 If ${\Delta} \Vdash {\mc{C}\,\mc{C}'}::{\Delta'}$, then for some $\Delta_1$ we have ${\Lambda_1} \Vdash {\mc{C}}::{\Lambda'_1, \Delta_1}$ and
 ${\Lambda_2,\Delta_1}\Vdash{\mc{C}'}::{\Lambda'_2}$, where $\Delta=\Lambda_1, \Lambda_2$ and $\Delta'=\Lambda'_1, \Lambda'_2$.
\end{lemma}
\begin{proof}
The proof is by a straightforward induction on the configuration typing rules.
\end{proof}

\subsubsection{Progress}
\begin{theorem}[Progress]

For any configuration $\mc{C}$, if $\mc{C}$ is a valid permutation of a configuration $\mc{C}''$ such that $ \Delta \Vdash \mc{C}'' :: \Delta'$,
then either $\mc{C} \mapsto \mc{C}'$ or $\mc{C}$ is poised. 
\end{theorem}
\begin{proof}
  The proof is standard and can be found in the literature of intuitionistic linear session types.
\end{proof}

\subsubsection{Preservation}
\label{apx:sec:progress-preservation:preservation}

\begin{theorem}[Preservation]
For any configuration $\mc{C}$ that is a valid permutation of a configuration $\mc{C}''$ such that $\Delta \Vdash \mc{C}'' :: \Delta'$,
if $\mc{C}\ostep \mc{C}'$, then $\mc{C}'$ is a valid permutation of a configuration $\mc{C}'''$ such that  $\Delta \Vdash \mc{C}''' :: \Delta'$.
\end{theorem}
\begin{proof}
  The proof is standard and can be found in the literature for intuitionistic linear session types.
\end{proof}

\subsection{IFC-typed processes}
\label{apx:sec:progress-preservation:ifc}

This section proves type safety of $\reflang$.

\begin{lemma}[Term variable substitution]\label{apx:lem:term-var-subst-ifc}
  The following substitutions are type-preserving and thus admissible:
  \begin{enumerate}
  \item\label{apx:lem:term-var-subst-ifc:right}
  If $\Psi; \Xi \vdash_{\Sigma} P@c :: x:A[c']$ with the tree invariant satisfied,
  then, for any fresh $y:A$, we have $\Psi; \Xi \vdash_{\Sigma} \subst{y}{x}{P}@c :: y:A[c']$
  with the tree invariant still satisfied.
  \item \label{apx:lem:term-var-subst-ifc:left}
  If $\Psi; \Xi, y : B[c''] \vdash_{\Sigma} P@c :: x:A[c']$ with the tree invariant satisfied,
  then, for any fresh $z:B$, we have $\Psi; \Xi, z : B[c'']  \vdash_{\Sigma} \subst{z}{y}{P}@c :: x:A[c']$
  with the tree invariant still satisfied.
  \end{enumerate}
  \end{lemma}

  \begin{lemma}[Making two forests out of one]\label{apx:lem:inv-ifc}
   If $\ctype{\Psi_0}{\Gamma}{\mathbb{C}\,\mathbb{C}'}{\Gamma'}$, then for some $\Gamma_1$ we have $\ctype{\Psi_0}{\Gamma''_1}{\mathbb{C}}{\Gamma'''_1, \Gamma_1}$ and
   $\ctype{\Psi_0}{\Gamma''_2,\Gamma_1}{\mathbb{C}'}{\Gamma'''_2}$, where $\Gamma=\Gamma''_1, \Gamma''_2$ and $\Gamma'=\Gamma'''_1, \Gamma'''_2$.
  \end{lemma}
\begin{proof}
The proof is a straightforward induction on the configuration typing rules.
\end{proof}

\begin{lemma}[Security variable substitution]\label{apx:lem:sec-var-subst}
If $\Psi; \Xi \vdash_{\Sigma} P@c :: x:B[d]$ with the tree invariant satisfied, then for every substitution $\Psi' \Vdash \gamma: \Psi$, we have
\[\hat{\gamma}(\Psi); \hat{\gamma}(\Xi) \vdash_{\Sigma} \hat{\gamma}(P)@\hat{\gamma}(c) :: x{:}B[\hat{\gamma}(d)],\]
with the tree invariant satisfied as well.
\end{lemma}

\begin{proof}
The proof is by induction on process term typing derivations $\Psi; \Xi \vdash_{\Sigma} P@c :: x:B[d]$.
We consider cases for the last step in the derivation. 
\begin{description}
 \item[Case 1.] \[\infer[\& R]{\Psi;\Xi \vdash_{\Sigma} (\mathbf{case}\, y^c(\ell \Rightarrow Q_\ell)_{\ell \in L})@d_1::  y:\&\{\ell:A_{\ell}\}_{\ell \in L}[c] }{\Psi; \Xi \vdash_{\Sigma} Q_k@c::  y:A_{k}[c] & \forall k \in L}\]
 By the induction hypothesis, for all $k \in L$, we have \[\hat{\gamma}(\Psi); \hat{\gamma}(\Xi) \vdash_{\Sigma} \hat{\gamma}(Q_k)@\hat{\gamma}(c)::  y{:}A_{k}[\hat{\gamma}(c)],\] which preserves the invariant. 
 
 By the $\&R$ rule, we get 
 \[\hat{\gamma}(\Psi);\hat{\gamma}(\Xi) \vdash_{\Sigma} \mathbf{case}\, y^{\hat{\gamma}(c)}(\ell \Rightarrow \hat{\gamma}(Q_\ell))_{\ell \in L})@\hat{\gamma}(d_1)::  y:\&\{\ell:A_{\ell}\}_{\ell \in L}[\hat{\gamma}(c)].\]
By the first part of the tree invariant for the original sequent, we get $\Psi \Vdash d_1 \sqsubseteq c $, and thus $\hat{\gamma}(\Psi) \Vdash \hat{\gamma}(d_1) \sqsubseteq \hat{\gamma}(c)$. The second part of the tree invariant is guaranteed by the induction hypothesis on the premise.
 
 By a simple rewrite according to how the lifting of $\gamma$ is defined we get
  \[\hat{\gamma}(\Psi);\hat{\gamma}(\Xi) \vdash_{\Sigma} \hat{\gamma}(\mathbf{case}\, y^{c}(\ell \Rightarrow Q_\ell)_{\ell \in L})@ \hat{\gamma}(d_1)::  y{:}\&\{\ell:A_{\ell}\}_{\ell \in L}[\hat{\gamma}(c)].\]
  
\item[Case 2.] \[\infer[\& L]{\Psi; \Xi, x:\&\{\ell:A_{\ell}\}_{\ell \in L}[c]\vdash_{\Sigma} (x^c.k; P)@d_1::  y:T[d] }{\deduce{\Psi;  \Xi, x:A_{k}[c] \vdash_{\Sigma}  P@d_1::  y:T[d]}{\Psi \Vdash d_1 \sqsubseteq c} & k \in L }\]

By the induction hypothesis on the first premise and definition of the lifting of $\gamma$ we have \[\hat{\gamma}(\Psi); \hat{\gamma}(\Xi), x{:}A_{k}[\hat{\gamma}(c)] \vdash_{\Sigma} \hat{\gamma}(P)@\hat{\gamma}(d_1)::  y{:}T[\hat{\gamma}(d)],\] which preserves the tree invariant.
By applying substitution on the second premise ($\Psi \Vdash d_1 \sqsubseteq c$) we get $\hat{\gamma}(\Psi) \Vdash \hat{\gamma}(d_1) \sqsubseteq \hat{\gamma}(c)$.
 
 By the $\&L$ rule, we get 
 \[\hat{\gamma}(\Psi);\hat{\gamma}(\Xi), x{:}\& \{A_{\ell}\}_{\ell \in L}[\hat{\gamma}(c)]\vdash_{\Sigma} x^{\hat{\gamma}(c)}.k;\hat{\gamma}(P) @\hat{\gamma}(d_1)::  y: T[\hat{\gamma}(d)],\]
 which satisfies the tree invariant since the premise satisfies it. 
Again by a simple rewrite according to how the lifting of $\gamma$ is defined we get
\[\hat{\gamma}(\Psi);\hat{\gamma}(\Xi, x{:}\& \{A_{\ell}\}_{\ell \in L}[c])\vdash_{\Sigma} \hat{\gamma}(x^{c}.k;P) @\hat{\gamma}(d_1)::  y: T[\hat{\gamma}(d)].\]

\item[Case 3.]  \[\infer[\otimes R]{\Psi;\Xi, z:A[d]\vdash_{\Sigma} (\mathbf{send}\,z^{d}\,y; P)@c::  y:(A[d] \otimes B)\,[d] }{\Psi;\Xi \vdash_{\Sigma}  P@c:: y: B[d]}\]

 By the induction hypothesis on the premise we have \[\hat{\gamma}(\Psi); \hat{\gamma}(\Xi) \vdash_{\Sigma} \hat{\gamma}(P)@\hat{\gamma}(c)::  y{:}B[\hat{\gamma}(d)],\]
which preserves the tree invariant and thus $\hat{\gamma}(\Psi) \Vdash \hat{\gamma}(d) \sqsubseteq \hat{\gamma}(d)$. 
 
 By the $\otimes R$ rule, we get 
 \[\hat{\gamma}(\Psi);\hat{\gamma}(\Xi), z{:} A[\hat{\gamma}(d)]\vdash_{\Sigma} \mathbf{send}\, z^{\hat{\gamma}(d)}\, y; \hat{\gamma}(P) @\hat{\gamma}(c)::  y: (A \otimes B)[\hat{\gamma}(d)],\]
 
the resulting sequent satisfies the tree invariant. 
Again by a simple rewrite according to how the lifting of $\gamma$ is defined we get
\[\hat{\gamma}(\Psi);\hat{\gamma}(\Xi, z{:} A[d])\vdash_{\Sigma} \hat{\gamma}(\mathbf{send}\, z^{d}\, y; P) @\hat{\gamma}(c)::  y: (A \otimes B)[\hat{\gamma}(d)].\]

\item[Case 4.]\[ \infer[\otimes L]{\Psi;\Xi, x:(A \otimes B)\,[d] \vdash_{\Sigma} (z^{d}\leftarrow \mathbf{recv}\, x; P)@c::  y:T[d_1]  }{\deduce{\Psi;\Xi, z:A[d], x:B[d] \vdash_{\Sigma}  P@c':: y: T[d_1]}{\Psi \Vdash c'= c\sqcup d }} \]

 By the induction hypothesis on the second premise and definition of the lifting of $\gamma$ we have
\[\hat{\gamma}(\Psi); \hat{\gamma}(\Xi), z{:}A[\hat{\gamma}(d)], x{:}B[\hat{\gamma}(d)] \vdash_{\Sigma} \hat{\gamma}(P)@\hat{\gamma}(c')::  y{:}T[\hat{\gamma}(d_1)],\]
which preserves the tree invariant.
Moreover, by applying the substitution on the first premise, we get
$\hat{\gamma}(\Psi)\Vdash \hat{\gamma}(c')= \hat{\gamma}(c) \sqcup \hat{\gamma}(d)$.
 
 By the $\otimes L$ rule, we get 
\[\hat{\gamma}(\Psi); \hat{\gamma}(\Xi), x{:}(A\otimes B)[\hat{\gamma}(d)] \vdash_{\Sigma} z^{\hat{\gamma}(d)} \leftarrow \mathbf{recv}\,x; \hat{\gamma}(P)@\hat{\gamma}(c)::  y{:}T[\hat{\gamma}(d_1)],\]
 
It is straightforward to observe that the resulting sequent satisfies the tree invariant and can be rewritten as before according to the definition of the lifting of $\gamma$. 

\item[Case 5.] Here is the interesting case!
\begin{mathpar}

\inferrule*[right=$\msc{Spawn}$]
  {\procdef{\Xi_1'}{X}{P}{\psi_0}{\omega_0}{x'}{A}{\psi}{\omega} \in \Sig \\
  \substmap{\Psi}{\delta_{\m{sec}}}{\Psi'} \\
  \latcstr{\substmapap{\delta_{\m{sec}}}{\psi}\sqsubseteq d} \\
  \latcstr{d_0 \sqsubseteq \substmapap{\delta_{\m{sec}}}{\psi_0}} \\
  \Xi_1, x{:}A[\substmapap{\delta_{\m{sec}}}{\psi_0}]\Vdash (\delta_{\m{sec}}, \delta_{\m{var}}) :: \Xi'_1, x'{:}A[\psi] \\
  \ptypb{\Xi_2, x{:}A\dlabel{\substmapap{\delta_{\m{sec}}}{\psi}}}{Q}{d_0}{f_0}{z}{C}{d}{f}}
  {\ptyp{\Xi_1, \Xi_2}{\refspwn{x}{\substmapap{\delta_{\m{sec}}}{\psi}}{}{X}{\Xi_1}{\substmapap{\delta_{\m{sec}}}{\psi_0}}{}{Q}}{d_0}{f_0}{z}{C}{d}{f}}
\end{mathpar}

By the induction hypothesis we have
\[\star\,\hat{\gamma}(\Psi); \hat{\gamma}(\Xi_2), x{:}A\dlabel{\substmapap{\gamma}{\substmapap{\delta_{\m{sec}}}{\psi}}} \vdash_{\Sigma} \hat{\gamma}(Q)@\hat{\gamma}(d_0)::  y{:}\hat{\gamma}(T)[\hat{\gamma}(d)],\] which preserves the tree invariants. By applying $\gamma$ on the 3rd and 4th premises, we get \[\hat{\gamma}(\Psi) \Vdash \hat{\gamma}(\hat{\delta_{\m{sec}}}(\psi)) \sqsubseteq \hat{\gamma}(d),\,  \hat{\gamma}(d_0) \sqsubseteq \hat{\gamma}(\hat{\delta_{\m{sec}}}(\psi_0)).\]
 
 By applying the substitution $\gamma$ on the 2nd premise, we get $\hat{\gamma}(\Psi) \Vdash \hat{\gamma}(\hat{\delta_{\m{sec}}}):: \Psi'$.
And from the 5th premise, we get 
\[\hat{\gamma}(Xi_1), x{:}A[\hat{\gamma}(\substmapap{\delta_{\m{sec}}}{\psi_0})]\Vdash (\hat{\gamma}(\delta_{\m{sec}}), \delta_{\m{var}}) :: \Xi'_1, x'{:}A[\psi]\]

 By the $\msc{Spawn}$ rule for the substitution $\gamma'=(\delta_{\m{var}},\gamma \circ \delta_{\m{sec}})$, we get 
 {\small\[\hat{\gamma}(\Psi);\hat{\gamma}(\Xi_1), \hat{\gamma}(\Xi_2)\vdash_{\Sigma} ((x^{[\hat{\gamma}(\hat{\delta_{\m{sec}}}(\psi))]} \leftarrow X[\gamma'] \leftarrow \hat{\gamma}(\Xi_1)@\hat{\gamma}(\hat{\delta_{\m{sec}}}(\psi_0)));\hat{\gamma}(Q)) @\hat{\gamma}(d_0)::  y: \hat{\gamma}(T)[\hat{\gamma}(d)].\]}
To show that this sequent satisfies the tree invariant, we use (a) the fact that the sequent $\star$ satisfies the tree invariant and (b) the condition on process definitions in the signature. 
The condition we imposed on process definitions ensures that $\Psi' \Vdash \psi_0 \sqsubseteq \psi$ and $\forall y{:}A[\psi_i] \in \Xi_1'. \Psi' \Vdash \psi_i \sqsubseteq \psi$.

By $\hat{\gamma}(\Psi) \Vdash \hat{\gamma}(\hat{\delta_{\m{sec}}}):: \Psi'$, we can rewrite the above judgment as

\[\forall y{:}A[\psi_i] \in \Xi_1'.\; \hat{\gamma}(\Psi) \Vdash \hat{\gamma}(\hat{\delta_{\m{sec}}}(\psi_i)) \sqsubseteq \hat{\gamma}(\hat{\delta_{\m{sec}}}(\psi)).\]
 By the tree invariant of $\star$ (i.e., $\hat{\gamma}(\Psi) \Vdash \hat{\gamma}(\hat{\delta_{\m{sec}}}(\psi)) \sqsubseteq \hat{\gamma}(d)$),

 \[\forall y{:}A[\psi_i] \in \Xi_1'.\; \hat{\gamma}(\Psi) \Vdash \hat{\gamma}(\hat{\delta_{\m{sec}}}(\psi_i)) \sqsubseteq \hat{\gamma}(\hat{\delta_{\m{sec}}}(\psi)) \sqsubseteq \hat{\gamma}(d).\]

 By $(\delta_{\m{var}}, \hat \gamma (\hat{\delta_{\m{sec}}})(\Xi'_1))=\hat{\gamma}(\Xi_1)$, 

 \[\forall y{:}\hat{\gamma}(A)[\hat{\gamma}(\psi_i)] \in \hat{\gamma }(\Xi_1).\; \hat{\gamma}(\Psi) \Vdash \hat{\gamma}((\psi_i)) \sqsubseteq \hat{\gamma}(\hat{\delta_{\m{sec}}}(\psi)) \sqsubseteq \hat{\gamma}(d).\]

 \item[Case 6.] The case for $\msc{d:fwd}$ is similar to the previous case.
\end{description}
\end{proof}

\subsubsection{Progress}\label{apx:sec:progress-preservation:progress}

Progress for IFC-typed processes in \reflang follows from the progress of \lang, as IFC typing can be viewed as a refinement typing using the same dynamics as the session typed processes.

\subsubsection{Preservation}\label{apx:sec:progress-preservation-ifc:preservation}

\begin{theorem}[Preservation]
For any configuration $\mathbb{C}$ that is a valid permutation of a configuration $\mathbb{C}''$ such that $\Psi_0; \Gamma \Vdash \mathbb{C}'' :: \Gamma'$,
if $\erasure{\mathbb{C}}\mapsto  \mc{C}'$, then there exists a security annotated program $\mathbb{C}'$ such that $\erasure{\mathbb{C}'}= \mathcal{C}'$ and is a valid permutation of a configuration $\mathbb{C}'''$ such that  $\Psi_0; \Gamma \Vdash \mathbb{C}''' :: \Gamma'$.
Moreover, the stepping preserves the tree invariant.
\end{theorem}

\begin{proof}
The proof is by considering different cases of $\erasure{\mathbb{C}}\mapsto \mathcal{C}'$. For each step we provide a security annotated configuration $\mathbb{C}'$ and then by inversion on the typing derivations show that it is IFC-typed. For the purpose of presentation, we put  security annotations of the post-steps for all possible steps in \Cref{apx:fig:dynamics-sec}. We provide the detailed proof by inversion for a couple of cases. The rest of the cases are similar.

\begin{figure*}
  \begin{center}
  \begin{small}
  \begin{tabbing}
  
  $\msc{Fwd} \qquad$ \= $\mb{proc}(y_\alpha[c], (y_\alpha \leftarrow x_\beta)@d_1) \ostep [x_\beta/y_\alpha]$  \` $(y_\alpha \not \in \Delta')$ \\[4pt]

  $\msc{Spawn}$ \= $\mb{proc}(y_\alpha[c],  (x^d \leftarrow X[\gamma] \leftarrow \Gamma_1)@d_2 ; Q@d_1)  \ostep$  \` $(\Psi'; \Gamma'_1 \vdash X = P@\psi' ::x': B[\psi] \in \Sigma)$ \\
\> $\mb{proc}(x_0[d], \hat{\gamma}(P)) \; \mb{proc}(y_\alpha[c],  \subst{x_0}{x}{Q} @d_1)$ \` $(\Gamma'_1, x'[\psi], \psi' \Vdash \gamma : \Gamma_1, x_0[d], d_2\, \, x_0\, \mi{fresh})$ \\[4pt]

$\msc{d:fwd}$ \= $\mb{proc}(y_\alpha[c],  F_Y[\gamma])  \ostep$  \` $(\psi=\psi; x'[\psi]:{C} \vdash F_y = \fwder{C,y'^\psi \leftarrow x'^\psi} @\psi  ::y': C[\psi] \in \Sigma)$ \\
\> $\mb{proc}(y_\alpha[c], \fwder{C, y^c_\alpha \leftarrow x^c_\beta} @ c)$ \` $( x'[\psi], y'[\psi] \Vdash \gamma : x_\beta[c],y_\alpha[c])$\\[4pt]

  $1_{\m{snd}}$ \> $\mb{proc}(y_\alpha[c],(\mathbf{close}\, y_\alpha)@d_1) \ostep \mb{msg}(\mathbf{close}\, y_\alpha)$ \\[4pt]
  
  $1_{\m{rcv}}$ \> $\mb{msg}(\mathbf{close}\, y_\alpha)\; \mb{proc}(x_\beta[{c'}], (\mathbf{wait}\,y_\alpha;Q)@d_1) \ostep \mb{proc}(x_\beta[{c'}], Q@(d_1 \sqcup c))$ \\[4pt]
  
  $\oplus_{\m{snd}}$ \>$\mb{proc}(y_\alpha[c], y_\alpha.k; P@d_1) \ostep \mb{proc}(y_{\alpha+1}[c], ([y_{\alpha+1}/y_{\alpha}]P)@d_1)\; \mb{msg}( y_\alpha.k)$ \\
  
  $\oplus_{\m{rcv}}$ \> $\mb{msg}(y_\alpha[{c}].k) \; \mb{proc}(u_\gamma[{c'}],\mb{case}\, y_\alpha (( \ell \Rightarrow P_\ell)_{\ell \in L})@d_1) \ostep \mb{proc}(u_{\gamma}[{c'}], ([y_{\alpha+1}/y_{\alpha}] P_k)@(d_1\sqcup c))$ \\[4pt]
  
  $\&_{\m{snd}}$ \> $\mb{proc}(y_\alpha[c], (x_\beta.k; P)@d_1) \ostep \mb{msg}( x_\beta.k)\; \mb{proc}(y_{\alpha}[c], ([x_{\beta+1}/x_{\beta}]P)@d_1)$ \\[4pt]
  
  $\&_{\m{rcv}}$ \> $\mb{proc}(y_\alpha[c],(\mb{case}\, y_\alpha ( \ell \Rightarrow P_\ell)_{\ell \in L})@d_1) \; \mb{msg}(y_\alpha.k) \ostep \mb{proc}(v_{\delta}[c],([y_{\alpha+1}/y_{\alpha}] P_k)@c)$ \\[4pt]
  
  $\otimes_{\m{snd}}$ \> $\mb{proc}(y_\alpha[c],(\mathbf{send}\,x_\beta\,y_\alpha; P)@d_1) \ostep \mb{proc}(y_{\alpha+1}[c], ([y_{\alpha+1}/y_{\alpha}]P)@d_1)\; \mb{msg}( \mathbf{send}\,x_\beta\,y_\alpha)$ \\[4pt]
  
  $\otimes_{\m{rcv}}$ \> $\mb{msg}(\mathbf{send}\,x_\beta\,y_\alpha) \; \mb{proc}(u_\gamma[{c'}],(w\leftarrow \mathbf{recv}\,y_\alpha; P)@d_1) \ostep \mb{proc}(u_{\gamma}[{c'}], ([x_\beta/w][y_{\alpha+1}/y_{\alpha}] P)@(d_1\sqcup c))$ \\[4pt]
  
  $\multimap_{\m{snd}}$ \> $\mb{proc}(y_\alpha[c],(\mathbf{send}\,x_\beta\,u_\gamma; P)@d_1) \ostep \mb{msg}( \mathbf{send}\,x_\beta\,u_\gamma)\; \mb{proc}(y_{\alpha}[c], ([u_{\gamma+1}/u_{\gamma}]P)@d_1)$ \\[4pt]
  
  $\multimap_{\m{rcv}}$ \> $\mb{proc}(y_\alpha[c],(w\leftarrow \mathbf{recv}\,y_\alpha; P)@d_1) \; \mb{msg}(\mathbf{send}\,x_\beta\,y_\alpha) \ostep \mb{proc}(v_\delta[c], ([x_\beta/w][y_{\alpha+1}/y_{\alpha}] P)@ c)$
  \end{tabbing}
  \caption{Annotated asynchronous dynamics--proof of preservation.}
  \label{apx:fig:dynamics-sec}
  \end{small}
  \end{center}
  \end{figure*}

\begin{description}

\item{\bf Case 1.} ($\otimes$)
\begin{description}
 \item{\bf Subcase 1.} (send)
   \[ \mathbb{C}_1 \mb{proc}(y_\alpha[c],(\mathbf{send}\,x_\beta^{c}\,y_\alpha; P)@d_1)\mathbb{C}_2 \ostep \mathbb{C}_1 \mb{proc}(y_{\alpha+1}[c], ([y_{\alpha+1}/y_{\alpha}]P)@d_1) \mb{msg}( \mathbf{send}\,x_\beta^{c}\,y_\alpha) \mathbb{C}_2 \]

{\bf By assumption of the theorem:} $\ctype{\Psi_0}{\Gamma}{\mathbb{C}_1 \mb{proc}(y_\alpha[c],(\mathbf{send}\,x_\beta^{c}\,y_\alpha; P)@d_1)\mathbb{C}_2}{\Gamma'}$.

{\bf By \Cref{apx:lem:inv-ifc}:} $\ctype{\Psi_0}{\Gamma^1_1}{\mathbb{C}_1}{\Gamma_1, \Gamma_2, \Gamma^1{'}_1}$ and $\ctype{\Psi_0}{\Gamma^{1''}_2,\Gamma'_1, x_\beta{:}A[c] }{\mb{proc}(y_\alpha[c],(\mathbf{send}\,x_\beta^{c}\,y_\alpha; P)@d_1)}{y_\alpha{:}(A\otimes B)[c]}$ and $\ctype{\Psi_0}{\Gamma^1_3, \Gamma_2}{\mathbb{C}_2}{\Gamma^{1'}_2}$.

Where $\Gamma=\Gamma^1_1,\Gamma^1_2,\Gamma^1_3$ and $\Gamma'=\Gamma^{1'}_1, \Gamma^{1'}_2$. If $x_\beta{:}A[c] \in \Gamma$ we have $\Gamma'_1= \Gamma_1$ and $\Gamma^{1''}_2, x_\beta{:}A[c] =\Gamma^1_2$,  and otherwise $\Gamma'_1, x_\beta{:}A[c] = \Gamma_1$ and $\Gamma^{1''}_2=\Gamma^1_2$. 

{\bf By inversion on $\mb{proc}$ rule}
$\ptype{\Psi_0}{\Gamma^{1''}_2,\Gamma'_1, x_\beta{:}A[c]}{(\mathbf{send}\,x_\beta^{c}\,y_\alpha^{c};P)@d_1}{y}{c}{\alpha}{(A\otimes B)}$. Moreover, $(\star)\; \; \forall u_\gamma{:}T[d_2] \in \Gamma^{1''}_2, \Gamma'_1, x_\beta{:}A[c]. \, \Psi_0 \Vdash d_2 \sqsubseteq c$. 

{\bf By inversion on $\otimes\, R$  rule}
$\ptype{\Psi_0}{\Gamma^{1''}_2,\Gamma'_1}{P@d_1}{y}{c}{\alpha}{B}$.

{\bf By substitution of $y_{\alpha+1}$ for $y_{\alpha}$:}
\[\ptype{\Psi_0}{\Gamma^{1''}_2,\Gamma'_1}{[y^{\alpha+1}/y^{\alpha}]P@d_1}{y}{c}{\alpha+1}{B}.\]

{\bf By $\m{proc}$ rule and $(\star)$:} 
$\dagger\; \ctype{\Psi_0}{\Gamma^{1''}_2, \Gamma'_1}{\mb{proc}(y_{\alpha+1}[c], [y^{\alpha+1}/y^{\alpha}]P@d_1)}{y_{\alpha+1}{:}B[c]}.$

{\bf By $\otimes R$ and $\m{fwd}$ rule:}
$\ptype{\Psi_0}{x_\beta{:}A[c], y_{\alpha+1}{:}B[c]}{(\mathbf{send}\,x_\beta^{c}\,y_\alpha^{c};y_{\alpha}^{c}\leftarrow y_{\alpha+1}^{c})@c}{y}{c}{\alpha}{(A\otimes B)}$.

{\bf By $\m{msg}$, $\star$, and $\dagger$:}
\[\ptype{\Psi_0}{\Gamma^{1''}_2, \Gamma_1', x_\beta{:}A[c]}{ \mb{proc}(y_{\alpha+1}[c], [y^{\alpha+1}/y^{\alpha}]P@d_1) \mb{msg}(\mathbf{send}\,x_\beta^{c}\,y_\alpha^{c};y_{\alpha}^{c}\leftarrow y_{\alpha+1}^{c})}{y}{c}{\alpha}{(A\otimes B)}.\]

{\bf By configuration typing rules:}
$\ctype{\Psi_0}{\Gamma}{\mc{C}_1\mb{proc}(y_{\alpha+1}[c], [y^{\alpha+1}/y^{\alpha}]P@d_1) \mb{msg}(\mathbf{send}\,x_\beta^{c}\,y_\alpha^{c};y_{\alpha}^{c}\leftarrow y_{\alpha+1}^{c}) \mc{C}_2}{\Gamma'}$  
   
\end{description}
\end{description}
\end{proof}

{\bf Remark}: To keep the proofs concise we may use $\Psi$ instead of the term secrecy lattice $\Psi_0$, whenever we are clearly working with configurations defined in the run-time. Moreover, from now on, we write  $\Psi_0; \Delta \Vdash \mc{C} :: \Delta'$ also when $\mc{C}$ is a valid permutation of a configuration $\mc{C}'$ such that $\Psi_0; \Delta \Vdash \mc{C}' :: \Delta'$.

\section{Session logical relation}
\label{apx:sec:logical-relation}
This section introduces the session logical relation and supporting definitions.

\subsection{Open configuration transition}\label{apx:sec:logical-relation:defs}

\begin{definition}
    The configuration $\Delta \Vdash \mc{C}:: \Delta'$ is ready to send along $\Lambda \subseteq \Delta, \Delta'$ iff for all $y_\alpha{:}C \in \Lambda$, there is a message $\mathbf{msg}(M)$ in $\mc{C}$ sending along $y_\alpha$, i.e. $M$ is of the form $y_\alpha.k$, $\mathbf{send}\,x_\beta\, y_\alpha$, or $\mathbf{close}\, y_\alpha$.

    \noindent
    Similarly, the configuration $\Delta \Vdash \mc{C}:: \Delta'$ is ready to receive along $\Lambda \subseteq \Delta, \Delta'$ iff for all $y_\alpha{:}C \in \Lambda$, there is a process $\mathbf{proc}(z_\delta, P)$ in $\mc{C}$ waiting to receive along $y_\alpha$, i.e. $P$ is of the form $\mathbf{case}\,y_\alpha(\ell \Rightarrow Q_\ell)_{\ell \in I}$, $w \leftarrow \mathbf{recv} y_\alpha;Q$, or $\mathbf{wait}\,y_\alpha; Q$.

    \end{definition}
\begin{definition}
    The set $\mathbf{Out}(\Delta \Vdash K)$, is defined as all channels with the sending semantics in $\Delta, K$, i.e., $y_\alpha \in \mathbf{Out}(\Delta \Vdash K)$ iff for some positive type $T$, we have $y_\alpha{:}T \in \Delta$ or for some negative type $T$, we have $y_\alpha{:}T \in K$.

\noindent
    Similarly, the set $\mathbf{In}(\Delta \Vdash K)$, is defined as all channels with the receiving semantics in $\Delta, K$, i.e., $y_\alpha \in \mathbf{Out}(\Delta \Vdash K)$ iff for some negative type $T$, we have $y_\alpha{:}T \in \Delta$ or for some positive type $T$, we have $y_\alpha{:}T \in K$.
\end{definition} 

\begin{definition}
    $\mathit{dom}(\Delta)$ is a set defined inductively as:
    \[\begin{array}{lll}
        \mathit{dom}(\Delta, y_{\alpha}:A) &=& \mathit{dom}(\Delta) \cup \{y\}\\
        \mathit{dom}(\cdot) &= & \emptyset
    \end{array}  
    \]
\end{definition}
    \begin{definition}[Open configuration transitions]\label{apx:def:config-transitions}
    We provide some notations used in the logical relation and proofs.
    
    \begin{itemize}
    \item The notation $\mapsto^*$ refers to taking none or many steps with $\mapsto$. The notation $\mapsto^j$ refers to taking $j$ steps with $\mapsto$. 
    
    \item We write $\mc{D}\mapsto^{{*}_{\Upsilon}}\mc{D}'$  stating that $\mc{D} \mapsto^*\mc{D}'$ and $\mc{D}'$ is ready to send along $\Upsilon$.
        
    \item We write $\mc{D}\mapsto^{{*}_{\Upsilon; \Theta}}\mc{D}'$  stating that $\mc{D} \mapsto^*\mc{D}'$ and $\mc{D}'$ is ready to send along $\Upsilon$ and ready to receive along $\Theta$.
   
\end{itemize}
    \end{definition}

\subsection{Session Logical Relation}\label{apx:sec:logical-relation:relation}

\Cref{apx:fig:logical-relation} defines the logical relation for intuitionistic linear logic session types with general recursive types.

\begin{figure}
\input{figs/logical-relation.tex}
\caption{Session logical relation}
\label{apx:fig:logical-relation}
\end{figure}

\section{Noninterference}
\label{apx:sec:noninterference}

This section introduces the fundamental theorem for progress-sensitive noninterference for \reflang.
We first start with defining supporting notions and proofs of necessary lemmas.

\subsection{Up-to equivalence, projections, and splitting up closed configuration}\label{apx:sec:noninterference:equiv-proj}

The fundamental theorem \Cref{apx:thm:ftlr} is stated in terms of the logical equivalence
{\small$(\Delta_1 \Vdash \mc{D}_1:: x_\alpha {:}A_1[{c_1}]) \equiv^{\Psi_0}_{\xi} (\Delta_2 \Vdash \mc{D}_2:: y_\beta {:}A_2[{c_2}])$},
expressing that two open configurations $\mc{D}_1$ and $\mc{D}_2$ are being related by the term interpretation
when plugged into arbitrary closing contexts and observed for any number of message exchanges {\small$m$}.
Next we define this equivalence as well as typing context projections, on which the former relies.

\begin{definition}[Typing context projections]\label{apx:def:typing-proj}
Downward projection on security linear contexts and $K^s$ is defined as follows:
\[\begin{array}{lclc}
\Gamma, x_\alpha{:}T[c] \Downarrow \xi& \defeq& \Gamma \Downarrow \xi,  x_\alpha{:}T[c] & \m{if}\; c\sqsubseteq \xi  \\
     \Gamma, x_\alpha{:}T[c] \Downarrow \xi& \defeq& \Gamma \Downarrow \xi & \m{if}\; c \not \sqsubseteq \xi  \\
    \cdot \Downarrow \xi& \defeq& \cdot &  \\
     x_\alpha{:}T[c] \Downarrow \xi& \defeq&   x_\alpha{:}T[c] & \m{if}\; c\sqsubseteq \xi  \\
    x_\alpha{:}T[c] \Downarrow \xi& \defeq& \_{:}1[\top] & \m{if}\; c\not\sqsubseteq \xi \\
\end{array}\]

\end{definition}\label{apx:def:highprovider}

\begin{definition}[High provider and High client]
   \[\begin{array}{lll}
     \cdot \,\,\in \mb{H\text{-}Provider}^\xi(\cdot) \\[5pt]
     \mc{B} \in \mb{H\text{-}Provider}^{\xi}(\Gamma, x_\alpha{:}A[c])\,  & \m{iff} & c \not \sqsubseteq \xi  \, \m{and}\,
     \mc{B} = \mc{B}' \mc{T} \,\, \m{and}\,\, \mc{B}' \in \mb{H\text{-}Provider}^\xi(\Gamma)\,\,\m{and}\\
    &&\mc{T} \in \m{Tree}(\cdot \Vdash x_\alpha{:}A), \mb{or}\\
     &&c \sqsubseteq \xi  \, \m{and}\,
    \mc{B} \in \mb{H\text{-}Provider}^\xi(\Gamma)\\[10pt]
 
     \mc{T} \in \mb{H\text{-}Client}^\xi( x_\alpha{:}A[c])\,  & \m{iff} & c \not \sqsubseteq \xi  \, \m{and}\, \mc{T} \in \m{Tree}(x_\alpha{:}A \Vdash \_:1), \mb{or}\\
     &&c \sqsubseteq \xi  \, \m{and}\, \mc{B}= \cdot
   \end{array}\]
 \end{definition}

\begin{definition}[Equivalence of trees by the logical relation upto the observer level]\label{apx:def:noninterference}
We define the relation \[(\Gamma_1 \Vdash \mc{D}_1:: x_\alpha {:}A_1[{c_1}]) \equiv^{\Psi_0}_{\xi} (\Gamma_2 \Vdash \mc{D}_2:: y_\beta {:}A_2[{c_2}])\,\,\m{as}\]
{\small\[
\begin{array}{lll}
  {\mc{D}_1} \in \m{Tree}(\erasure{\Gamma_1} \Vdash x_\alpha {:}A_1)\,\, \m{and}\,\,{\mc{D}_2} \in \m{Tree}(\erasure{\Gamma_2} \Vdash y_\beta {:}A_2)\,\, \m{and}\\
  \Gamma_1 {\Downarrow} \xi= \Gamma_2 {\Downarrow} \xi=\Gamma\,\,\m{and}\,\, x_\alpha {:}A_1[{c_1}]{\Downarrow} \xi= y_\beta {:}A_2[{c_2}] {\Downarrow} \xi= K^s\,\,\m{and}\\
  \forall \, \mc{B}_1 \in \mb{H\text{-}Provider}^\xi(\Gamma_1).\, \forall \, \mc{B}_2 \in \mb{H\text{-}Provider}^\xi(\Gamma_2). \,\forall \mc{T}_1 \in \mb{H\text{-}Client}^\xi(x_\alpha {:}A_1[{c_1}]). \,\forall \mc{T}_2 \in \mb{H\text{-}Client}^\xi(y_\beta {:}A_2[{c_2}]).
\end{array}  
\]
}
\[\begin{array}{l}  \forall\,m.\,(\mc{B}_1\mc{D}_1\mc{T}_1, \mc{B}_2\mc{D}_2\mc{T}_2) \in \mc{E}\llbracket \erasure{\Gamma} \Vdash \erasure{K^s} \rrbracket^{m},\,
\m{and} \\[1pt]
\forall\,m.\,(\mc{B}_2\mc{D}_2\mc{T}_2, \mc{B}_1\mc{D}_1\mc{T}_1) \in \mc{E}\llbracket \erasure{\Gamma} \Vdash \erasure{K^s} \rrbracket^{m}.
\end{array}
\]
\end{definition}

\subsection{Quasi-running secrecy and relevant nodes}\label{apx:sec:noninterference:quasi-and-relevant}

The proof of the fundamental theorem \Cref{apx:thm:ftlr} relies on the notion of a relevant node, maintaining the invariant that the relevant nodes of the two program runs execute the same code.
The latter is guaranteed by \Cref{apx:lem:invariant}.
Next we define the notion of a relevant node, expressed locally for an asynchronous semantics.
The notion of a relevant node relies on the notion of quasi-running secrecy, also defined below.

\begin{definition}[Quasi-running secrecy]\label{apx:def:quasi}
In an open configuration $\Psi_0; \Gamma \Vdash \mathbb{C} :: \Gamma'$, the quasi-running secrecy of a message or process is determined
by its running secrecy, its process term, and the running secrecy of its parent. 
\begin{itemize}
    \item If the node is a process with a process term other than $\mb{recv}$ or $\mb{case}$, then its quasi-running secrecy is equal to its running secrecy.
    \item If the process term is of the form $\mb{case}\,y^c_{\alpha}(\ell \Rightarrow P_\ell)_{\ell \in L} @d_1$ or $x \leftarrow \mb{recv}\,y^c_{\alpha};P_{x} @d_1$, then its quasi-running secrecy is $d_1 \sqcup c$.
    \item If the node is a message of a negative type along channel $y^c_\alpha$, its quasi-running secrecy is $c$.
    \item  If the node is a message of a positive type along channel $y^c_\alpha$ and it has a parent with quasi-running secrecy $d_1$, its quasi-running secrecy is $d_1 \sqcup c$, otherwise its quasi-running secrecy is $c$.
\end{itemize}
The quasi-running secrecy can be determined by traversing the tree top to bottom.

\end{definition}

\begin{definition}[Relevant node]\label{apx:def:relevant-node}
Consider a configuration $\Gamma \Vdash \mathbb{D}:: K^s$ and observer level $\xi$.
A channel is relevant in $\mathbb{D}$ if
(1) it has a maximal secrecy level less than or equal to $\xi$, and
(2) it is either an observable channel or it shares a process or message with quasi-running secrecy less than or equal to $\xi$ with a relevant channel.
(A channel shares a process with another channel if they are siblings or one is the parent of another.)
A process or message is relevant if
(1) it has quasi-running secrecy less than or equal to $\xi$, and
(2) it has at least one relevant channel.
We denote the relevant processes and messages (\ie nodes) in $\mathbb{D}$ by $\cproj{\mathbb{D}}{\xi}$.
We write $\cproj{\mathbb{D}_1}{\xi}=_{\xi}\cproj{\mathbb{D}_2}{\xi}$ if the relevant nodes in $\mathbb{D}_1$ are identical to those in $\mathbb{D}_2$ up to renaming of channels with higher or incomparable secrecy than the observer.

\end{definition}

We can build the set of all relevant nodes in a configuration by traversing the tree bottom-up as explained in~\Cref{apx:sec:noninterference:quasi-and-relevant}.
\Cref{apx:def:relevant-node} provides us with a local guide to identify whether a process is relevant or not.
We note that if $K^s$ is observable, then by the tree invariant, every channel in $\mathbb{D}$ is relevant.

Now we can state and prove the invariant that the relevant nodes of the two program runs execute the same code.
The proof of the fundamental theorem \Cref{apx:thm:ftlr} crucially relies on this lemma.

\begin{lemma}[Keeping relevant nodes in sync]\label{apx:lem:invariant}
Consider $\ctype{\Psi_0}{\Gamma}{\mathbb{D}_i}{K^s}$ for $i \in \{1,2\}$ with the relevant nodes in $\mathbb{D}_1$ and $\mathbb{D}_2$ are identical, i.e.,
$\cproj{\mathbb{D}_1}{\xi}=\cproj{\mathbb{D}_2}{\xi}$, with $\Gamma \Downarrow \xi = \Gamma$ and $K^s \Downarrow \xi = K^s$.
If $\erasure{\mathbb{D}_1} \mapsto \erasure{\mathbb{D}'_1}$, then there exists a $\mathbb{D}'_2$ such that $\erasure{\mathbb{D}_2} \mapsto^{0,1} \erasure{\mathbb{D}'_2}$, i.e., $\erasure{\mathbb{D}_2}$ steps to $\erasure{\mathbb{D}'_2}$ in zero or one step, and $\cproj{\mathbb{D}_1}{\xi}=\cproj{\mathbb{D}_2}{\xi}$.
\end{lemma}

\begin{proof}
The proof is by cases on the possible  steps of ${\erasure{\mathbb{D}_1}}  \mapsto {\erasure{\mathbb{D}^1_1}}$. 
 In each case we prove that either the step does not change relevancy of any process in $\mathbb{D}_1$ or we can step $\erasure{\mathbb{D}_2}$ such that the same change of relevancy occurs in $\mathbb{D}_2$ too. Note that  in all cases, we get $({\mathbb{D}^1_1}; {\mathbb{D}^1_2}) \in \m{Tree}(\erasure{\Gamma} \Vdash \erasure{K^s})$ by the preservation of session-typed processes.\\ In the following cases, we annotate the post-step of the configuration to reflect the annotations required by the preservation of IFC-typed configurations.

 {\color{ForestGreen} \bf Case 1. $\mathbb{D}_1=\mathbb{D}'_1\mb{proc}(y_\alpha[c],  y^c_\alpha.k ;P @d_1) \mathbb{D}''_1$ and
\[ \erasure{\mathbb{D}'_1\mb{proc}(y_\alpha[c],  y^c_\alpha.k ;P @d_1) \mathbb{D}''_1}\mapsto \erasure{\mathbb{D}'_1\mb{proc}(y_{\alpha+1}[c], [y^c_{\alpha+1}/y^c_{\alpha}] P@d_1) \mb{msg}(y^c_\alpha.k) \mathbb{D}''_1}\]}
where $\mathbb{D}^1_1= \mathbb{D}'_1\mb{proc}(y_{\alpha+1}[c], [y^c_{\alpha+1}/y^c_{\alpha}] P @d_1) \mb{msg}(y^c_\alpha.k) \mathbb{D}''_1$.
We consider subcases based on relevancy of process offering along $y_\alpha[c]$: 
\begin{description}
\item {\bf Subcase 1. $\mb{proc}(y_\alpha[c],  y^c_\alpha.k ;P @d_1)$ is not relevant.} By inversion on the typing rules $d_1 \sqsubseteq c$. By definition either $d_1 \not\sqsubseteq \xi $ or none of the channels connected to $P$ including its offering channel $y_\alpha^c$ are relevant. In both cases neither $\mb{proc}(y_{\alpha+1}[c], [y^c_{\alpha+1}/y^c_{\alpha}] P@d_1)$, nor  $\mb{msg}(y^c_\alpha.k)$ are relevant in the post step. Note that from $d_1 \sqsubseteq c$ and $d_1 \not\sqsubseteq \xi$, we get $c \not\sqsubseteq \xi $. Channel $y^c_\alpha$ is not relevant in the pre-step, and both $y^c_\alpha$ and $y^c_{\alpha+1}$ are not relevant in pre-step and post-step configurations. Every not relevant resource of $\mb{proc}(y^c_\alpha,  y^c_\alpha.k ;P @d_1)$ will remain irrelevant in the post-step too.

In this subcase, it is enough to show 
\[\cproj{\mathbb{D}'_1\mb{proc}(y_{\alpha+1}[c], [y^c_{\alpha+1}/y^c_{\alpha}] P @d_1) \mb{msg}(y^c_\alpha.k) \mathbb{D}''_1}{\xi} \meq \cproj{\mathbb{D}'_1 \mathbb{D}''_1}{\xi}\meq \cproj{\mathbb{D}_1}{\xi}\meq \cproj{\mathbb{D}_2}{\xi}.\]
To prove this we need two observations:
\begin{itemize}
    \item Neither $\mb{proc}(y_{\alpha+1}[c], [y^c_{\alpha+1}/y^c_{\alpha}] P@d_1)$ nor  $\mb{msg}(y^c_\alpha.k)$ are relevant and they will be dismissed by the projection. (As explained above.)
    \item Replacing $\mb{proc}(y_\alpha[c],  y^c_\alpha.k ;P @d_1)$ with these two nodes, does not affect relevancy of the rest of processes in $\mathbb{D}'_1\mathbb{D}''_1$. Relevancy of processes in $\mathbb{D}''_1$ remains intact since $y^c_\alpha$ and $y^c_{\alpha+1}$ are irrelevant. 
    
    The relevancy of processes in $\mathbb{D}'_1$ remains intact too as we replace their irrelevant root with another irrelevant process. However, we need to be careful about the changes in the quasi-running secrecy of a process and their effect on its (grand)children. The quasi-running secrecy of the process offering along $y^c_{\alpha+1}$ may be higher or incomparable to $d_1$ based on the code of $P$ (if it starts with a $\mb{recv}$ or $\mb{case}$). This is of significance only if $d_1 \sqsubseteq \xi$, and in the pre-step the process has a relevant channel $x:\_[d]$ as its resource  where $d\sqsubseteq \xi$.
    But by the assumption of the subcase, either $d_1 \not \sqsubseteq \xi$ or  $x:\_[d]$ is not a relevant channel in the pre-step. 
\end{itemize}

This completes the proof.

\item {\bf Subcase 2. $\mb{proc}(y_\alpha[c],  y^c_\alpha.k ;P @d_1)$ is relevant.}
By assumption ($\cproj{\mathbb{D}_1}{\xi}\meq\cproj{\mathbb{D}_2}{\xi}$) we have :
\[\mathbb{D}_2=\mathbb{D}'_2\mb{proc}(y_\alpha[c],  y^c_\alpha.k ;P@ d_1) \mathbb{D}''_2.\]
and
{\color{red} \[\mathbb{D}_2\mapsto \mathbb{D}'_2\mb{proc}(y_{\alpha+1}[c], [y^c_{\alpha+1}/y^c_\alpha]P @d_1) \mb{msg}(y^c_\alpha.k)\mathbb{D}''_2\]}
and $\mathbb{D}^1_2= \mathbb{D}'_2\mb{proc}(y_{\alpha+1}[c], [y^c_{\alpha+1}/y^c_\alpha]P @d_1) \mb{msg}(y^c_\alpha.k; )\mathbb{D}''_2$.
 
We need to show: 
{\small\[\cproj{\mathbb{D}'_1\mb{proc}(y_{\alpha+1}[c], [y^c_{\alpha+1}/y^c_{\alpha}] P @d_1) \mb{msg}(y^c_\alpha.k) \mathbb{D}''_1}{\xi} \meq \cproj{\mathbb{D}'_2\mb{proc}(y_{\alpha+1}[c], [y^c_{\alpha+1}/y^c_\alpha]P @d_1) \mb{msg}(y^c_\alpha.k)\mathbb{D}''_2}{\xi}.\]}

\begin{itemize}
\item If $c\not\sqsubseteq \xi$, then $\mb{msg}(y^c_\alpha.k)$  is not relevant in both runs, and will be dismissed by the projections. Moreover, in this case $y^c_\alpha$ is not relevant in the pre-step and post-step configurations. Thus the relevancy of processes in $\mathbb{D}''_1$ and $\mathbb{D}''_2$ will remain intact.

\item If $c \sqsubseteq \xi$, then $y^c_\alpha$ is relevant in the pre-step in both runs. We need to consider the possibility of change in quasi-running secrecy in the post-step. The quasi-running secrecy of the processes offering along $y^c_{\alpha+1}$ may increase based on their code (if the code of $P$ starts with a $\mb{recv}$ or $\mb{case}$). But in this case, it cannot become irrelevant since by the tree invariant it is always bounded by the max secrecy of the offering channel, $c \sqsubseteq \xi$. 
Thus, in the post-step of both runs, $y^c_{\alpha+1}$ is relevant. 
Relevancy of message $\mb{msg}(y^c_\alpha.k)$ in the post-steps of the first and second run is determined by the quasi-running secrecies ($d$ and $d'$) of their parents ($X$ and $X'$) in $\mathbb{D}''_1$ and $\mathbb{D}''_2$. If $d\sqsubseteq \xi$, then the parent ($X$) is relevant in the first run and by assumption is identical to a relevant $X'$ in the second run. Thus messages $\mb{msg}(y^c_\alpha.k)$ are relevant in both runs and $y^c_\alpha$ is relevant in the post-step too. The same holds when $d'\sqsubseteq \xi$. Otherwise, in both runs the quasi-running secrecy of the parent is higher than or incomparable to the observer (the parents are both irrelevant). Thus messages $\mb{msg}(y^c_\alpha.k)$ are not relevant in the post step of both runs, and will be dismissed by the projections. The channel $y^c_\alpha$ will be irrelevant in the post-step of both runs too. However, this does not affect the processes in $\mathbb{D}''_1$ and $\mathbb{D}''_2$ as the parents of messages ($X$ and $X'$) are already irrelevant in the pre-step.
\end{itemize}

Finally, we need to show that projections of $\mathbb{D}'_1$ and $\mathbb{D}'_2$ are equal in the post-step too. 

Consider the resources of the process that offeres along $y_\alpha[c]$. By our writing convention, they are all in $\mathbb{D}'_1$ and $\mathbb{D}'_2$: (i) Those resources offered by $\mathbb{D}'_1$ and $\mathbb{D}'_2$ with max secrecies higher than or incomparable to the observer $\xi$  in the pre-step will remain higher than or incomparable to the observer and thus irrelevant in the post-step too. (ii) Those resources with max secrecies lower than the observer $\xi$ in the pre-step, i.e., a sub-tree $\mathbb{T}_i$ in $\mathbb{D}'_i$ offering along a channel $w[c']$ where $c' \sqsubseteq \xi$, are relevant in the pre-step and will remain  relevant in the post-step (again, in this case, by the tree invariant, the quasi-running secrecy of the processes offering along $y^c_{\alpha+1}$ cannot increase in the post-step).

\end{description}

 {\color{ForestGreen}\bf Case 2. $\mathbb{D}_1=\mathbb{D}'_1\mb{proc}(x_\beta[d],  y^c_\alpha.k ;P @d_1) \mathbb{D}''_1$ and
 \[ \mathbb{D}'_1\mb{proc}(x_\beta[d],  y^c_\alpha.k ;P @d_1) \mathbb{D}''_1\mapsto \mathbb{D}'_1\mb{msg}(y^c_\alpha.k)\mb{proc}(x_\beta[d], [y^c_{\alpha+1}/y^c_{\alpha}]P @d_1)  \mathbb{D}''_1\]}
We consider subcases based on relevancy of the process offering along $x_d^\beta$: 
\begin{description}
\item {\bf Subcase 1. $\mb{proc}(x_\beta[d],  y^c_\alpha.k ;P @d_1)$ is irrelevant.} By inversion on the typing rules, $d_1 \sqsubseteq c \sqsubseteq d$. By definition either $d_1 \not \sqsubseteq \xi$  or none of the channels connected to $P$ including $y^c_\alpha$ and $x^d_\beta$ are relevant. In both cases, neither $\mb{msg}(y^c_\alpha.k)$ nor $\mb{proc}(x_\beta[d], [y^c_{\alpha+1}/y^c_{\alpha}]P @d_1)$ are relevant. Channel $x^d_\beta$  is irrelevant in the pre-step and post-step configurations.

Channel $y^c_\alpha$  is irrelevant in the pre-step, and both $y^c_\alpha$ and $y^c_{\alpha+1}$ are irrelevant in pre-step and post-step configurations. Every other irrelevant resource of $\mb{proc}(x_\beta[d],  y^c_\alpha.k ;P @d_1)$ will remain irrelevant in the post-step too.

In this subcase, our goal is to show 
\[\cproj{\mathbb{D}'_1\mb{msg}(y^c_\alpha.k)\mb{proc}(x_\beta[d], [y^c_{\alpha+1}/y^c_{\alpha}]P @d_1) \mathbb{D}''_1}{\xi} \meq \cproj{\mathbb{D}'_1 \mathbb{D}''_1}{\xi}\meq \cproj{\mathbb{D}_1}{\xi}\meq \cproj{\mathbb{D}_2}{\xi}\]

With a same argument as in {\bf \color{ForestGreen}Case 1.} {\bf Subcase 1.}, we can prove that the relevancy of processes in $\mathbb{D}'_1$ and $\mathbb{D}''_1$ remain intact.

\item {\bf Subcase 2. $\mb{proc}(x_\beta[d],  y^c_\alpha.k ;P @d_1)$ is relevant.} 
By assumption that $\cproj{\mathbb{D}_1}{\xi}\meq\cproj{\mathbb{D}_2}{\xi}$, and definition of $\meq$: 
\[\mathbb{D}_2=\mathbb{D}'_2\mb{proc}(x_\beta[d],  y^c_\alpha.k ;P @d_1) \mathbb{D}''_2,\] and we have {\color{red}\[\mathbb{D}_2 \mapsto \mathbb{D}'_2\mb{msg}(y^c_\alpha.k)\mb{proc}(x_{\beta}[d], [y^c_{\alpha+1}/y^c_\alpha]P @d_1) \mathbb{D}''_2\]}

It remains to show 
{\small\[\cproj{\mathbb{D}'_1\mb{msg}(y^c_\alpha.k)\mb{proc}(x_{\beta}[d], [y^c_{\alpha+1}/y^c_\alpha]P @d_1) \mathbb{D}''_1}{\xi} \meq \cproj{\mathbb{D}'_2\mb{msg}(y^c_\alpha.k)\mb{proc}(x_{\beta}[d], [y^c_{\alpha+1}/y^c_\alpha]P @d_1) \mathbb{D}''_2}{\xi}.\]}

\begin{itemize}
  \item If $d \sqsubseteq \xi$, then $x_\beta$ is relevant in the pre-steps of both runs and  remains relevant in the post-steps. Even if the quasi-running secrecy increases based on the code of $P$, by the tree invaiant it will be less than or equal to $d \sqsubseteq \xi$, and thus remains observable. As a result, the relevancy of processes in $\mathbb{D}''_i$ remain intact. Moreover, every resource of the processes in $\mathbb{D}'_i$ is relevant in the pre-steps and post-steps of both runs.
  \item If $d \not \sqsubseteq \xi$, then $x_\beta$ is irrelevant in the pre-step and remains irrelevant in the post-steps of both runs, and thus the relevancy of processes in $\mc{D}''_i$ remain intact. It remains to show that the projections of $\mc{D}'_1$ and $\mc{D}'_2$ in post-steps are still equal.

  We first condider the trees offered along $y^c_\alpha$ in both runs. The quasi running secrecy of the negative message $\mb{msg}(y^c_\alpha.k)$ is $c$ in the post-steps. 
\begin{itemize}
\item If $c\sqsubseteq \xi$ then the same message exists in both runs, and the tree offered along $y^c_\alpha$ is relevant in the pre-steps and post-steps.

\item If $c \not \sqsubseteq \xi$ then the message is irrelevant in both runs. $y^c_\alpha$ is irrelevant in the pre-steps of both runs and remain irrelevant in the post-steps too. 
\end{itemize}

Finally, we need to show that projections of the rest of the trees in $\mathbb{D}'_1$ and $\mathbb{D}'_2$ are equal in the post-step too. 
Consider the resources of the process that offeres along $x_\beta[d]$. By our writing convention, they are all in $\mathbb{D}'_1$ and $\mathbb{D}'_2$. (i) Those resources offered by $\mathbb{D}'_1$ and $\mathbb{D}'_2$ with max secrecies higher than or incomparable to the observer $\xi$  in the pre-step will remain higher than or incomparable to the observer and thus irrelevant in the post-step too. (ii) Those resources with max secrecies lower than the observer $\xi$ in the pre-step, i.e., a sub-tree $\mathbb{T}_i$ in $\mathbb{D}'_i$ offering along a channel $w[c']$ where $c' \sqsubseteq \xi$, are relevant in the pre-step and thus $\mathbb{T}_1= \mathbb{T}_2$. The quasi-running secrecy of the process offering along $x_\beta[d]$ after stepping is the same in both runs. If it is still observable, the sub-tree $\mathbb{T}_1= \mathbb{T}_2$ remains relevant in the post-step of both runs. Otherwise, if the quasi-running secrecies become nonobservable, $\mathbb{T}_1= \mathbb{T}_2$ may become irrelevant. However, since $\mathbb{T}_1= \mathbb{T}_2$ and by the tree invariant it only consists of low-secrecy channels, $\mathbb{T}_1$ becomes irrelevant in the post-step of the first run iff $\mathbb{T}_2$ becomes irrelevant in the post-step of the second run. Which completes the proof of this case.
\end{itemize}

\end{description}
{\color{ForestGreen} \bf Case 3. $\mathbb{D}_1=\mathbb{D}'_1\mathbb{T}_1\mb{proc}(y_\alpha[c], \mb{send} x^d_\beta\, y^c_\alpha @d_1) \mathbb{D}''_1$ and
{\small\[\erasure{\mathbb{D}'_1\mathbb{T}_1\mb{proc}(y_\alpha[c], \mb{send} x^c_\beta\, y^c_\alpha ;P @d_1) \mathbb{D}''_1}\mapsto \erasure{\mathbb{D}'_1 \mathbb{T}_1\mb{proc}(y_{\alpha+1}[c], [y^c_{\alpha+1}/y^c_{\alpha}]P @d_1)\mb{msg}(\mb{send} x^c_\beta\, y^c_\alpha)  \mathbb{D}''_1}\]}}
such that $\Gamma=\Gamma'\Gamma_t$ and $\Psi_0;\Gamma' \Vdash \mathbb{D}'_1:: \Gamma''$ and  $\Psi_0;\Gamma_t \Vdash \mathbb{T}_1:: (x_{\beta}{:}A[c])$ and $\Gamma'', x_{\beta}{:}A[c] \Vdash  \mb{proc}(y_\alpha[c], \mb{send} x^c_\beta\, y^c_\alpha ;P @d_1):: (y_\alpha{:}A\otimes B[c])$. (In the case that $\mathbb{T}_1$ is empty, we have $\Gamma_t= x_\beta{:}A[c]$.)

We consider subcases based on relevancy of process offering along $y^c_\alpha$: 
\begin{description}
\item {\bf Subcase 1. $\mb{proc}(y_c^\alpha, \mb{send} x^d_\beta\, y^c_\alpha ;P @d_1)$ is not relevant.} 

By inversion on the typing rules $d_1 \sqsubseteq c$. By definition either $d_1 \not \sqsubseteq \xi $ or none of the channels connected to $P$ including $y^c_\alpha$ and $x^c_\beta$ are relevant. In both cases, neither $\mb{proc}(y_{\alpha+1}[c], [y^c_{\alpha+1}/y^c_{\alpha}]P @d_1)$ nor $\mb{msg}(\mb{send} x^c_\beta\, y^c_\alpha)$ are relevant.

Channel $y^c_\alpha$  is irrelevant in the pre-step, and both $y^c_\alpha$ and $y^c_{\alpha+1}$ are irrelevant in pre-step and post-step configurations. In this subcase, our goal is to show

\[\cproj{\mathbb{D}'_1\mathbb{T}_1\mb{proc}(y_{\alpha+1}[c], [y^c_{\alpha+1}/y^c_{\alpha}]P @d_1) \mb{msg}(\mb{send} x^c_\beta\, y^c_\alpha) \mathbb{D}''_1}{\xi} \meq \cproj{\mathbb{D}'_1 \mathbb{T}_1 \mathbb{D}''_1}{\xi}\meq \cproj{\mathbb{D}_1}{\xi}\meq \cproj{\mathbb{D}_2}{\xi}.\]

We first prove that the relevancy status of  $\mathbb{T}_1$ remain intact. Note that all channels in $\mathbb{T}_1$, except $x^c_\beta$ have the same connections in the pre-step and post-step. So it is enough to consider the changes made to the tree rooted at $x^c_\beta$.  We show that the relevancy of $x^c_\beta$ remains intact after the step: 

\begin{itemize}
\item If $c \not \sqsubseteq \xi$, then $x_\beta^c$ is irrelevant in the pre-step, and remains irrelevant in the post-step. 
Moreover, the message is irrelevant in the post-step since its quasi-running secrecy of is higher than or incomparable to the observer.
\item If $c \sqsubseteq \xi$, then by $d_1\sqsubseteq c$, we have $d_1 \sqsubseteq \xi$, and thus both $x_c^\beta$ and $y^c_\alpha$ must be irrelevant in the pre-step. In the post-step, the parent of the message has to be irrelevant, otherwise $y^c_\alpha$ would be relevant in the pre-step. Since the parent of the message is irrelevant, we know that $x^c_\beta$ remains irrelevant in the post-step.  As a result, the message with three irrelebant channels connected to it is irrelevant in the post-step. 
\end{itemize}
In both cases, the relevancy status of $\mathbb{T}_1$ remains intact: the tree $\mathbb{T}_1$ rooted at $x^c_\beta$ offers to an irrelevant node before and after the step. With a same argument as in {\bf \color{ForestGreen}Case 1.} {\bf Subcase 1.} and the one given for $\mathbb{T}_1$, we can prove that the relevancy status of processes in $\mathbb{D}'_1$ and $\mathbb{D}''_1$ remain intact.

\item {\bf Subcase 2. $\mb{proc}(y^c_\alpha, \mb{send} x^c_\beta\, y^c_\alpha ;P @d_1)$ is relevant.}
By assumption that $\cproj{\mathbb{D}_1}{\xi}\meq\cproj{\mathbb{D}_2}{\xi}$, and definition of $\meq$:
\[\mathbb{D}_2=\mathbb{D}'_2\mathbb{T}''_2\mb{proc}(y_\alpha[c], \mb{send} x^c_\beta\, y^c_\alpha ;P @d_1)\mathbb{D}''_2\]
 
 We have {\color{red}\[\erasure{\mathbb{D}_2} \mapsto \erasure {\mathbb{D}'_2 \mathbb{T}''_2\mb{proc}(y_{\alpha+1}[c], [y^c_{\alpha+1}/y^c_\alpha]P @d_1) \mb{msg}(\mb{send} x^c_\beta\, y^c_\alpha)\mathbb{D}''_2}\]}
 If $c \not\sqsubseteq \xi$, then $\mb{msg}(\mb{send} x^c_\beta\, y^c_\alpha)$ is not relevant in both runs, and will be dismissed by the projections. Moreover,  $y^c_\alpha$ is not relevant in the pre-step and  post-step configurations. Thus the relevancy of processes in $\mathbb{D}''_1$ and $\mathbb{D}''_2$ will remain intact. Moreover, in this case $x^c_\beta$  is irrelevant in both pre-steps and post-steps. Which means that relevancy status of $\mathbb{T}''_1$ and $\mathbb{T}''_2$ remains intact.

If $ c\sqsubseteq \xi$, then $y^c_\alpha$ is relevant in the pre-step in both runs. We also know that $x_\beta$ is relevant in the pre-step and $\mathbb{T}''_1=\mathbb{T}''_2$ are relevant in the pre-step. In the post-step, $y^c_{\alpha+1}$ is relevant in both runs. Relevancy of messages $\mb{msg}(\mb{send} x^c_\beta\, y^c_\alpha)$  in the post-steps are determined by the quasi-running secrecy ($d'$ and $d''$) of their parents ($X$ and $X'$) in $\mathbb{D}''_1$ and $\mathbb{D}''_2$. If $d'\sqsubseteq \xi$, then the parent ($X$) is relevant in the first run and by assumption is equal to a relevant $X'$ in the second run. Thus messages $\mb{msg}(\mb{send} x^c_\beta\, y^c_\alpha)$ are relevant in both runs, $y^c_\alpha$ are relevant in the post-step, and trees $\mathbb{T}''_1=\mathbb{T}''_2$ and their offering channels $x^c_\beta$ are relevant in the post-steps too.  The same holds when $d'\sqsubseteq \xi$.

Otherwise, in both runs the quasi-running secrecy of the parent is higher than or incomparable to the observer level (the parents are both irrelevant). Thus messages $\mb{msg}(\mb{send} x^c_\beta\, y^c_\alpha)$ are not relevant in the post step of both runs, and will be dismissed by the projections. The channels $y^c_\alpha$ will be irrelevant in the post-step too. However, this does not affect the processes in $\mathbb{D}''_1$ and $\mathbb{D}''_2$ as the parents of messages ($X$ and $X'$) are already irrelevant in the pre-step. The channels $x^c_\beta$ both become irrelevant in the post-steps. However, we still have $\cproj{\mathbb{T}''_1}{\xi} =\cproj{\mathbb{T}''_2}{\xi}$ as they have the same type and their parents have the same quasi-running secrecy.

We show that projections of $\mathbb{D}'_1$ and $\mathbb{D}'_2$ are equal in the post-step too. The  resources with secrecy level higher than or incomparable to the observer level offered along $\mathbb{D}'_1$ and $\mathbb{D}'_2$ in the pre-step will remain higher than or incomparable to and thus irrelevant in the post-step too. For a relevant resources ($w$) offered along $\mathbb{D}_1$ and $\mathbb{D}_2$, we need to consider the change in quasi-running secrecy as in {\bf \color{ForestGreen}Case 1.} {\bf Subcase 2.} Moreover, we need to consider the scenario that a relevant resource ($w^{c'}$) in the pre-step loses its relevancy in the post-step because the channel offered along $x_\beta^c$ is transferred to the message. This case only happens if $c', c \sqsubseteq \xi$ and thus the trees $\mathbb{T}_1$ and $\mathbb{T}_2$ offered along $w^{c'}$ is present in both runs and $\mathbb{T}_1=\mathbb{T}_2$. We know that $w^{c'}$ is irrelevant in the post-step of both runs, and the quasi-running secrecy of the processes using the resource $w^{c'}$ in both runs are the same.

The relevant and irrelevant processes in $\mathbb{D}''_i$ remain intact.

\end{description}

{\color{ForestGreen} 
\bf Case 4. $\mathbb{D}_1=\mathbb{D}'_1\mathbb{T}''_1\mb{proc}(w_\eta[c'], \mb{send} x^c_\beta\, y^c_\alpha @d_1) \mathbb{D}''_1$ and
{\small\[\erasure{\mathbb{D}'_1\mathbb{T}''_1\mb{proc}(w_\eta[c'], \mb{send} x^c_\beta\, y^c_\alpha ;P @d_1) \mathbb{D}''_1}\mapsto \erasure{\mathbb{D}'_1\mathbb{T}''_1\mb{msg}(\mb{send} x^c_\beta\, y^c_\alpha)\mb{proc}(w_\eta[c'], [y^c_{\alpha+1}/y^c_{\alpha}]P @d_1)  \mathbb{D}''_1}\]}}

such that $\Gamma=\Gamma_1\Gamma_2$ and $\Gamma_1 \Vdash \mathbb{D}'_1:: \Gamma', y_\alpha{:}(A\multimap B)[c]$ and  $\Gamma_2 \Vdash \mathbb{T}''_1:: (x_{\beta}{:}A[c])$ and \[\Gamma', y_\alpha{:}(A\multimap B)[c], x_{\beta}{:}A[c] \Vdash  \mb{proc}(w_\eta[c'], \mb{send} x^c_\beta\, y^c_\alpha ;P @d_1):: (w_\eta{:}C[c']).\] In the case where $\mathbb{T}''_1$ is empty we have  $\Gamma_2=x_{\beta}{:}A[c]$. We proceed by considering subcases based on relevancy of the process offering along $w_\eta[c']$:

\item {\bf Subcase 1. $ \mb{proc}(w_\eta[c'], \mb{send} x^c_\beta\, y^c_\alpha ;P @d_1)$ is not relevant.} By inversion on the typing rules $ d_1 \sqsubseteq c \sqsubseteq c'$. By definition either $d_1 \not \sqsubseteq \xi$  or none of the channels connected to $P$ including $y^c_\alpha$, and $x^c_\beta$ are relevant. In both cases, neither $\mb{msg}(\mb{send} x^c_\beta\, y^c_\alpha)$ nor $\mb{proc}(w_\eta[c'], [y^c_{\alpha+1}/y^c_{\alpha}]P @d_1)$  are relevant. Channel $w^{c'}_\eta$  is irrelevant in the pre-step and post-step configurations.

Channel $y^c_\alpha$  is irrelevant in the pre-step, and both $y^c_\alpha$ and $y^c_{\alpha+1}$ are irrelevant in pre-step and post-step configurations. Every other irrelevant resource of the process in the pre-step will remain irrelevant in the post-step too. See {\bf \color{ForestGreen}Case 3.} {\bf Subcase 1.} for the discussion on the relevancy of $\mathbb{T}''_1$.

\[\cproj{\mathbb{D}'_1\mathbb{T}''_1 \mb{msg}(\mb{send} x^c_\beta\, y^c_\alpha) \mb{proc}(w_\eta[c'], [y^c_{\alpha+1}/y^c_{\alpha}]P @d_1) \mathbb{D}''_1}{\xi} \meq \cproj{\mathbb{D}'_1 \mathbb{D}''_1}{\xi}\meq \cproj{\mathbb{D}_1}{\xi}\meq \cproj{\mathbb{D}_2}{\xi}\]

\item {\bf Subcase 2. $\mb{proc}(w^{c'}_\eta, \mb{send} x^c_\beta\, y^c_\alpha ;P @d_1)$ is relevant.}
By assumption that $\cproj{\mathbb{D}_1}{\xi}\meq\cproj{\mathbb{D}_2}{\xi}$, and definition of $\meq$:

\[\mathbb{D}_2=\mathbb{D}'_2\mathbb{T}''_2\mb{proc}(w_\eta[c'], \mb{send} x^c_\beta\, y^c_\alpha ;P @d_1)\mathbb{D}''_2\]
 
 We have {\color{red}\[\erasure{{\mathbb{D}_2}} \mapsto\erasure{\mathbb{D}'_2\mathbb{T}''_2\mb{msg}(\mb{send} x^c_\beta\, y^c_\alpha) \mb{proc}(w_\eta[c'], [y^c_{\alpha+1}/y^c_\alpha]P @d_1)\mathbb{D}''_2}\]}
 
 With the same argument as in {\bf \color{ForestGreen}Case 2.} {\bf Subcase 2.} we can show that relevancy of $\mathbb{D}''_i$ remains intact.
 
 For $\mathbb{T}''_i$, we argue that if $c \sqsubset \xi$, then $\mathbb{T}''_1=\mathbb{T}''_2$ is relevant in the pre-step and remains relevant in the post-step too. If $c \not \sqsubseteq \xi$, then the relevancy of $\mathbb{T}''_i$ remain intact from pre-step to post-step. 
 See {\bf \color{ForestGreen}Case 3.} {\bf Subcase 2.} for a more detailed discussion on transferring a tree via message.
 
The discussion on relevancy of $\mathbb{D}'_i$ is similar to the previous cases.

{\color{ForestGreen} \bf Case 5. $\mathbb{D}_1=\mathbb{D}'_1\mb{msg}(y^c_\alpha.k) \mb{proc}(x_\beta[d], \mb{case} \,y^c_\alpha (\ell \Rightarrow P_\ell)_{\ell \in L} @d_1)  \mathbb{D}''_1$ and
{\small\[ \erasure{\mathbb{D}'_1\mb{msg}(y^c_\alpha.k)\mb{proc}(x_\beta[d], \mb{case} \,y^c_\alpha (\ell \Rightarrow P_\ell)_{\ell \in L} @d_1)  \mathbb{D}''_1} \mapsto \erasure{\mathbb{D}'_1\mb{proc}(w_{\beta}[d], [y^c_{\alpha+1}/y^c_{\alpha}] P_k@(d_1 \sqcup c)) \mathbb{D}''_1}\]}}
We consider sub-cases based on relevancy of process offering along $w^d_\beta$. Observe that  $y^c_\alpha$  is relevant if an only if $v^c$ is relevant, since they share a message of secrecy $c$.
\begin{description}
\item {\bf Subcase 1. $\mb{proc}(x_\beta[d], \mb{case} \,y^c_\alpha (\ell \Rightarrow P_\ell)_{\ell \in L} @d_1)$ is not relevant.} By definition either $d_1 \sqcup c \not \sqsubseteq \xi  $ or none of the channels connected to $P$ including its offering channel $y_\alpha^c$ are relevant.  In both cases the messages $\mb{msg}(y^c_\alpha.k)$ and the continuation process $\mb{proc}(w_{\beta}[d], [y^c_{\alpha+1}/y^c_{\alpha}] P_k@ (d_1 \sqcup c))$ are not relevant either. It is straightforward to see that   
\[\cproj{\mathbb{D}'_1\mb{proc}(w_{\beta}[d], [y^c_{\alpha+1}/y^c_{\alpha}] P_k@ (d_1\sqcup c))\mathbb{D}''_1}{\xi} \meq \cproj{\mathbb{D}'_1 \mathbb{D}''_1}{\xi}\meq \cproj{\mathbb{D}_1}{\xi}\meq \cproj{\mathbb{D}_2}{\xi}.\]

\item {\bf Subcase 2. $\mb{proc}(x_\beta[d], \mb{case} \,y^c_\alpha (\ell \Rightarrow P_\ell)_{\ell \in L} @d_1)$ is relevant.}
By definition of relevancy, we get that $c \sqcup d_1\sqsubseteq \xi$ and thus $y^c_\alpha$ is relevant. This means that $\mb{msg}(y^c_\alpha.k)$ is relevant too.
By assumption that $\cproj{\mathbb{D}_1}{\xi}\meq\cproj{\mathbb{D}_2}{\xi}$, and definition of $\meq$: \[\mathbb{D}_2=\mathbb{D}'_2\mb{msg}(y^c_\alpha.k)\mb{proc}(x_\beta[d],   \mb{case}\,y^c_\alpha( \ell \Rightarrow {P_\ell})_{\ell\in L} @d_1)  \mathbb{D}''_2,\] such that ${P_\ell}$ is equal to $P_\ell$ modulo renaming of some channels with secrecy level higher than or incomparable to the observer.

We have {\color{red}\[{\erasure{\mathbb{D}_2}} \mapsto \erasure{\mathbb{D}'_2\mb{proc}(x_{\beta}[d], [y^c_{\alpha+1}/y^c_\alpha]{P_k} @(d_1 \sqcup c)) \mathbb{D}''_2}\]}
This completes the proof of the subcase as we know that the relevancy of channels in $\mathbb{D}'_i$ and $\mathbb{D}''_i$ remain intact.

\end{description}

{\color{ForestGreen} \bf Case 6. $\mathbb{D}_1=\mathbb{D}'_1\mb{proc}(y_\alpha[c], \mb{case} \,y^c_\alpha (\ell \Rightarrow P_\ell)_{\ell \in L} @d_1) \mb{msg}(y^c_\alpha.k) \mathbb{D}''_1$ and
\[ \erasure{\mathbb{D}'_1\mb{proc}(y_\alpha[c], \mb{case} \,y^c_\alpha (\ell \Rightarrow P_\ell)_{\ell \in L} @d_1) \mb{msg}(y^c_\alpha.k) \mathbb{D}''_1} \mapsto \erasure{\mathbb{D}'_1\mb{proc}(x_1[c], [y^c_{\alpha+1}/y^c_{\alpha}] P_k@c) \mathbb{D}''_1}\]}
We consider sub-cases based on relevancy of process offering along $y^c_\alpha$. Observe that $y^c_\alpha$  is relevant if an only if $x_1^c$ is relevant, since they share a message of secrecy $c$.
\begin{description}
\item {\bf Subcase 1. $\mb{proc}(y_\alpha[c], \mb{case} \,y^c_\alpha (\ell \Rightarrow P_\ell)_{\ell \in L} @d_1)$ is not relevant.} By definition either $d_1 \sqcup c = c \not \sqsubseteq \xi$ or none of the channels connected to $P$ including its offering channel $y_\alpha^c$ are relevant.  In both cases, means that $y_{\alpha+1}^c$ is not relevant and $\mb{msg}(y^c_\alpha.k)$ is not relevant either.
Moreover the continuation process $\mb{proc}(y_{\alpha+1}[c], [y_{\alpha+1}^c/y^c_{\alpha}] P_k@c)$ won't be relevant. And   
\[\cproj{\mathbb{D}'_1\mb{proc}(y_{\alpha+1}[c], [y_{\alpha+1}^c/y^c_{\alpha}] P_k@c)\mathbb{D}''_1}{\xi} \meq \cproj{\mathbb{D}'_1 \mathbb{D}''_1}{\xi}\meq \cproj{\mathbb{D}_1}{\xi}\meq \cproj{\mathbb{D}_2}{\xi}\]
\item {\bf Subcase 2. $\mb{proc}(y_\alpha[c], \mb{case} \,y^c_\alpha (\ell \Rightarrow P_\ell)_{\ell \in L} @d_1)$ is relevant.}
By definition of relevancy, we get that $c \sqcup d_1=c \sqsubseteq \xi$ and thus $y^c_\alpha$ is relevant. This means that $\mb{msg}(y^c_\alpha.k)$ is relevant too.
By assumption that $\cproj{\mathbb{D}_1}{\xi}\meq\cproj{\mathbb{D}_2}{\xi}$, and definition of $\meq$: \[\mathbb{D}_2=\mathbb{D}'_2\mb{proc}(y_\alpha[c],   \mb{case}\,y^c_\alpha( \ell \Rightarrow {P_\ell})_{\ell\in L} @d_1) \mb{msg}(y^c_\alpha.k) \mathbb{D}''_2,\] such that ${P_\ell}$ is equal to $P_\ell$ modulo renaming of some channels with secrecy level higher than or incomparable to the observer.

We have {\color{red}\[\erasure{\mathbb{D}_2} \mapsto  \erasure{\mathbb{D}'_2\mb{proc}(y_{\alpha+1}[c], [y_{\alpha+1}^c/y^c_\alpha]{P_k} @c) \mathbb{D}''_2}\]}
This completes the proof of the subcase as we know that the relevancy of channels in $\mathbb{D}'_i$ and $\mathbb{D}''_i$ remain intact.

\end{description}

{\color{ForestGreen}
\bf Case 7. $\mathbb{D}_1=\mathbb{D}'_1\mb{msg}(\mb{send}x_\eta^c y_\alpha^c) \mb{proc}(v_\eta[c'],  w^c \leftarrow \mb{recv} y^c_\alpha; P @d_1)  \mathbb{D}''_1$ and
{\small\[ \mathbb{D}'_1\mb{msg}(\mb{send} x_\eta^c\, y_\alpha^c) \mb{proc}(v_\eta[c'],  w^c \leftarrow \mb{recv} y^c_\alpha; P @d_1)\mathbb{D}''_1 \mapsto \mathbb{D}'_1\mb{proc}(v_{\eta}[c'], [x_\eta/w][y_{\alpha+1}^c/y^c_{\alpha}] P@c\sqcup d_1) \mathbb{D}''_1\]}}
We consider sub-cases based on relevancy of process offering along $v^{c'}_\eta$. 
\begin{description}
\item {\bf Subcase 1. $ \mb{proc}(v_\eta[c'],  w^c \leftarrow \mb{recv} y^c_\alpha; P @d_1)$ is not relevant.} By definition either $d_1 \sqcup c \not \sqsubseteq \xi$ or none of the channels connected to $P$ including $y_\alpha^c$ are relevant.

In both cases by the definition of quasi-running secrecy we know that neither $\mb{msg}(\mb{send}x_\eta^c y_\alpha^c)$ nor the continuation process $\mb{proc}(v_{\eta}[c'], [x_\eta/w][y_{\alpha+a}^c/y^c_{\alpha}] P@c\sqcup d_1)$  are relevant. It is then straightforward to see that

\[\cproj{\mathbb{D}'_1\mb{proc}(v_{\eta}[c'], [x_\eta/w][y_{\alpha+1}^c/y^c_{\alpha}] P@c\sqcup d_1) \mathbb{D}''_1}{\xi} \meq \cproj{\mathbb{D}'_1 \mathbb{D}''_1}{\xi}\meq \cproj{\mathbb{D}_1}{\xi}\meq \cproj{\mathbb{D}_2}{\xi}.\]

\item {\bf Subcase 2. $\mb{proc}(v_\eta[c'],  w \leftarrow \mb{recv} y^c_\alpha; P @d_1)$ is relevant.}
By definition of relevancy, we get that $c \sqcup d_1\sqsubseteq \xi$. This implies that $y^c_\alpha$ is relevant in the pre-step. From relevancy of $y^c_\alpha$ and the quasi-running secrecy lower than or equal to the observer of the positive message $\mb{msg}(\mb{send} x_\eta^c\, y_\alpha^c;)$ we get that the message is relevant too.
By assumption:
\[\mathbb{D}_2=\mathbb{D}'_2\mb{msg}(\mb{send} x_\eta^c\, y_\alpha^c)\mb{proc}(u_\gamma[c'],  w \leftarrow \mb{recv} y^c_\alpha; P @d_1)  \mathbb{D}''_2.\]

We have {\color{red}\[\erasure{\mathbb{D}_2} \mapsto \erasure{\mathbb{D}'_2\mb{proc}(u_{\gamma}[c'], [x_\eta/w][y_{\alpha+1}^c/y^c_{\alpha}] {P} @d_1 \sqcup c) \mathbb{D}''_2}\]}
We need to consider  that the quasi-running secrecy of the process may increases in the post step based on the code of $P$. The argument for this case is similar to the previous cases of the proof. See {\bf \color{ForestGreen}Case 1.}{\bf  Subcase 2.}. One interesting situation is when  the relevancy of chain of positive and relevant messages in the pre-step of $\mathbb{D}'_i$ changes in the post-step. By relevancy in the pre-step we know that these chains exist in both runs, so the same chain of messages will become irrelevant in the post-step of both runs.  

\end{description}

{\color{ForestGreen}
\bf Case 8. $\mathbb{D}_1=\mathbb{D}'_1\mb{proc}(y_\alpha[c],  w \leftarrow \mb{recv} y^c_\alpha; P @d_1)\mb{msg}(\mb{send}x_\eta^c y_\alpha^c)   \mathbb{D}''_1$ and
{\small\[ \erasure{\mathbb{D}'_1 \mb{proc}(y_\alpha[c],  w \leftarrow \mb{recv} y^c_\alpha; P @d_1)\mb{msg}(\mb{send}x_\eta^c y_\alpha^c) \mathbb{D}''_1} \mapsto \erasure{\mathbb{D}'_1\mb{proc}(y_{\alpha+1}[c], [x_\eta/w][y_{\alpha+1}^c/y^c_{\alpha}] P@d_1) \mathbb{D}''_1}\]}
}

We consider sub-cases based on relevancy of process offering along $y^{c}_\alpha$.
\begin{description}
\item {\bf Subcase 1. $\mb{proc}(y_\alpha[c],  w \leftarrow \mb{recv} y^c_\alpha; P @d_1)$ is not relevant.} By definition either $d_1 \sqcup c =c \not \sqsubseteq \xi $ or none of the channels connected to $P$ including $y_\alpha^c$ are relevant. In both cases the negative message $\mb{msg}(\mb{send}x_\eta^c y_\alpha^c)$  and the continuation process $\mb{proc}(y_{\alpha+1}[c], [x_\eta/w][y_{\alpha+c}^c/y^c_{\alpha}] P@ d_1)$  are not relevant either. It is then straightforward to see that   
\[\cproj{\mathbb{D}'_1\mb{proc}(y_{\alpha+1}[c], [x_\eta/w][y_{\alpha+1}^c/y^c_{\alpha}] P@ d_1) \mathbb{D}''_1}{\xi} \meq \cproj{\mathbb{D}'_1 \mathbb{D}''_1}{\xi}\meq \cproj{\mathbb{D}_1}{\xi}\meq \cproj{\mathbb{D}_2}{\xi}.\]

\item {\bf Subcase 2. $\mb{proc}(y_\alpha[c],  w \leftarrow \mb{recv} y^c_\alpha; P @d_1)$ is relevant.}
By definition of relevancy, we get that $c \sqcup d_1=c\sqsubseteq \xi$ and $y^c_\alpha$ is relevant. This means that $\mb{msg}(\mb{send} x_\eta^c\, y_\alpha^c)$ and the channel $x^c_\eta$ are relevant. By assumption that $\cproj{\mathbb{D}_1}{\xi}\meq\cproj{\mathbb{D}_2}{\xi}$, and definition of $\meq$:
\[\mathbb{D}_2=\mathbb{D}'_2\mb{proc}(y_\alpha[c],    w \leftarrow \mb{recv} y^c_\alpha; P @d_1)\mb{msg}(\mb{send} x_\eta^c\, y_\alpha^c)  \mathbb{D}''_2.\] 

We have {\color{red}\[\erasure{\mathbb{D}_2} \mapsto \erasure{\mathbb{D}'_2\mb{proc}(y_{\alpha+1}[c], [x^c_\eta/w][y_{\alpha+1}^c/y^c_{\alpha}] {P} @d_1) \mathbb{D}''_2}{\mathbb{F}_2}\]}

The proof is similar to previous cases.

\end{description}

{\color{ForestGreen} \bf Case 9.  {\small$\mathbb{D}_1=\mathbb{D}'_1\mb{proc}(y_\alpha[c], (x^{[\substmapap{\gamma_{\m{sec}}}{\psi}]} \leftarrow X[\gamma] \leftarrow \Gamma_1) @\substmapap{\gamma_{\m{sec}}}{\psi'}; Q @d_1)\mathbb{D}''_1$}and

{\small$\mathbf{def}(\Psi'; \Xi' \vdash X = P @{\psi'} ::x: B'[\psi])\,\,$}
and
{\[
\begin{array}{c}
  \erasure{\mathbb{D}'_1\mb{proc}(y_\alpha[c], (x^{[\substmapap{\gamma_{\m{sec}}}{\psi}]} \leftarrow X[\gamma] \leftarrow \Gamma_1) @\substmapap{\gamma_{\m{sec}}}{\psi'}; Q @d_1)\mathbb{D}''_1}\\
   \mapsto \\
  \erasure{ \mathbb{D}'_1 \mb{proc}(x_0[d],  \hat{\gamma}(P) @d_2)\mb{proc}(y_\alpha[c], [x^d_0/x^d] Q @d_1) \mathbb{D}''_1}
\end{array}  
\]}}
where $\gamma=(\gamma_{\m{sec}}, \gamma_{\m{var}})$ $\substmapap{\gamma}{\psi}= d$,  $\substmapap{\gamma}{\psi'}=d_2$, and  $\substmapap{\gamma}{\Xi'}=\Gamma_1$.
We consider sub-cases based on relevancy of process offering along $y^c_\alpha$.
\begin{description}
\item {\bf Subcase 1. $\mb{proc}(y_\alpha[c], (x^{[\substmapap{\gamma_{\m{sec}}}{\psi}]} \leftarrow X[\gamma] \leftarrow \Gamma_1) @\substmapap{\gamma_{\m{sec}}}{\psi'}; Q @d_1)$ is not relevant.} By definition either $d_1 \not \sqsubseteq \xi $ or none of the channels of this process including $y_\alpha^c$ are relevant.  

In both cases, it means that $\mb{proc}(x_0[d], \hat{\gamma} (P) @d_2)$ and $\mb{proc}(y_\alpha[c], [x^d_0/x^d] Q @d_1)$ are not relevant either. Note that $d_1 \sqsubseteq d_2$ and thus $d_2 \not \sqsubseteq \xi$ if $d_1 \not \sqsubseteq \xi$.

\[\cproj{\mathbb{D}'_1\mb{proc}(y_\alpha[c], (x^{[\substmapap{\gamma_{\m{sec}}}{\psi}]} \leftarrow X[\gamma] \leftarrow \Gamma_1) @\substmapap{\gamma_{\m{sec}}}{\psi'}; Q @d_1) \mathbb{D}''_1}{\xi} \meq \cproj{\mathbb{D}'_1 \mathbb{D}''_1}{\xi}\meq \cproj{\mathbb{D}_1}{\xi}\meq \cproj{\mathbb{D}_2}{\xi}\]

\item {\bf Subcase 2. $\mb{proc}(y_\alpha[c], (x^{[\substmapap{\gamma_{\m{sec}}}{\psi}]} \leftarrow X[\gamma] \leftarrow \Gamma_1) @\substmapap{\gamma_{\m{sec}}}{\psi'}; Q @d_1)$ is relevant.}
By definition of relevancy, we get that $d_1 \sqsubseteq \xi$ and all  channels  of this process with secrecy levels lower than or equal to the observer level are relevant. By assumption that $\cproj{\mathbb{D}_1}{\xi}\meq\cproj{\mathbb{D}_2}{\xi}$, and definition of $\meq$: \[\mathbb{D}_2=\mathbb{D}'_2\mb{proc}(y_\alpha[c], (x^{[\substmapap{\gamma_{\m{sec}}}{\psi}]} \leftarrow X[\gamma] \leftarrow \Gamma_1) @\substmapap{\gamma_{\m{sec}}}{\psi'}; Q @d_1)\mathbb{D}''_2.\] 

We have {\color{red}\[\erasure{\mathbb{D}_2} \mapsto\erasure{\mathbb{D}'_2 \mb{proc}(x_0[d], \hat{\gamma} (P) @d_2)\mb{proc}(y_\alpha[c], [x^d_0/x^d] Q @d_1)\mathbb{D}''_2}\]}
Remark: we can assume that the fresh channel being spawned will be $x_0$ in both runs. Moreover, the (unique) substitution $\hat{\gamma}$ is the same when stepping both runs.

Note that if $x^d_0$ has secrecy level lower than or equal to the observer level, then taking this step won't change relevancy of any channels. Otherwise some resource of the process may become irrelevant after this step, e.g. $x^d_0$ may block their relevancy path if $d \not \sqsubseteq \xi$ or in the case where $d_2\not \sqsubseteq \xi$. But this happens to the processes in the both runs and can only change relevant processes/messages in the pre-step to irrelevant processes/messages in the post-step. (similar to the cases 3 and 4 for $\otimes$ and $\multimap$)
\end{description}

{\color{ForestGreen} \bf Case 10.  {\small$\mathbb{D}_1=\mathbb{D}'_1\mb{proc}(y_\alpha[c], F_Y[\gamma] @ c)\mathbb{D}''_1$} and

{\small$\mathbf{def}(\Psi'; u{:}C[\psi] \vdash F_Y = \fwder{C, w^{\psi}\leftarrow u^{\psi}} @{\psi} :: w: C[\psi])\,\,$}
and
{\[
\begin{array}{c}
  \erasure{\mathbb{D}'_1 \mb{proc}(y_\alpha[c], F_Y[\gamma] @ c) \mathbb{D}''_1}\,
   \mapsto \,
  \erasure{\mathbb{D}'_1 \mb{proc}(y_\alpha[c], \fwder{C, y_\alpha^c\leftarrow z_\beta^c} @d_1) \mathbb{D}''_1}
\end{array}  
\]}}
where $\gamma=(\gamma_{\m{sec}}, \gamma_{\m{var}})$ $\substmapap{\gamma}{\psi}= c$,  $\substmapap{\gamma}{w}=y_\alpha$, and $\substmapap{\gamma}{u}=z_\beta.$

The proof of this case is similar to the proof of Case 9(Spawn) by considering sub-cases based on relevancy of process offering along $y^c_\alpha$.
\end{proof}

\subsection{Fundamental theorem}\label{apx:sec:noninterference:ftlr}
\begin{theorem}[Fundamental Theorem]\label{apx:thm:ftlr}
For all security levels $\xi$, and configurations \[{\Psi; \Gamma_1 \Vdash \mathbb{D}_1:: u_\alpha {:}T_1[c_1]},\,\, \textit{and}\,\, {\Psi; \Gamma_2 \Vdash \mathbb{D}_2:: v_\beta {:}T_2[c_2]},\]  with $\cproj{\mathbb{D}_1}{\xi} \meq \cproj{\mathbb{D}_2}{\xi}$, $\Gamma_1 \Downarrow \xi = \Gamma_2 \Downarrow \xi$, and $u_\alpha {:}T_1[c_1]\Downarrow \xi = v_\beta {:}T_2[c_2] \Downarrow \xi$  we have \[(\Gamma_1 \Vdash \mathbb{D}_1:: u_\alpha {:}T_1[c_1])  \equiv^\Psi_{\xi} (\Gamma_2 \Vdash \mathbb{D}_2:: v_\beta {:}T_2[c_2]).\]  
\end{theorem}
\begin{proof}

Our goal is to show:
\begin{quote}
For all $\mathbb{D}_1$ and $\mathbb{D}_2$ that are IFC-typed, i.e.,${\Psi; \Gamma_1 \Vdash \mathbb{D}_1:: u_\alpha {:}T_1[c_1]}$ and ${\Psi; \Gamma_2 \Vdash \mathbb{D}_2:: v_\beta {:}T_2[c_2]}$,  with $\cproj{\mathbb{D}_1}{\xi} \meq \cproj{\mathbb{D}_2}{\xi}$, and $\Gamma_1 \Downarrow \xi = \Gamma_2 \Downarrow \xi= \Gamma$, and $u_\alpha {:}T_1[c_1]\Downarrow \xi = v_\beta {:}T_2[c_2] \Downarrow \xi= K^s$ we have 
{\small $ \forall \mc{B}_1 \in \mb{H\text{-}Provider}^\xi(\Gamma_1).\, \forall \, \mc{B}_2 \in \mb{H\text{-}Provider}^\xi(\Gamma_2). \,\forall \mc{T}_1 \in \mb{H\text{-}Client}^\xi(x_\alpha {:}A_1[{c_1}]). \,\forall \mc{T}_2 \in \mb{H\text{-}Client}^\xi(y_\beta {:}A_2[{c_2}]).$}
{ \[ \forall\,m.\,(\mc{B}_1\erasure{\mathbb{D}_1}\mc{T}_1, \mc{B}_2\erasure{\mathbb{D}}_2\mc{T}_2) \in \mc{E}\llbracket \erasure{\Gamma} \Vdash \erasure{K^s} \rrbracket^{m},\,
\m{and}\, \,
\forall\,m.\,(\mc{B}_2\erasure{\mathbb{D}}_2\mc{T}_2, \mc{B}_1\erasure{\mathbb{D}}_1\mc{T}_1) \in \mc{E}\llbracket \erasure{\Gamma} \Vdash \erasure{K^s} \rrbracket^{m}. .\]}
\end{quote}

Without loss of generality, we only prove one part of this symmetric relation, i.e.:
\begin{quote}
  {\small For all $\mathbb{D}_1$ and $\mathbb{D}_2$ that are IFC-typed, i.e.,${\Psi; \Gamma_1 \Vdash \mathbb{D}_1:: u_\alpha {:}T_1[c_1]}$ and ${\Psi; \Gamma_2 \Vdash \mathbb{D}_2:: v_\beta {:}T_2[c_2]}$,  with $\cproj{\mathbb{D}_1}{\xi} \meq \cproj{\mathbb{D}_2}{\xi}$, and $\Gamma_1 \Downarrow \xi = \Gamma_2 \Downarrow \xi= \Gamma$, and $u_\alpha {:}T_1[c_1]\Downarrow \xi = v_\beta {:}T_2[c_2] \Downarrow \xi= K^s$ we have }
   {\small$ \forall \mc{B}_1 \in \mb{H\text{-}Provider}^\xi(\Gamma_1).\, \forall \, \mc{B}_2 \in \mb{H\text{-}Provider}^\xi(\Gamma_2). \,\forall \mc{T}_1 \in \mb{H\text{-}Client}^\xi(x_\alpha {:}A_1[{c_1}]). \,\forall \mc{T}_2 \in \mb{H\text{-}Client}^\xi(y_\beta {:}A_2[{c_2}]).$}
  { \[ \forall\,m.\,(\mc{B}_1\erasure{\mathbb{D}}_1\mc{T}_1, \mc{B}_2\erasure{\mathbb{D}}_2\mc{T}_2) \in \mc{E}\llbracket \erasure{\Gamma} \Vdash \erasure{K^s} \rrbracket^{m}\]}
  \end{quote}

The goal is equivalent to:
\begin{quote}
  $(\star_1)\,\,$For all $m$, for all $\mathbb{D}_1$ and $\mathbb{D}_2$ that are IFC-typed, i.e.,${\Psi; \Gamma_1 \Vdash \mathbb{D}_1:: u_\alpha {:}T_1[c_1]}$ and ${\Psi; \Gamma_2 \Vdash \mathbb{D}_2:: v_\beta {:}T_2[c_2]}$,  with $\cproj{\mathbb{D}_1}{\xi} \meq \cproj{\mathbb{D}_2}{\xi}$, and $\Gamma_1 \Downarrow \xi = \Gamma_2 \Downarrow \xi= \Gamma$, and $u_\alpha {:}T_1[c_1]\Downarrow \xi = v_\beta {:}T_2[c_2] \Downarrow \xi= K^s$ we have 
  $ \forall \mc{B}_1 \in \mb{H\text{-}Provider}^\xi(\Gamma_1).\, \forall \, \mc{B}_2 \in \mb{H\text{-}Provider}^\xi(\Gamma_2). \,\forall \mc{T}_1 \in \mb{H\text{-}Client}^\xi(x_\alpha {:}A_1[{c_1}]). \,\forall \mc{T}_2 \in \mb{H\text{-}Client}^\xi(y_\beta {:}A_2[{c_2}]).$
 { \[ (\mc{B}_1\erasure{\mathbb{D}}_1\mc{T}_1, \mc{B}_2\erasure{\mathbb{D}}_2\mc{T}_2) \in \mc{E}\llbracket \erasure{\Gamma} \Vdash \erasure{K^s} \rrbracket^{m}\]}
  \end{quote}

 First, we show that it is enough to prove the following goal.

\begin{quote}
 {\small $(\star_2)\,\,$For all $m$, for all $\mathbb{D}'_1$ and $\mathbb{D}'_2$ that are IFC-typed, i.e., ${\Psi; \Gamma \Vdash \mathbb{D}'_1:: K^s}$ and ${\Psi; \Gamma \Vdash \mathbb{D}'_2:: K^s}$,  with $\cproj{\mathbb{D}'_1}{\xi} \meq \cproj{\mathbb{D}'_2}{\xi}$, and $\Gamma \Downarrow \xi = \Gamma$, and $K^s\Downarrow \xi = K^s$ we have }
 { \[ (\erasure{\mathbb{D}'_1}, \erasure{\mathbb{D}'_2}) \in \mc{E}\llbracket \erasure{\Gamma} \Vdash \erasure{K^s} \rrbracket^{m}\]}
  \end{quote}

  We need to show that from $\star_1$, we get $\star_2$:
  \begin{quote}
  Apply a for all introduction on $\star_2$ to get an arbitrary number $m$, configurations $\mathbb{D}_1$ and $\mathbb{D}_2$, and high providers $\mc{B}_1$ and $\mc{B}_2$ and high clients $\mc{T}_1$ and $\mc{T}_2$ satisfying the given assumptions. We build well-typed configurations stel $\mathbb{D}'_1$ and $\mathbb{D}'_2$ with the following steps:
\begin{enumerate}
\item  Let's assume $\mc{T}_1\neq \cdot$ and $\mc{T}_2\neq \cdot$, i.e., $\mc{T}_1 \in \m{Tree}(u_\alpha {:}T_1 \Vdash \_{:}1)$ and $c_1 \not \sqsubseteq \xi$
and $\mc{T}_2 \in \m{Tree}(v_\beta {:}T_2 \Vdash \_{:}1)$ and $c_1 \not \sqsubseteq \xi$. We annotate all channels in $\mc{T}_1$ with security level $c_1$ and all processes with running secrecy $c_1$. We add the same security variable $\psi$ correspondingly to all process variables defined in the signature. We also provide a mapping for the spawns appearing in the process terms such that they map all security variables $\psi$ in the process definition to the security variable $c_1$.  It is straightforward to see that this annotation of $\mc{T}_1$ which we call $\mathbb{T}_1$ is IFC-typed.
With a similar approach we can build the IFC-typed annotation of $\mc{T}_1$ which we call $\mathbb{T}_2$.

\item If $\mc{T}_1= \cdot$ and $\mc{T}_2= \cdot$, then define $\mathbb{T}_1=\mathbb{T}_2=\cdot$, and observe that they are trivially IFC-typed.

\item Consider every tree $\mc{A}_1 \in \mc{B}_1$, we know that $\mc{A}_1 \in \m{Tree}(x_\gamma{:}A)$ for some $x_\gamma{:}A[d]\in \Gamma_1$ such that $d \not \sqsubseteq \xi$, we build a security-annotated version of $\mc{A}_1$ as $\mathbb{A}_1$, similar to (1), by annotating all channels/running secrecy/substitutions with the level $d$. From all such annotated trees we build the annotated forest $\mathbb{B}_1$. Similarly, we can build annotated forest $\mathbb{B}_2$. Again it is straightforward to show that both $\mathbb{B}_1$ and $\mathbb{B}_2$ are IFC-typed. 
\end{enumerate}

We define $\mathbb{D}'_1= \mathbb{B}_1 \mathbb{D}_1 \mathbb{T}_1$ and $\mathbb{D}'_2= \mathbb{B}_2 \mathbb{D}_2 \mathbb{T}_2$. Note that these two configurations are both IFC-typed, and also we have 
$\Psi; \Gamma \Vdash \mathbb{D}'_1:: K^s$ and $\Psi; \Gamma \Vdash \mathbb{D}'_2:: K^s$. We also know that $\mathbb{D}'_1 \Downarrow \xi = \mathbb{D}'_2 \Downarrow \xi$. Observe that all relevant nodes occur in $\mathbb{D}_1$ and $\mathbb{D}_2$, for which by assumption we know $\mathbb{D}_1 \Downarrow \xi = \mathbb{D}_2 \Downarrow \xi$. Now we can apply $\star_2$ to get what we want.

\end{quote}
It remains to prove $(\star_2)$:

\begin{quote}
  $(\star_2)\,\,$For all $m$, for all $\mathbb{D}_1$ and $\mathbb{D}_2$ that are IFC-typed, i.e.,${\Psi; \Gamma \Vdash \mathbb{D}_1:: K^s}$ and ${\Psi; \Gamma \Vdash \mathbb{D}_2:: K^s}$,  with $\cproj{\mathbb{D}_1}{\xi} \meq \cproj{\mathbb{D}_2}{\xi}$, and $\Gamma \Downarrow \xi = \Gamma$, and $K^s\Downarrow \xi = K^s$ we have 
 { \[ (\erasure{\mathbb{D}_1}; \erasure{\mathbb{D}_2}) \in \mc{E}\llbracket \erasure{\Gamma} \Vdash \erasure{K^s} \rrbracket^{m}\]}
\end{quote}

The proof is by induction on the index $m$.

{\bf Base case. $m=0$.} We consider arbitrary configurations ($\forall I$ on the goal). By the configuration typing, we know that 
\[ (\erasure{\mathbb{D}_1} ; \erasure{\mathbb{D}_2} ) \in\m{Tree}(\erasure{\Gamma} \Vdash \erasure{K^s}),\]
which is enough to complete the proof.

{\bf Inductive case. $m=m'+1$.}

\[ (\mathcal{B}_1\erasure{\mathbb{D}_1} \mathcal{T}_1;\mathcal{B}_2 \erasure{\mathbb{D}_2} \mathcal{T}_2) \in\m{Tree}(\erasure{\Gamma} \Vdash \erasure{K^s}).\]

Consider an arbitrary $\mathbb{D}'_1$ such that $ \erasure{\mathbb{D}_1} \mapsto^{{*}_{\Upsilon; \Theta}} \erasure{\mathbb{D}'_1}$. 
By Preservation we know that $\mathbb{D}'_1$ is IFC-typed for the same interface ${\Psi; \Gamma \Vdash \mathbb{D}'_1:: K^s}$ 
By \Cref{apx:lem:invariant}, for some $\mathbb{D}'_2$, we get $\erasure{\mathbb{D}_2} \mapsto^{{*}_{\Upsilon}} \erasure{\mathbb{D}'_2} $, such that ${\Psi; \Gamma \Vdash \mathbb{D}'_2:: K^s}$ with $\mathbb{D}'_2\Downarrow \xi = \mathbb{D}'_2 \Downarrow \xi$.

We need to show that \noindent 
\[
\begin{array}{l}
\;(\dagger_1) \; \forall u'_\delta \in \mathbf{Out}(\erasure{\Gamma} \Vdash \erasure{K}).\, \mathbf{if}\, u'_\delta \in \Upsilon.\, \mathbf{then}\, (\erasure{\mathbb{D}'_1} ; \erasure{\mathbb{D}'_2} ) \in \mathcal{V}\llbracket \erasure{\Gamma} \Vdash \erasure{K^s} \rrbracket_{\cdot;u'_\delta}^{m}\; \mathbf{and}\\
\;(\dagger_2) \;\forall u'_\delta \in \mathbf{In}(\erasure{\Gamma} \Vdash \erasure{K}).\,\mathbf{if}\, u'_\delta \in \Theta.\, \m{then}\,  (\erasure{\mathbb{D}'_1} ; \erasure{\mathbb{D}'_2} ) \in \mathcal{V}\llbracket \erasure{\Gamma} \Vdash \erasure{K^s}\rrbracket_{u'_\delta;\cdot}^{m} 
\end{array}
\]

We prove both $\dagger_1$, and $\dagger_2$. We consider an arbitrary $u'_\delta$ and provide a case analysis based on its type. If there is no such $u'_\delta$ in the specified set, then the statements trivially holds. There might be a channel in both of these sets, if and only if $\mathbb{D}_1= \cdot$.

\begin{description}
\item{\bf Proof of $\dagger_1$.} Consider an arbitrary channel $u'_\delta{:}T \in \mathbf{Out}(\erasure{\Gamma} \Vdash \erasure{K}^s)$ and assume $u'_\delta{:}T \in \Upsilon$. We consider different cases based on the session type $T$.
\begin{description}
{{\bf Case 1.} \color{RedViolet} $u'_\delta{:}T[c']=K^s=u_\alpha{:}T_1[c_1]=u_\alpha{:}1[c_1]$}, and $\mathbb{D}'_1= \mathbb{D}''_1\msg{\mb{close}\,u^{c_1}_\alpha}$. 

By the typing rules and $\mathbb{D}'_1= \mathbb{D}''_1\msg{\mb{close}\,u^{c_1}_\alpha}$ we know that $\mathbb{D}''_1=\cdot$, and $\Gamma_1= \cdot$.

Note that by the definition of relevancy, all processes in $\mathbb{D}'_1$ and $\mathbb{D}'_2$ has to be relevant.
As a result, \[\mathbb{D}'_{2}=\mathbf{msg}(\mb{close}\,u^{c_1}_\alpha).\]
Thus, we satisfy the conditions required by Line 1 of the logical relation to establish
\[\dagger_1 (\erasure{\mathbb{D}'_1}; \erasure{\mathbb{D}'_2}) \in \mc{V}\llbracket \cdot \Vdash \erasure{u_\alpha{:}1[c_1]}\rrbracket_{\cdot;u'_\delta}^{m}\]
as needed.

{\bf Case 2.} {\color{RedViolet} $u'_\delta{:}T[c']= K^s=u_\alpha{:}T_1[c_1]=u_\alpha{:}\oplus\{\ell{:}A_{\ell}\}_{\ell \in I}[c_1],$}  and $\mathbb{D}'_{1}=\mathbb{D}''_1\mathbf{msg}(u_\alpha^{c_1}.k)$.

    By the assumption of the theorem, $\mathbb{D}'_1$ and $\mathbb{D}'_2$ are both relevant, and we have $\mathbb{D}'_{2}=\mathbb{D}''_2\mathbf{msg}(u_\alpha^{c_1}.k)$. Removing $\mb{msg}(u^{c_1}_\alpha.k)$ from $\mathbb{D}'_1$ and $\mathbb{D}'_2$ does not change relevancy of the remaining configuration when $c_1 \sqsubseteq \xi$: \[\cproj{\mathbb{D}''_1}{\xi}=\cproj{\mathbb{D}''_2}{\xi}.\]
 
We can apply the induction hypothesis on the index $m' < m$ to get 
  \[\begin{array}{l}
 (\erasure{\mathbb{D}''_1} ,\erasure{\mathbb{D}''_2}) \in \mathcal{E}\llbracket \erasure{\Gamma} \Vdash \erasure{u_{\alpha+1}{:}A_k[c_1]}\rrbracket^{m'}.
\end{array}\]

By line (2) in the definition of $\mathcal{V}$:
{\small   \[\begin{array}{l}
 \dagger_1  (\erasure{\mathbb{D}''_1 \msg{u^{c_1}_\alpha.k}} ,\erasure{\mathbb{D}''_2\msg{u^{c_1}_\alpha.k}})\in \mathcal{V}\llbracket \erasure{\Gamma} \Vdash \erasure{u_\alpha{:}\oplus\{\ell:A_{\ell}\}_{\ell \in I}[c_1]}\rrbracket_{\cdot; u'_\delta}^{m'+1}.
\end{array}
\]}\\
 {\bf Case 3.} {\color{RedViolet} $u'_\delta{:}T[c']=K^s=u_\alpha{:}T_1[c_1]=u_\alpha{:} (A \otimes B)[c_1]$}, and
    \[\mathbb{D}'_1=\mathbb{D}''_1\mathbb{A}_1\mathbf{msg}(\mb{send}x_\beta^{c_1}\,u_\alpha^{c_1})\]
     where $\Gamma=\Gamma', \Gamma''$, and $\Psi; \Gamma''\Vdash \mathbb{A}_1::(x_\beta{:}A[c_1])$, and $\Psi; \Gamma'\Vdash \mathbb{D}''_1::(u_{\alpha+1}{:}A[c_1])$.

    Since $\mathbb{D}'_1$ and $\mathbb{D}'_2$ are both IFC-typed, they enjoy the tree invariant. Thus both of them are relevant and since $\mathbb{D}'_1 \Downarrow \xi= \mathbb{D}'_2 \Downarrow \xi$, we have \[{\mathbb{D}'_{2}=\mathbb{D}''_2\mathbb{A}_2\mathbf{msg}(\mb{send}x_\beta^{c_1}\, u_\alpha^{c_1}),}\] such that $\Psi; \Gamma''\Vdash \mathbb{A}_2::(x_\beta{:}A[c_1])$, and $\Psi; \Gamma'\Vdash \mathbb{D}''_2::(u_{\alpha+1}{:}A[c_1])$. Moreover, by relevancy of $\mathbb{D}'_i$ (and relevancy of $x_\beta^{c_1}$) we get $\cproj{\mathbb{D}''_1}{\xi}=\cproj{\mathbb{D}''_2}{\xi}$ and 
    $\cproj{\mathbb{A}_1}{\xi}=\cproj{\mathbb{A}_2}{\xi}$. 

Note that in the particular case with $x_\beta{:}A[c_1] \in \Gamma$, we have $\mathbb{A}_i=\cdot$ and $\Gamma''=x_\beta{:}A[c_1]$.
     
We apply the induction hypothesis on the index $m'<m$ for $\Psi; \Gamma''\Vdash \mathbb{A}_i::(x_{\beta}{:}B[c_1])$ to get 

\[\begin{array}{l}
(\erasure{\mathbb{A}_1}, \erasure{\mathbb{A}_2}) \in \mathcal{E}\llbracket \erasure{\Gamma''} \Vdash \erasure{x_{\beta}{:}A[c_1]}\rrbracket^{m'}.
\end{array}\]

and apply the induction hypothesis on the index $m'<m$ for $\Psi; \Gamma'\Vdash \mathbb{D}''_i::(u_{\alpha+1}{:}A[c_1])$ 

\[\begin{array}{l}
  (\erasure{\mathbb{D}''_1}, \erasure{\mathbb{D}''_1})  \in \mathcal{E}\llbracket \erasure{\Lambda_1} \Vdash \erasure{w_{\gamma}{:}B[c_1]}\rrbracket^{m'}.
\end{array}\]

By line (4) in the definition of $\mathcal{V}$:
{\small  
\[\begin{array}{l}
    \dagger_1  (\erasure{\mathbb{A}_1 \mathbb{D}''_1\msg{\mb{send}x_\beta^{c_1}\,u_\alpha^{c_1}}}, \erasure{\mathbb{A}_2\mathbb{D}''_1\msg{\mb{send}x_\beta^{c_1}\,u_\alpha^{c_1}}}) \in \mathcal{V}\llbracket \erasure{\Delta} \Vdash \erasure{u_\alpha{:}A \otimes B[c_1]}\rrbracket_{\cdot;u'_\delta}^{m'+1}.
\end{array}\]}

{{\bf Case 4.} \color{RedViolet} $\Gamma=\Gamma', x_\gamma{:}\&\{\ell:A_\ell\}_{\ell \in L}[c]$}, and $u'_\delta{:}T[c']=x_\gamma{:}\&\{\ell:A_\ell\}_{\ell \in L}[c]$, and $\mathbb{D}'_1=\msg{x^{c}_\gamma.k}\mathbb{D}''_1$.

By the assumption of the theorem, $\mathbb{D}'_{2}=\msg{x^{c}_\gamma.k}\mathbb{D}''_2.$ Moreover,  $\cproj{\mathbb{D}''_1}{\xi}=\cproj{\mathbb{D}''_2}{\xi}$. The reason is $c\sqsubseteq \xi$ and thus $x^c_{\gamma+1}$ is relevant in $\mathbb{D}''_i$ and no relevancy changes in the configurations after removing the negative message. 
 
We can apply the induction hypothesis on the smaller index $m'<m$ to get 
 \[\begin{array}{l}
 (\erasure{\mathbb{D}''_1};\erasure{\mathbb{D}''_2}) \in \mathcal{E}\llbracket \erasure{\Gamma', x_{\gamma+1}{:}A_k[c]}\Vdash \erasure{K^s}\rrbracket^{m'}.
\end{array}\]

By line (8) in the definition of $\mathcal{V}$:
 \  \[\begin{array}{l}
 \dagger_1  (\erasure{\msg{x^{c}_\gamma.k} \,\mathbb{D}''_1};\erasure{\msg{x^{c}_\gamma.k}\,\mathbb{D}''_2})\in \mathcal{V}\llbracket \erasure{\Gamma', x_{\gamma}{:}\&\{\ell{:}A_\ell\}_{\ell\in L}[c]}\Vdash \erasure{K^s}\rrbracket_{\cdot; u'_\delta}^{m'+1}.
\end{array}
\]

 {\bf Case 5. 
 \color{RedViolet} $\Gamma=\Gamma', \Gamma'', x_\gamma{:}(A\multimap B)[c]$}, and $u'_\delta{:}T'[c']=x_\gamma{:}(A \multimap B)[c]$, and
 \[  \mathbb{D}'_1= \mathbb{A}_1\msg{\mb{send} y^{c}_\delta\,x^{c}_\gamma} \mathbb{D}''_1,\] such that $\ctype{\Psi_0}{\Gamma''}{\mathbb{A}_1}{y_\delta{:}A[c]}$ and $\ctype{\Psi_0}{\Gamma',x_{\gamma+1}{:}B[c]}{\mathbb{D}''_1}{K^s}$.

The message $\msg{\mb{send}y^{c}_\delta\,x^{c}_\gamma}$ is relevant in $\mathbb{D}'_1$. By assumption of the theorem,\\ {\small$\mathbb{D}'_{2}=\mathbb{A}_2\msg{\mb{send}y^{c}_\delta\,x^{c}_\gamma}\mathbb{D}''_2,$ } such that $\ctype{\Psi_0}{\Gamma''}{\mathbb{A}_2}{y_\delta{:}A[c]}$ and $\ctype{\Psi_0}{\Gamma^h_2,\Gamma',x_{\gamma+1}{:}B[c]}{\mathbb{D}''_2}{K^s}$ with $\Gamma^h_2$ being the context of all the channels in $\Gamma_2$ with security higher than or incomparable to the observable level $\xi$. Again because of the tree invariant every resource of $\mathbb{A}_2$ has a secrecy less than or equal to the observer. Also, by assumption we know that  $\cproj{\mathbb{D}'_1}{\xi}=\cproj{\mathbb{D}'_2}{\xi}$. By definition of relevancy, we know that $\msg{\mb{send}y^{c}_{\delta}\, x^c_\gamma}$ and the tree $\mathbb{A}_1$ connected to it are relevant in $\mathbb{D}_1$ and thus $\mathbb{A}_1$ is equal to $\mathbb{A}_2$, i.e, $\mathbb{A}_1 \Downarrow \xi = \mathbb{A}_2 \Downarrow \xi$.

Removing the tress $\mathbb{A}_1$ and $\mathbb{A}_2$ and the messages from both configurations does not change relevancy of the rest of the configuration since $x^c_{\gamma+1}$ will remain relevant, i.e., $\cproj{\mc{D}''_1}{\xi}=\cproj{\mc{D}''_2}{\xi}$.
 
We can apply the induction hypothesis on the smaller index $m'<m$ and $\mathbb{A}_1 \Downarrow \xi = \mathbb{A}_2 \Downarrow \xi$  and $\mathbb{D}''_1 \Downarrow \xi = \mathbb{D}''_2 \Downarrow \xi$ to get 
 \[\begin{array}{l}
 (\erasure{\mathbb{A}_1}, \erasure{\mathbb{A}_2})\in \mc{E}\llbracket \erasure{\Gamma''} \Vdash \erasure{y_\delta{:}A[c]}\rrbracket^{m'},\,\, \m{and}\\
  (\erasure{\mc{D}''_1}, \erasure{\mathcal{D}''_2})\in \mc{E}\llbracket \erasure{\Gamma', x_{\gamma+1}{:}B[c]}\Vdash \erasure{K^s}\rrbracket^{m'}.
\end{array}\]

By line (10) in the definition of $\mc{V}$:
  \[\begin{array}{l}
 \dagger_1  (\erasure{\mathbb{A}_1\msg{\mb{send} y^{c}_\delta\,x^{c}_\gamma} \mathbb{D}''_1}  ;\ \erasure{\mathbb{A}_2 \msg{\mb{send} y^{c}_\delta\,x^{c}_\gamma} \mathbb{D}''_2}) \in \mc{V}\llbracket \erasure{\Gamma'_1, x_{\gamma}{:}A \multimap B[c]}\Vdash \erasure{K^s} \rrbracket_{\cdot; u'_\delta}^{m'+1}.
\end{array}
\]

\end{description}

\item{\bf Proof of $\dagger_2$.} Consider an arbitrary channel $u'_\delta{:}T \in \mathbf{In}(\erasure{\Gamma} \Vdash \erasure{K})$  We consider different cases based on the session type $T$.
\begin{description}
{{\bf Case 1.} \color{RedViolet} $u'_\delta{:}T[c']= K^s=u_\alpha{:}T_1[c_1]=u_\alpha{:}\&\{\ell{:}A_{\ell}\}_{\ell \in I}[c_1]$}.

 We need to apply the induction hypothesis on the index $m' < m$, but first we need to show that the invariant of the induction holds. Consider an arbitrary label $k \in L$.From $\cproj{\mathbb{D}'_1}{\xi}=\cproj{\mathbb{D}'_2}{\xi}$ and  $c_1 \sqsubseteq \xi$, we get \[\cproj{\mathbb{D}'_1\msg{u^{c_1}_\alpha.k}}{\xi}=\cproj{\mathbb{D}'_2\msg{u^{c_1}_\alpha.k}}{\xi}.\]
 
By induction hypothesis
 \[\begin{array}{l}
        (\erasure{\mathbb{D}'_1},\erasure{\mathbb{D}'_2}) \in \mc{E}\llbracket \erasure{\Gamma} \Vdash \erasure{u_{\alpha+1}{:}A_k[c_1]}\rrbracket^{m'}.
\end{array}\]
By line (3) in the definition of $\mc{V}$:
{
\[\begin{array}{l}
\dagger_2  ( \erasure{\mathbb{D}'_1 \msg{u^{c_1}_\alpha.k}},\erasure{\mathbb{D}'_2 \msg{u^{c_1}_\alpha.k}}) \in \mc{V}\llbracket \erasure{\Gamma} \Vdash \erasure{u_\alpha{:}\&\{\ell:A_{\ell}\}_{\ell \in I}[c_1]}\rrbracket_{u'_\delta; \cdot}^{m'+1}.
\end{array}\]}

{\bf Case 2.}  {\color{RedViolet} $u'_\delta{:}T[c']=K^s=u_\alpha{:}T_1[c_1]=u_\alpha{:}A \multimap B[c_1]$.}  
By assumption and $c_1 \sqsubseteq \xi$, we know that $\mathbb{D}'_1$ and $\mathbb{D}'_2$ are relevant. Consider an arbitrary channel $x_\beta$, which is not a free name in $\Gamma \Vdash K^s$. We know $c_1 \sqsubseteq \xi$, and thus adding a negative message sending $x_\beta[c_1]$ along $u_\alpha[c_1]$ does not change the relevancy of the rest of processes: \[\cproj{\mathbb{D}'_1\msg{\mb{send}x^{c_1}_{\beta}\, u^{c_1}_{\alpha}}}{\xi}=\cproj{\mathbb{D}'_2\msg{\mb{send}x^{c_1}_{\beta}\, u^{c_1}_{\alpha}}}{\xi}\]

We can apply the induction hypothesis on $m' < m$ to get
\small\[\begin{array}{l}{ 
         (\erasure{\mathbb{D}'_1 \msg{\mb{send}x^{c_1}_{\beta}\, u^{c_1}_{\alpha}}}, \erasure{\mathbb{D}'_2 \msg{\mb{send}x^{c_1}_{\beta}\, u^{c_1}_{\alpha}}}) \in \mc{E}\llbracket \erasure{\Gamma, x_\beta{:}A[c_1]} \Vdash \erasure{u_{\alpha+1}{:}B[c_1]}\rrbracket^{m'}.}
\end{array}\]

By line (5) in the definition of $\mc{V}$:
{\small 
\[\begin{array}{l}
\dagger_2  (\erasure{\mathbb{D}'_1}, \erasure{\mathbb{D}'_2})\in \mc{V}\llbracket \erasure{\Gamma}\Vdash \erasure{u_\alpha{:}A\multimap B[c_1]}\rrbracket_{u'_\delta; \cdot}^{m'+1}.
\end{array}\]}

{{\bf Case 3.} \color{RedViolet} $\Gamma=\Gamma', x_\gamma{:}1[c]$}, and $u'_\delta{:}T'[c']=x_\gamma{:}1[c]$

 We first briefly explain why the invariant of induction holds after we bring the closing message inside $\mathbb{D}'_i$ and remove the channel $x_\gamma^c$ from $\Delta$, i.e.
\[\cproj{\msg{\mb{close}\,x^{c}_\gamma}\mathbb{D}'_1}{\xi}= \cproj{\msg{\mb{close}\,x^{c}_\gamma}\mathbb{D}'_2}{\xi}. \]
\begin{itemize}
  \item If the parent of $\msg{\mb{close}\,x^{c}_\gamma}$ in $\mathbb{D}'_1$ is relevant in  $\mathbb{D}'_1$ before bringing the message inside then by the assumption of the theorem, it is the same as the parent of $\msg{\mb{close}\,x^{c}_\gamma}$ in $\mathbb{D}'_2$. If after adding  the message $\mathbb{D}'_1$ the parent still remains relevant, it means that it has at least one other relevant channel other than $x^c_\gamma$ in  $\mathbb{D}'_1$ which also exists in $\mathbb{D}'_2$ and will be relevant after adding the message to $\mathbb{D}'_2$. If after adding the message, the message's parent in $\mathbb{D}'_1$ becomes irrelevant, it means that it does not have a relevant path to any other channel in $\Delta'$ and $K$. By the assumption of theorem the parent of the message in $\mathbb{D}'_2$ does not have such path either. The same argument holds for any other node that becomes irrelevant because of adding the closing message to $\mathbb{D}'_i$. As a result the same processes becomes irrelevant in both $\mathbb{D}'_1$ and $\mathbb{D}'_2$ after adding the message to them and the proof of this case is complete. The same argument holds for the case in which the parent of $\msg{\mb{close}\,x^{c}_\gamma}$ in $\mathbb{D}'_2$ before adding the message is relevant. 
  \item Otherwise, the parent of  $\msg{\mb{close}\,x^{c}_\gamma}$ is irrelevant in $\mathbb{D}'_1$ and $\mathbb{D}'_2$ before adding the message and remains irrelevant after that too. The proof in this case is straightforward.
\end{itemize}

Now that the invariant holds, we can apply the induction hypothesis on the smaller index $m'<m$: 
\[ (\erasure{\msg{\mb{close}\,x^{c}_\gamma}\mathbb{D}'_1} ; \erasure{\msg{\mb{close}\,x^{c}_\gamma}\mathbb{D}'_2}) \in \mc{E}\llbracket \erasure{\Gamma'} \Vdash \erasure{K^s}\rrbracket^{m'}.\]

By line (6) in the definition of $\mc{V}$,  
{ \[\dagger_2 \,(\erasure{\mathbb{D}'_1};\erasure{\mathbb{D}'_2}) \in \mc{V}\llbracket \erasure{\Gamma', x_\gamma{:}1[c]} \Vdash \erasure{K^s}\rrbracket_{u'_\delta; \cdot}^{m'+1}\]}

{{\bf Case 4.} \color{RedViolet} $\Gamma=\Gamma', x_\gamma{:}\oplus\{\ell:A_\ell\}_{\ell \in L}[c]$}, and $u'_\delta{:}T'[c']=x_\gamma{:}\oplus\{\ell:A_\ell\}_{\ell \in L}[c]$.

Consider an arbitrary label $k\in L$. We have $\cproj{\msg{x^{c}_\gamma.k}\mathbb{D}'_1}{\xi}=\cproj{\msg{x^{c}_\gamma.k}\mathbb{D}'_2}{\xi}$, since $x^c_{\gamma+1}$ is relevant and the positive messages $\msg{x^{c}_\gamma.k}$ in both configurations are relevant if and only if their parents are. Thus adding the message does not change relevancy of any other process.
 
We can apply the induction hypothesis on the smaller index $m'<m$ to get
 \[\begin{array}{l}
 (\erasure{\msg{x^{c}_\gamma.k}\mathbb{D}'_1} ; \erasure{\msg{x^{c}_\gamma.k}\mathbb{D}'_2} ) \in \mc{E}\llbracket \erasure{\Gamma', x_{\gamma+1}:A_k[c]}\Vdash \erasure{K^s}\rrbracket^{m'}.
\end{array}\]

By line (7) in the definition of $\mc{V}$:
{
\[\begin{array}{l}
\dagger_2  (\erasure{\mathbb{D}'_1} ; \mathbb{D}'_2 )\in \mc{V}\llbracket \erasure{\Gamma', x_\gamma{:}\oplus\{\ell:A_{\ell}\}_{\ell \in L}[c]}\Vdash \erasure{K^s}\rrbracket_{u'_\delta;\cdot}^{m'+1}.
\end{array}\]}

{\bf Case 5.
\color{RedViolet} $\Gamma=\Gamma', x_\gamma{:}(A \otimes B)[c]$}, and $u'{:}T'[c']= x_\gamma{:}(A \otimes B)[c]$.
 
Consider an arbitrary channel $y_\eta$ which is not a free name in $\Gamma \Vdash K^s$. We have \[\msg{\mb{send}y_\eta^{c}\,x_\gamma^{c}}\cproj{\mathbb{D}'_1}{\xi}\meq\cproj{\msg{\mb{send}y_\eta^{c}\,x_\gamma^{c}}\mathbb{D}'_2}{\xi}:\] The quasi-running secrecy of  $\msg{\mb{send}y_\eta^{c}\,x_\gamma^{c}}$ is lower than or equal to the observer level if the quasi-running secrecy of its parent is lower than or equal to the observer level. So the relevancy of the parent of the message and thus the rest of configurations do not change by adding the message to the configuration.

We can apply the induction hypothesis on  the smaller index $m'<m$ to get 
 \[\begin{array}{l}
        (\erasure{\msg{\mb{send}y_\eta^{c}\,x_\gamma^{c}}\mathbb{D}'_1} ; \erasure{\msg{\mb{send}y_\eta^{c}\,x_\gamma^{c}}}) \in  \mc{E}\llbracket \erasure{\Gamma',  y_\eta{:}A[c], x_{\gamma+1}{:}B[c]}\Vdash \erasure{K^s}\rrbracket^{m'}.
\end{array}\]

By line (9) in the definition of $\mc{V}$:

\[\begin{array}{l}
  (\erasure{\mathbb{D}'_1} ; \erasure{\mathcal{D}'_2} ) \in \mc{V}\llbracket \erasure{\Gamma', x_\gamma{:}A \otimes B[c]}\Vdash \erasure{K^s}\rrbracket_{u'_\delta; \cdot}^{m'+1}.
\end{array}\]

\end{description}
\end{description}
 
\end{proof}

\section{Diamond property, confluence, and backward/forward closures}
\label{apx:sec:confluence-closures}

This section introduces various supporting lemmas, asserting the diamond property and confluence, as well as forward and backward closure.
These lemmas are used in the proofs of \Cref{apx:lem:moving-existential} and \Cref{apx:lem:compositionality} introduced in \Cref{apx:sec:confluence-closures},
which are in turn instrumental in proving transitivity (see \Cref{apx:sec:equivalence}) and adequacy (see \Cref{apx:sec:adequacy}).

\subsection{Diamond property, confluence, and minimal sending configuration}\label{apx:sec:confluence-closures:diamond-confluence}

Subsequent lemmas rely on the notions of \emph{active} and \emph{produced} processes and messages, which we define next.

\begin{definition}[Produced processes and messages]\label{apx:def:produced-procs-msgs}
For each dynamic step in \Cref{apx:fig:dynamics}, we say that a process or message is \emph{produced} in the step if it occurs in the post-step (right side of $\mapsto$ notation).
For example, the transition step for $\m{Spawn}$, produces two processes \[\mb{proc}(x_0, (\subst{x_0, \Lambda}{x, \Lambda'} \hat{\gamma}\,{P}))\, \m{and}\,\mb{proc}(y_\alpha,  (\subst{x_0}{x}{Q})),\] and  the transition step for $\multimap_\m{snd}$ produces a message and a process \[ \mb{msg}( \mathbf{send}\,x_\beta\,u_\gamma)\,\m{and}\, \mb{proc}(y_{\alpha}, ([u_{\gamma+1}/u_{\gamma}]P)).\]

\defend
\end{definition}

\begin{definition}[Active processes and messages]\label{apx:def:active-procs-msgs}
For each dynamic step in \Cref{apx:fig:dynamics}, we define the \emph{active} configuration $\mc{A}$ as the set of processes and messages occurring in the pre-step (the left side of $\mapsto$ notation).
For example, in the transition step for $\m{Spawn}$, the active configuration is the single process
\[\mb{proc}(y_\alpha, (x \leftarrow X \leftarrow \Lambda) ; Q),\]
and in the transition step for $\multimap_\m{rcv}$, the active configuration is
\[\mb{proc}(y_\alpha,(w\leftarrow \mathbf{recv}\,y_\alpha; P))\,\mb{msg}(\mathbf{send}\,x_\beta\,y_\alpha).\]

We define the set of active configurations, i.e. $\m{active}$, for $\mc{D}_1 \mapsto^{*_{\Upsilon}}_{\Delta \Vdash K} \mc{D}'_1$ as the union of every step's active configuration.  A process is called active in $\mc{D}_1 \mapsto^{*_{\Upsilon}} \mc{D}'_1$, if it is in the set  $\m{active}$.

\defend
\end{definition}

\begin{lemma}[Uniqueness of process productions]\label{apx:lem:uniqueprod}
Consider $ \Delta \Vdash \mc{D}:: K$, and $\mc{D}\mapsto^* \mc{D}'\mapsto \mc{D}_1 \mb{proc}(x,P)\mc{D}_2$, such that $\mb{proc}(x,P)$ is produced in the last step. Then process $\mb{proc}(x,P)$ does not occur in any of the previous steps $\mc{D}\mapsto^* \mc{D}'$.

The same result hold for the production of a message.
\end{lemma}
\begin{proof}
Observe the following invariant in the dynamics
\begin{enumerate}
    \item In each configuration, there exists at most one process or message offering along a particular generation of $x$.
    \item The offering channel of processes that are not active or produced in the step does not change. 
    \item The offering channel of a negative message is a fresh generation of a channel (and no process offers along it before the message is received).
    \item In the production of $\mb{proc}(x_\alpha,P)$, either (a) $x_\alpha$ is freshly generated (in the case of $\m{Spawn}$ for the callee), or (b) $\mb{proc}(x_\alpha,P)$ replaces another process $\mb{proc}(x_\alpha,P')$ that offers along $x^\alpha$ and has a larger process term, i.e. $|P|< |P'|$ (in $1_\m{rcv}, \oplus_{\m{rcv}}, \otimes_{\m{rcv}}, \&_{\m{snd}}, \multimap_{\m{snd}}$, and $\m{spawn}$ for the continuation of the caller), or (c) $x_\alpha$ is a fresh generation of $x$ (in $ \oplus_{\m{snd}}, \otimes_{\m{snd}}, \&_{\m{rcv}}, \multimap_{\m{rcv}}$), or (d) $\mb{proc}(x_\alpha,P)$ is of the form $\mb{proc}(x_\alpha, \mathsf{Fwd}_{C,x_\alpha \leftarrow y_\beta})$ such that $x_\alpha$ has never occurred as the offering channel of another forwarder process before. 
\end{enumerate} 
Consider production of a process $\mb{proc}(x_\alpha,P)$ and the cases described in 4. If 4.(a) or 4.(c) hold, then by freshness of $x_\alpha$, such process has not been produced before.  It is enough to consider case 4.(b). Assume that there is another occurrence of process $\mb{proc}(x_\alpha,P)$ before this production; By the observations 1, 2 and 4 we get to a contradiction: there must be a chain of productions satisfying 4.(b) with decreasing sizes from the earlier $\mb{proc}(x_\alpha,P)$ to the later one $\mb{proc}(x_\alpha,P)$, which is contradictory.

With a similar reasoning we can prove that if $ \Delta \Vdash \mc{D}:: K$, and $\mc{D}\mapsto^* \mc{D}'\mapsto \mc{D}_1 \mb{msg}(M)\mc{D}_2$, such that $\mb{msg}(M)$ is produced in the last step and $\overline{u}$ is the name of the message resources, then $\mb{msg}(M)$ does not occur in any of $\mc{D}\mapsto^* \mc{D}'$ steps.

\end{proof}

\begin{lemma}[Diamond Property]\label{apx:lem:diamond}
If $\Delta \Vdash \mc{D}_1:: K$ and  $\dagger\; \mc{D}_1 \mapsto^{\Upsilon}_{\Delta \Vdash K}\mc{D}'_1$ and $\dagger'\;\mc{D}_1 \mapsto^{\Upsilon'}_{\Delta \Vdash K} \mc{D}''_1$, and $\mc{D}'_1 \neq \mc{D}''_1$ then there is a configuration $\mc{D}$ such that 
$\star\;\mc{D}'_1 \mapsto^{\Upsilon''}_{\Delta \Vdash K} \mc{D}$, and $\star'\;\mc{D}''_1 \mapsto^{\Upsilon''}_{\Delta \Vdash K} \mc{D}$, where {\small$\Upsilon \cup \Upsilon'= \Upsilon''$}. The messages produced along the channels $\Upsilon \cap \Upsilon'$  are identical in $\mc{D}'_1$ and $\mc{D}''_1$ and $\mc{D}$. 

Moreover, every process in  $\mc{D}'_1$ that is not an active process of $\mc{D}_1 \mapsto^{\Theta';\Upsilon'}_{\Delta \Vdash K} \mc{D}''_1$ is in $\mc{D}$. And every process in  $\mc{D}''_1$ that is not an active process of $\mc{D}_1 \mapsto^{\Theta;\Upsilon}_{\Delta \Vdash K} \mc{D}'_1$ is in $\mc{D}$.
\end{lemma}
\begin{proof}
The proof is straightforward by cases. The key is to build $\star$ (locally) identical to $\dagger'$, and $\star'$ (locally) identical to $\dagger$. 
\end{proof}

\begin{lemma}[Confluence]\label{apx:lem:confluence}
If $ \Delta \Vdash \mc{D}_1:: K$ and  $\dagger\,\mc{D}_1\mapsto^{m'_{\Upsilon}}_{\Delta \Vdash K}\mc{D}'_1$ and $\dagger'\,\mc{D}_1 \mapsto^{n'_{\Upsilon'}}_{\Delta \Vdash K} \mc{D}''_1$, then there is a configuration $\mc{D}$ such that 
$\star\,\mc{D}'_1 \mapsto^{j_{\Upsilon''}}_{\Delta \Vdash K}  \mc{D}$ for some $j \le n'$, and $\star'\,\mc{D}''_1 \mapsto^{k_{\Upsilon''}}_{\Delta \Vdash K} \mc{D}$ for some $k \le m'$, where {\small$\Upsilon \cup \Upsilon' =\Upsilon''$}.  The messages produced along the channels $\Upsilon \cap \Upsilon'$ are identical in $\mc{D}'_1$ and $\mc{D}''_1$ and $\mc{D}_1$. 
The steps in $\star$ are (locally) identical to the steps of $\dagger'$ that do not occur in $\dagger$. And the steps in $\star'$ are (locally) identical to the steps of $\dagger$ that do not occur in $\dagger'$.

Moreover, every process in  $\mc{D}'_1$ that is not an active process of  $\mc{D}_1 \mapsto^{*_{\Upsilon'}}_{\Delta \Vdash K} \mc{D}''_1$ is in $\mc{D}$. And every process in  $\mc{D}''_1$ that is not an active process of $\mc{D}_1\mapsto^{*_{\Upsilon}}_{\Delta \Vdash K}\mc{D}'_1$ is in $\mc{D}$.
\end{lemma}
\begin{proof}
It follows by standard inductions from the diamond property (\Cref{apx:lem:diamond}). The induction is on the pair $(n',m')$. If $n'=0$ and $m'=0$, the proof is straightforward. Similarly, if $n'=1$ and $m'=0$ or $n'=0$ and $m'=1$ the proof is straightforward. For $n'=1$ and $m'=1$, we apply the diamond property (\Cref{apx:lem:diamond}). Assume that $n'=n+1$ and $m'=m+1$. We form the following diagram to sketch the structure of the proof.

  \includegraphics[scale=0.31]{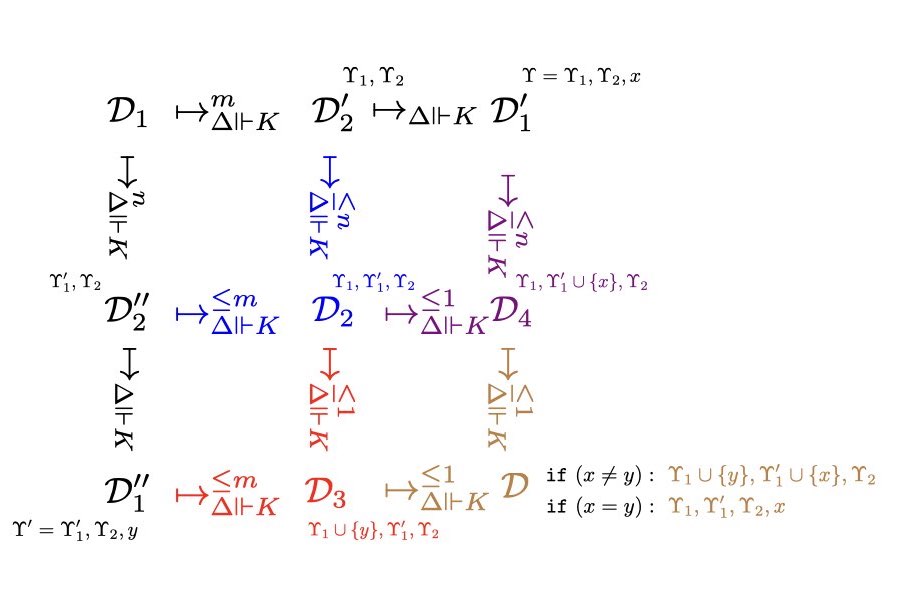}
 
 By induction hypothesis, from $\mc{D}''_2$ and $\mc{D}'_2$, we build $\mc{D}_2$ (with the blue steps) that satisfies the required properties. Then, again by induction we build $\mc{D}_3$ (with the red steps) and $\mc{D}_4$ (with the violet steps). And finally, with a last induction, we build $\mc{D}$ (with the brown steps). The diagram depicts how the required properties move along the steps. For each configuration, we put the set of channels that it sends along them near the configuration. In particular, for $\mc{D}'_1$, we put  $\Upsilon=\Upsilon_1,\Upsilon_2,x$ and for $\mc{D}''_1$, we put $\Upsilon'=\Upsilon'_1,\Upsilon_2,y$. The set $\Upsilon_2$  is in both $\Upsilon$ and $\Upsilon'$, and we have $\Upsilon_1 \cap \Upsilon'_1= \emptyset$, i.e. we put all the common channels (except possibly $x$ and $y$) in $\Upsilon_2$. We assume that the step $\mc{D}'_2 \mapsto \mc{D}'_1$ produces a message along channel $x$. In the case that this step does not produce any message we can simply ignore $x$. Similarly, we assume that the step $\mc{D}''_2 \mapsto \mc{D}''_1$ produces a message along channel $y$. In the case that this step does not produce any message we can simply ignore $y$. 
 At the end, we can build $\mc{D}$ with at most $m+1$ steps from $\mc{D}''_1$ and at most $n+1$ steps from $\mc{D}'_1$. The configuration $\mc{D}$ sends messages along the union of $\Upsilon$ and $\Upsilon'$, and by induction the messages along the intersection of $\Upsilon \cap \Upsilon'$ are identical in $\mc{D}'_1$ and $\mc{D}''_1$ and $\mc{D}$.
 The proof of the second part of the lemma is straightforward by stating the required property for each inductive step and passing them down to $\mc{D}$ in the diagram.
 
\end{proof}

As a straightforward corollary to the first part of the confluence lemma, we get that the messages produced along the channels $\Upsilon \cap \Upsilon'$ are identical in ${\mc{D}'_1}$ and ${\mc{D}''_1}$,
\ie identical messages will be produced along the same channels, independent of the non-deterministic path that we take to produce them. 

\begin{corollary}[Active set independent of non-determinism]\label{apx:cor:uniquepre}
If $\mc{D}_1 \mapsto^* \mc{D}'_1 \mc{A} \mc{D}''_1 \mapsto \mc{D}^1_1 \mb{proc}(x,P) \mc{D}^2_1$ and $\mc{D}_1 \mapsto^* \mc{D}'_2 \mc{A}'\mc{D}''_2 \mapsto \mc{D}^1_2 \mb{proc}(x,P)\mc{D}^2_2$, where $\mc{A}$ and $\mc{A}'$ are the active parts that produce $\mb{proc}(x,P)$, then $\mc{A}=\mc{A}'$. 
\end{corollary}
\begin{proof}
This is another corollary to the second part of the confluence lemma (\Cref{apx:lem:confluence}). First observe that $\mc{A}$ is not active in  $\mc{D}_1 \mapsto^* \mc{D}'_2 \mc{A}'\mc{D}''_2$: if it is active, we produce $\mb{proc}(x,P)$ twice in  $\mc{D}_1 \mapsto^* \mc{D}'_2 \mc{A}'\mc{D}''_2 \mapsto \mc{D}^1_2 \mb{proc}(x,P)\mc{D}^2_2$, which contradicts with uniqueness of process productions (\Cref{apx:lem:uniqueprod}). Similarly, $\mc{A}'$ is not active in  $\mc{D}_1 \mapsto^* \mc{D}'_1 \mc{A}\mc{D}''_1$. With the second part of the confluence lemma, for some $\mc{D}$, we have $\mc{D}'_1\mc{A}\mc{D}''_1\mapsto^*\mc{D}$ and $\mc{D}'_2\mc{A}'\mc{D}''_2\mapsto^*\mc{D}$ such that both $\mc{A}$ and $\mc{A}'$ occur in $\mc{D}$. If $\mc{A} \neq \mc{A}'$, we can produce $\mb{proc}(x,P)$ twice from $\mc{D}$ which by \Cref{apx:lem:uniqueprod} is contradictory. Thus, know that $\mc{A}=\mc{A}'$.

Note: here we rely on the fact that the steps in \Cref{apx:fig:dynamics} produce the post-steps uniquely from the pre-steps. In particular, in the spawn rule we assume that the fresh channel name is uniquely determined based on the offering channel and the process term of the caller.
\end{proof}

Similarly, we can prove that if $\mc{D}_1 \mapsto^* \mc{D}'_1 \mc{A}\mc{D}''_1 \mapsto \mc{D}^1_1 \mb{msg}(M)\mc{D}^2_1$ and $\mc{D}_1 \mapsto^* \mc{D}'_2 \mc{A}' \mc{D}''_2 \mapsto \mc{D}^1_2 \mb{msg}(M)\mc{D}^2_2$,
where $\mc{A}$ and $\mc{A}'$ are the active parts that produce $\mb{msg}(M)$, then $\mc{A}=\mc{A}'$.

In the proof of \Cref{apx:cor:uniquepre}, we used the fact that the pre-step for each substitution is unique.
This can be proved independently by a straightforward observation that if $\mc{D}_1 \mapsto^* \mc{D}'_1 \mc{A}\mc{D}''_1 \mapsto \mc{D}^1_1 [x_\beta/y_\alpha] \mc{D}^2_1$ and $\mc{D}_1 \mapsto^* \mc{D}'_2 \mc{A}' \mc{D}''_2 \mapsto \mc{D}^1_2 [x_\beta/y_\alpha] \mc{D}^2_2$, where $\mc{A}$ and $\mc{A}'$ are forwarding processes that step by renaming the resource $y_\alpha$ to $x_\beta$ ($[x_\beta/y_\alpha]$), then $\mc{A}=\mc{A}'$.

\begin{lemma}[Building minimal sending configuration]\label{apx:lem:least}

Consider {\small$\Delta \Vdash \mc{D}_2:: K$}, a set of channels {\small$\Upsilon_1 \subseteq \Delta, K$} and {\small$\mc{D}_2 \mapsto^{*_{\Upsilon_2}} \mc{D}'_2$} for some {\small$\Upsilon_2 \supseteq \Upsilon_1$}. There exists a set {\small$\Upsilon$} and a configuration {\small$\mc{D}''_2$} such that {\small$\Upsilon_1 \subseteq \Upsilon \subseteq \Upsilon_2$} and {\small$\mc{D}_2\mapsto^{*_{\Upsilon}} \mc{D}''_2$} and {\small$\forall \mc{D}^1_2, \Upsilon_3 \supseteq \Upsilon_1.\, \m{if} \,\mc{D}_2 \mapsto^{*_{\Upsilon_3}} \mc{D}^1_2\,\m{then}\,\mc{D}''_2 \mapsto^{*_{\Upsilon_3}} \mc{D}^1_2$}. We call {\small$\mc{D}''_2$} and {\small$\Upsilon$} the minimal sending configuration and the minimal sending set with respect to {\small$\Upsilon_1$} and {\small$\mc{D}_2$}, respectively. 
\end{lemma}
\begin{proof}
We first provide an algorithm to build {\small$\mc{D}''_2$} and {\small$\Upsilon$} based on the transition steps in {\small$\mc{D}_2 \mapsto^{*_{\Upsilon_2}} \mc{D}'_2$} and the set {\small$\Upsilon_1$} and we show that {\small$\mc{D}_2 \mapsto^{*_{\Upsilon}} \mc{D}''_2$} and {\small$\mc{D}'_2 \mapsto^{*_{\Upsilon_2}} \mc{D}''_2$} and {\small$\Upsilon_1 \subseteq \Upsilon \subseteq \Upsilon_2$}. Then, we prove that if we apply the algorithm on every {\small$\mc{D}^1_2$} with {\small$\mc{D}_2 \mapsto^{*_{\Upsilon_1}} \mc{D}^1_2$} and {\small$\Upsilon_1 \subseteq \Upsilon_3$}, we build the same {\small$\mc{D}''_2$} and {\small$\Upsilon$} that satisfies {\small$\mc{D}''_2 \mapsto^{*_{\Upsilon_3}} \mc{D}^1_2$}.

\begin{algorithm}
\caption{Building the minimal sending configuration}\label{apx:alg:leastsending}
\begin{algorithmic}
\Require A set $\Upsilon_1$ of channels, configuration $\mc{D}_2$, and a configuration $\mc{D}$ with $\mc{D}_2 \mapsto^{*_{\Upsilon_2}} \mc{D}$ and $\Upsilon_1 \subseteq \Upsilon_2$.
\Ensure A set $\Upsilon$, and a configuration $\mc{D}''_2$ with $\mc{D}_2 \mapsto^{*_{\Upsilon}} \mc{D}''_2$ and $\mc{D}''_2 \mapsto^{*_{\Upsilon_2}} \mc{D}$ and $\Upsilon_1 \subseteq \Upsilon \subseteq \Upsilon_2$.
\State $S:= \mathtt{the \; local\; transition \; steps \; in\;} \mc{D}_2 \mapsto^{*_{\Upsilon_2}} \mc{D}$
\State $i := 0$
\State $X_0:=\emptyset$
\State $M:= \mathtt{the\; messages\; in \;\mc{D} \;along\; } \Upsilon_1$
\State $A:=\emptyset$
\While{($M \neq \emptyset$)}
\For{($s\, \mathtt{in}\,  S$)}
    \If{( $\exists p\in post(s)$. $p \in M$)} \Comment{post(s) is the list on the right-hand side of s.}
        \State $A:=A \cup \{pre(s)\}$ \Comment{pre(s) is the set of processes on the left-hand side of s. }
        \State $X_i:=X_i \cup \{s\}$ 
    \EndIf
\EndFor
\State $i:=i+1$
\State $X_i=\emptyset$
\State $M:=A$
\State $A:= \emptyset$

\EndWhile\\ 
\State $\m{Cfg}:=\mc{D}_2$
\State $\m{Tns}:=\cdot$ 
\State $i=i-1$
\While{($i>=0$)}
\For{($s \; \m{in}\;  X_i$)}
    \State $\m{Cfgpost}:=\m{the\; global\; post\mbox{-}state\; when\; the\; local \; step\;} s \m{\; applies\; to\; the\; configuration\;} \m{Cfg}.$
    \State $\m{Tns}.\m{append}(\m{Cfg}\, \mapsto \m{Cfgpost})$
    \State $\m{Cfg}:=\m{Cfgpost}$
\EndFor
\State $i=i-1$
\EndWhile
\State $\mc{D}''_2:= \m{Cfg}$
\State $\Upsilon:= \m{Send}(\mc{D}''_2)$\Comment{$\Upsilon$ is a set of channels along which $\mc{D}''_2$ is ready to send.}\\
\Return{$\mc{D}''_2, \m{Tns}, \Upsilon$} 
\end{algorithmic}
\end{algorithm}
 For every given $\Upsilon_1$, $\mc{D}_2$, and $\mc{D}$ with $\mc{D}_2 \mapsto^{*_{\Upsilon_2}} \mc{D}$ and $\Upsilon_1 \subseteq \Upsilon_2$, 
 \Cref{apx:alg:leastsending} returns a configuration $\mc{D}''$ and set $\Upsilon$ and builds the dynamic transition $\mc{D}_2 \mapsto^{*_{\Upsilon}} \mc{D}''_2$  in $\m{Tns}$. Observe that the steps of $\mc{D}_2 \mapsto^{*_{\Upsilon}} \mc{D}''_2$ are all local transitions of $\mc{D}_2 \mapsto^{*_{\Upsilon_2}} \mc{D}$. Thus, by the confluence lemma (\Cref{apx:lem:confluence}), we know that $\mc{D}''_2 \mapsto^{*_{\Upsilon_2}} \mc{D}$.
 
It remains to be shown that the $\mc{D}''_2$ that \Cref{apx:alg:leastsending} builds is independent of the choice of $\mc{D}$ and $\Upsilon_2$, and is uniquely identified for each $\mc{D}_2$ and $\Upsilon_1$. By the confluence lemma, we know that for each $\mc{D}$ with $\mc{D}_2\mapsto^{*_{\Upsilon_2}}\mc{D}$ and $\Upsilon_1 \subseteq \Upsilon_2$, the algorithm initialized set $M$ with the same messages. In the first $\mathbf{for}$ loop, the algorithm collects the generators of the set $M$ in $A$ and the local steps that produces the set $M$ in $X_i$. By \Cref{apx:cor:uniquepre},  the set of generators of a set $M$ is the same for each $\mc{D}$ with $\mc{D}_2\mapsto^{*_{\Upsilon_2}}\mc{D}$. Similarly, the set of local steps that produces the set $M$ from its generators is the same despite the choice of the configuration $\mc{D}$ with $\mc{D}_2\mapsto^{*_{\Upsilon_2}}\mc{D}$. As a result, we collect the same local transition steps in all $X_i$s  for every $\mc{D}$ with $\mc{D}_2\mapsto^{*_{\Upsilon_2}}\mc{D}$.  The local transition steps in all $X_i$s are those with which we construct $\mc{D}''_2$ from $\mc{D}_2$ and from $\mc{D}''_2$ we can uniquely identify the set $\Upsilon$, and the proof is complete. 
\end{proof}

\subsection{Backward closure}\label{apx:sec:confluence-closures:backward-closure}

\begin{lemma}[Backward closure on the second run]\label{apx:lem:backsecondrun} The second run enjoys backward closure:
\begin{enumerate}
 \item If $(\mc{D}_1;\mc{D}_2) \in \mc{E}\llbracket \Delta\Vdash K \rrbracket^k$ 
 and for $\mc{D}''_2 \in \m{Tree}(\Delta \Vdash K)$, we have $\mc{D}''_2\mapsto^* \mc{D}_2$ then
 $(\mc{D}_1;\mc{D}''_2) \in \mc{E}\llbracket \Delta\Vdash K \rrbracket^k$. 

\item If $(\mc{D}_1;\mc{D}_2) \in \mc{V}\llbracket \Delta\Vdash K \rrbracket_{\cdot;y_\alpha}^{k+1}$ and for $\mc{D}''_2 \in \m{Tree}(\Delta \Vdash K)$, we have $\mc{D}''_2\mapsto^* \mc{D}_2$ with $\mc{D}''_2$ sending along channel $y_\alpha$, 
then
$(\mc{D}_1;\mc{D}''_2) \in \mc{V}\llbracket \Delta\Vdash K \rrbracket_{\cdot;y_\alpha}^{k+1}.$ 
\item If $(\mc{D}_1;\mc{D}_2) \in \mc{V}\llbracket \Delta\Vdash K \rrbracket_{y_\alpha;\cdot}^{k+1}$ and for $\mc{D}''_2 \in \m{Tree}(\Delta \Vdash K)$, we have $\mc{D}''_2\mapsto^* \mc{D}_2$, then
$(\mc{D}_1;\mc{D}''_2)\in \mc{V}\llbracket \Delta\Vdash K \rrbracket_{y_\alpha; \cdot}^{k+1}.$
\end{enumerate}

\end{lemma}
\begin{proof}
We prove the first statement separately and then use it to prove the second and third statements.

\begin{enumerate}
    \item If $k=0$, the proof is trivial. Consider $k=m+1$.
    By assumption, we have $(\mc{D}_1;\mc{D}_2)  \in \mc{E}\llbracket \Delta\Vdash K \rrbracket^{m+1}$.
    By line (12) of the logical relation we know

   \[\begin{array}{l}
(\star)\,\,\forall\, \Upsilon_1, \Theta_1, \mc{D}'_1. \,\m{if}\,  \mc{D}_1 \mapsto^{*_{ \Upsilon_1; \Theta_1}} \mc{D}'_1,\, \m{then}\, \exists \Upsilon_2, \mc{D}'_2 \,\m{such\, that} \, \,\mc{D}_2\mapsto^{{*}_{ \Upsilon_2}} \mc{D}'_2\, \m{and}\\ 
 \; \; \forall\, x_\alpha \in \mb{Out}(\Delta \Vdash K).\, \mb{if}\, x_\alpha \in \Upsilon_1.\, \mb{then}\,\, (\mc{D}'_1; \mc{D}'_2)\in \mc{V}^\xi_{\Psi}\llbracket  \Delta \Vdash K \rrbracket_{\cdot;x_\alpha}^{m+1}\,\m{and}\\ 
 \; \; \forall\, x_\alpha \in \mb{In}(\Delta \Vdash K).\, \mb{if}\, x_\alpha \in \Theta_1.\, \mb{then}\,\, (\mc{D}'_1; \mc{D}'_2)\in \mc{V}^\xi_{\Psi}\llbracket  \Delta \Vdash K \rrbracket_{x_\alpha;\cdot}^{m+1}.
\end{array}\]
Consider $\mc{D}''_2$, for which by the assumption we have $\mc{D}''_2\mapsto^* \mc{D}_2$. Our goal is to prove $( \mc{D}_1;\mc{D}''_2)  \in \mc{E}\llbracket \Delta\Vdash K \rrbracket^{m+1}$.
We need to show that 
\[\begin{array}{l}
    (\dagger)\,\,\forall\, \Upsilon_1, \Theta_1, \mc{D}'_1.\, \m{if}\,  \mc{D}_1 \mapsto^{*_{ \Upsilon_1; \Theta_1}} \mc{D}'_1,\, \m{then}\, \exists \Upsilon_2, \mc{D}'_2 \,\m{such\, that} \, \,\mc{D}''_2\mapsto^{{*}_{ \Upsilon_2}} \mc{D}'_2\, \m{and}\, \Upsilon_1 \subseteq \Upsilon_2\, \m{and}\\ 
     \; \; \forall\, x_\alpha \in \mb{Out}(\Delta \Vdash K).\, \mb{if}\, x_\alpha \in \Upsilon_1.\, \mb{then}\,\, (\mc{D}'_1; \mc{D}'_2)\in \mc{V}^\xi_{\Psi}\llbracket  \Delta \Vdash K \rrbracket_{\cdot;x_\alpha}^{m+1}\,\m{and}\\
      \; \; \forall\, x_\alpha \in \mb{In}(\Delta \Vdash K).\,\mb{if}\, x_\alpha \in \Theta_1.\, \mb{then}\,\,  (\mc{D}'_1; \mc{D}'_2)\in \mc{V}^\xi_{\Psi}\llbracket  \Delta \Vdash K \rrbracket_{x_\alpha;\cdot}^{m+1}.
    \end{array}\]

With a $\forall$-Introduction, an $\m{if}$-Introduction on the goal followed by a $\forall$-Elimination and an $\m{if}$-Elimination on the assumption, we get the assumption 

\[\begin{array}{l}
    (\star')\,\,\exists \Upsilon_2, \mc{D}'_2 \,\m{such\, that} \, \,\mc{D}_2\mapsto^{{*}_{ \Upsilon_2}} \mc{D}'_2\, \m{and}\\ 
     \; \; \forall\, x_\alpha \in \mb{Out}(\Delta \Vdash K).\, \mb{if}\, x_\alpha \in \Upsilon_1.\, \mb{then}\,\, (\mc{D}'_1; \mc{D}'_2)\in \mc{V}^\xi_{\Psi}\llbracket  \Delta \Vdash K \rrbracket_{\cdot;x_\alpha}^{m+1}\,\m{and}\\
      \; \; \forall\, x_\alpha \in \mb{In}(\Delta \Vdash K).\,\mb{if}\, x_\alpha \in \Theta_1.\, \mb{then}\,\,  (\mc{D}'_1; \mc{D}'_2)\in \mc{V}^\xi_{\Psi}\llbracket  \Delta \Vdash K \rrbracket_{x_\alpha;\cdot}^{m+1}.
    \end{array}\]

and the goal 
\[\begin{array}{l}
    (\dagger')\,\,\exists \Upsilon_2,\,\mc{D}'_2 \,\m{such\, that} \, \,\mc{D}''_2\mapsto^{{*}_{ \Upsilon_2}} \mc{D}'_2\, \m{and}\\ 
     \; \; \forall\, x_\alpha \in \mb{Out}(\Delta \Vdash K).\, \mb{if}\, x_\alpha \in \Upsilon_1.\, \mb{then}\,\, (\mc{D}'_1; \mc{D}'_2)\in \mc{V}^\xi_{\Psi}\llbracket  \Delta \Vdash K \rrbracket_{\cdot;x_\alpha}^{m+1}\,\m{and}\\
      \; \; \forall\, x_\alpha \in \mb{In}(\Delta \Vdash K).\,\mb{if}\, x_\alpha \in \Theta_1.\, \mb{then}\,\,  (\mc{D}'_1; \mc{D}'_2)\in \mc{V}^\xi_{\Psi}\llbracket  \Delta \Vdash K \rrbracket_{x_\alpha;\cdot}^{m+1}.
    \end{array}\]
We apply an $\exists$-Elimination on the assumption to get $\Upsilon_2$ and $\mc{D}'_2$ that satisfies the conditions and use the same $\Upsilon_2$ and $\mc{D}'_2$ to instantiate the existential quantifier in the goal. Since $\mc{D}''_2\mapsto^* \mc{D}_2$, we get $\mc{D}''_2\mapsto^{{*}_{\Upsilon_2}} \mc{D}'_2$, and the proof is complete.

\item The proof is by cases on the row of the logical relation that makes the assumption $(\mc{D}_1;\mc{D}_2) \in \mc{V}\llbracket \Delta\Vdash K \rrbracket_{\cdot;y_\alpha}^{k+1}$ true.
Here we only consider an interesting cases, the proof of other cases is similar.

\begin{description}
\item[\bf Row 4.]
By the conditions of this row, we know that $\Delta=\Delta',\Delta''$, and $K=y_\alpha{:}A\otimes B$. Moreover, we have \[\mathcal{D}_1=\mathcal{D}'_1\mathcal{T}_1\mathbf{msg}( \mb{send}x_\beta^c\,y_\alpha^c)\, \m{for}\, \mc{T}_1 \in \m{Tree}_\Psi(\Delta'' \Vdash x_\beta{:}A),\]
\[\mathcal{D}_2=\mathcal{D}'_2\mc{T}_2 \mathbf{msg}( \mb{send},x_\beta^c\,y_\alpha^c)\,\m{for}\, \mc{T}_2 \in \m{Tree}_\Psi(\Delta'' \Vdash x_\beta{:}A),\]
\[ (\dagger_1) \,\,({\mc{T}_1} ; {\mc{T}_2} )  \in  \mathcal{E}^\xi_{ \Psi}\llbracket \Delta'' \Vdash x_{\beta}{:}A[c]\rrbracket^{k}, \, \m{and}\]
\[(\dagger_2)\,\, ({\mc{D}'_1}; {\mc{D}'_2})  \in  \mathcal{E}^\xi_{ \Psi}\llbracket \Delta' \Vdash y_{\alpha+1}{:}B\rrbracket^{k}\]
These are also the statements that we need to prove when replacing $\mc{D}_2$ with $\mc{D}''_2$.
By the assumption that $\mc{D}''_2$ is sesstion-typed and sends along $y$,  uniqueness of channels, and $\mc{D}''_2 \mapsto^* \mc{D}_2$ we get 
\[\mathcal{D}''_2=\mathcal{D}^1_2\mc{T}^1_2 \mathbf{msg}( \mb{send},x_\beta^c\,y_\alpha^c)\,\m{for}\, \mc{T}^1_2 \in \m{Tree}_\Psi(\Delta'' \Vdash x_\beta{:}A[c])\]
Moreover, since $\mc{T}^1_2$ and $\mathcal{D}^1_2$  are disjoint sub-trees with the common parent $\mathbf{msg}( \mb{send},x_\beta^c\,y_\alpha^c)$ and cannot communicate with each other internally, we have
$\mc{T}^1_2 \mapsto^* \mc{T}_2$ and $\mc{D}^1_2 \mapsto^* \mc{D}'_2$.

Now we can apply the the result of part (1) of this lemma on $(\dagger_1)$ and $(\dagger_2)$ to get
\[ (\dagger'_1) \,\,({\mc{T}_1} ; {\mc{T}^1_2} )  \in  \mathcal{E}^\xi_{ \Psi}\llbracket \Delta'' \Vdash x_{\beta}{:}A[c]\rrbracket^{k}, \, \m{and}\]
\[(\dagger'_2)\,\, ({\mc{D}'_1}; {\mc{D}^1_2})  \in  \mathcal{E}^\xi_{ \Psi}\llbracket \Delta' \Vdash y_{\alpha+1}{:}B\rrbracket^{k}\]
and the proof of this subcase is complete.
\end{description}

\item The proof is by considering the row of the logical relation that ensures the assumption {\small$( \mc{D}_1;\mc{D}_2) \in \mc{V}\llbracket \Delta\Vdash K \rrbracket_{y_\alpha;\cdot}^{k+1}$}. We provide the detailed proof for an interesting case, the proof of other cases is similar.

\begin{description}
\item[\bf Row 5.] By the conditions of this row, we know that $K=y_\alpha{:}A\multimap B$. We have 
 $(\mc{D}_1;\mc{D}_2) \in \m{Tree}_\Psi(\Delta \Vdash y_\alpha{:}A\multimap B)$ and \[\;\;(\dagger)\,\forall\, x_\beta \not\in \mathit{dom}(\Delta,K). \, ({\mc{D}_1\mathbf{msg}( \mb{send}x_\beta^c\,y_\alpha^c)}; {\mc{D}_2\mathbf{msg}(\mb{send}x_\beta^c\,y_\alpha^c)})  \in  \mathcal{E}^\xi_{ \Psi}\llbracket \Delta, x_\beta{:}A\Vdash y_{\alpha+1}{:}B\rrbracket^{k}\]

 These are also the statements that we need to prove when replacing $\mc{D}_2$ with $\mc{D}''_2$. The first tree statement is straightforward by the assumption that $\mc{D}''_2$ is session-typed.
Using the local transition steps, we get $\mc{D}''_2\mathbf{msg}(\mb{send}x_\beta^c\,y_\alpha^c)\mapsto^*\mc{D}_2\mathbf{msg}(\mb{send}x_\beta^c\,y_\alpha^c)$. We can apply the result of part (1) of this lemma on $(\dagger)$ to get 
 \[(\dagger')\,\forall\, x_\beta \not \in  \mathit{dom}(\Delta,K).\,(\mc{D}_1\mathbf{msg}( \mb{send}x_\beta^c\,y_\alpha^c); \mc{D}''_2\mathbf{msg}(\mb{send}x_\beta^c\,y_\alpha^c))  \in  \mathcal{E}^\xi_{ \Psi}\llbracket \Delta, x_\beta{:}A\Vdash y_{\alpha+1}{:}B\rrbracket^{k}\]
 which completes the proof of this case.
\item[\bf Row 9.] By the conditions of this row, we know that $\Delta= \Delta', y_\alpha{:}A\otimes B$. We have 

$ \forall x_\beta  \not \in  \mathit{dom}(\Delta;, y_\alpha{:}A \otimes B, K).\,({\mc{D}_1};{\mc{D}_2}) \in \m{Tree}_\Psi(\Delta', y_\alpha{:}A\otimes B \Vdash K)$ and  
\[ (\dagger)\, ({\mathbf{msg}( \mb{send}x_\beta^c\,y_\alpha^c)\mc{D}_1}; {\mathbf{msg}( \mb{send}x_\beta^c\,y_\alpha^c)\mc{D}_2})  \in  \mathcal{E}^\xi_{ \Psi}\llbracket \Delta',x_\beta{:}A, y_{\alpha+1}{:}B \Vdash K \rrbracket^{k} \]
 These are also the statements that we need to prove when replacing $\mc{D}_2$ with $\mc{D}''_2$. The first tree statement is straightforward by the assumption that $\mc{D}''_2$ is session-typed. Using the local dynamic steps we get $\mathbf{msg}( \mb{send}x_\beta^c\,y_\alpha^c) \mc{D}''_2 \mapsto^* \mathbf{msg}( \mb{send}x_\beta^c\,y_\alpha^c)\mc{D}_2$.
  We apply the result of part (1) of this lemma on $\dagger$ to get
 {\small \[ (\dagger')\, \forall x_\beta  \not \in  \mathit{dom}(\Delta;, y_\alpha{:}A \otimes B, K).\,(\mathbf{msg}( \mb{send}x_\beta^c\,y_\alpha^c)\mc{D}_1; \mathbf{msg}( \mb{send}x_\beta^c\,y_\alpha^c) \mc{D}''_2)  \in  \mathcal{E}^\xi_{ \Psi}\llbracket \Delta',x_\beta{:}A, y_{\alpha+1}{:}B \Vdash K \rrbracket^{k} \]}

\end{description}
\end{enumerate}
\end{proof}

\subsection{Forward closure}\label{apx:sec:confluence-closures:forward-closue}

\begin{lemma}[Forward closure on the first run]\label{apx:lem:fwdfirst}
Consider $(\mc{D}_1;\mc{D}_2) \in \mc{E}\llbracket \Delta\Vdash K \rrbracket^k$ and   $\mc{D}_1\mapsto^{*} \mc{D}''_1$. We have
$( \mc{D}''_1;\mc{D}_2) \in \mc{E}\llbracket \Delta\Vdash K \rrbracket^k.$
\end{lemma}
\begin{proof}
If $k=0$ the proof is trivial. Consider $k=m+1$.
By assumption, we have $(\mc{D}_1;\mc{D}_2)  \in \mc{E}\llbracket \Delta\Vdash K \rrbracket^{m+1}$.

By line (12) of the logical relation we get  
 
\[\begin{array}{l}
    \star\,\, \forall\, \Upsilon_1, \Theta_1, \mc{D}'_1. \,\m{if}\,  \mc{D}_1 \mapsto^{*_{ \Upsilon_1; \Theta_1}} \mc{D}'_1,\, \m{then}\, \exists \Upsilon_2, \mc{D}'_2 \,\m{such\, that} \, \,\mc{D}_2\mapsto^{{*}_{ \Upsilon_2}} \mc{D}'_2\, \m{and}\, \Upsilon_1 \subseteq \Upsilon_2\, \m{and}\\ 
     \; \; \forall\, x_\alpha \in \mb{Out}(\Delta \Vdash K).\, \mb{if}\, x_\alpha \in \Upsilon_1.\, \mb{then}\,\, (\mc{D}'_1; \mc{D}'_2)\in \mc{V}^\xi_{\Psi}\llbracket  \Delta \Vdash K \rrbracket_{\cdot;x_\alpha}^{m+1}\,\m{and}\\ 
     \; \; \forall\, x_\alpha \in \mb{In}(\Delta \Vdash K). \,\mb{if}\, x_\alpha \in \Theta_1. \,\mb{then}\, (\mc{D}'_1; \mc{D}'_2)\in \mc{V}^\xi_{\Psi}\llbracket  \Delta \Vdash K \rrbracket_{x_\alpha;\cdot}^{m+1}.
    \end{array}\]

Consider $\mc{D}''_1$, for which by the assumption we have $\mc{D}_1\mapsto^* \mc{D}''_1$. Our goal is to prove $( \mc{D}''_1;\mc{D}_2)  \in \mc{E}\llbracket \Delta\Vdash K \rrbracket^{m+1}$.
We need to show that 
\[\begin{array}{l}
    \dagger,\,\forall\, \Upsilon_1, \Theta_1, \mc{D}'_1. \,\m{if}\,  \mc{D}''_1 \mapsto^{*_{ \Upsilon_1; \Theta_1}} \mc{D}'_1,\, \m{then}\, \exists \mc{D}'_2 \,\m{such\, that} \, \,\mc{D}_2\mapsto^{{*}_{ \Upsilon_2}} \mc{D}'_2\, \m{and}\, \Upsilon_1 \subseteq \Upsilon_2\, \m{and}\\ 
     \; \; \forall\, x_\alpha \in \mb{Out}(\Delta \Vdash K).\, \mb{if}\, x_\alpha \in \Upsilon_1.\, \mb{then}\,\, (\mc{D}'_1; \mc{D}'_2)\in \mc{V}^\xi_{\Psi}\llbracket  \Delta \Vdash K \rrbracket_{\cdot;x_\alpha}^{m+1}\,\m{and}\\ 
     \; \; \forall\, x_\alpha \in \mb{In}(\Delta \Vdash K).\, \,\mb{if}\, x_\alpha \in \Theta_1. \,\mb{then}\,  (\mc{D}'_1; \mc{D}'_2)\in \mc{V}^\xi_{\Psi}\llbracket  \Delta \Vdash K \rrbracket_{x_\alpha;\cdot}^{m+1}.
    \end{array}\]

With a $\forall$-Introduction and an $\m{if}$-Introduction on the goal, we assume $ \mc{D}''_1 \mapsto^{*_{ \Upsilon_1; \Theta_1}} \mc{D}'_1$. By assumption of $\mc{D}_1 \mapsto^* \mc{D}''_1$ we get $\mc{D}_1 \mapsto^{*_{ \Upsilon_1; \Theta_1}} \mc{D}'_1$. We use this to apply $\forall$-Elimination and $\m{if}$- Elimination on the assumption, and get 

\[\begin{array}{l}
    \star'\,\,\exists \Upsilon_2, \mc{D}'_2 \,\m{such\, that} \, \,\mc{D}_2\mapsto^{{*}_{ \Upsilon_2}} \mc{D}'_2\, \m{and} \Upsilon_1 \subseteq \Upsilon_2\, \m{and}\\ 
     \; \; \forall\, x_\alpha \in \mb{Out}(\Delta \Vdash K).\, \mb{if}\, x_\alpha \in \Upsilon_1.\, \mb{then}\,\, (\mc{D}'_1; \mc{D}'_2)\in \mc{V}^\xi_{\Psi}\llbracket  \Delta \Vdash K \rrbracket_{\cdot;x_\alpha}^{m+1}\,\m{and}\\ 
     \; \; \forall\, x_\alpha \in \mb{In}(\Delta \Vdash K).\,\mb{if}\, x_\alpha \in \Theta_1. \,\mb{then}\,  (\mc{D}'_1; \mc{D}'_2)\in \mc{V}^\xi_{\Psi}\llbracket  \Delta \Vdash K \rrbracket_{x_\alpha;\cdot}^{m+1}.
    \end{array}\]

Which exactly matches our goal and the proof is complete.  
\end{proof}

\begin{lemma}[Forward closure on the second run with some specific conditions]\label{apx:lem:fwdsecond}
Consider $(\mc{D}_1;\mc{D}_2) \in \mc{E}\llbracket \Delta\Vdash K \rrbracket^k$ and  $\mc{D}_2\mapsto^{*_{\Upsilon}} \mc{D}''_2$ such that if $\mc{D}_1$ sends along the set $\Upsilon_1$, we have $\Upsilon_1 \subseteq \Upsilon$ and also $\mc{D}''_2$ is the minimal configuration built by \Cref{apx:lem:least} given the set $\Upsilon_1$ and configuration $\mc{D}_2$. We have
$(\mc{D}_1;\mc{D}''_2) \in \mc{E}\llbracket \Delta\Vdash K \rrbracket^k.$
\end{lemma}
\begin{proof}
If $k=0$ the proof is trivial. Consider $k=m+1$.
By assumption, we have $(\mc{D}_1;\mc{D}_2)  \in \mc{E}\llbracket \Delta\Vdash K \rrbracket^{m+1}$.
    By line (12) of the logical relation we get  

\[\begin{array}{l}
\star\,\,\forall\, \Upsilon_1, \Theta_1, \mc{D}'_1. \, \m{if}\,  \mc{D}_1 \mapsto^{*_{ \Upsilon_1; \Theta_1}} \mc{D}'_1,\, \m{then}\, \exists \mc{D}'_2 \,\m{such\, that} \, \,\mc{D}_2\mapsto^{{*}_{ \Upsilon_2}} \mc{D}'_2\, \m{and}\, \Upsilon_1 \subseteq \Upsilon_2\, \m{and}\\ 
    \; \; \forall\, x_\alpha \in \mb{Out}(\Delta \Vdash K).\, \mb{if}\, x_\alpha \in \Upsilon_1.\, \mb{then}\,\, (\mc{D}'_1; \mc{D}'_2)\in \mc{V}^\xi_{\Psi}\llbracket  \Delta \Vdash K \rrbracket_{\cdot;x_\alpha}^{m+1}\,\m{and}\\ 
    \; \; \forall\, x_\alpha \in \mb{In}(\Delta \Vdash K).\,\mb{if}\, x_\alpha \in \Theta_1.\, (\mc{D}'_1; \mc{D}'_2)\in \mc{V}^\xi_{\Psi}\llbracket  \Delta \Vdash K \rrbracket_{x_\alpha;\cdot}^{m+1}.
\end{array}\]

Consider $\mc{D}''_2$, for which by the assumption we have $\mc{D}_2\mapsto^{*_\Upsilon} \mc{D}''_2$. Our goal is to prove $(\mc{D}_1;\mc{D}''_2)  \in \mc{E}\llbracket \Delta\Vdash K \rrbracket^{m+1}$.
We need to show that 
\[\begin{array}{l}
\dagger\,\forall\, \Upsilon_1, \Theta_1, \mc{D}'_1. \,\,\m{if}\,  \mc{D}_1 \mapsto^{*_{ \Upsilon_1; \Theta_1}} \mc{D}'_1,\, \m{then}\, \exists \mc{D}'_2 \,\m{such\, that} \, \,\mc{D}''_2\mapsto^{{*}_{ \Upsilon_2}} \mc{D}'_2\, \m{and} \Upsilon_1 \subseteq \Upsilon_2\, \m{and}\\ 
\; \; \forall\, x_\alpha \in \mb{Out}(\Delta \Vdash K).\, \mb{if}\, x_\alpha \in \Upsilon_1.\, \mb{then}\,\, (\mc{D}'_1; \mc{D}'_2)\in \mc{V}^\xi_{\Psi}\llbracket  \Delta \Vdash K \rrbracket_{\cdot;x_\alpha}^{m+1}\,\m{and}\\ 
\; \; \forall\, x_\alpha \in \mb{In}(\Delta \Vdash K).\,\mb{if}\, x_\alpha \in \Theta_1.\, (\mc{D}'_1; \mc{D}'_2)\in \mc{V}^\xi_{\Psi}\llbracket  \Delta \Vdash K \rrbracket_{x_\alpha;\cdot}^{m+1}.
\end{array}\]

With an $\m{if}$-Introduction on the goal followed by an $\m{if}$- Elimination on the assumption, we get the assumption 
\[\begin{array}{l}
    \star'\,\,\exists \Upsilon_2, \mc{D}'_2 \,\m{such\, that} \, \,\mc{D}_2\mapsto^{{*}_{ \Upsilon_1; \Theta_1}} \mc{D}'_2\, \m{and}\Upsilon_1 \subseteq \Upsilon_2\, \m{and}\\ 
        \; \; \forall\, x_\alpha \in \mb{Out}(\Delta \Vdash K).\, \mb{if}\, x_\alpha \in \Upsilon_1.\, \mb{then}\,\, (\mc{D}'_1; \mc{D}'_2)\in \mc{V}^\xi_{\Psi}\llbracket  \Delta \Vdash K \rrbracket_{\cdot;x_\alpha}^{m+1}\,\m{and}\\ 
        \; \; \forall\, x_\alpha \in \mb{In}(\Delta \Vdash K).\,\mb{if}\, x_\alpha \in \Theta_1.\, (\mc{D}'_1; \mc{D}'_2)\in \mc{V}^\xi_{\Psi}\llbracket  \Delta \Vdash K \rrbracket_{x_\alpha;\cdot}^{m+1}.
    \end{array}\]

and the goal 
\[\begin{array}{l}
    \dagger'\,\,\exists \Upsilon_2, \mc{D}'_2 \,\m{such\, that} \, \,\mc{D}''_2\mapsto^{{*}_{ \Upsilon_2}} \mc{D}'_2\, \m{and} \Upsilon_1 \subseteq \Upsilon_2\, \m{and}\\ 
    \; \; \forall\, x_\alpha \in \mb{Out}(\Delta \Vdash K).\, \mb{if}\, x_\alpha \in \Upsilon_1.\, \mb{then}\,\, (\mc{D}'_1; \mc{D}'_2)\in \mc{V}^\xi_{\Psi}\llbracket  \Delta \Vdash K \rrbracket_{\cdot;x_\alpha}^{m+1}\,\m{and}\\ 
    \; \; \forall\, x_\alpha \in \mb{In}(\Delta \Vdash K).\,\mb{if}\, x_\alpha \in \Theta_1.\,x (\mc{D}'_1; \mc{D}'_2)\in \mc{V}^\xi_{\Psi}\llbracket  \Delta \Vdash K \rrbracket_{x_\alpha;\cdot}^{m+1}.
    \end{array}\]
We apply an $\exists$-Elimination on the assumption to get $\mc{D}'_2$ that satisfies the conditions, i.e., $\mc{D}_2 \mapsto^{*_{\Upsilon_2}} \mathcal{D}'_2$ and $\Upsilon_1 \subseteq \Upsilon_2$. We use the same $\mc{D}'_2$ to instantiate the existential quantifier in the goal, we need to show that $\mc{D}''_2 \mapsto^{*_{\Upsilon_1}} \mathcal{D}'_2$. Since $\mc{D}''_2$ is the minimal configuration built for $\Upsilon_1$ and $\mc{D}_2$, and $\Upsilon_1 \subseteq \Upsilon_2$, by Lemma~\ref{apx:lem:least} we get $\mc{D}''_2\mapsto^{*_{\Upsilon_2}} \mc{D}'_2$,  and the proof is complete.  
\end{proof}

\section{Moving existential and compositionality}
\label{apx:sec:extsiscompose}
This section introduces two lemmas, \Cref{apx:lem:moving-existential} and \Cref{apx:lem:compositionality},
which are instrumental in proving transitivity (see \Cref{apx:sec:equivalence}) and adequacy (see \Cref{apx:sec:adequacy}).

\subsection{Moving existential over universal quantifier}\label{apx:sec:confluence-closures:moving-existential}
\begin{lemma} [Moving existential over universal quantifier]\label{apx:lem:moving-existential}
if we have 
 
\[\begin{array}{l}
    (\dagger)\,\,\forall m.\, \forall\, \Upsilon_1, \Theta_1, \mathcal{D}'_1.\, \, \m{if}\,  \mc{D}_1 \mapsto^{*_{ \Upsilon_1; \Theta_1}} \mc{D}'_1,\, \m{then}\, \exists \Upsilon_2, \mc{D}'_2 \,\m{such\, that} \, \,\mc{D}_2\mapsto^{{*}_{ \Upsilon_2}} \mc{D}'_2\, \m{and}\, \Upsilon_1 \subseteq \Upsilon_2 \, \m{and}\\ 
     \; \; \forall\, x_\alpha \in \mb{Out}(\Delta \Vdash K).\, \mb{if}\, x_\alpha \in \Upsilon_1.\, \mb{then}\,\, (\mc{D}'_1; \mc{D}'_2)\in \mc{V}\llbracket  \Delta \Vdash K \rrbracket_{\cdot;x_\alpha}^{m+1}\,\m{and}\\ 
     \; \; \forall\, x_\alpha \in \mb{In}(\Delta \Vdash K).\mb{if}\, x_\alpha \in \Theta_1.\,\mb{then}\, (\mc{D}'_1; \mc{D}'_2)\in \mc{V}\llbracket  \Delta \Vdash K \rrbracket_{x_\alpha;\cdot}^{m+1}.
    \end{array}\]

then 
\[\begin{array}{l}
    \forall\, \Upsilon_1, \Theta_1, \mathcal{D}'_1. \m{if}\,  \mc{D}_1 \mapsto^{*_{ \Upsilon_1; \Theta_1}} \mc{D}'_1,\, \m{then}\, \exists \Upsilon_2,\mc{D}'_2 \,\m{such\, that} \, \,\mc{D}_2\mapsto^{{*}_{ \Upsilon_2}} \mc{D}'_2\, \m{and}\, \Upsilon_1 \subseteq \Upsilon_2 \, \m{and}\\ 
    \; \; \forall\, x_\alpha \in \mb{Out}(\Delta \Vdash K).\, \mb{if}\, x_\alpha \in \Upsilon_1.\, \mb{then}\,\, \forall \,k.\, (\mc{D}'_1; \mc{D}'_2)\in \mc{V}\llbracket  \Delta \Vdash K \rrbracket_{\cdot;x_\alpha}^{k+1}\,\m{and}\\ 
    \; \; \forall\, x_\alpha \in \mb{In}(\Delta \Vdash K). \mb{if}\, x_\alpha \in \Theta_1.\,\mb{then}\,\,\forall k.\, (\mc{D}'_1; \mc{D}'_2)\in \mc{V}\llbracket  \Delta \Vdash K \rrbracket_{x_\alpha;\cdot}^{k+1}.
\end{array}\]
\end{lemma}

\begin{proof}
First put $m=1$ to apply $\forall\, \mathbf{E}.$ on the assumption (instantiating $\forall\, m$ only). We get as an assumption
\[\begin{array}{l}
 (\dagger')\,\, \forall\, \Upsilon_1, \Theta_1, \mathcal{D}'_1.\, \, \m{if}\,  \mc{D}_1 \mapsto^{*_{ \Upsilon_1; \Theta_1}} \mc{D}'_1,\, \m{then}\, \exists \Upsilon_2, \mc{D}'_2 \,\m{such\, that} \, \,\mc{D}_2\mapsto^{{*}_{ \Upsilon_2}} \mc{D}'_2\, \m{and}\, \Upsilon_1 \subseteq \Upsilon_2 \, \m{and}\\ 
     \; \; \forall\, x_\alpha \in \mb{Out}(\Delta \Vdash K).\, \mb{if}\, x_\alpha \in \Upsilon_1.\, \mb{then}\,\, (\mc{D}'_1; \mc{D}'_2)\in \mc{V}\llbracket  \Delta \Vdash K \rrbracket_{\cdot;x_\alpha}^{0+1}\,\m{and}\\ 
     \; \; \forall\, x_\alpha \in \mb{In}(\Delta \Vdash K).\mb{if}\, x_\alpha \in \Theta_1.\,\mb{then}\, (\mc{D}'_1; \mc{D}'_2)\in \mc{V}\llbracket  \Delta \Vdash K \rrbracket_{x_\alpha;\cdot}^{0+1}.
    \end{array}\]

Next, apply a $\forall \mathbf{I}.$ and $\mathbf{if}\, \mathbf{I}.$ on the goal followed by a corresponding $\forall \mathbf{E}.$ and $\mathbf{if}\, \mathbf{E}.$ on the assumption $(\dagger')$.  Now apply $\exists \mathbf{E}.$ on the assumption $(\dagger')$ to get $\mc{D}'_2$ such that $\mc{D}_2 \mapsto^{*_{\Upsilon_2}} \mc{D}'_2$ and $\Upsilon_1 \subseteq \Upsilon_2$. Given $\mc{D}'_2$, by \Cref{apx:lem:least}, we can build the minimal $\mc{D}''_2$ such that $\mc{D}_2 \mapsto^{*_\Upsilon} \mc{D}''_2$ and $\Upsilon_1 \subseteq \Upsilon$. Moreover, we know that for every $\mc{D}$ such that $\mc{D}_2 \mapsto^{*_{\Upsilon_3}} \mc{D}$ and $\Upsilon_1 \subseteq \Upsilon_3$, we get $\mc{D}''_2 \mapsto^{*_{\Upsilon_3}} \mc{D}$.

We use this minimal $\mc{D}''_2$ to instantiate the existential ($\exists\, \mathbf{I}.$) in the goal, and use $\forall \, \mathbf{I}.$ on the goal. In particular, we instantiate $k$ with an arbitrary natural number. 
The goals are:
{\[\begin{array}{l}
 (\mc{D}'_1; \mc{D}''_2)\in \mc{V}\llbracket  \Delta \Vdash K \rrbracket_{\cdot;x_\alpha}^{k+1}\,\,\m{and}\,\, ({\mc{D}'_1}; {\mc{D}''_2})\in \mc{V}\llbracket  \Delta \Vdash K \rrbracket_{x_\alpha;\cdot}^{k+1}.
\end{array}\]}

Next, we instantiate the $\forall$ quantifier in the original assumption ($\dagger$) once again, this time with $m=k$ followed by a $\forall \mathbf{E}$, and $\mathbf{if} \mathbf{E}$ instantiating the quantifiers with similar $\mc{D}_1'$, $\Upsilon_1$, and $\Theta_1$ as the first time. We get as an assumption:

\[\begin{array}{l}
    (\dagger'')\,\,\exists \Upsilon_2, \mc{D}'_2 \,\m{such\, that} \, \,\mc{D}_2\mapsto^{{*}_{ \Upsilon_2}} \mc{D}'_2\, \m{and}\, \Upsilon_1 \subseteq \Upsilon_2 \, \m{and}\\ 
     \; \; \forall\, x_\alpha \in \mb{Out}(\Delta \Vdash K).\, \mb{if}\, x_\alpha \in \Upsilon_1.\, \mb{then}\,\, (\mc{D}'_1; \mc{D}'_2)\in \mc{V}\llbracket  \Delta \Vdash K \rrbracket_{\cdot;x_\alpha}^{k+1}\,\m{and}\\ 
     \; \; \forall\, x_\alpha \in \mb{In}(\Delta \Vdash K).\mb{if}\, x_\alpha \in \Theta_1.\,\mb{then}\, (\mc{D}'_1; \mc{D}'_2)\in \mc{V}\llbracket  \Delta \Vdash K \rrbracket_{x_\alpha;\cdot}^{k+1}.
    \end{array}\]

Next, apply $\exists\, \mathbf{I}.$ to get a $\Upsilon'$ and $\mc{D}'$ that satisfies the conditions, i.e., $\mathcal{D} \mapsto^{*_{\Upsilon'}} \mathcal{D}'$ and $\Upsilon_1 \subseteq \Upsilon'$. Instantiate $x_\alpha$ as those chosen for the goal. We have as assumptions: 
{\[\begin{array}{l}
 (\dagger''')\,\,({\mc{D}'_1}; {\mc{D}'_2})\in \mc{V}\llbracket  \Delta \Vdash K \rrbracket_{\cdot;x_\alpha}^{k+1}\,\,\m{and}\,\, ({\mc{D}'_1}; {\mc{D}'_2})\in \mc{V}\llbracket  \Delta \Vdash K \rrbracket_{x_\alpha;\cdot}^{k+1}.
\end{array}\]}
Since $\mc{D}''_2$ is the minimal configuration built for $\Upsilon_1$ and $\mc{D}_2$, we know that $\mc{D}''_2 \mapsto^{*_{\Upsilon'}} \mc{D}'$. We can apply the backward closure results of \Cref{apx:lem:backsecondrun} to get the goal from the assumptions $ (\dagger''')$, and this completes the proof.

\end{proof}

\subsection{Compositionality}\label{apx:sec:confluence-closures:compositionality}

\begin{corollary}[Semantic Cut (Compositionality)]\label{apx:lem:compositionality}
    $\forall m.\, (\mc{D}_1;\mc{D}_2) \in \mc{E}\llbracket \Delta, u_\alpha{:}T \Vdash K \rrbracket^m$ iff for all $\mc{T}_2$ and $\mc{T}_2$ s.t. $\dagger_2\, \forall m.\, (\mc{T}_1;\mc{T}_2) \in \mc{E}\llbracket \Delta'\Vdash  u_\alpha{:}T \rrbracket^m$ we have $\dagger\,\forall k.\, (\mc{T}_1 \mc{D}_1;\mc{T}_2\mc{D}_2) \in \mc{E}\llbracket \Delta', \Delta\Vdash K \rrbracket^k$.
    \end{corollary}
\begin{proof}
    The left to right direction is a corollary of Lemma~\ref{apx:lem:generalized-compositionality} in which we compose multiple configurations instead of just two.
    For the right to left direction, we put $\mc{T}_i=\cdot$, and $\Delta'=u_\alpha{:}T$, and the rest of the proof is straightforward. 
\end{proof}
\begin{lemma}[Generalized compositionality]\label{apx:lem:generalized-compositionality}
    For $i\in \{1,2\}$, and index set $I$ consider session typed tree-shaped configurations $\Delta^n \Vdash \mc{B}^n_i:: K^n$
    such that $n \in I$ and their compositions form session typed tree-shaped configurations $\Delta \Vdash \mc{D}_i::K$, i.e.,  $\mc{D}_i= \{\mc{B}^n_i\}_{n\in I}$. If for all $n \in I$, we have 
 $\dagger_n\,\,\forall m.\, (\mc{B}^n_1;\mc{B}^n_2) \in \mc{E}\llbracket \Delta^n \Vdash K^n \rrbracket^m$ then $\dagger\,\forall k.\, (\mc{D}_1;\mc{D}_2) \in \mc{E}\llbracket \Delta\Vdash K \rrbracket^k$.

\end{lemma}
\begin{proof}
Our goal is to prove the following:
\begin{quote} For all index set $I$ and all session-typed configurations $\Delta^n \Vdash \mc{B}^n_i:: K^n$ with $n\in I$ such that
$\dagger_1\,\, \forall m.\, (\mc{B}^n_1;\mc{B}^n_2) \in \mc{E}\llbracket \Delta^n \Vdash K^n \rrbracket^m$ we have $\dagger \,\,\forall k.\, (\{\mc{B}^n_1\}_{n \in I};\{\mc{B}^n_2\}_{n \in I}) \in \mc{E}\llbracket \Delta\Vdash K \rrbracket^k$.
\end{quote}

This is equivalent to the following statement which we prove:
\begin{quote} For all natural numbers, $k$, and for all index set $I$, and for all session-typed configurations $\Delta^n \Vdash \mc{B}^n_i:: K^n$ such that
$\dagger_n\,\,\forall m.\, (\mc{B}^n_1;\mc{B}^n_2) \in \mc{E}\llbracket \Delta^n \Vdash K^n \rrbracket^m$ we have $\dagger'\,\, (\{\mc{B}^n_1\}_{n \in I};\{\mc{B}^n_2\}_{n \in I}) \in \mc{E}\llbracket \Delta \Vdash K \rrbracket^k$.
\end{quote}

We proceed the proof by an induction on k.

{\bf Base case ($k=0$). }The proof is straightforward, since by the definition of the logical relation for session-typed configurations $\Delta^n \Vdash \mc{B}^n_i:: K^n$ with $\Delta \Vdash \{\mc{B}^n_i\}_{n \in I}::K$, we have $(\{\mc{B}^n_1\}_{n \in I};\{\mc{B}^n_2\}_{n \in I}) \in \mc{E}\llbracket \Delta \Vdash K \rrbracket^0$.

{\bf Inductive case ($k=k'+1$). }
Our goal is to prove the following:
\begin{quote} 
     For all index set $I$, and for all session-typed configurations $\Delta^n \Vdash \mc{B}^n_i:: K^n$ where $n\in I$ such that
    $\dagger_n\,\,\forall m.\, (\mc{B}^n_1;\mc{B}^n_2) \in \mc{E}\llbracket \Delta^n \Vdash K^n \rrbracket^m$ we have $\dagger'\,\, (\{\mc{B}^n_1\}_{n \in I};\{\mc{B}^n_2\}_{n \in I}) \in \mc{E}\llbracket \Delta \Vdash K \rrbracket^{k'+1}$
\end{quote}

where $\dagger'$ is defined in line (12) of the logical relation as 
{\small\[\begin{array}{l}
    \forall \,{\Upsilon_1}, \Theta_1, \mc{D}'_1.\,  \forall \, j \in \mathbb{N}.\, \mathbf{if}\, {\{\mc{B}^n_1\}_{n \in I}}\mapsto^{{j}_{ \Upsilon_1; \Theta_1}} {\mc{D}'_1}\,\m{then}\, \exists \Upsilon_2, \mc{D}'_2 \,\m{such\, that} \, \,\{\mc{B}^n_2\}_{n \in I}\mapsto^{{*}_{ \Upsilon_2}} \mc{D}'_2\, \m{and} \, \Upsilon_1 \subseteq \Upsilon_2\, \m{and}\\ 
     \; \; \forall\, x_\gamma \in \mb{Out}( \Delta  \Vdash K).\, \mb{if}\, x_\gamma \in \Upsilon_1.\, \mb{then}\,\, (\mc{D}'_1; \mc{D}'_2)\in \mc{V}\llbracket   \Delta \Vdash K \rrbracket_{\cdot;x_\gamma}^{k'+1}\,\m{and}\\ 
     \; \; \forall\, x_\gamma \in \mb{In}( \Delta  \Vdash K).\,\mb{if}\, x_\gamma \in \Theta_1.\, \mb{then}\,\, (\mc{D}'_1; \mc{D}'_2)\in \mc{V}\llbracket   \Delta \Vdash K \rrbracket_{x_\gamma;\cdot}^{k'+1}.
 \end{array}\]}

Again, we rewrite the above goal as an  equivalent statement as follows:

\begin{quote} For all natural numbers $j$, for all $ I$, and for all session-typed configurations $\Delta^n \Vdash \mc{B}^n_i:: K^n$ with $n \in I$ such that
    $\dagger_n\,\,\forall m.\, (\mc{B}^n_1;\mc{B}^n_2) \in \mc{E}\llbracket \Delta^n \Vdash K^n \rrbracket^m$ we have 
\[\begin{array}{l}
    \forall \,\Upsilon_1, \Theta_1, \mc{D}'_1.\, \, \mathbf{if}\, {\{\mc{B}^n_1\}_{n \in I}}\mapsto^{{j}_{ \Upsilon_1; \Theta_1}} {\mc{D}'_1}\,\m{then}\, \exists \Upsilon_2, \mc{D}'_2 \,\m{such\, that} \, \,\{\mc{B}^n_2\}_{n \in I}\mapsto^{{*}_{ \Upsilon_2}} \mc{D}'_2\, \m{and}\, \Upsilon_1 \subseteq \Upsilon_2\, \m{and}\\ 
     \; \; \forall\, x_\gamma \in \mb{Out}(\Delta  \Vdash K).\, \mb{if}\, x_\gamma \in \Upsilon_1.\, \mb{then}\,\, (\mc{D}'_1; \mc{D}'_2)\in \mc{V}\llbracket   \Delta \Vdash K \rrbracket_{\cdot;x_\gamma}^{k'+1}\,\m{and}\\ 
     \; \; \forall\, x_\gamma \in \mb{In}( \Delta  \Vdash K).\mb{if}\, x_\gamma \in \Theta_1.\, \mb{then}\, (\mc{D}'_1; \mc{D}'_2)\in \mc{V}\llbracket  \Delta \Vdash K \rrbracket_{x_\gamma;\cdot}^{k'+1}.
 \end{array}\]
\end{quote}

We proceed the proof by a nested induction on $j$. 

\begin{description}
\item{\bf Base case ($j=0$)}. Consider an arbitrary index set $I$ and an arbitrary session-typed configurations $\Delta^n\Vdash \mc{B}^n_i:: K^n$ that satisfy the conditions $\dagger_n$. We need to show that

\[\begin{array}{l}
    \forall \Upsilon_1, \Theta_1, \,\mc{D}'_1.\, \mathbf{if}\, {\{\mc{B}^n_1\}_{n \in I}}\mapsto^{{0}_{ \Upsilon_1; \Theta_1}} {\mc{D}'_1}\,\m{then}\, \exists \Upsilon_2, \mc{D}'_2 \,\m{such\, that} \, \,\{\mc{B}^n_2\}_{n \in I}\mapsto^{{*}_{ \Upsilon_2}} \mc{D}'_2\, \m{and} \, \Upsilon_1 \subseteq \Upsilon_2\, \m{and}\\ 
     \; \; \forall\, x_\gamma \in \mb{Out}( \Delta  \Vdash K).\, \mb{if}\, x_\gamma \in \Upsilon_1.\, \mb{then}\,\, (\mc{D}'_1; \mc{D}'_2)\in \mc{V}\llbracket \Delta \Vdash K \rrbracket_{\cdot;x_\gamma}^{k'+1}\,\m{and}\\ 
     \; \; \forall\, x_\gamma \in \mb{In}( \Delta \Vdash K).\, \mb{if}\, x_\gamma \in \Theta_1. \,(\mc{D}'_1; \mc{D}'_2)\in \mc{V}\llbracket   \Delta \Vdash K \rrbracket_{x_\gamma;\cdot}^{k'+1}.
 \end{array}\]

Consider an arbitrary $\Upsilon_1$, $\Theta_1$, ans $\mc{D}'_1$, and apply $\mathbf{If}\, \mathbf{I}.$ on the goal. By the assumption $\{\mc{B}^n_1\}_{n \in I}\mapsto^{{0}_{ \Upsilon_1; \Theta_1}} {\mc{D}'_1}$  we know that $\mc{D}'_1=\{\mc{B}^n_1\}_{n \in I}$, and $\{\mc{B}^n_1\}_{n \in I}$ sends along $\Upsilon_1$ and receives along $\Theta_1$. Our goal is to show the following:

\[\begin{array}{l}
   \star\,\, \exists \Upsilon_2, \mc{D}'_2 \,\m{such\, that} \, \,\{\mc{B}^n_2\}_{n \in I}\mapsto^{{*}_{ \Upsilon_2}} \mc{D}'_2\, \m{and}  \, \Upsilon_1 \subseteq \Upsilon_2\, \m{and}\\ 
     \; \; \forall\, x_\gamma \in \mb{Out}(\Delta \Vdash K).\, \mb{if}\, x_\gamma \in \Upsilon_1.\, \mb{then}\,\, (\{\mc{B}^n_1\}_{n \in I}; \mc{D}'_2)\in \mc{V}\llbracket   \Delta \Vdash K \rrbracket_{\cdot;x_\gamma}^{k'+1}\,\m{and}\\ 
     \; \; \forall\, x_\gamma \in \mb{In}( \Delta \Vdash K). \mb{if}\, x_\gamma \in \Theta_1.\, \mb{then} \,(\{\mc{B}^n_1\}_{n \in I}; \mc{D}'_2)\in \mc{V}\llbracket  \Delta \Vdash K \rrbracket_{x_\gamma;\cdot}^{k'+1}.
 \end{array}\]

By the definition of the logical relation, and from assumptions $\dagger_n$ for $n \in I$ we get

\[\begin{array}{l}
    \dagger'_n\,\forall m. \forall \Upsilon_n, \Theta_n, \mc{B}^{n'}_1.\,\m{if}\, \mc{B}^n_1 \mapsto^{*_{\Upsilon_n; \Theta_n}} \mc{B}^{n'}_1\, \m{then}\, \exists \Upsilon_{n_2}, \mc{B}^{n'}_2 \,\m{such\, that} \, \,\mc{B}^n_2\mapsto^{{*}_{ \Upsilon_{n_2}}} \mc{B}^{n'}_2\, \m{and} 
    \, \Upsilon_n \subseteq \Upsilon_{n_2}\, \m{and}\\ 
     \; \; \forall\, x_\gamma \in \mb{Out}(\Delta^n \Vdash K^n).\, \mb{if}\, x_\gamma \in \Upsilon_n.\, \mb{then}\,\, (\mc{B}^{n'}_1; \mc{B}^{n'}_2)\in \mc{V}\llbracket  \Delta^n \Vdash K^n\rrbracket_{\cdot;x_\gamma}^{m+1}\,\m{and}\\ 
     \; \; \forall\, x_\gamma \in \mb{In}(\Delta^n \Vdash K^n).\,\mb{if}\, x_\gamma \in \Theta_n.\, \mb{then} \,(\mc{B}^{n'}_1; \mc{B}^{n'}_2)\in \mc{V}\llbracket \Delta^n \Vdash K^n\rrbracket_{x_\gamma;\cdot}^{m+1}.
 \end{array}\]

By \Cref{apx:lem:moving-existential}, we get 
\[\begin{array}{l}
    \dagger''_n\,\,\forall \Upsilon_n, \Theta_n, \mc{B}^{n'}_1.\,\m{if}\, \mc{B}^n_1 \mapsto^{*_{\Upsilon_n; \Theta_n}} \mc{B}^{n'}_1\, \m{then}\, \exists \Upsilon_{n_2} \mc{B}^{n'}_2 \,\m{such\, that} \, \,\mc{B}^n_2\mapsto^{{*}_{ \Upsilon_{n_2}}} \mc{B}^{n'}_2\, \m{and} 
    \, \Upsilon_n \subseteq \Upsilon_{n_2}\, \m{and}\\ 
     \; \; \forall\, x_\gamma \in \mb{Out}(\Delta^n \Vdash K^n).\, \mb{if}\, x_\gamma \in \Upsilon_n.\, \mb{then}\,\, \forall\, m.\,(\mc{B}^{n'}_1; \mc{B}^{n'}_2)\in \mc{V}\llbracket  \Delta^n \Vdash K^n\rrbracket_{\cdot;x_\gamma}^{m+1}\,\m{and}\\ 
     \; \; \forall\, x_\gamma \in \mb{In}(\Delta^n \Vdash K^n).  \, \mb{if}\, x_\gamma \in \Theta_n.\, \mb{then}\, \forall\, m. \,(\mc{B}^{n'}_1; \mc{B}^{n'}_2)\in \mc{V}\llbracket \Delta^n \Vdash K^n\rrbracket_{x_\gamma;\cdot}^{m+1}.
 \end{array}\]

We instantiate the for all quantifier in $\dagger''_n$ by $\mathcal{B}^n_1$ and the sets $\Upsilon_n$ and $\Theta_n$ along which $\mathcal{B}^n_1$ sends and receives. Note that by definition we have $\Upsilon_1 \subseteq \bigcup\{\Upsilon_n\}_{n \in I}$ and $\Theta_1 \subseteq \bigcup\{\Theta_n\}_{n \in I}$. We get:

\[\begin{array}{l}
    \exists \Upsilon_{n_2}, \mc{B}^{n'}_2 \,\m{such\, that} \, \,\mc{B}^n_2\mapsto^{{*}_{ \Upsilon_{n_2}}} \mc{B}^{n'}_2\, \m{and}\\ 
    \; \; \forall\, x_\gamma \in \mb{Out}(\Delta^n \Vdash K^n).\, \mb{if}\, x_\gamma \in \Upsilon_n.\, \mb{then}\,\, \forall\, m.\,(\mc{B}^{n}_1; \mc{B}^{n'}_2)\in \mc{V}\llbracket  \Delta^n \Vdash K^n\rrbracket_{\cdot;x_\gamma}^{m+1}\,\m{and}\\ 
    \; \; \forall\, x_\gamma \in \mb{In}(\Delta^n \Vdash K^n). \mb{if}\, x_\gamma \in \Theta_n.\, \mb{then}\, \forall\, m. \,(\mc{B}^{n}_1; \mc{B}^{n'}_2)\in \mc{V}\llbracket \Delta^n \Vdash K^n\rrbracket_{x_\gamma;\cdot}^{m+1}.
 \end{array}\]

 By existential elimination, for all $n \in I$, we get a $\mc{B}^{n'}_2$ and $\Upsilon_{n_2}$ such that  $\mc{B}_2\mapsto^{*_{\Upsilon_{n_2}}} \mc{B}^{n'}_2$, and  $ \Upsilon_n \subseteq \Upsilon_{n_2}$ and

 \[\begin{array}{l}
    \dagger'''_n\, 
     \; \; \forall\, x_\gamma \in \mb{Out}(\Delta^n \Vdash K^n).\, \mb{if}\, x_\gamma \in \Upsilon_n.\, \mb{then}\,\, \forall\, m.\,(\mc{B}^{n}_1; \mc{B}^{n'}_2)\in \mc{V}\llbracket  \Delta^n \Vdash K^n\rrbracket_{\cdot;x_\gamma}^{m+1}\,\m{and}\\ 
     \; \; \forall\, x_\gamma \in \mb{In}(\Delta^n \Vdash K^n). \mb{if}\, x_\gamma \in \Theta_n.\, \mb{then}\, \forall\, m. \,(\mc{B}^{n}_1; \mc{B}^{n'}_2)\in \mc{V}\llbracket \Delta^n \Vdash K^n\rrbracket_{x_\gamma;\cdot}^{m+1}.
 \end{array}\]

 Note that by definition $\mb{Out}(\Delta \Vdash K)\subseteq \bigcup_{n \in I}\mb{Out}(\Delta^n \Vdash K^n)$ and  $\mb{In}(\Delta \Vdash K)\subseteq \bigcup_{n \in I}\mb{In}(\Delta^n \Vdash K^n)$.

We apply \Cref{apx:lem:least} to get the minimal sending configurations $\mc{B}^{n''}_2$ and set $\Upsilon'_n$ for each given $\Upsilon_n$ and $\mc{B}^n_2$. We get $\mc{B}^{n''}_2$ such that  $\mc{B}^n_2\mapsto^{*_{\Upsilon'_n}} \mc{B}^{n''}_2$. Since these configurations are minimal, for all $n \in I$, we have $\mc{B}^{n''}_2\mapsto^{*_{\Upsilon_{n_2}}} \mc{B}^{n'}_2$.

Since $\Upsilon_n \subseteq \Upsilon'_n \subseteq \Upsilon_{n_2}$,  we can apply the backward closure (\Cref{apx:lem:backsecondrun}) on $\dagger'''_n$ to get

\[\begin{array}{l}
    \dagger^4_n\, 
    \; \; \forall\, x_\gamma \in \mb{Out}(\Delta^n \Vdash K^n).\, \mb{if}\, x_\gamma \in \Upsilon_n.\, \mb{then}\,\, \forall\, m.\,(\mc{B}^{n}_1; \mc{B}^{n''}_2)\in \mc{V}\llbracket  \Delta^n \Vdash K^n\rrbracket_{\cdot;x_\gamma}^{m+1}\,\m{and}\\ 
     \; \; \forall\, x_\gamma \in \mb{In}(\Delta^n \Vdash K^n). \,  \mb{if}\, x_\gamma \in \Theta_n.\forall\, m. \,(\mc{B}^{n}_1; \mc{B}^{n''}_2)\in \mc{V}\llbracket \Delta^n \Vdash K^n\rrbracket_{x_\gamma;\cdot}^{m+1}.
 \end{array}\]

 Moreover, we apply forward closure on the second run \Cref{apx:lem:fwdsecond} on assumptions $\dagger_n$ to get for all $n \in I$:
 \[\circ_n\,\,\forall m.\, (\mc{B}^n_1;\mc{B}^{n''}_2) \in \mc{E}\llbracket \Delta^n \Vdash K^n\rrbracket^m\]
 We can apply the forward closure on the second run since the conditions of \Cref{apx:lem:fwdsecond} are satisfied, i.e. $\Upsilon_n\subseteq \Upsilon'_n$ and $\mc{B}^{n''}_2$is the minimal sending configuration with respect to $\mc{B}^n_2$ and $\Upsilon_n$.

We build $\mc{D}'_2$ to be $\{\mc{B}^{n''}_2\}_{n \in I}$. We have $\{\mc{B}^{n}_2\}_{n \in I} \mapsto^{*_{\Upsilon_2}}\{\mc{B}^{n''}_2\}_{n \in I}$ with $\Upsilon_1 \subseteq \Upsilon_2$, and $\{\mc{B}^{n''}_2\}_{n \in I} \mapsto^{*_{\Upsilon_3}} \{\mc{B}^{n'}_2\}_{n \in I}$ with $\Upsilon_2 \subseteq \Upsilon_3$.
We, then instantiate the existential quantifier in the goal ($\star$) with $\Upsilon_2$ and $\mc{D}'_2$ that we built for which we know $\{\mc{B}^{n}_2\}_{n \in I} \mapsto^{*_{\Upsilon_2}}\mc{D}'_2$.  We need to show  
\[\begin{array}{l}
    \star'
      \; \; \forall\, x_\gamma \in \mb{Out}( \Delta \Vdash K).\, \mb{if}\, x_\gamma \in \Upsilon_1.\, \mb{then}\,\, (\{\mc{B}^{n}_1\}_{n \in I} ; \{\mc{B}^{n''}_2\}_{n \in I} )\in \mc{V}\llbracket   \Delta \Vdash K \rrbracket_{\cdot;x_\gamma}^{k'+1}\,\m{and}\\ 
      \; \; \forall\, x_\gamma \in \mb{In}( \Delta \Vdash K). \, \mb{if}\, x_\gamma \in \Theta_1.\, \mb{then}\, \,(\{\mc{B}^{n}_1\}_{n \in I} ; \{\mc{B}^{n''}_2\}_{n \in I} )\in \mc{V}\llbracket   \Delta \Vdash K \rrbracket_{x_\gamma;\cdot}^{k'+1}.
  \end{array}\]

  \begin{description}
\item{\bf Part 1.} Consider arbitrary  $x_\gamma \in \mb{Out}( \Delta \Vdash K)$ and assume $x_\gamma \in \Upsilon_1$. By the structure of the configurations, for some $n \in I$, $x_\gamma \in \Delta^n, K^n$ and thus $x_\gamma \in \mb{Out}(\Delta^n \Vdash K^n)$ and $x_\gamma \in \Upsilon_n$. The goal is to prove
\[ \star_1\,(\{\mc{B}^{n}_1\}_{n \in I} ; \{\mc{B}^{n''}_2\}_{n \in I} )\in \mc{V}\llbracket   \Delta \Vdash K \rrbracket_{\cdot;x_\gamma}^{k'+1}\]

\item{\bf Part 2.}
Consider arbitrary  $x_\gamma \in \mb{In}( \Delta \Vdash K)$ and assume $x_\gamma \in \Theta_1$. By the structure of the configurations, for some $n \in I$, $x_\gamma \in \Delta^n, K^n$ and thus $x_\gamma \in \mb{In}(\Delta^n\Vdash K^n)$ and $x_\gamma \in \Theta_n$. The goal is to prove
\[\star_2\,(\{\mc{B}^{n}_1\}_{n \in I} ; \{\mc{B}^{n''}_2\}_{n \in I} )\in \mc{V}\llbracket   \Delta \Vdash K \rrbracket_{x_\gamma;\cdot}^{k'+1}\]
\end{description}

In both parts, we continue the proof by considering the type of $x_\gamma$. The type of $x_\gamma$ determines whether we need to prove {\bf Part 1.} or {\bf Part 2}. We provide the detailed proof for two interesting cases, the proof of the rest of cases is similar.

\begin{description}
\item[\bf Subcase 1.] $x_\gamma{:}A \otimes B \in K$. This case corresponds to {\bf Part 1.} of the goal in which we have $x_\gamma \in \mb{Out}(\Delta \Vdash x _\gamma{:}A \otimes B )$ and $x_\gamma \in \Upsilon_1$. By the structure of the configuration, there exists a tree $\mc{B}_1^\kappa$ that provides the root channel $K= K^\kappa= x_\gamma{:}A \otimes B$. We use assumption $\dagger^4_n$ for that specific channel ($\dagger^4_\kappa$), we have
\[ \dagger^5_\kappa\,\, \forall\, m.\,(\mc{B}^{\kappa}_1; \mc{B}^{\kappa''}_2)\in \mc{V}\llbracket  \Delta^\kappa \Vdash K^\kappa\rrbracket_{\cdot;x_\gamma}^{m+1}\]

First, instantiate the for all quantifier with $m=0$.
By {\bf Row 4.} of the logical relation, we have $\Delta^\kappa=\Delta^\kappa_1, \Delta^\kappa_2$ and for some $y_\beta \in \mathbf{chnl}$, we have:
$\mc{B}^\kappa_1 = \mc{B}^{\kappa'}_1\mc{A}_1\msg{\mathbf{send}\, y_\beta x_\gamma} $ and $\mc{B}^{\kappa''}_2= \mc{B}^{\kappa'''}_2\mc{A}_2\msg{\mathbf{send}\, y_\beta x_\gamma}$.

We want to prove \[\circ'\,\forall\, m. \,(\mc{A}_1, \mc{A}_2) \in \mc{E}\llbracket \Delta^\kappa_2\Vdash y_\beta{:}A  \rrbracket^m\,\,\quad\,
\m{and}
\,\,\quad \, \circ'' \forall \, m.\, (\mc{B}^{\kappa'}_1, \mc{B}^{\kappa'''}_2) \in \mc{E}\llbracket \Delta^\kappa_1 \Vdash x_{\gamma+1}{:}B  \rrbracket^m.\]

Consider an arbitrary $m$ given by $\forall \mathbf{I}$ on the goals $\circ'$ and $\circ''$. Once again, instantiate the quantifier in $\dagger^5_\kappa$, this time with the arbitrary $m$. Again, we get $\Delta^\kappa=\Delta^\kappa_1, \Delta^\kappa_2$ and for some $y_\beta \in \mathbf{chnl}$, we have:
$\mc{B}^\kappa_1 = \mc{B}^{\kappa'}_1\mc{A}_1\msg{\mathbf{send}\, y_\beta x_\gamma} $ and $\mc{B}^{\kappa''}_2= \mc{B}^{\kappa'''}_2\mc{A}_2\msg{\mathbf{send}\, y_\beta x_\gamma}$. Moreover, 
\[(\mc{A}_1, \mc{A}_2) \in \mc{E}\llbracket \Delta^\kappa_2\Vdash y_\beta{:}A  \rrbracket^m\,\,\quad\,
\m{and}
\,\,\quad \,  (\mc{B}^{\kappa'}_1, \mc{B}^{\kappa'''}_2) \in \mc{E}\llbracket \Delta^\kappa_1 \Vdash x_{\gamma+1}{:}B  \rrbracket^m.\]

Since the naming in the configuration is unique, we get the above for the same $y_\beta$ as we got in the case of $m=0$ and the proof of $\circ'$ and $\circ''$ is complete.

From this we can prove 
\[\begin{array}{ll}
    \mc{D}_1=\{\mc{B}^n_1\}_{n \in I}= \{\mc{B}^n_1\}_{n \in I \& n \neq \kappa}\,\, \mc{B}^{\kappa'}_1\mc{A}_1\msg{\mathbf{send}\, y_\beta x_\gamma}\,\m{and}\\
    \mc{D}'_2=\{\mc{B}^{n''}_2\}_{n \in I }= \{\mc{B}^{n''}_2\}_{n \in I \& n \neq \kappa}\,\,\mc{B}^{\kappa'''}_2\mc{A}_2\msg{\mathbf{send}\, y_\beta x_\gamma}
\end{array}\]
First, observe that by the structure of the configuration, there is no tree $\mc{B}_1^n$ or $\mc{B}_2^{n''}$ using $K^\kappa=x_\gamma{:}A \otimes B$ as its resource, i.e., $x_\gamma$ is the root.
We can break down the resources $\Delta^\kappa_1$ and $\Delta^\kappa_2$ as $\Delta^\kappa_1= \Delta^{\kappa'}_1, \Delta^{\kappa''}_1$ and $\Delta^\kappa_2= \Delta^{\kappa'}_2, \Delta^{\kappa''}_2$, such that  $\Delta^{\kappa'}_1$ and  $\Delta^{\kappa'}_2$ are in the interface of $\mc{D}_1$ and $\mc{D}'_2$ and $\Delta^{\kappa''}_1$ and  $\Delta^{\kappa''}_2$ are the resources provided by other trees. We can partition $I\backslash \{\kappa\}$ into two disjoint sets $I_1$ and $I_2$ such that the configurations $\{\mc{B}^{n}_1\}_{n \in I_1}$ and $\{\mc{B}^{n''}_2\}_{n \in I_1}$ provide the resources in  $\Delta^{\kappa''}_1$ and configurations $\{\mc{B}^{n}_1\}_{n \in I_2}$ and $\{\mc{B}^{n''}_2\}_{n \in I_2}$ provide the resources in $\Delta^{\kappa''}_2$. In other words, we have $\Delta= \Delta_1, \Delta^{\kappa'}_1,\Delta_2, \Delta^{\kappa'}_2$ and $K=x_\gamma{:}A \otimes B$ and
{\small \[
\begin{array}{llll}
 (i) &  \Delta_1 \Vdash \{\mc{B}^{n}_1\}_{n \in I_1}:: \Delta^{\kappa''}_1 &\quad \Delta^{\kappa'}_1, \Delta^{\kappa''}_1 \Vdash \mc{B}^{\kappa'}_1::x_{\gamma+1}{:}B & \quad \Delta_1, \Delta^{\kappa'}_1  \Vdash \{\mc{B}^{n}_1\}_{n \in I_1} \mc{B}^{\kappa'}_1:: x_{\gamma+1}{:}B\\
  &  \Delta_1 \Vdash \{\mc{B}^{n''}_2\}_{n \in I_1}:: \Delta^{\kappa''}_1 &\quad \Delta^{\kappa'}_1, \Delta^{\kappa''}_1 \Vdash \mc{B}^{\kappa'''}_2::x_{\gamma+1}{:}B & \quad \Delta_1, \Delta^{\kappa'}_1  \Vdash \{\mc{B}^{n''}_2\}_{n \in I_1}\mc{B}^{\kappa'''}_2:: x_{\gamma+1}{:}B\\[4pt]
  (ii)&  \Delta_2 \Vdash \{\mc{B}^{n}_1\}_{n \in I_2}:: \Delta^{\kappa''}_2 &\quad \Delta^{\kappa'}_2, \Delta^{\kappa''}_2 \Vdash \mc{A}_1::y_{\beta}{:}A  & \quad \Delta_2, \Delta^{\kappa'}_2  \Vdash \{\mc{B}^{n}_1\}_{n \in I_2} \mc{A}_1:: y_{\beta}{:}A\\
  & \Delta_2 \Vdash \{\mc{B}^{n''}_2\}_{n \in I_2}:: \Delta^{\kappa''}_2 & \quad \Delta^{\kappa'}_2, \Delta^{\kappa''}_2 \Vdash \mc{A}_2::y_{\beta}{:}A&\quad \Delta_2, \Delta^{\kappa'}_2  \Vdash \{\mc{B}^{n''}_2\}_{n \in I_2}\mc{A}_2:: y_{\beta}{:}A\\
\end{array}    
\]}
We also have 
{\small\[\begin{array}{lll}
  \circ''' &  \mc{D}_1=\{\mc{B}^n_1\}_{n \in I}= \{\mc{B}^n_1\}_{n \in I \& n \neq \kappa}\,\, \mc{B}^{\kappa'}_1\mc{A}_1\msg{\mathbf{send}\, y_\beta x_\gamma}= \{\mc{B}^n_1\}_{n \in I_1}\, \{\mc{B}^n_1\}_{n \in I_2}\, \mc{B}^{\kappa'}_1\mc{A}_1\msg{\mathbf{send}\, y_\beta x_\gamma} \,\m{and}\\
  &  \mc{D}'_2=\{\mc{B}^{n''}_2\}_{n \in I }= \{\mc{B}^{n''}_2\}_{n \in I \& n \neq \kappa}\,\,\mc{B}^{\kappa'''}_2\mc{A}_2\msg{\mathbf{send}\, y_\beta x_\gamma}= \{\mc{B}^{n''}_2\}_{n \in I_1}\{\mc{B}^{n''}_2\}_{n \in I_2}\,\,\mc{B}^{\kappa'''}_2\mc{A}_2\msg{\mathbf{send}\, y_\beta x_\gamma}
\end{array}\]
}
Recall that earlier we proved 
\[\circ'\,\forall\, m. \,(\mc{A}_1, \mc{A}_2) \in \mc{E}\llbracket \Delta^\kappa_2\Vdash y_\beta{:}A  \rrbracket^m\,\,\quad\,
\m{and}
\,\,\quad \, \circ'' \forall \, m.\, (\mc{B}^{\kappa'}_1, \mc{B}^{\kappa'''}_2) \in \mc{E}\llbracket \Delta^\kappa_1 \Vdash x_{\gamma+1}{:}B  \rrbracket^m,\]
and we also have for all $n \in I_1 \cup I_2,$ 
\[\circ_2\,\,\forall m.\, (\mc{B}^n_1;\mc{B}^{n''}_2) \in \mc{E}\llbracket \Delta^{n}\Vdash K^n\rrbracket^m\]

We can apply the induction hypothesis on the smaller index $k'$, (i), $\circ''$, and $\circ_2$ for those trees indexed in $I_1$ to get 
    \[\circ_3\,(\{\mc{B}^n_1\}_{n \in I_1}\, \mc{B}^{\kappa'}_1;\{\mc{B}^{n''}_2\}_{n \in I_1}\, \mc{B}^{\kappa'''}_2) \in \mc{E}\llbracket \Delta_1,\Delta^{\kappa'}_1\Vdash  x_{\gamma+1}{:}B\rrbracket^{k'}.\]

Similarly, we can apply the induction hypothesis on the smaller index $k'$, (ii), $\circ'$, and $\circ_2$ for those trees indexed in $I_2$ to get 
    \[\circ_3\,(\{\mc{B}^n_1\}_{n \in I_2}\,\mc{A}_1;\{\mc{B^{''}}^n_2\}_{n \in I_2}\,\mc{A}_2) \in \mc{E}\llbracket \Delta_1,\Delta^{\kappa'}_2\Vdash  y_{\beta}{:}A\rrbracket^{k'}.\]

By {\bf Row 4.} of the logical relation, this is enough to establish the goal
\[ \star_1\,(\mc{D}_1; \mc{D}'_2)\in \mc{V}\llbracket  \Delta \Vdash K \rrbracket_{\cdot;x_\gamma}^{k'+1}\]

\item[\bf Subcase 2.] $x_\gamma{:}A \otimes B \in \Delta$, i.e., $\Delta=\Delta', x_\gamma{:} A \otimes B$. This case corresponds to {\bf Part 2.} of the goal in which we have $x_\gamma \in \mb{In}(\Delta\Vdash K)$ and $x_\gamma \in \Theta_1$, and the goal is to prove
$(\{\mc{B}^{n}_1\}_{n \in I} ; \{\mc{B}^{n''}_2\}_{n \in I} )\in \mc{V}\llbracket   \Delta \Vdash K \rrbracket_{x_\gamma;\cdot}^{k'+1}$.

By the structure of the configuration, for some index $\kappa \in I$, we have $x_\gamma{:} A \otimes B \in \Delta^\kappa$, i.e. $\Delta^\kappa=\Delta^{\kappa'}, x_\gamma{:} A \otimes B$.

By assumption $\dagger^4_\kappa$, we have
\[  \forall\, m.\,(\mc{B}^\kappa_1; \mc{B}^{\kappa''}_2)\in \mc{V}\llbracket  \Delta^{\kappa'}, x_\gamma{:}A \otimes B \Vdash K^\kappa \rrbracket_{x_\gamma; \cdot}^{m+1}\]
By {\bf Row 9.} of the logical relation, we get:

{\small\[ \dagger^5_\kappa\,\, \forall y_\beta \not \in \mathit{dom}(\Delta^{\kappa'}, x_{\gamma}{:} A \otimes B, K^\kappa)\, \forall m.\,(\msg{\mathbf{send}\, y_\beta x_\gamma}\mc{B}^\kappa_1; \msg{\mathbf{send}\, y_\beta x_\gamma}\mc{B}^{\kappa''}_2)\in \mc{E}\llbracket  \Delta^{\kappa'}, y_\beta{:}A,  x_{\gamma+1}{:}B\Vdash K^\kappa \rrbracket^{m}\]}

Moreover, we have
\[\Delta, y_\beta{:}A ,x_\gamma{:} B \Vdash  \msg{\mathbf{send}\, y_\beta x_\gamma}\{\mc{B}^{n}_1\}_{n \in I}::K \qquad \Delta, y_\beta{:}A ,x_\gamma{:} B \Vdash  \msg{\mathbf{send}\, y_\beta x_\gamma}\{\mc{B}^{n''}_2\}_{n \in I}::K \]

Recall that for all $n \in I\backslash \{\kappa\}$ we also have \[\circ_2\,\,\forall m.\, (\mc{B}^{n}_1;\mc{B}^{n''}_2) \in \mc{E}\llbracket \Delta^n\Vdash  K^n\rrbracket^m\]

We can apply the induction hypothesis on the smaller index $k'$, $\circ_2$, and $\dagger^5_\kappa$ to get 
{\small\[\forall y_\beta \not \in \mathit{dom}(\Delta,  x_{\gamma}{:}A \otimes B,  K).(\msg{\mathbf{send}\, y_\beta x_\gamma}\{\mc{B}^{n}_1\}_{n \in I};\msg{\mathbf{send}\, y_\beta x_\gamma}\{\mc{B}^{n''}_1\}_{n \in I}) \in \mc{E}\llbracket \Delta, y_\beta{:}A,  x_{\gamma+1}{:}B\Vdash  K\rrbracket^{k'}.\]}
By {\bf Row 9.} of the logical relation, this is enough to establish the goal.

\end{description}

\item{\bf Inductive case ($j=j'+1$).}
Consider an arbitrary index set $I$ and an arbitrary session-typed configurations $\Delta^n\Vdash \mc{B}^n_i:: K^n$ that satisfy the conditions $\dagger_n$. 
We need to show that

\[\begin{array}{l}
    \forall \Upsilon_1, \Theta_1,\,\mc{D}'_1.\, \mathbf{if}\, {\{\mc{B}^n_1\}_{n \in I}}\mapsto^{{j'+1}_{ \Upsilon_1;\Theta_1}} {\mc{D}'_1}\,\m{then}\, \exists \Upsilon_2, \mc{D}'_2 \,\m{such\, that} \, \,\{\mc{B}^n_2\}_{n \in I}\mapsto^{{*}_{ \Upsilon_2}} \mc{D}'_2\, \m{and}\,\Upsilon_1 \subseteq \Upsilon_2 \, \m{and}\\ 
     \; \; \forall\, x_\gamma \in \mb{Out}( \Delta  \Vdash K).\, \mb{if}\, x_\gamma \in \Upsilon_1.\, \mb{then}\,\, (\mc{D}'_1; \mc{D}'_2)\in \mc{V}\llbracket \Delta \Vdash K \rrbracket_{\cdot;x_\gamma}^{k'+1}\,\m{and}\\ 
     \; \; \forall\, x_\gamma \in \mb{In}( \Delta \Vdash K). \mb{if}\, x_\gamma \in \Theta_1.\, \mb{then} \,(\mc{D}'_1; \mc{D}'_2)\in \mc{V}\llbracket   \Delta \Vdash K \rrbracket_{x_\gamma;\cdot}^{k'+1}.
 \end{array}\]

We apply a $\forall \mathbf{I}.$ and $\mathbf{if}\,\mathbf{I}.$ on the goal: consider an arbitrary $\Upsilon_1$, $\Theta_1$, and $\mc{D}'_1$ and the first step of ${\{\mc{B}^n_1\}_{n \in I}}\mapsto^{{j'+1}_{ \Upsilon_1}; \Theta_1} {\mc{D}'_1}$. There are two cases to consider:

\item{\bf Case 1. the first step is an internal step but not a communication between the subtrees.}

Without loss of generality, let's assume that the communication is internal to $\mc{B}^\kappa_1$ for some $\kappa \in I$ and all other trees $\mc{B}^n_1$ for $n \neq \kappa$ remain intact, i.e. we have $\{\mc{B}^n_1\}_{n \in I}=\{\mc{B}^n_1\}_{n \in I_1}\, \mc{B}^{\kappa}_1 \{\mc{B}^n_1\}_{n \in I_2}$ and \[ \{\mc{B}^n_1\}_{n \in I_1}\, \mc{B}^{\kappa}_1 \{\mc{B}^n_1\}_{n \in I_2}\;\mapsto \{\mc{B}^n_1\}_{n \in I_1}, \mc{B}^{\kappa'}_1 \{\mc{B}^n_1\}_{n \in I_2} \;\mapsto^{{j'}_{\Upsilon_1; \Theta_1}} {\mc{D}'_1}.\]

By forward closure (\Cref{apx:lem:fwdfirst}) on the assumption $\dagger_\kappa$ (i.e.,$\dagger_n$for $n = \kappa$), we get
\[\dagger'_\kappa\,\,\forall m.\, (\mc{B}^{\kappa'}_1;\mc{B}_2) \in \mc{E}\llbracket \Delta^\kappa\Vdash K^\kappa \rrbracket^m.\]

Recall that by $\dagger_n$, we also have for $n \neq \kappa$
\[\forall m.\, (\mc{B}^n_1;\mc{B}^n_2) \in \mc{E}\llbracket \Delta^n \Vdash K^n\rrbracket^m.\]
 We can apply the induction hypothesis on the number of steps $j'$ to get:
 {\small\[\begin{array}{l}
    \forall \Upsilon_1, \Theta_1, \,\mc{D}'_1.\,\forall\, {\Upsilon_1}. \, \mathbf{if}\,  \{\mc{B}^n_1\}_{n \in I_1}\, \mc{B}^{\kappa'}_1 \{\mc{B}^n_1\}_{n \in I_2} \mapsto^{{j'}_{ \Upsilon_1; \Theta_1}} {\mc{D}'_1}\,\m{then}\, \exists \Upsilon_2\,\mc{D}'_2. \, \,\{\mc{B}^n_2\}_{n \in I}\mapsto^{{*}_{ \Upsilon_2}} \mc{D}'_2\, \m{and}\, \Upsilon_1 \subseteq \Upsilon_2\,\m{and}\\ 
     \; \; \forall\, x_\gamma \in \mb{Out}(\Delta  \Vdash K).\, \mb{if}\, x_\gamma \in \Upsilon_1.\, \mb{then}\,\, (\mc{D}'_1; \mc{D}'_2)\in \mc{V}\llbracket   \Delta \Vdash K \rrbracket_{\cdot;x_\gamma}^{k'+1}\,\m{and}\\ 
     \; \; \forall\, x_\gamma \in \mb{In}( \Delta  \Vdash K).\,\mb{if}\, x_\gamma \in \Theta_1.\, \mb{then}\, (\mc{D}'_1; \mc{D}'_2)\in \mc{V}\llbracket  \Delta \Vdash K \rrbracket_{x_\gamma;\cdot}^{k'+1}.
 \end{array}\]}

 Since $\{\mc{B}^n_1\}_{n \in I_1}\, \mc{B}^{\kappa}_1 \{\mc{B}^n_1\}_{n \in I_2}\;\mapsto \{\mc{B}^n_1\}_{n \in I_1}, \mc{B}^{\kappa'}_1 \{\mc{B}^n_1\}_{n \in I_2} \;\mapsto^{{j'}_{\Upsilon_1; \Theta_1}} {\mc{D}'_1},$
we can prove the goal:
\[\begin{array}{l}
    \exists \Upsilon_2, \mc{D}'_2 \,\m{such\, that} \, \,\{\mc{B}^n_2\}_{n \in I}\mapsto^{{*}_{ \Upsilon_2}} \mc{D}'_2\, \m{and}\, \Upsilon_1 \subseteq \Upsilon_2\, \m{and}\\ 
     \; \; \forall\, x_\gamma \in \mb{Out}(\Delta  \Vdash K).\, \mb{if}\, x_\gamma \in \Upsilon_1.\, \mb{then}\,\, (\mc{D}'_1; \mc{D}'_2)\in \mc{V}\llbracket   \Delta \Vdash K \rrbracket_{\cdot;x_\gamma}^{k'+1}\,\m{and}\\ 
     \; \; \forall\, x_\gamma \in \mb{In}( \Delta  \Vdash K).\,\mb{if}\, x_\gamma \in \Theta_1.\, (\mc{D}'_1; \mc{D}'_2)\in \mc{V}\llbracket  \Delta \Vdash K \rrbracket_{x_\gamma;\cdot}^{k'+1}.
 \end{array}\]
which completes the proof of this case.

\item{\bf Case 2. the first step is a communication between two sub-configurations.} Without loss of generality, we assume that the communication is between trees indexed by $\kappa, \lambda \in I$, i.e., $\mc{B}^\kappa_1$ offers a resource $u_\alpha{:}T$ and $\mc{B}^\lambda_1$ uses the resource, and there is a message available along $u_\alpha$ that is received in the first step.
The proof proceeds by case analysis on type of $u_\alpha^c{:}T$. We provide the detailed proof for one case. The proof of the rest of the cases is similar.

\begin{description}
\item{\bf  Subcase 1. $T= A \otimes B$}, i.e., we have
\[\begin{array}{llll}
   (i)& \Delta^\kappa \Vdash \mc{B}^{\kappa}_1:: K^\kappa & \quad \m{where}\,\, K^\kappa= u_\alpha{:}A \otimes B\\ &\Delta^\kappa \Vdash \mc{B}^{\kappa}_2 ::K^\kappa &  \\[4pt]
   (ii)& \Delta^\lambda \Vdash \mc{B}^{\lambda}_1:: K^\lambda & \quad \m{where}\,\, \Delta^\lambda= \Delta^{\lambda'}, u_\alpha{:}A \otimes B \\ &\Delta^\lambda \Vdash \mc{B}^{\lambda}_2 ::K^\lambda& \\\
\end{array}
   \]

By the assumptions of this case, we know that there is a message sent along $u_\alpha^c{:}A \otimes B$ in $\mc{B}^\kappa_1$ that is received by a process in $\mc{B}^\lambda_1$, i.e., $\mc{B}^{\kappa}_1= \mc{B}^{\kappa'}_1 \mc{A}_1\mb{msg}(\mb{send}y_\beta\, u_\alpha)$ and

\[\begin{array}{l}
    \{\mc{B}^n_1\}_{n \in I}=\{\mc{B}^n_1\}_{n \in I_1}, \mc{B}^{\kappa}_1 \mc{B}^{\lambda}_1 \{\mc{B}^n_1\}_{n \in I_2}= \{\mc{B}^n_1\}_{n \in I_1}, \mc{B}^{\kappa'}_1 \mc{A}_1\,\mb{msg}(\mb{send}y_\beta\, u_\alpha) \mc{B}^{\lambda}_1 \{\mc{B}^n_1\}_{n \in I_2}\,\, \m{and}\\
    \{\mc{B}^n_1\}_{n \in I_1}, \mc{B}^{\kappa'}_1 \mc{A}_1\,\mb{msg}(\mb{send}y_\beta\, u_\alpha) \mc{B}^{\lambda}_1 \{\mc{B}^n_1\}_{n \in I_2}\;\mapsto \{\mc{B}^n_1\}_{n \in I_1}, \mc{B}^{\kappa'}_1 \mc{A}_1\mc{B}^{\lambda'}_1 \{\mc{B}^n_1\}_{n \in I_2} \;\mapsto^{{j'}_{\Upsilon_1; \Theta_1}} {\mc{D}'_1}.
   
\end{array} \]

By $\dagger_{\kappa}$, \Cref{apx:lem:moving-existential}, and $\mc{B}^{\kappa'}_1 \mc{A}_1\mb{msg}(\mb{send}y_\beta\, u_\alpha)\mapsto^{0_{\Upsilon'_\kappa; \Theta'_\kappa}}\mc{B}^{\kappa'}_1 \mc{A}_1\mb{msg}(\mb{send}y_\beta\, u_\alpha)$, we get

{\small\[\begin{array}{ll}
    \,\exists \Upsilon_{\kappa_2}\mc{B}^{\kappa'}_2 \,\m{such\, that} \, \,\mc{B}^\kappa_2\mapsto^{{*}_{ \Upsilon_{\kappa_2}}} \mc{B}^{\kappa'}_2\, \m{and}\, \Upsilon'_\kappa \subseteq \Upsilon_{\kappa_2}\\ 
    \; \; \forall\, x_\gamma \in \mb{Out}(\Delta^\kappa \Vdash u_{\alpha}{:}A \otimes B ).\, \mb{if}\, x_\gamma \in \Upsilon'_\kappa.\, \mb{then}\,\,\forall\, m.\, (\mc{B}^{\kappa'}_1 \mc{A}_1\mb{msg}(\mb{send}y_\beta\, u_\alpha); \mc{B}^{\kappa'}_2)\in \mc{V}\llbracket  \Delta^ \kappa \Vdash u_{\alpha}{:}A \otimes B\rrbracket_{\cdot;x_\gamma}^{m+1}\\
    \,\m{and}\\ 
    \; \; \forall\, x_\gamma \in \mb{In}(\Delta^\kappa \Vdash u_{\alpha}{:} A \otimes B).\, \mb{if}\, x_\gamma \in \Theta'_\kappa.\, \mb{then}\, \forall\,m. \,(\mc{B}^{\kappa'}_1 \mc{A}_1\mb{msg}(\mb{send}y_\beta\, u_\alpha); \mc{B}^{\kappa'}_2)\in \mc{V}\llbracket \Delta^\kappa \Vdash u_{\alpha}{:}A \otimes B\rrbracket_{x_\gamma;\cdot}^{m+1}
\end{array}\]}

In particular, we know that $u_\alpha{:}A \otimes B \in \mb{Out}(\Delta^\kappa \Vdash  u_{\alpha}{:} A \otimes B)$ and $u_\alpha \in \Upsilon'_\kappa$, which gives us:
\[\begin{array}{ll}
\forall\, m.\, (\mc{B}^{\kappa'}_1 \mc{A}_1\mb{msg}(\mb{send}y_\beta\, u_\alpha); \mc{B}^{\kappa'}_2)\in \mc{V}\llbracket  \Delta^\kappa \Vdash u_{\alpha}{:}A \otimes B\rrbracket_{\cdot;u_\alpha}^{m+1}
\end{array}\]

By {\bf Row 4.} of the logical relation, we have $\Delta^{\kappa}= \Delta^{\kappa}_1, \Delta^{\kappa}_2$ and
\[\begin{array}{ll}
    \mc{B}^{\kappa'}_2=\mc{B}^{\kappa''}_2\mc{A}_2 \mb{msg}(\mb{send}y_\beta\, u_\alpha),\,\m{and}\\
\circ'\,   \forall\, m.\, (\mc{A}_1 ; \mc{A}_2)\in \mc{E}\llbracket  \Delta^{\kappa}_2 \Vdash y_{\beta}{:}A \rrbracket^{m}\,\,\quad\, \m{and}\,\quad \,\,
  \circ''\,  \forall\, m.\, (\mc{B}^{\kappa'}_1 ; \mc{B}^{\kappa''}_2)\in \mc{E}\llbracket  \Delta^{\kappa}_1 \Vdash u_{\alpha+1}{:}B\rrbracket^{m}.\\
\end{array}\]

Next, we consider $\mc{B}^\lambda_1$.
By $\dagger_\lambda$, \Cref{apx:lem:moving-existential}, and $\mc{B}^\lambda_1 \mapsto^{0_{\Upsilon'_\lambda}}\mc{B}^\lambda_1$, we get
{\small\[\begin{array}{ll}
    \,\exists \Upsilon_{\lambda_2} \mc{B}^{\lambda'}_2 \,\m{such\, that} \, \,\mc{B}^{\lambda}_2\mapsto^{{*}_{ \Upsilon_{\lambda_2}}} \mc{B}^{\lambda'}_2\, \m{and}\\ 
    \; \; \forall\, x_\gamma \in \mb{Out}(\Delta^{\lambda'}, u_{\alpha}{:}A \otimes B \Vdash K^\lambda).\, \mb{if}\, x_\gamma \in \Upsilon'_\lambda.\, \mb{then}\,\,\forall\, m.\, (\mc{B}^\lambda_1; \mc{B}^{\lambda'}_2)\in \mc{V}\llbracket  \Delta^{\lambda'}, u_{\alpha}{:}A \otimes B \Vdash K^\lambda \rrbracket_{\cdot;x_\gamma}^{m+1}\,\m{and}\\ 
    \; \; \forall\, x_\gamma \in \mb{In}(\Delta^{\lambda'}, u_{\alpha}{:} A \otimes B \Vdash K^\lambda). \,\mb{if}\, x_\gamma \in \Theta'_\lambda.\, \forall\,m. \,(\mc{B}^{\lambda}_1; \mc{B}^{\lambda'}_2)\in \mc{V}\llbracket  \Delta^{\lambda}, u_{\alpha}{:}A \otimes B \Vdash K^{\lambda} \rrbracket_{x_\gamma;\cdot}^{m+1}
\end{array}\]}

In particular, we know that $u_\alpha{:}A \otimes B \in \mb{In}(\Delta^{\lambda'}, u_{\alpha}{:} A \otimes B \Vdash K^\lambda)$ and also $u_\alpha \in \Theta'_\lambda$, which gives us:
\[\begin{array}{ll}
\forall\, m.\, (\mc{B}^\lambda_1; \mc{B}^{\lambda'}_2)\in \mc{V}\llbracket  \Delta^{\lambda'}, u_{\alpha}{:}A \otimes B \Vdash K^{\lambda}\rrbracket_{u_\alpha; \cdot}^{m+1}
\end{array}\]

By {\bf Row 9.} of the logical relation for the specific channel $y_\beta$ (for which by the tree structure, we know $y_\beta \not \in \Delta^\lambda, u_{\alpha}{:}A \otimes B, K^\lambda$) we have

\[\begin{array}{ll}
    \forall m.\, (\mb{msg}(\mb{send}y_\beta\, u_\alpha)\mc{B}^\lambda_1 ; \mb{msg}(\mb{send}y_\beta\, u_\alpha)\mc{B}^{\lambda'}_2)\in \mc{E}\llbracket  \Delta^\lambda, y_{\beta}{:}A, u_{\alpha+1}{:}B \Vdash K^\lambda \rrbracket^{m}.\\
\end{array}\]

By forward closure (\Cref{apx:lem:fwdfirst}) and $\mb{msg}(\mb{send}y_\beta\, u_\alpha)\mc{B}^\lambda_1 \mapsto \mc{B}^{\lambda'}_1$ we have:
\[\begin{array}{ll}
   \circ'''\,\, \forall m.\, (\mc{B}^{\lambda'}_1 ; \mb{msg}(\mb{send}y_\beta\, u_\alpha)\mc{B}^{\lambda'}_2)\in \mc{E}\llbracket  \Delta, y_{\beta}{:}A, u_{\alpha+1}{:}B \Vdash K \rrbracket^{m}\\
\end{array}\]

Put $I'=I,\ell$, where  $\ell$ does not occur in I and define $\mc{B}^\ell_1= \mc{A}_1$ and $\mc{B}^\ell_2= \mc{A}_2$.
By induction on the number of steps, $\circ'$, $\circ''$, $\circ'''$ and $\dagger_n$ for $n \neq \kappa, \lambda$ we get 

{\small\[\begin{array}{l}
    \forall \Upsilon_1, \Theta_1,\,\mc{D}'_1. \, \mathbf{if}\, {\{\mc{B}^n_1\}_{n \in I_1} \mc{B}^{\kappa'}_1 \mc{A}_1\mc{B}^{\lambda'}_1 \{\mc{B}^n_1\}_{n \in I_2}}\mapsto^{{j'}_{ \Upsilon_1; \Theta_1}} {\mc{D}'_1}\\ \qquad  \qquad \,\mb{then}\, \exists \Upsilon_2, \mc{D}'_2 \,\m{such\, that} \, \,\{\mc{B}^n_2\}_{n \in I_1} \mc{B}^{\kappa''}_2 \mc{A}_2 \mb{msg}(\mb{send}y_\beta\, u_\alpha)\mc{B}^{\lambda'}_2 \{\mc{B}^n_2\}_{n \in I_2} 
    \mapsto^{{*}_{ \Upsilon_2}} \mc{D}'_2\, \m{and}\Upsilon_1 \subseteq \Upsilon_2\, \m{and}\\ 
     \; \; \forall\, x_\gamma \in \mb{Out}(\Delta  \Vdash K).\, \mb{if}\, x_\gamma \in \Upsilon_1.\, \mb{then}\,\, (\mc{D}'_1; \mc{D}'_2)\in \mc{V}\llbracket   \Delta \Vdash K \rrbracket_{\cdot;x_\gamma}^{k'+1}\,\m{and}\\ 
     \; \; \forall\, x_\gamma \in \mb{In}( \Delta  \Vdash K).\,\mb{if}\, x_\gamma \in \Theta_1.\, \mb{then}\,\, (\mc{D}'_1; \mc{D}'_2)\in \mc{V}\llbracket  \Delta \Vdash K \rrbracket_{x_\gamma;\cdot}^{k'+1}.
 \end{array}\]}

By a $\forall \mathbf{E}.$ and ${\{\mc{B}^n_1\}_{n \in I_1} \mc{B}^{\kappa'}_1 \mc{A}_1\mc{B}^{\lambda'}_1 \{\mc{B}^n_1\}_{n \in I_2}}\mapsto^{{j'}_{ \Upsilon_1}} {\mc{D}'_1}$ we get:
\[\begin{array}{l}
    \exists \Upsilon_2, \mc{D}'_2 \,\m{such\, that} \, \,\{\mc{B}^n_2\}_{n \in I_1} \mc{B}^{\kappa''}_2 \mc{A}_2 \mb{msg}(\mb{send}y_\beta\, u_\alpha)\mc{B}^{\lambda'}_2 \{\mc{B}^n_2\}_{n \in I_2} 
    \mapsto^{{*}_{ \Upsilon_2}} \mc{D}'_2\, \m{and}\Upsilon_1 \subseteq \Upsilon_2\, \m{and}\\ 
     \; \; \forall\, x_\gamma \in \mb{Out}(\Delta  \Vdash K).\, \mb{if}\, x_\gamma \in \Upsilon_1.\, \mb{then}\,\, (\mc{D}'_1; \mc{D}'_2)\in \mc{V}\llbracket   \Delta \Vdash K \rrbracket_{\cdot;x_\gamma}^{k'+1}\,\m{and}\\ 
     \; \; \forall\, x_\gamma \in \mb{In}( \Delta  \Vdash K).\,\mb{if}\, x_\gamma \in \Theta_1.\, \mb{then}\,\, (\mc{D}'_1; \mc{D}'_2)\in \mc{V}\llbracket  \Delta \Vdash K \rrbracket_{x_\gamma;\cdot}^{k'+1}.
 \end{array}\]

{\small We also know that $\{\mc{B}^n_2\}_{n \in I}=\{\mc{B}^n_2\}_{n \in I_1} \mc{B}^{\kappa}_2 \mc{B}^{\lambda}_2 \{\mc{B}^n_2\}_{n \in I_2}\mapsto^* \{\mc{B}^n_2\}_{n \in I_1} \mc{B}^{\kappa''}_2  \mc{A}_2 \mb{msg}(\mb{send}y_\beta\, u_\alpha) \mc{B}^{\lambda'}_2 \{\mc{B}^n_2\}_{n \in I_2}.$}
We get the following for the same $\Upsilon_2$ and $\mc{D}'_2$:
\[\begin{array}{l}
    \exists \Upsilon_2, \mc{D}'_2 \,\m{such\, that} \, \,\{\mc{B}^n_2\}_{n \in I_1} \mc{B}^{\kappa''}_2 \mc{A}_2 \mb{msg}(\mb{send}y_\beta\, u_\alpha)\mc{B}^{\lambda'}_2 \{\mc{B}^n_2\}_{n \in I_2} 
    \mapsto^{{*}_{ \Upsilon_2}} \mc{D}'_2\, \m{and}\Upsilon_1 \subseteq \Upsilon_2\, \m{and}\\ 
     \; \; \forall\, x_\gamma \in \mb{Out}(\Delta  \Vdash K).\, \mb{if}\, x_\gamma \in \Upsilon_1.\, \mb{then}\,\, (\mc{D}'_1; \mc{D}'_2)\in \mc{V}\llbracket   \Delta \Vdash K \rrbracket_{\cdot;x_\gamma}^{k'+1}\,\m{and}\\ 
     \; \; \forall\, x_\gamma \in \mb{In}( \Delta  \Vdash K).\,\mb{if}\, x_\gamma \in \Theta_1.\, \mb{then}\,\, (\mc{D}'_1; \mc{D}'_2)\in \mc{V}\llbracket  \Delta \Vdash K \rrbracket_{x_\gamma;\cdot}^{k'+1}.
 \end{array}\]
\end{description}
\end{description}

\end{proof}

\section{Logical equivalence}
\label{apx:sec:equivalence}
\subsection{Equivalence}

\begin{lemma}[Reflexivity]\label{apx:lem:reflexivity-eq}
  For all security levels $\xi$ and configurations $\Delta \Vdash \mathcal{D}:: x_\alpha {:}T$, we have 
  {
  \[(\Delta \Vdash\mathcal{D}::x_\alpha {:}T)  \equiv (\Delta \Vdash \mathcal{D}:: x_\alpha {:}T).\]
  }
  \end{lemma}
  \begin{proof}
 The proof is straightforward by applying the reflexivity Lemma (\Cref{apx:lem:reflexivity}) proved in the next section on a trivial lattice that has only one element $\perp$. We annotate all channels and processes with max secrecies and running secrecies equal to $\perp$ and annotate process definitions in the signature all with one secrecy variable $\psi$. The spawn terms in a process use a substitution that maps $\psi$ to $\perp$.
 With this translation, all session-typed configurations are also IFC-typed under the trivial lattice.
  \end{proof}
  
  \begin{lemma}[Symmetry]\label{apx:lem:symmetry-eq}
  For all configurations $\mc{D}_1$ and $\mc{D}_2$, we have 
  {\small
  $(\Delta \Vdash \mathcal{D}_1:: x_\alpha {:}T)  \equiv (\Delta \Vdash \mathcal{D}_2:: x_\alpha {:}T),$
  iff 
  $(\Delta \Vdash \mathcal{D}_2:: x_\alpha {:}T)  \equiv (\Delta \Vdash \mathcal{D}_1:: x_\alpha {:}T),$
  }
  \end{lemma}
  \begin{proof}
  The proof is straightforward by the definition of logical equivalence (\Cref{def:eq}).
  \end{proof}
  
  \begin{lemma}[Transitivity]\label{apx:lem:transitivity-eq}
  For all  configurations $\mc{D}_1$, $\mc{D}_2$, and $\mc{D}_3$, we have 

  {\small
   \[
  \begin{array}{l}
    \m{if}\,(\Delta \Vdash \mathcal{D}_1:: x_\alpha {:}T)  \equiv (\Delta \Vdash \mathcal{D}_2:: x_\alpha {:}T),\, \m{and}\, (\Delta \Vdash \mathcal{D}_2:: x_\alpha {:}T)  \equiv (\Delta \Vdash \mathcal{D}_3:: x_\alpha {:}T)\\
    \m{then}\,(\Delta \Vdash \mathcal{D}_1:: x_\alpha {:}T)  \equiv (\Delta \Vdash \mathcal{D}_3:: x_\alpha {:}T).
  \end{array}  
  \]
  }
  \end{lemma}
\begin{proof}
  The proof follows from \Cref{apx:cor:transitivity} when we instantiate $\Psi_0$ with the trivial lattice only containing the element $\perp$ and $\xi$ to be $\perp$.
\end{proof}
\subsection{Noninterference}
\begin{lemma}[Reflexivity -- only for IFC-typed processes]\label{apx:lem:reflexivity}
For all security levels $\xi$ and configurations $\Psi_0; \Gamma \Vdash \mathbb{D}:: x_\alpha {:}T[c]$, we have 
{
\[(\Gamma \Vdash \erasure{\mathbb{D}}:: x_\alpha {:}T[c])  \equiv^{\Psi_0}_{\xi} (\Gamma \Vdash \erasure{\mathbb{D}}:: x_\alpha {:}T[c]).\]
}
\end{lemma}
\begin{proof}
Corollary of the fundamental theorem (\Cref{apx:thm:ftlr}). 
\end{proof}

\begin{lemma}[Symmetry]\label{apx:lem:symmetry}
For all security levels $\xi$ and configurations $\mc{D}_1$ and $\mc{D}_2$, we have 
{\small
$(\Gamma_1 \Vdash \mathcal{D}_1:: x_\alpha {:}T_1[c_1])  \equiv^{\Psi_0}_{\xi} (\Gamma_2 \Vdash \mathcal{D}_2:: y_\beta {:}T_2[c_2]),$
iff 
$(\Gamma_2 \Vdash \mathcal{D}_2:: y_\beta {:}T_2[c_2]) \equiv^{\Psi_0}_{\xi} (\Gamma_1 \Vdash \mathcal{D}_1:: x_\alpha {:}T_1[c_1]).$
}
\end{lemma}
\begin{proof}
The proof is straightforward by \Cref{apx:def:noninterference}
\end{proof}

\begin{corollary}[Transitivity]\label{apx:cor:transitivity}
For all security levels $\xi$, and configurations $\mc{D}_1$, $\mc{D}_2$, and $\mc{D}_3$, we have 
{\small
 \[
\begin{array}{l}
  \m{if}\,(\Gamma_1 \Vdash \mathcal{D}_1:: x_\alpha {:}T_1[c_1])  \equiv^{\Psi_0}_{\xi} (\Gamma_2 \Vdash \mathcal{D}_2:: y_\beta {:}T_2[c_2]),\, \m{and}\, (\Gamma_2 \Vdash \mathcal{D}_2:: y_\beta {:}T_2[c_2])  \equiv^{\Psi_0}_{\xi} (\Gamma_3 \Vdash \mathcal{D}_3:: z_\eta {:}T_3[c_3])\\
  \m{then}\,(\Gamma_1 \Vdash \mathcal{D}_1:: x_\alpha {:}T_1[c_1]) \equiv^{\Psi_0}_{\xi} (\Gamma_3 \Vdash \mathcal{D}_3:: z_\eta {:}T_3[c_3]).
\end{array}  
\]
}
\end{corollary}
\begin{proof}
  Consider arbitrary $\mc{C}_1$, $\mc{F}_1$ and $\mc{C}_3$, and $\mc{F}_3$ we need to show \[\forall m.\, (\mc{C}_1\mc{D}_1\mc{F}_1; \mc{C}_3\mc{D}_3\mc{F}_3) \in \mc{E}^\llbracket \Delta \Vdash K\rrbracket^m \,\, \m{and}\,\,
  \forall m. \,(\mc{C}_1\mc{D}_1\mc{F}_1; \mc{C}_3\mc{D}_3\mc{F}_3) \in \mc{E}\llbracket \Delta \Vdash K\rrbracket^m\]
  By the assumptions, we get 
  $(\mc{C}_1\mc{D}_1\mc{F}_1; \mc{C}_2\mc{D}_2\mc{F}_2) \in \mc{E}\llbracket \Delta \Vdash K\rrbracket^m$
  and
  $(\mc{C}_2\mc{D}_2\mc{F}_2; \mc{C}_3\mc{D}_3\mc{F}_3) \in \mc{E}\llbracket \Delta \Vdash K\rrbracket^m.$
 The result follows by \Cref{apx:lem:transitivity-g}.
\end{proof}

\begin{lemma}[Transitivity of the term relation]\label{apx:lem:transitivity-g}
  If $\dagger_1\, \forall m. \,(\mc{D}_1;\mc{D}_2) \in \mc{E}\llbracket \Delta \Vdash K\rrbracket^m$ and $\dagger_2\,\forall m.\,(\mc{D}_2;\mc{D}_3) \in \mc{E}\llbracket \Delta \Vdash K\rrbracket^m$ then $\star\,\forall k.\, (\mc{D}_1;\mc{D}_3) \in \mc{E}\llbracket \Delta \Vdash K\rrbracket^k$.
  \end{lemma}

 \begin{proof}

Our goal is to prove for all $\mc{D}_1$, $\mc{D}_2$, and $\mc{D}_3$ with $\forall m. \,(\mc{D}_1;\mc{D}_2) \in \mc{E}\llbracket \Delta \Vdash K\rrbracket^m$ and $\forall m.\,(\mc{D}_2;\mc{D}_3) \in \mc{E}\llbracket \Delta \Vdash K\rrbracket^m$ then $\forall k.\, (\mc{D}_1;\mc{D}_3) \in \mc{E}\llbracket \Delta \Vdash K\rrbracket^k$. We prove an equivalent statement that says 
 for all $k$, $\mc{D}_1$, $\mc{D}_2$, and $\mc{D}_3$ with $\forall m. \,(\mc{D}_1;\mc{D}_2) \in \mc{E}\llbracket \Delta \Vdash K\rrbracket^m$ and $\forall m.\,(\mc{D}_2;\mc{D}_3) \in \mc{E}\llbracket \Delta \Vdash K\rrbracket^m$ then $(\mc{D}_1;\mc{D}_3) \in \mc{E}\llbracket \Delta \Vdash K\rrbracket^k$.
We proceed by induction on $k$. 

{\bf Base case. ($k=0$)} Consider arbitrary configurations $\mc{D}_1$, $\mc{D}_2$, and $\mc{D}_3$. By the assumptions, we know that $(\mc{D}_1; \mc{D}_2) \in \m{Tree}(\Delta \Vdash K)$
and $(\mc{D}_2; \mc{D}_3) \in \m{Tree}(\Delta \Vdash K)$, which gives us $(\mc{D}_1; \mc{D}_3) \in \m{Tree}(\Delta \Vdash K)$. It is enough to complete the proof in this case.

{\bf Inductive case. ($k=k'+1$)} Consider arbitrary configurations $\mc{D}_1$, $\mc{D}_2$, and $\mc{D}_3$.
Our goal is to show

\[\begin{array}{l}
  \forall\, \mathcal{D}'_1.\, \forall \,\Upsilon_1. \, \m{if}\,  \mc{D}_1 \mapsto^{*_{ \Upsilon_1}} \mc{D}'_1,\, \m{then}\, \exists \mc{D}'_3 \,\m{such\, that} \, \,\mc{D}_3\mapsto^{{*}_{ \Upsilon_1}} \mc{D}'_3\, \m{and}\\ 
   \; \; \forall\, x_\alpha \in \mb{Out}(\Delta \Vdash K).\, \mb{if}\, x_\alpha \in \Upsilon_1.\, \mb{then}\,\, (\mc{D}'_1; \mc{D}'_3)\in \mc{V}\llbracket  \Delta \Vdash K \rrbracket_{\cdot;x_\alpha}^{k'+1}\,\m{and}\\ 
   \; \; \forall\, x_\alpha \in \mb{In}(\Delta \Vdash K).\, (\mc{D}'_1; \mc{D}'_3)\in \mc{V}\llbracket  \Delta \Vdash K \rrbracket_{x_\alpha;\cdot}^{k'+1}.
  \end{array}\]

  Consider an arbitrary $\mathcal{D}'_1$ and $\Upsilon_1$, and assume $\mc{D}_1 \mapsto^{*_{ \Upsilon_1}} \mc{D}'_1$. Our goal is to prove 

  \[\begin{array}{l}
   \exists \mc{D}'_3 \,\m{such\, that} \, \,\mc{D}_3\mapsto^{{*}_{ \Upsilon_1}} \mc{D}'_3\, \m{and}\\ 
     \; \; \forall\, x_\alpha \in \mb{Out}(\Delta \Vdash K).\, \mb{if}\, x_\alpha \in \Upsilon_1.\, \mb{then}\,\, (\mc{D}'_1; \mc{D}'_3)\in \mc{V}\llbracket  \Delta \Vdash K \rrbracket_{\cdot;x_\alpha}^{k'+1}\,\m{and}\\ 
     \; \; \forall\, x_\alpha \in \mb{In}(\Delta \Vdash K).\, (\mc{D}'_1; \mc{D}'_3)\in \mc{V}\llbracket  \Delta \Vdash K \rrbracket_{x_\alpha;\cdot}^{k'+1}.
    \end{array}\]
  
  By assumption $\dagger_1$ and \Cref{apx:lem:moving-existential}, we have
  \[\begin{array}{l}
    \dagger'_1\,\forall\, \mathcal{D}'_1.\, \forall \,\Upsilon_1. \, \m{if}\,  \mc{D}_1 \mapsto^{*_{ \Upsilon_1}} \mc{D}'_1,\, \m{then}\, \exists \mc{D}'_2 \,\m{such\, that} \, \,\mc{D}_2\mapsto^{{*}_{ \Upsilon_1}} \mc{D}'_2\, \m{and}\\ 
      \; \; \forall\, x_\alpha \in \mb{Out}(\Delta \Vdash K).\, \mb{if}\, x_\alpha \in \Upsilon_1.\, \mb{then}\,\, \forall m. (\mc{D}'_1; \mc{D}'_2)\in \mc{V}\llbracket  \Delta \Vdash K \rrbracket_{\cdot;x_\alpha}^{m+1}\,\m{and}\\ 
      \; \; \forall\, x_\alpha \in \mb{In}(\Delta \Vdash K).\, \forall\, m.(\mc{D}'_1; \mc{D}'_2)\in \mc{V}\llbracket  \Delta \Vdash K \rrbracket_{x_\alpha;\cdot}^{m+1}.
    \end{array}\]

    And by assumption $\dagger_2$ and \Cref{apx:lem:moving-existential}, we have
    \[\begin{array}{l}
      \dagger'_2\,\forall\, \mathcal{D}'_2.\, \forall \,\Upsilon_1. \, \m{if}\,  \mc{D}_2 \mapsto^{*_{ \Upsilon_1}} \mc{D}'_2,\, \m{then}\, \exists \mc{D}'_3 \,\m{such\, that} \, \,\mc{D}_3\mapsto^{{*}_{ \Upsilon_1}} \mc{D}'_3\, \m{and}\\ 
        \; \; \forall\, x_\alpha \in \mb{Out}(\Delta \Vdash K).\, \mb{if}\, x_\alpha \in \Upsilon_1.\, \mb{then}\,\, \forall\, m.\,(\mc{D}'_2; \mc{D}'_3)\in \mc{V}\llbracket  \Delta \Vdash K \rrbracket_{\cdot;x_\alpha}^{m+1}\,\m{and}\\ 
        \; \; \forall\, x_\alpha \in \mb{In}(\Delta \Vdash K).\,\forall\,m.\, (\mc{D}'_2; \mc{D}'_3)\in \mc{V}\llbracket  \Delta \Vdash K \rrbracket_{x_\alpha;\cdot}^{m+1}.
      \end{array}\]

We apply $\forall \mathbf{E}.$ on $\dagger'_1$ by instantiating the existential quantifiers with $\mc{D}'_1$ and $\Upsilon_1$. We can apply $\mathbf{if} \mathbf{E}.$ since we know
$\mc{D}_1 \mapsto^{*_{ \Upsilon_1}} \mc{D}'_1$. We get a $\mc{D}'_2$ with $\mc{D}_2 \mapsto^{*_{ \Upsilon_1}} \mc{D}'_2$. Next, we apply $\forall \mathbf{E}.$ on $\dagger_2$ by instantiating the existential quantifiers with $\mc{D}'_2$ and $\Upsilon_1$. We apply $\mathbf{if} \mathbf{E}.$ as we know
$\mc{D}_2 \mapsto^{*_{ \Upsilon_1}} \mc{D}'_2$. As a result, we get a $\mc{D}'_3$ with 
$\mc{D}_3 \mapsto^{*_{ \Upsilon_1}} \mc{D}'_3$ as required by the goal. We need to prove:

\[\begin{array}{l}
  \; \; \forall\, x_\alpha \in \mb{Out}(\Delta \Vdash K).\, \mb{if}\, x_\alpha \in \Upsilon_1.\, \mb{then}\,\, (\mc{D}'_1; \mc{D}'_3)\in \mc{V}\llbracket  \Delta \Vdash K \rrbracket_{\cdot;x_\alpha}^{k'+1}\,\m{and}\\ 
  \; \; \forall\, x_\alpha \in \mb{In}(\Delta \Vdash K).\, (\mc{D}'_1; \mc{D}'_3)\in \mc{V}\llbracket  \Delta \Vdash K \rrbracket_{x_\alpha;\cdot}^{k'+1}.
  \end{array}\]

By $\dagger'_1$, and $\dagger'_2$, we have as assumptions
\[\begin{array}{l}
  \dagger''_1\,
    \; \forall\, x_\alpha \in \mb{Out}(\Delta \Vdash K).\, \mb{if}\, x_\alpha \in \Upsilon_1.\, \mb{then}\,\, \forall m. (\mc{D}'_1; \mc{D}'_2)\in \mc{V}\llbracket  \Delta \Vdash K \rrbracket_{\cdot;x_\alpha}^{m+1}\,\m{and}\\ 
    \; \;\quad \forall\, x_\alpha \in \mb{In}(\Delta \Vdash K).\, \forall\, m.(\mc{D}'_1; \mc{D}'_2)\in \mc{V}\llbracket  \Delta \Vdash K \rrbracket_{x_\alpha;\cdot}^{m+1}.
  \end{array}\]

and
  \[\begin{array}{l}
    \dagger''_2\,
      \; \forall\, x_\alpha \in \mb{Out}(\Delta \Vdash K).\, \mb{if}\, x_\alpha \in \Upsilon_1.\, \mb{then}\,\, \forall\, m.\,(\mc{D}'_2; \mc{D}'_3)\in \mc{V}\llbracket  \Delta \Vdash K \rrbracket_{\cdot;x_\alpha}^{m+1}\,\m{and}\\ 
      \; \;\quad \forall\, x_\alpha \in \mb{In}(\Delta \Vdash K).\,\forall\,m.\, (\mc{D}'_2; \mc{D}'_3)\in \mc{V}\llbracket  \Delta \Vdash K \rrbracket_{x_\alpha;\cdot}^{m+1}.
    \end{array}\]

There are two parts to prove:
\begin{description}
\item{\bf Part 1.}
Consider an arbitrary $x_\alpha \in \mb{Out}(\Delta \Vdash K)$ and assume $x_\alpha \in \Upsilon_1$. Our goal is to show
\[\begin{array}{l}
  (\mc{D}'_1; \mc{D}'_3)\in \mc{V}\llbracket  \Delta \Vdash K \rrbracket_{\cdot;x_\alpha}^{k'+1}
\end{array}\]

And by $\dagger''_1$, and $\dagger''_2$, we have as assumptions
\[\begin{array}{ll}
  \dagger'''_1\,
    \forall m. (\mc{D}'_1; \mc{D}'_2)\in \mc{V}\llbracket  \Delta \Vdash K \rrbracket_{\cdot;x_\alpha}^{m+1}\, \quad \m{and}& \quad
    \dagger'''_2\,
    \forall\, m.\,(\mc{D}'_2; \mc{D}'_3)\in \mc{V}\llbracket  \Delta \Vdash K \rrbracket_{\cdot;x_\alpha}^{m+1}.
    \end{array}\]
We consider cases based on the type of $x_\alpha$. We provide the detailed proof for a few interesting cases. The proof of other cases is similar.
\begin{description}
  \item[{\bf Case 1. ($K=x_{\alpha}{:}A \otimes B $)}.]
  By $\dagger'''_1$ and $\dagger'''_2$we get
   \[\begin{array}{lll}
    \Delta=\Delta_1, \Delta_2 & \quad \m{and} \quad \mc{D}'_1= \mc{D}''_1 \mc{A}_1\mb{msg}(\mb{send}y_\beta\, x_\alpha)& \quad \m{and}  \quad \mc{D}'_2= \mc{D}''_2 \mc{A}_2\mb{msg}(\mb{send}y_\beta\, x_\alpha)\\ 
   & \quad \m{and}\quad  \mc{D}'_2= \mc{D}''_2 \mc{A}_1\mb{msg}(\mb{send}y_\beta\, x_\alpha)& \quad \m{and}  \quad \mc{D}'_3= \mc{D}''_3 \mc{A}_2\mb{msg}(\mb{send}y_\beta\, x_\alpha)
   \end{array}
   \]  

  Moreover,
  \[\begin{array}{ll}
    \forall m. (\mc{D}''_1; \mc{D}''_2)\in \mc{E}\llbracket  \Delta_1 \Vdash x_{\alpha+1}{:}B \rrbracket^{m} & \qquad \forall m. (\mc{A}_1; \mc{A}_2)\in \mc{E}\llbracket  \Delta_2 \Vdash y_{\beta}{:}A \rrbracket^{m}\\
    \forall m. (\mc{D}''_2; \mc{D}''_3)\in \mc{E}\llbracket  \Delta_1 \Vdash x_{\alpha+1}{:}B \rrbracket^{m} & \qquad \forall m. (\mc{A}_2; \mc{A}_3)\in \mc{E}\llbracket  \Delta_2 \Vdash y_{\beta}{:}A \rrbracket^{m}
  \end{array}
\]

We apply the induction hypothesis on a smaller observation index $k'$ to get 
\[\begin{array}{ll}
  (\mc{D}''_1; \mc{D}''_3)\in \mc{E}\llbracket  \Delta_1 \Vdash x_{\alpha+1}{:}B \rrbracket^{k'} & \qquad (\mc{A}_1; \mc{A}_3)\in \mc{E}\llbracket  \Delta_2 \Vdash y_{\beta}{:}A \rrbracket^{k'}
\end{array}
\]
Which by {\bf Row 4} of the logical relation is wnough to prove the goal of this subcase, i.e.,
\[ (\mc{D}'_1; \mc{D}'_3)\in \mc{V}\llbracket  \Delta \Vdash K \rrbracket_{\cdot;x_\alpha}^{k'+1}\]
\end{description}

\item {\bf Part 2.}
Consider an arbitrary $x_\alpha \in \mb{In}(\Delta \Vdash K)$ and assume $x_\alpha \in \Theta_1x$. 
Our goal is to show
\[\begin{array}{l}
   (\mc{D}'_1; \mc{D}'_3)\in \mc{V}\llbracket  \Delta \Vdash K \rrbracket_{x_\alpha;\cdot}^{k'+1}.
  \end{array}\]

By $\dagger''_1$, and $\dagger''_2$, we have as assumptions
\[\begin{array}{ll}
  \dagger'''_1\,\forall\, m.(\mc{D}'_1; \mc{D}'_2)\in \mc{V}\llbracket  \Delta \Vdash K \rrbracket_{x_\alpha;\cdot}^{m+1}\quad \m{and} \quad
  \dagger'''_2\,\forall\,m.\, (\mc{D}'_2; \mc{D}'_3)\in \mc{V}\llbracket  \Delta \Vdash K \rrbracket_{x_\alpha;\cdot}^{m+1}.
    \end{array}\]
    We consider cases based on the type of $x_\alpha$. We provide the detailed proof for a few interesting cases. The proof of other cases is similar.
    \begin{description}
      \item[{\bf Case 1. ($\Delta=\Delta', x_{\alpha}{:}A \otimes B $)}]
      By $\dagger'''_1$ and $\dagger'''_2$we get for all $   y_\beta \not\in \mathit{dom}(\Delta, x_{\alpha}:A \otimes B, K).$
  \[\begin{array}{ll}
    \forall m. (\mb{msg}(\mb{send}y_\beta\, u_\alpha)\mc{D}'_1; \mb{msg}(\mb{send}y_\beta\, x_\alpha)\mc{D}'_2)\in \mc{E}\llbracket  \Delta, y_\beta:A, x_{\alpha+1}:B \Vdash K \rrbracket^{m}\\
    \forall m. (\mb{msg}(\mb{send}y_\beta\, u_\alpha)\mc{D}'_2; \mb{msg}(\mb{send}y_\beta\, x_\alpha)\mc{D}'_3)\in \mc{E}\llbracket  \Delta, y_\beta:A, x_{\alpha+1}:B \Vdash K \rrbracket^{m}
  \end{array}
\]

We apply the induction hypothesis on a smaller observation index $k'$ to get 
{\small\[\begin{array}{ll}
  \forall y_\beta \not\in \mathit{dom}(\Delta, x_{\alpha}:A \otimes B, K).\, (\mb{msg}(\mb{send}y_\beta\, u_\alpha)\mc{D}'_1; \mb{msg}(\mb{send}y_\beta\, x_\alpha)\mc{D}'_3)\in \mc{E}\llbracket  \Delta, y_\beta:A, x_{\alpha+1}:B \Vdash K \rrbracket^{k'}
\end{array}
\]}

Which by {\bf Row 9} of the logical relation is wnough to prove the goal of this subcase, i.e.,
\[ (\mc{D}'_1; \mc{D}'_3)\in \mc{V}\llbracket  \Delta \Vdash K \rrbracket_{x_\alpha; \cdot}^{k'+1}\]
    \end{description}
\end{description}
\end{proof}

\section{Adequacy}
\label{apx:sec:adequacy}
To show that the \slr is adequate, we prove that configurations related by the logical relation are bisimilar and vice versa.
To facilitate this proof, we first give the definition of \emph{asynchronous bisimilarity} \myCite{SangiorgiWalkerBook2001}
(\Cref{apx:def:asynchronousbisim}).
The definition relies on the labeled transition system resulting from our \slr displayed in \Cref{apx:fig:lts}
and on weak transitions and free names defined next.

\begin{figure}
\centering
\begin{small}
\input{figs/lts.tex}
\end{small}
\caption{Labeled transition system induced by \slr.}
\label{apx:fig:lts}
\end{figure}

\begin{definition}[Free names of a configuration]
  For $\Delta \Vdash \mc{D}:: \Delta'$, we define $\m{fn}(\mc{D})$ as $\mathit{dom}(\Delta,\Delta')$.
\end{definition}

\begin{definition}[Weak transition relations]$\,$\\

  \begin{enumerate}
    \item[(1)] $\xRightarrow{}$ is the reflexive and transitive closure of $\xrightarrow{\tau}$.
    \item[(2)] $\xRightarrow{\alpha}$ is $\xRightarrow{} \xrightarrow{\alpha}$.
  \end{enumerate}

\end{definition}
\begin{definition}[Asynchronous bisimilarity]\label{apx:def:asynchronousbisim}
Asynchronous bisimilarity is the largest symmetric relation such that whenever { $\mc{D}_1 \approx_a \mc{D}_2$}, we have
\begin{enumerate}
  \item{$(\tau-\mathtt{step})$} if  { $\mc{D}_1 \xrightarrow{\tau} \mc{D}'_1$} then {$\exists \mc{D}'_2. \, \mc{D}_2 \xRightarrow{\tau} \mc{D}'_2$} and  {$\mc{D}'_1 \approx_a \mc{D}'_2,$}
  \item{$(\m{output})$} if $\mc{D}_1 \xrightarrow{\overline{x_\alpha}\,q} \mc{D}'_1$ then $\exists \mc{D}'_2. \, \mc{D}_2 \xRightarrow{\overline{x_\alpha}\,q} \mc{D}'_2$ and {$\mc{D}'_1 \approx_a \mc{D}'_2$}
  \item{$(\m{left\, input})$} for all $q\not \in \m{fn}(\mc{D}_1).$, if $\mc{D}_1 \xrightarrow{{\mathbf{L}\, x_\alpha\,q}} \mc{D}'_1$ then  $\exists \mc{D}'_2. \, \mc{D}_2 \xRightarrow{\tau} \mc{D}'_2$ and {$\mc{D}'_1 \approx_a \mathbf{msg}(x_\alpha.q) \mc{D}'_2,$}
\item{$(\m{right\, input})$} for all $q\not \in \m{fn}(\mc{D}_1).$, if $\mc{D}_1 \xrightarrow{{\mathbf{R}\, x_\alpha\,q}} \mc{D}'_1$ then  $\exists \mc{D}'_2. \, \mc{D}_2 \xRightarrow{\tau} \mc{D}'_2$ and   {$\mc{D}'_1 \approx_a \mc{D}'_2\mathbf{msg}(x_\alpha.q),$}
\end{enumerate}
where $\mb{msg}(x_\alpha.q)$ is defined as $\mb{msg}(\mathbf{close}\, x_\alpha)$ if $q=\mathbf{close}$, $\mb{msg}(x_\alpha.k)$ if $q=k$, and $\mb{msg}(\mathbf{send}\,z_\delta\, x_\alpha)$ if $q=z_\delta$.
\end{definition}

\begin{definition}[High provider and High client]
 We repeat the definition of high provider and high client configurations (\Cref{apx:def:highprovider}) here.
 {\small \[\begin{array}{lll}
    \cdot \,\,\in \mb{H\text{-}Provider}^\xi(\cdot) \\[5pt]
    \mc{B} \in \mb{H\text{-}Provider}^\xi(\Gamma, x_\alpha{:}A[c])\,  & \m{iff} & c \not \sqsubseteq \xi  \, \m{and}\,
    \mc{B} = \mc{B}' \mc{T} \,\, \m{and}\,\, \mc{B}' \in \mb{H\text{-}Provider}^\xi(\Gamma)\\
    && \qquad \qquad \,\,\m{and}\,\, \mc{T} \in \m{Tree}(\cdot \Vdash x_\alpha{:}A), \mb{or}\\
    &&c \sqsubseteq \xi  \, \m{and}\,
   \mc{B} \in \mb{H\text{-}Provider}^\xi(\Gamma)\\[10pt]

    \mc{T} \in \mb{H\text{-}Client}^\xi( x_\alpha{:}A[c])\,  & \m{iff} & c \not \sqsubseteq \xi  \, \m{and}\, \mc{T} \in \m{Tree}(x_\alpha{:}A \Vdash \_:1), \mb{or}\\
    &&c \sqsubseteq \xi  \, \m{and}\, \mc{B}= \cdot
  \end{array}\]}
\defend
\end{definition}

\begin{definition}\label{apx:def:simequiv}
  For $\mc{D}_1\in \m{Tree}(\erasure{\Gamma_1} \Vdash x_\alpha {:}A_1)$, $\mc{D}_2\in \m{Tree}(\erasure{\Gamma_2} \Vdash y_\beta {:}A_2)$ we define
   \[\Gamma_1 \Vdash \mc{D}_1 :: x_\alpha {:}A_1[c_1] \approx^\xi_a \Gamma_2 \Vdash \mc{D}_2:: y_\beta {:}A_2[c_2]\,\,\m{as}\]
{\small\[\begin{array}{ll}
  \Gamma_1 \Downarrow \xi= \Gamma_2,  \Downarrow \xi\,\,\,\m{and}\,\,\, y_\beta {:}A_2[c_2] \Downarrow \xi =  x_\alpha {:}A_1[c_1] \Downarrow \xi\,\, \,\m{and}\, \\

  \forall \mc{B}_1 \in \m{H\text{-}Provider}^\xi(\Gamma_1), \mc{B}_2\in \m{H\text{-}Provider}^\xi(\Gamma_2), \mc{T}_1 \in \m{H\text{-}CLient}^\xi(x_\alpha{:}A_1[c_1]),\, \mc{T}_2 \in \m{H\text{-}Client}^\xi(y_\beta{:}A_2[c_2]).\\ \mc{B}_1\mc{D}_1\mc{T}_1 \approx_a \mc{B}_2\mc{D}_2\mc{T}_2.
\end{array}
\] }
\end{definition}

\begin{corollary}\label{apx:cor:eq-soundness-completeness}
  For all $\mc{D}_1\in \m{Tree}(\erasure{\Gamma_1} \Vdash x_\alpha {:}A_1)$ and $\mc{D}_2\in \m{Tree}(\erasure{\Gamma_2} \Vdash y_\beta {:}A_2)$, we have
 $(\Gamma_1 \Vdash \mc{D}_1:: x_\alpha {:}A_1[{c_1}]) \equiv^{\Psi_0}_{\xi} (\Gamma_2 \Vdash \mc{D}_2:: y_\beta {:}A_2[{c_2}])$ iff $(\Gamma_1 \Vdash \mc{D}_1:: x_\alpha {:}A_1[{c_1}]) \approx_a^{\xi} (\Gamma_2 \Vdash \mc{D}_2:: y_\beta {:}A_2[{c_2}])$.
  \end{corollary}
  \begin{proof}
    The proof is straightforward by considering \Cref{apx:def:noninterference}, \Cref{apx:def:simequiv}, and the following corollary (\Cref{apx:cor:soundness-completeness}).
  \end{proof}

\begin{corollary}\label{apx:cor:soundness-completeness}
  For all $(\mc{D}_1, \mc{D}_2) \in \m{Tree}(\Delta \Vdash K)$, we have $\forall m. (\mc{D}_1, \mc{D}_2) \in \mc{E}\llbracket \Delta \Vdash K\rrbracket^m$ and $\forall m. (\mc{D}_2, \mc{D}_1) \in \mc{E}\llbracket \Delta \Vdash K\rrbracket^m$  iff $\mc{D}_1 \approx_a \mc{D}_2$.
  \end{corollary}
  \begin{proof}
It is a corollary of the following lemma (\Cref{apx:lem:soundness-completeness})
  \end{proof}
\begin{lemma}\label{apx:lem:soundness-completeness}
Consider a pair of session-typed forests $(\mc{C}_1, \mc{C}_2) \in \m{Forest}(\Delta \Vdash \Delta')$ consisting of multiple session-typed trees indexed in the set $I$, i.e., $\mc{C}_i=\{\mc{C}^j_i\}_{j\in I}$ for $i \in \{1,2\}$, and $(\mc{C}^j_1, \mc{C}^j_2) \in \m{Tree}(\Delta^j \Vdash K^j)$ with $\Delta= \{\Delta^j\}_{j\in I}$ and $\Delta'= \{K^j\}_{j\in I}$. We have:
 \[\forall j\in I. \forall m.\, (\mc{C}^j_1, \mc{C}^j_2) \in \mc{E}\llbracket \Delta^j \Vdash K^j\rrbracket^m \,\, \m{and}\,\, (\mc{C}^j_2, \mc{C}^j_1) \in \mc{E}\llbracket \Delta^j \Vdash K^j\rrbracket^m \,\,\m{iff} \,\,\mc{C}_1 \approx_a \mc{C}_2.\]
\end{lemma}
\begin{proof}
  The proof consists of two parts.
\begin{description}
  \item[\bf (1) Soundness.]
   We want to prove for any arbitrary pair of forests $\mc{C}_1=\{\mc{C}^j_1\}_{j\in I}$ and $\mc{C}_2=\{\mc{C}^j_2\}_{j\in I}$ such that $\forall j \in I. \, (\mc{C}^j_1, \mc{C}^j_2) \in \m{Tree}(\Delta^j \Vdash K^j)$, if $\forall j \in I. \forall m. (\mc{C}^j_1, \mc{C}^j_2) \in \mc{E}\llbracket \Delta^j \Vdash K^j\rrbracket^m$  and $\forall j \in I. \forall m. (\mc{C}^j_2, \mc{C}^j_1) \in \mc{E}\llbracket \Delta^j \Vdash K^j\rrbracket^m$ then $\{\mc{C}^j_1\}_{j \in I} \approx_a \{\mc{C}^j_2\}_{j \in I}$.
  
  We proceed the proof by coinduction on the generating function of the bisimilarity. Consider an arbitrary $\mc{C}_1=\{\mc{C}^j_1\}_{j\in I}$ and $\mc{C}_2=\{\mc{C}^j_2\}_{j\in I}$ such that $\forall j \in I. \, (\mc{C}^j_1, \mc{C}^j_2) \in \m{Tree}(\Delta^j \Vdash K^j)$, and assume $\forall j \in I. \forall m. (\mc{C}^j_1, \mc{C}^j_2) \in \mc{E}\llbracket \Delta^j \Vdash K^j\rrbracket^m$. Our goal is to show $\{\mc{C}^j_1\}_{j \in I} \approx_a \{\mc{C}^j_2\}_{j \in I}$. We prove it by showing the items (1)-(4) \Cref{apx:def:asynchronousbisim} by assuming $\forall j \in I. \forall m. (\mc{C}^j_1, \mc{C}^j_2) \in \mc{E}\llbracket \Delta^j \Vdash K^j\rrbracket^m$  only. The symmetric relation can be established with the other assumption $\forall j \in I. \forall m. (\mc{C}^j_2, \mc{C}^j_1) \in \mc{E}\llbracket \Delta^j \Vdash K^j\rrbracket^m$ We continue by establishing items (1)-(4) in \Cref{apx:def:asynchronousbisim}:

  \begin{enumerate}
    \item Assume $\mc{C}_1 \xrightarrow{\tau} \mc{C}'_1$, then by the definition of $\tau-$transition, we have $\mc{C}_1 \mapsto \mc{C}'_1$. In other words we have $\{\mc{C}^j_1\}_{j \in I} \mapsto \{\mc{C}^{j'}_1\}_{j \in I}$ for $\mc{C}'_1=\{\mc{C}^{j'}_1\}_{j \in I}$. To establish item (1) in the definition, it is enough to prove that $\{\mc{C}^{j'}_1\}_{j \in I} \approx_a \{\mc{C}^j_2\}_{j \in I}$. Since we unfolded the coinductive definition of the bisimilarity's generating function once, we can apply the coinductive case if the assumptions are satisfied, i.e., we need to prove that each tree in the post-step forest is session-typed and the pairs are related by the logical relation.
    
    By assumption $\forall j \in I. \forall m. (\mc{C}^j_1, \mc{C}^j_2) \in \mc{E}\llbracket \Delta^j \Vdash K^j\rrbracket^m$ and by the forward closure lemma (\Cref{apx:lem:fwdfirst}), we get $\forall j \in I. \forall m. (\mc{C}^{j'}_1, \mc{C}^j_2) \in \mc{E}\llbracket \Delta^j \Vdash K^j\rrbracket^m$.  This along with type-preservation of session-typed programs, is enough to get $\{\mc{C}^{j'}_1\}_{j \in I} \approx_a \{\mc{C}^j_2\}_{j \in I}$ by coinduction.

    \item Assume $\mc{C}_1 \xrightarrow{\overline{x_\alpha}\,q} \mc{C}'_1$, then by \Cref{apx:fig:lts}  we have $\mc{C}_1= \mc{C}^1_1\mathbf{msg}(x_\alpha.q)\,\mc{C}^2_1$, with $\mc{C}'_1= \mc{C}^1_1, \mc{C}^2_1$. 
    
    This means that for some $\kappa \in I$, we have either  $\mc{C}^\kappa_1= \mc{C}^{\kappa'}_1 \mathbf{msg}(x_\alpha.q)$ or $\mc{C}^\kappa_1= \mathbf{msg}(x_\alpha.q) \mc{C}^{\kappa'}_1$, and $\mc{C}'_1= \{\mc{C}^{j}_1\}_{j \in I-\{\kappa\}}\mc{C}^{\kappa'}_1$. In particular, we know that $\forall m. (\mc{C}^{\kappa}_1, \mc{C}^\kappa_2) \in \mc{E}\llbracket \Delta^\kappa \Vdash K^\kappa\rrbracket^m$, and $\mc{C}^{\kappa}_1\mapsto^{0_{\Upsilon^\kappa_1;\Theta^\kappa_1}} \mc{C}^{\kappa}_1$, such that $x_\alpha \in \Upsilon^\kappa_1$.
   By the definition of the term interpretation and \Cref{apx:lem:moving-existential}, we get 
   \[\begin{array}{l}
    \forall\, \Upsilon^\kappa_1, \Theta^\kappa_1, \mathcal{C}^{\kappa'}_1. \m{if}\,  \mc{C}^\kappa_1 \mapsto^{*_{ \Upsilon^\kappa_1; \Theta^\kappa_1}} \mc{C}^{\kappa'}_1,\, \m{then}\, \exists \Upsilon^\kappa_2,\mc{C}^{\kappa'}_2 \,\m{such\, that} \, \,\mc{C}^\kappa_2\mapsto^{{*}_{ \Upsilon^\kappa_2}} \mc{C}^{\kappa'}_2\, \m{and}\, \Upsilon^\kappa_1 \subseteq \Upsilon^\kappa_2 \, \m{and}\\ 
    \; \; \forall\, x_\alpha \in \mb{Out}(\Delta^\kappa \Vdash K^\kappa).\, \mb{if}\, x_\alpha \in \Upsilon^\kappa_1.\, \mb{then}\,\, \forall \,k.\, (\mc{C}^{\kappa'}_1; \mc{C}^{\kappa'}_2)\in \mc{V}\llbracket  \Delta^\kappa \Vdash K^\kappa \rrbracket_{\cdot;x_\alpha}^{k+1}\,\m{and}\\ 
    \; \; \forall\, x_\alpha \in \mb{In}(\Delta^\kappa \Vdash K^\kappa). \mb{if}\, x_\alpha \in \Theta^\kappa_1.\,\mb{then}\,\,\forall k.\, (\mc{C}^{\kappa'}_1; \mc{C}^{\kappa'}_2)\in \mc{V}\llbracket  \Delta^\kappa \Vdash K^\kappa \rrbracket_{x_\alpha;\cdot}^{k+1}.
\end{array}\]
   
We can instantiate the $\forall$ quantifiers and apply an $\mathbf{if}-$Elimination rule to get
   $\mc{C}^{\kappa}_2\mapsto^{*_{\Upsilon^\kappa_2}} \mc{C}^{\kappa'}_2$ for some $\Upsilon^\kappa_2 \supseteq \Upsilon^\kappa_1$ and some $\mc{C}^{\kappa'}_2$. Since $x_\alpha \in \Upsilon_1^\kappa$, we get by the definition of the logical relation $\forall m. (\mc{C}^{\kappa}_1, \mc{C}^{\kappa'}_2) \in \mc{V}_{\Psi_0}^{\xi}\llbracket \Delta^\kappa \Vdash K^\kappa\rrbracket_{\cdot;x_\alpha}^{m+1}$. This (depending on the type of the channel $x_\alpha$) gives us $ \mc{C}^{\kappa'}_2= \mc{C}_{\kappa''_2} \mathbf{msg}(x_\alpha.q)$ or $ \mc{C}^{\kappa'}_2=  \mathbf{msg}(x_\alpha.q) \mc{C}_{\kappa''_2}$. Note that if $q$ is a channel name, $\mc{C}^{\kappa''}_2$ is a forest, rather than a tree, and if $q$ is the label $\mb{close}$, the configuration $\mc{C}_2^{\kappa''}$ is empty.

  Put $\mc{C}'_2= \{\mc{C}^{j}_2\}_{j \in I-\{\kappa\}}, \mc{C}^{\kappa''}_2$. To establish item (2) in the definition, it is enough to prove that $\mc{C}_1' \approx_a \mc{C}_2'$. Since we unfolded the coinductive definition of the bisimilarity's generating function once, we can apply the coinductive case if the assumptions are satisfied, i.e., we need to prove that each tree in the post-step forest is session-typed and the pairs are related by the logical relation.

  The proof proceeds by cases on the type of $x_\alpha$. Here we only provide one interesting case, the proof of other cases is similar.
    
    \begin{description}
    \item[\bf Subcase 1. $K^\kappa= x_\alpha{:} A\otimes B \in \Delta'$.]
    We have $\Delta^\kappa=\Delta^\kappa_1, \Delta^\kappa_2$ and  $\mc{C}^\kappa_1= \mc{C}^{\kappa'}_1 \mathbf{msg}(\mathbf{send}\,y_\delta\,x_\alpha)$ with $\mc{C}^{\kappa'}_1= \mc{C}^{\kappa''}_1\mc{A}_1$.
    Moreover, we have $\mc{C}^{\kappa'}_2= \mc{C}^{\kappa''}_2 \mathbf{msg}(\mathbf{send}\,y_\delta\,x_\alpha)$ with $\mc{C}^{\kappa''}_2= \mc{C}^{\kappa'''}_2\mc{A}_2$.
    And $\forall m. (\mc{C}^{\kappa''}_1; \mc{C}^{\kappa'''}_2) \in \mc{E}\llbracket \Delta^\kappa_1 \Vdash x_{\alpha+1}{:}B\rrbracket^m$ and $\forall m. (\mc{A}_1; \mc{A}_2) \in \mc{E}\llbracket \Delta^\kappa_2 \Vdash y_{\delta}{:}A\rrbracket^m$. 
    
    Recall that all other trees in the forest are still related, i.e. $\forall j \in I-\{\kappa\}$, we have $\forall m. (\mc{C}^{j}_1; \mc{C}^{j}_2) \in \mc{E}\llbracket \Delta^j \Vdash K^j\rrbracket^m$
    This is enough to apply the coinductive argument and get $\mc{C}_1' \approx_a \mc{C}_2'$ where  $\mc{C}'_1= \{\mc{C}^{j}_1\}_{j \in I-\{\kappa\}}, \mc{C}^{\kappa'}_1$ and   $\mc{C}'_2= \{\mc{C}^{j}_2\}_{j \in I-\{\kappa\}}, \mc{C}^{\kappa''}_2$.
    \end{description}

    \item Assume for an arbitrary $q$, we have $\mc{C}_1 \xrightarrow{{\mathbf{L}\,x_\alpha\,q}} \mc{C}'_1$.
   By \Cref{apx:fig:lts}, there is a process $\mb{proc}(z_\beta, P_{x_\alpha}) \in \mc{C}_1$ which is waiting to receive a message along $x_\alpha$. In particular, for some tree $\mc{C}^\kappa_1$ for $\kappa \in I$, we have $\mb{proc}(z_\beta, P_{x_\alpha}) \in \mc{C}^\kappa_1$. 

   We know that $\forall m. (\mc{C}^{\kappa}_1, \mc{C}^\kappa_2) \in \mc{E}\llbracket \Delta^\kappa \Vdash K^\kappa\rrbracket^m$, and $\mc{C}^{\kappa}_1\mapsto^{0_{\Upsilon^\kappa_1;\Theta^\kappa_1}} \mc{C}^{\kappa}_1$, such that $x_\alpha \in \Theta^\kappa_1$.
   By the definition of the term interpretation and \Cref{apx:lem:moving-existential}, we get 
   \[\begin{array}{l}
    \forall\, \Upsilon^\kappa_1, \Theta^\kappa_1, \mathcal{C}^{\kappa'}_1. \m{if}\,  \mc{C}^\kappa_1 \mapsto^{*_{ \Upsilon^\kappa_1; \Theta^\kappa_1}} \mc{C}^{\kappa'}_1,\, \m{then}\, \exists \Upsilon^\kappa_2,\mc{C}^{\kappa'}_2 \,\m{such\, that} \, \,\mc{C}^\kappa_2\mapsto^{{*}_{ \Upsilon^\kappa_2}} \mc{C}^{\kappa'}_2\, \m{and}\, \Upsilon^\kappa_1 \subseteq \Upsilon^\kappa_2 \, \m{and}\\ 
    \; \; \forall\, x_\alpha \in \mb{Out}(\Delta^\kappa \Vdash K^\kappa).\, \mb{if}\, x_\alpha \in \Upsilon^\kappa_1.\, \mb{then}\,\, \forall \,k.\, (\mc{C}^{\kappa'}_1; \mc{C}^{\kappa'}_2)\in \mc{V}\llbracket  \Delta^\kappa \Vdash K^\kappa \rrbracket_{\cdot;x_\alpha}^{k+1}\,\m{and}\\ 
    \; \; \forall\, x_\alpha \in \mb{In}(\Delta^\kappa \Vdash K^\kappa). \mb{if}\, x_\alpha \in \Theta^\kappa_1.\,\mb{then}\,\,\forall k.\, (\mc{C}^{\kappa'}_1; \mc{C}^{\kappa'}_2)\in \mc{V}\llbracket  \Delta^\kappa \Vdash K^\kappa \rrbracket_{x_\alpha;\cdot}^{k+1}.
\end{array}\]
   
We can instantiate the $\forall$ quantifiers and apply an $\mathbf{if}-$Elimination rule to get
   $\mc{C}^{\kappa}_2\mapsto^{*_{\Upsilon^\kappa_2}} \mc{C}^{\kappa'}_2$ for some $\Upsilon^\kappa_2 \supseteq \Upsilon^\kappa_1$ and some $\mc{C}^{\kappa'}_2$. Since $x_\alpha \in \Theta_1^\kappa$, we get by the definition of the logical relation $\forall m. (\mc{C}^{\kappa}_1, \mc{C}^{\kappa'}_2) \in \mc{V}_{\Psi_0}^{\xi}\llbracket \Delta^\kappa \Vdash K^\kappa\rrbracket_{x_\alpha; \cdot}^{m+1}$. We put $\mc{C}'_2= \{\mc{C}^{j}_2\}_{j \in I-\{\kappa\}}, \mc{C}^{\kappa'}_2$.
 
 The proof proceeds by cases on the type of $x_\alpha$. Here we only provide one interesting case, the proof of other cases is similar.

 \begin{description}
  \item[\bf Subcase 1. $x_\alpha{:} A\otimes B \in \Delta^\kappa \subseteq \Delta$ (i.e., $\Delta^\kappa=\Delta^{\kappa'}, x_\alpha{:} A\otimes B$)] In this case, we know that $\mc{C}^\kappa_1= \mc{C}^{\kappa_1}_1\, \mb{proc}(z_\delta, w \leftarrow \mb{recv}\, x_\alpha)\,\mc{C}^{\kappa_2}_1$ and $q$ is a channel $y_\delta$ which is fresh in both $\mc{C}^\kappa_1$ and $\mc{C}^{\kappa'}_2$. Moreover, we have $\mc{C}^{\kappa'}_1=\mb{msg}(\mb{send}\, y_\delta\, x_\alpha)\mc{C}^{\kappa}_1$, and thus $\mc{C}'_1=\mb{msg}(\mb{send}\, y_\delta\, x_\alpha)\mc{C}_1$.

By the logical relation, we get \\
$\forall m. (\mb{msg}(\mb{send}\, y_\delta\, x_\alpha)\mc{C}^{\kappa}_1; \mb{msg}(\mb{send}\, y_\delta\, x_\alpha)\mc{C}^{\kappa'}_2) \in \mc{E}\llbracket \Delta^{\kappa'}, y_\delta{:}A, x_{\alpha+1}{:}B \Vdash K^\kappa \rrbracket^m$.
  
  Recall that all other trees in the forest are still related, i.e. $\forall j \in I-\{\kappa\}$, we have $\forall m. (\mc{C}^{j}_1; \mc{C}^{j}_2) \in \mc{E}\llbracket \Delta^j \Vdash K^j\rrbracket^m$.\\
  This is enough to apply the coinductive argument and get $\mc{C}'_1 \approx_a \mb{msg}(\mb{send}\, y_\delta\, x_\alpha)\mc{C}'_2 $ as required.
  \end{description}
  \item Assume for an arbitrary $q$, we have $\mc{C}_1 \xrightarrow{{\mathbf{R}\,x_\alpha\,q}} \mc{C}'_1$. The proof is similar to the previous case:
  
  By \Cref{apx:fig:lts}, there is a process $\mb{proc}(x_\alpha, P_{x_\alpha}) \in \mc{C}_1$ which is waiting to receive a message along $x_\alpha$. In particular, for some tree $\mc{C}^\kappa_1$ for $\kappa \in I$, we have $\mb{proc}(x_\alpha, P) \in \mc{C}^\kappa_1$. 

   We know that $\forall m. (\mc{C}^{\kappa}_1, \mc{C}^\kappa_2) \in \mc{E}\llbracket \Delta^\kappa \Vdash K^\kappa\rrbracket^m$, and $\mc{C}^{\kappa}_1\mapsto^{0_{\Upsilon^\kappa_1;\Theta^\kappa_1}} \mc{C}^{\kappa}_1$, such that $x_\alpha \in \Theta^\kappa_1$.
   By the definition of the term interpretation and \Cref{apx:lem:moving-existential}, we get 
   \[\begin{array}{l}
    \forall\, \Upsilon^\kappa_1, \Theta^\kappa_1, \mathcal{C}^{\kappa'}_1. \m{if}\,  \mc{C}^\kappa_1 \mapsto^{*_{ \Upsilon^\kappa_1; \Theta^\kappa_1}} \mc{C}^{\kappa'}_1,\, \m{then}\, \exists \Upsilon^\kappa_2,\mc{C}^{\kappa'}_2 \,\m{such\, that} \, \,\mc{C}^\kappa_2\mapsto^{{*}_{ \Upsilon^\kappa_2}} \mc{C}^{\kappa'}_2\, \m{and}\, \Upsilon^\kappa_1 \subseteq \Upsilon^\kappa_2 \, \m{and}\\ 
    \; \; \forall\, x_\alpha \in \mb{Out}(\Delta^\kappa \Vdash K^\kappa).\, \mb{if}\, x_\alpha \in \Upsilon^\kappa_1.\, \mb{then}\,\, \forall \,k.\, (\mc{C}^{\kappa'}_1; \mc{C}^{\kappa'}_2)\in \mc{V}\llbracket  \Delta^\kappa \Vdash K^\kappa \rrbracket_{\cdot;x_\alpha}^{k+1}\,\m{and}\\ 
    \; \; \forall\, x_\alpha \in \mb{In}(\Delta^\kappa \Vdash K^\kappa). \mb{if}\, x_\alpha \in \Theta^\kappa_1.\,\mb{then}\,\,\forall k.\, (\mc{C}^{\kappa'}_1; \mc{C}^{\kappa'}_2)\in \mc{V}\llbracket  \Delta^\kappa \Vdash K^\kappa \rrbracket_{x_\alpha;\cdot}^{k+1}.
\end{array}\]
   
We can instantiate the $\forall$ quantifiers and apply an $\mathbf{if}-$Elimination rule to get
   $\mc{C}^{\kappa}_2\mapsto^{*_{\Upsilon^\kappa_2}} \mc{C}^{\kappa'}_2$ for some $\Upsilon^\kappa_2 \supseteq \Upsilon^\kappa_1$ and some $\mc{C}^{\kappa'}_2$. Since $x_\alpha \in \Theta_1^\kappa$, we get by the definition of the logical relation $\forall m. (\mc{C}^{\kappa}_1, \mc{C}^{\kappa'}_2) \in \mc{V}_{\Psi_0}^{\xi}\llbracket \Delta^\kappa \Vdash K^\kappa\rrbracket_{x_\alpha; \cdot}^{m+1}$. We put $\mc{C}'_2= \{\mc{C}^{j}_2\}_{j \in I-\{\kappa\}}, \mc{C}^{\kappa'}_2$.
 
 The proof proceeds by cases on the type of $x_\alpha$. Here we only provide one interesting case, the proof of other cases is similar.

 \begin{description}
  \item[\bf Subcase 1. $K^\kappa=x_\alpha{:} A\multimap B \in \Delta'$.] In this case, we know that\\ $\mc{C}^\kappa_1= \mc{C}^{\kappa_1}_1\, \mb{proc}(x_\alpha, w \leftarrow \mb{recv}\, x_\alpha)\,\mc{C}^{\kappa_2}_1$ and $q$ is a channel $y_\delta$ which is not free in $\mc{C}^\kappa_1$ (and by the typing not in $\mc{C}^{\kappa'}_2$ either). Moreover, we have $\mc{C}^{\kappa'}_1=\mc{C}^{\kappa}_1\mb{msg}(\mb{send}\, y_\delta\, x_\alpha)$, and thus $\mc{C}'_1=\mc{C}_1\mb{msg}(\mb{send}\, y_\delta\, x_\alpha)$.

By the logical relation, we get $\forall m. (\mc{C}^{\kappa}_1\mb{msg}(\mb{send}\, y_\delta\, x_\alpha); \mc{C}^{\kappa'}_2\mb{msg}(\mb{send}\, y_\delta\, x_\alpha)) \in \mc{E}\llbracket \Delta^\kappa, y_\delta{:}A \Vdash x_{\alpha+1}{:}B\rrbracket^m$.
  
  Recall that all other trees in the forest are still related, i.e. $\forall j \in I-\{\kappa\}$, we have $\forall m. (\mc{C}^{j}_1; \mc{C}^{j}_2) \in \mc{E}\llbracket \Delta^j \Vdash K^j\rrbracket^m$
  This is enough to apply the coinductive argument and get\\ $\mc{C}'_1 \approx_a \mc{C}'_2 \mb{msg}(\mb{send}\, y_\delta\, x_\alpha)$ as required.
  \end{description}

  \end{enumerate}

  \item[\bf (2) Completeness.] We want to prove for any arbitrary pair of forests $\mc{C}_1=\{\mc{C}^j_1\}_{j\in I}$ and $\mc{C}_2=\{\mc{C}^j_2\}_{j\in I}$ such that $\forall j \in I. \, (\mc{C}^j_1, \mc{C}^j_2) \in \m{Tree}(\Delta^j \Vdash K^j)$, if $\{\mc{C}^j_1\}_{j \in I} \approx_a \{\mc{C}^j_2\}_{j \in I}$ then $\forall j \in I. \forall m. (\mc{C}^j_1, \mc{C}^j_2) \in \mc{E}\llbracket \Delta^j \Vdash K^j\rrbracket^m$.
  
  We instead prove an equivalent statement that says
\begin{quote}
  For any natural number $m$ and any arbitrary pair of forests $\mc{C}_1=\{\mc{C}^j_1\}_{j\in I}$ and $\mc{C}_2=\{\mc{C}^j_2\}_{j\in I}$ such that $\forall j \in I. \, (\mc{C}^j_1, \mc{C}^j_2) \in \m{Tree}(\Delta^j \Vdash K^j)$, if $\{\mc{C}^j_1\}_{j \in I} \approx_a \{\mc{C}^j_2\}_{j \in I}$ then $\forall j \in I. (\mc{C}^j_1, \mc{C}^j_2) \in \mc{E}\llbracket \Delta^j \Vdash K^j\rrbracket^m$.
\end{quote}
We proceed by induction on $m$. 
\begin{description}
\item{\bf Base case. ($m=0$)} The proof is straightforward by the definition of logical relation and the fact that the configurations are session-typed.
\item{\bf Inductive case. ($m=m'+1$)} Consider an arbitrary $\kappa\in I$. By the definition of the term interpretation, we need to prove 

\[\begin{array}{l}
  \forall\, \Upsilon^\kappa_1, \Theta^\kappa_1, \mathcal{C}^\kappa_1. \m{if}\,  \mc{C}^{\kappa'}_1 \mapsto^{*_{ \Upsilon^\kappa_1; \Theta^\kappa_1}} \mc{C}^{\kappa'}_1,\, \m{then}\, \exists \Upsilon^\kappa_2,\mc{C}^{\kappa'}_2 \,\m{such\, that} \, \,\mc{C}^\kappa_2\mapsto^{{*}_{ \Upsilon^\kappa_2}} \mc{C}^{\kappa'}_2\, \m{and}\, \Upsilon^\kappa_1 \subseteq \Upsilon^\kappa_2 \, \m{and}\\ 
  \; \; \forall\, x_\alpha \in \mb{Out}(\Delta^\kappa \Vdash K^\kappa).\, \mb{if}\, x_\alpha \in \Upsilon^\kappa_1.\, \mb{then}\,\, (\mc{C}^{\kappa'}_1; \mc{C}^{\kappa'}_2)\in \mc{V}\llbracket  \Delta^\kappa\Vdash K^\kappa \rrbracket_{\cdot;x_\alpha}^{m'+1}\,\m{and}\\ 
  \; \; \forall\, x_\alpha \in \mb{In}(\Delta^\kappa \Vdash K^\kappa). \mb{if}\, x_\alpha \in \Theta^\kappa_1.\,\mb{then}\,\, (\mc{C}^{\kappa'}_1; \mc{C}^{\kappa'}_2)\in \mc{V}\llbracket  \Delta^\kappa \Vdash K^\kappa \rrbracket_{x_\alpha;\cdot}^{m'+1}.
\end{array} \]

Consider an arbitrary $\Upsilon^\kappa_1$, $\Theta^\kappa_1$, and $\mc{C}^{\kappa'}_1$, and assume $\mc{C}^{\kappa'}_1 \mapsto^{*_{ \Upsilon^\kappa_1; \Theta^\kappa_1}} \mc{C}^{\kappa'}_1$, we need to prove 

\[\begin{array}{l}
  \exists \Upsilon^\kappa_2,\mc{C}^{\kappa'}_2 \,\m{such\, that} \, \,\mc{C}^\kappa_2\mapsto^{{*}_{ \Upsilon^\kappa_2}} \mc{C}^{\kappa'}_2\, \m{and}\, \Upsilon^\kappa_1 \subseteq \Upsilon^\kappa_2 \, \m{and}\\ 
  \; \; \forall\, x_\alpha \in \mb{Out}(\Delta^\kappa \Vdash K^\kappa).\, \mb{if}\, x_\alpha \in \Upsilon^\kappa_1.\, \mb{then}\,\, (\mc{C}^{\kappa'}_1; \mc{C}^{\kappa'}_2)\in \mc{V}\llbracket  \Delta^\kappa\Vdash K^\kappa \rrbracket_{\cdot;x_\alpha}^{m'+1}\,\m{and}\\ 
  \; \; \forall\, x_\alpha \in \mb{In}(\Delta^\kappa \Vdash K^\kappa). \mb{if}\, x_\alpha \in \Theta^\kappa_1.\,\mb{then}\,\, (\mc{C}^{\kappa'}_1; \mc{C}^{\kappa'}_2)\in \mc{V}\llbracket  \Delta^\kappa \Vdash K^\kappa \rrbracket_{x_\alpha;\cdot}^{m'+1}.
\end{array} \]

By local transition steps, we know that $\mc{C}_1\mapsto^* \mc{C}'_1$, where $\mc{C}'_1=\{\mc{C}^{\kappa}_1\}_{j \in I-\{\kappa\}}\mc{C}^{\kappa'}_1$. By $\mc{C}_1 \approx_a \mc{C}_2$, we can apply the clause for $\tau$-transitions in \Cref{apx:def:asynchronousbisim} for zero or multiple times, to get $\mc{C}_2 \mapsto^* \mc{C}'_2$ such that $\mc{C}'_1 \approx_a \mc{C}'_2$. Consider the channels in the sets $\Upsilon^\kappa_1$ and $\Theta^k_1$. There are three cases:
\begin{enumerate}
  \item[A.]  By item (2) of \Cref{apx:def:asynchronousbisim} and $\mc{C}'_1 \approx_a \mc{C}'_2$, we know that for all $x_\alpha \in \mb{Out}(\Delta \Vdash K)$ such that $x_\alpha \in \Upsilon^\kappa_1$, we have $\mc{C}'_2 \mapsto^* \mc{C}^{x_s}_2$ such that $\mc{C}^{x_s}_2$ sends along the channel $x_\alpha$ and  $\mc{C}'_1= \mc{C}^1_1 \mathbf{msg}(x_\alpha.q)\mc{C}^2_1$ and  $\mc{C}^{x_s}_2= \mc{C}^{x'_s}_2 \mathbf{msg}(x_\alpha.q)\mc{C}^{x''_s}_2$, and we have $\mc{C}^1_1 \mc{C}^2_1 \approx_a \mc{C}^{x'_s}_2 \mc{C}^{x''_s}_2$.
  \item[B.] By item (3) of \Cref{apx:def:asynchronousbisim} and $\mc{C}'_1 \approx_a \mc{C}'_2$, we know that for all $x_\alpha \in \mb{In}(\Delta \Vdash \cdot)$ such that $x_\alpha \in \Theta^\kappa_1$, we have $\mc{C}'_2 \mapsto^* \mc{C}^{x_r}_2$ such that $\mathbf{msg}(x_\alpha.q) \mc{C}'_1 \approx_a \mathbf{msg}(x_\alpha.q)\mc{C}^{x_r}_2$. 
  \item[C.] By item (4) of \Cref{apx:def:asynchronousbisim} and $\mc{C}'_1 \approx_a \mc{C}'_2$, we know that for all $x_\alpha \in \mb{In}(\cdot \Vdash K)$ such that $x_\alpha \in \Theta^\kappa_1$, we have $\mc{C}'_2 \mapsto^* \mc{C}^{x_r}_2$ such that $\mc{C}'_1 \mathbf{msg}(x_\alpha.q)\approx_a \mc{C}^{x_r}_2\mathbf{msg}(x_\alpha.q)$. 
\end{enumerate}
Apply the confluence lemma (\Cref{apx:lem:confluence}) on (i) $\mc{C}'_2 \mapsto^* \mc{C}^{x_s}_2$ for  all $x_\alpha \in \mb{Out}(\Delta \Vdash K)$ such that $x_\alpha \in \Upsilon^\kappa_1$ and (ii) $\mc{C}'_2 \mapsto^* \mc{C}^{x_r}_2$ for all $x_\alpha \in \mb{In}(\Delta \Vdash K)$ such that $x_\alpha \in \Theta^\kappa_1$ to build  $\mc{C}''_2$. In particular, by the confluence lemma, we get $\mc{C}'_2\mapsto^{*_{\Upsilon_2}} \mc{C}''_2$ where $\Upsilon^{\kappa}_1 \subseteq \Upsilon_2,$ and (i') $\mc{C}^{x_s}_2\mapsto^* \mc{C}''_2$ for  all $x_\alpha \in \mb{Out}(\Delta \Vdash K)$ such that $x_\alpha \in \Upsilon^\kappa_1$ and (ii) $ \mc{C}^{x_r}_2 \mapsto^* \mc{C}''_2$ for all $x_\alpha \in \mb{In}(\Delta \Vdash K)$ such that $x_\alpha \in \Theta^\kappa_1$.

More precisely, by the forest structure, we have $\mc{C}''_2= \{\mc{C}^{\kappa''}_2\}_{j \in I-\{\kappa\}}\mc{C}^{\kappa''}_2$, with $\mc{C}^\kappa_2 \mapsto^{*_{\Upsilon^\kappa_2}} \mc{C}^{\kappa''}_2$ with $\Upsilon^\kappa_1 \subseteq \Upsilon^\kappa_2$. We use $\Upsilon^\kappa_2$ and $\mc{C}^{\kappa''}_2$ to instantiate the existential quantifier in the goal. We need to prove that
\[\begin{array}{l}
  \; \; \forall\, x_\alpha \in \mb{Out}(\Delta^\kappa \Vdash K^\kappa).\, \mb{if}\, x_\alpha \in \Upsilon^\kappa_1.\, \mb{then}\,\, (\mc{C}^{\kappa'}_1; \mc{C}^{\kappa''}_2)\in \mc{V}\llbracket  \Delta^\kappa\Vdash K^\kappa \rrbracket_{\cdot;x_\alpha}^{m'+1}\,\m{and}\\ 
  \; \; \forall\, x_\alpha \in \mb{In}(\Delta^\kappa \Vdash K^\kappa). \mb{if}\, x_\alpha \in \Theta^\kappa_1.\,\mb{then}\,\, (\mc{C}^{\kappa'}_1; \mc{C}^{\kappa''}_2)\in \mc{V}\llbracket  \Delta^\kappa \Vdash K^\kappa \rrbracket_{x_\alpha;\cdot}^{m'+1}.
\end{array} \]
By the forward closure lemma for the bisimulation (\Cref{apx:lem:bisimforward}), we get
\begin{enumerate}
 \item[A'.]  By item $[A.]$ we get, for all $x_\alpha \in \mb{Out}(\Delta \Vdash K)$ such that $x_\alpha \in \Upsilon^\kappa_1$, we have $\mc{C}'_2 \mapsto^* \mc{C}''_2$ such that $\mc{C}''_2$ sends along the channel $x_\alpha$, and we have $\mc{C}'_1= \mc{C}^1_1 \mathbf{msg}(x_\alpha.q)\mc{C}^2_1$  and  $\mc{C}''_2= \mc{C}^{x'_s} \mathbf{msg}(x_\alpha.q)\mc{C}^{x''_s}$ and $\mc{C}^1_1 \mc{C}^2_1 \approx_a \mc{C}^{x'_s} \mc{C}^{x''_s}$. 
 \item[B'.] By item $\textit{B}.$, the confluence lemma, and the forward closure lemma for the bisimulation (\Cref{apx:lem:bisimforward}), we get:  for all $x_\alpha \in \mb{In}(\Delta \Vdash \cdot)$ such that $x_\alpha \in \Theta^\kappa_1$, we have $\mc{C}'_2 \mapsto^* \mc{C}''_2$ such that $\mathbf{msg}(x_\alpha.q) \mc{C}'_1 \approx_a \mathbf{msg}(x_\alpha.q)\mc{C}''_2$. 
 \item[C'.] By item $\textit{C}.$, the confluence lemma, and the forward closure lemma for the bisimulation (\Cref{apx:lem:bisimforward}), we get:  for all $x_\alpha \in \mb{In}(\cdot \Vdash K)$ such that $x_\alpha \in \Theta^\kappa_1$, we have $\mc{C}'_2 \mapsto^* \mc{C}''_2$ such that $ \mc{C}'_1 \mathbf{msg}(x_\alpha.q)\approx_a \mc{C}''_2 \mathbf{msg}(x_\alpha.q)$. 
\end{enumerate}
There are two parts we need to prove:
\begin{description}
  \item{\bf Part 1.} $\forall\, x_\alpha \in \mb{Out}(\Delta^\kappa \Vdash K^\kappa).\, \mb{if}\, x_\alpha \in \Upsilon^\kappa_1.\, \mb{then}\,\, (\mc{C}^{\kappa'}_1; \mc{C}^{\kappa''}_2)\in \mc{V}\llbracket  \Delta^\kappa\Vdash K^\kappa \rrbracket_{\cdot;x_\alpha}^{m'+1}$
  
  Assume an arbitrary $x_\alpha \in\mb{Out}(\Delta^\kappa \Vdash K^\kappa)$ with $x_\alpha \in \Upsilon_1$. We consider cases based on the type of $x_\alpha$. Here we provide the detailed proof for one case, the proof of other cases is similar.
  \begin{description}
    \item {\bf Subcase 1. $K^\kappa= x_\alpha{:} A \otimes B$}.  We need to prove $(\mc{C}^{\kappa'}_1, \mc{C}^{\kappa''}_2) \in \mathcal{V}\llbracket \Delta^\kappa \Vdash x_\alpha{:} A \otimes B\rrbracket^{m'+1}_{\cdot;x_\alpha}.$ We first establish $\mc{C}^{\kappa'}_1= \mc{C}^{\kappa''}_1 \mc{A}_1 \mb{msg}(\mb{send}y_\beta\,x_\alpha)$ and  $\mc{C}^{\kappa''}_2= \mc{C}^{\kappa'''}_2 \mc{A}_2 \mb{msg}(\mb{send}y_\beta\,x_\alpha)$, using item [A'] above and well-typedness of programs. Next, we apply the inductive hypothesis to show $(\mc{A}_1; \mc{A}_2) \in \mc{E}\llbracket  \Delta^\kappa_1 \Vdash y_\beta{:}A \rrbracket^{m'}$ and 
    also $( \mc{C}^{\kappa''}_1 ;  \mc{C}^{\kappa'''}_2 ) \in \mc{E}\llbracket  \Delta^\kappa_2 \Vdash x_{\alpha+1}{:}B  \rrbracket^{m'}$ with $\Delta^\kappa=\Delta^\kappa_1, \Delta^\kappa_2$. We get this by configuration typing, $\mc{C}^1_1 \mc{C}^2_1 \approx_a \mc{C}^{x'_s} \mc{C}^{x''_s}$ given in item $[A']$ above and the fact that both $\mc{A}_1$, $\mc{A}_2$, $\mc{C}^{\kappa''}_1$ and $\mc{C}^{\kappa'''}_2$ are separate trees in the forests $\mc{C}^1_1 \mc{C}^2_1$ and $\mc{C}^{x'_s} \mc{C}^{x''_s}$.
  \end{description}
 
  \item{\bf Part 2.} $\forall\, x_\alpha \in \mb{In}(\Delta^\kappa \Vdash K^\kappa). \mb{if}\, x_\alpha \in \Theta^\kappa_1.\,\mb{then}\,\, (\mc{C}^{\kappa'}_1; \mc{C}^{\kappa''}_2)\in \mc{V}\llbracket  \Delta^\kappa \Vdash K^\kappa \rrbracket_{x_\alpha;\cdot}^{m'+1}.$
  Assume an arbitrary $x_\alpha \in\mb{In}(\Delta^\kappa \Vdash K^\kappa)$ with $x_\alpha \in \Theta_1$. We consider cases based on the type of $x_\alpha$. Here we provide the detailed proof for one case, the proof of other cases is similar.
  \begin{description}
    \item {\bf Subcase 1. $\Delta^\kappa=\Delta^\kappa_1 x_\alpha{:} A \otimes B$}.  We need to prove $(\mc{C}^{\kappa'}_1, \mc{C}^{\kappa''}_2) \in \mathcal{V}\llbracket \Delta^\kappa_1, x_\alpha{:}A  \otimes B \Vdash K^\kappa \rrbracket^{m'+1}_{x_\alpha;\cdot}.$ 
  We apply the inductive hypothesis on $\mathbf{msg}(\mb{send}\, y_\beta\, x_\alpha) \mc{C}'_1 \approx_a \mathbf{msg}(\mb{send} y_\beta\,x_\alpha)\mc{C}''_2$ which is given by item [B'] above to get\\
  $(\mathbf{msg}(\mb{send} y_\beta\,x_\alpha)\mc{C}^{\kappa'}_1, \mathbf{msg}(\mb{send} y_\beta\,x_\alpha)\mc{C}^{\kappa''}_2) \in \mc{E}\llbracket  \Delta^\kappa_1, y_\beta{:}A, x_{\alpha+1}{:}B \Vdash K^\kappa \rrbracket^{m'}$
    which completes the proof of this case.
  \end{description}
\end{description}
\end{description}
\end{description}
\end{proof}
\begin{lemma}[Forward closure]\label{apx:lem:bisimforward}
   For all $\mc{D}_1, \mc{D}_2$ if $\mc{D}_1 \approx_a \mc{D}_2$ and $\mc{D}_2 \Rightarrow \mc{D}'_2$, then $\mc{D}_1 \approx_a \mc{D}'_2$.
\end{lemma}
\begin{proof}
  The proof is by coinduction on the generative function of the bisimulation ($\mc{D}_1 \approx_a \mc{D}'_2$). We consider four cases required to establish $\mc{D}_1 \approx_a \mc{D}'_2$.
   \begin{enumerate}
    \item{$(\tau-\mathtt{step})$} if  { $\mc{D}_1 \xrightarrow{\tau} \mc{D}'_1$} then by the assumption $\mc{D}_2 \xRightarrow{\tau} \mc{D}''_2$ and $\mc{D}'_1 \approx_a \mc{D}''_2$. This gives us $\mc{D}_2 \mapsto^* \mc{D}''_2$. We also know by assumption that $\mc{D}_2 \mapsto^* \mc{D}'_2$. By the confluence lemma(\Cref{apx:lem:confluence}), we can build a $\mc{D}$ such that $\mc{D}''_2 \mapsto^* \mc{D}$ and $\mc{D}'_2 \mapsto^* \mc{D}$. By coinduction on $\mc{D}'_1 \approx_a \mc{D}''_2$ having the assumption $\mc{D}''_2 \mapsto^* \mc{D}$, we get  $\mc{D}'_1 \approx_a \mc{D}$, which along with $\mc{D}'_2 \xRightarrow{\tau} \mc{D}$ is enough to get the result.
    \item{$(\m{output})$} if $\mc{D}_1 \xrightarrow{\overline{x_\alpha}\,q} \mc{D}'_1$ then by the assumption $\mc{D}_2 \xRightarrow{\overline{x_\alpha}\,q} \mc{D}^4_2$ and $\mc{D}'_1 \approx_a \mc{D}^4_2$. This gives us $\mc{D}_2 \mapsto^* \mc{D}''_2$, and $\mc{D}''_2 \xrightarrow{\overline{x_\alpha}\,q}\mc{D}'''_2$, and $\mc{D}'''_2 \mapsto^* \mc{D}^4_2$. We also know by assumption that $\mc{D}_2 \mapsto^* \mc{D}'_2$. By the confluence lemma(\Cref{apx:lem:confluence}), we can build a $\mc{D}$ such that $\mc{D}''_2 \mapsto^* \mc{D}$ and $\mc{D}'_2 \mapsto^* \mc{D}$. Moreover, for some $\mc{D}'$ we have $\mc{D} \xrightarrow{\overline{x_\alpha}\,q}\mc{D}'$, such that $\mc{D}'''_2 \mapsto^* \mc{D}'$.  Recall that we also have $\mc{D}'''_2 \mapsto^* \mc{D}^4_2$. Again, we apply the confluence lemma(\Cref{apx:lem:confluence}) to get a $\mc{D}''$ such that $\mc{D}^4_2\mapsto^* \mc{D}''$ and $\mc{D}'\mapsto^* \mc{D}''$. 
   By coinduction on $\mc{D}'_1 \approx_a \mc{D}^4_2$ having the assumption $\mc{D}^4_2 \mapsto^* \mc{D}''$, we get  $\mc{D}'_1 \approx_a \mc{D}''$. Moreover, we have $\mc{D}'_2 \xRightarrow{\overline{x_\alpha}\,q} \mc{D}''$, i.e., $\mc{D}'_2 \mapsto^* \mc{D}$, and $\mc{D} \xrightarrow{\overline{x_\alpha}\,q}\mc{D}'$, and $\mc{D}'\mapsto^* \mc{D}''$.
   This is enough to get the result.
    \item{$(\m{left\, input})$} Consider an arbitrary $q$ which is not free in $\mc{D}_1$, and assume $\mc{D}_1 \xrightarrow{{\mathbf{L}\, x_\alpha\,q}} \mc{D}'_1$, then $\exists \mc{D}''_2. \, \mc{D}_2 \xRightarrow{\tau} \mc{D}''_2$ and {$\mc{D}'_1 \approx_a \mathbf{msg}(x_\alpha.q) \mc{D}''_2,$}.  This gives us $\mc{D}_2 \mapsto^* \mc{D}''_2$. We also know by assumption that $\mc{D}_2 \mapsto^* \mc{D}'_2$. By the confluence lemma(\Cref{apx:lem:confluence}), we can build a $\mc{D}$ such that $\mc{D}''_2 \mapsto^* \mc{D}$ and $\mc{D}'_2 \mapsto^* \mc{D}$. This also gives us $\mathbf{msg}(x_\alpha.q)\mc{D}''_2 \mapsto^* \mathbf{msg}(x_\alpha.q)\mc{D}$.
    By coinduction on $\mc{D}'_1 \approx_a \mathbf{msg}(x_\alpha.q) \mc{D}''_2$ having the assumption $\mathbf{msg}(x_\alpha.q)\mc{D}''_2 \mapsto^* \mathbf{msg}(x_\alpha.q)\mc{D}$, we get  $\mc{D}'_1 \approx_a \mathbf{msg}(x_\alpha.q) \mc{D}$. Moreover, we have $\mc{D}'_2 \xRightarrow{} \mc{D}$, i.e., $\mc{D}'_2 \mapsto^* \mc{D}$. This is enough to complete the proof of this case.
   
  \item{$(\m{right\, input})$} for all $q\not \in \m{fn}(\mc{D}_1)$, if $\mc{D}_1 \xrightarrow{{\mathbf{R}\, x_\alpha\,q}} \mc{D}'_1$ then  $\exists \mc{D}'_2. \, \mc{D}_2 \xRightarrow{\tau} \mc{D}'_2$ and   {$\mc{D}'_1 \approx_a \mc{D}'_2\mathbf{msg}(x_\alpha.q).$} The proof is similar to the previous case.
  \end{enumerate}
\end{proof}

\section{Biorthogonality}
\label{apx:sec:biorthogonal}
\subsection{Equivalence}
\begin{definition}[Session-typed environment]
    A session-typed environment with the interface $\Delta \Vdash K$, is of the form $\mc{C}[\,\,]\mc{F}$, such that for some $\Delta'$ and $K'$, we have $\Delta' \Vdash \mc{C}:: \Delta$ and $K \Vdash \mc{F}:: K'$.
\end{definition}

\begin{definition}[Program- and environment- relations]
A session program-relation is a binary relation between session-typed programs, i.e., open configurations of the form $\Delta \Vdash \mc{D}:: K$. Given the interface $\Delta \Vdash K$, we write $\mathit{PRel}(\Delta \Vdash K)$ for the set of all program relations that relate programs $\Delta \Vdash \mc{D}:: K$.

A session environment-relation is a binary relation between session-typed environments. Given the interface $\Delta \Vdash K$, we write $\mathit{ERel}(\Delta \Vdash K)$ for the set of all environment relations that relate environments $\mc{C}[\,\,]\mc{F}$ with the interface $\Delta \Vdash K$.
\end{definition}

\begin{definition}[The $(\_)^\top$ operation on relations]
    Given any interface $\Delta \Vdash K$ and $r \in \mathit{PRel}(\Delta \Vdash K)$, we define $r^\top \in \mathit{ERel}(\Delta \Vdash K)$ by
    \[(\mc{C}_1 [\,\,]\mc{F}_1, \mc{C}_2 [\,\,] \mc{F}_2) \in r^\top  \,\, \m{iff}\,\, \forall (\mc{D}_1, \mc{D}_2) \in r. (\mc{C}_1\mc{D}_1\mc{F}_1 \approx_a \mc{C}_2\mc{D}_2\mc{F}_2);\]
    and given any $s \in \mathit{ERel}(\Delta \Vdash K)$, we define $s^\top \in \mathit{PRel}(\Delta \Vdash K)$ by
    \[(\mc{D}_1, \mc{D}_2) \in s^\top  \,\, \m{iff}\,\, \forall (\mc{C}_1[\,\,]\mc{F}_1, \mc{C}_2[\,\,]\mc{F}_2) \in s. (\mc{C}_1\mc{D}_1\mc{F}_1 \approx_a \mc{C}_2\mc{D}_2\mc{F}_2);\]

    As explained in Pitts \myCite{PittsMSCS2000}, just by virtue of how these definitions are defined, we get that the operation is a Galois connection, which is inflationary, and idempotent.
\end{definition}

\begin{definition}\label{apx:def:eq} $\,$
\begin{itemize}
  \item  Define the relation $\Delta \Vdash \mc{D}_1 :: K \equiv \Delta \Vdash \mc{D}_2 :: K $  as $(\mc{D}_1, \mc{D}_2)\in \m{Tree}(\Delta \Vdash K)$ and $\forall m.\, (\mc{D}_1, \mc{D}_2)\in \mc{E}\llbracket \Delta \Vdash K \rrbracket$ and $\forall m. (\mc{D}_2, \mc{D}_1) \in \mc{E}\llbracket \Delta \Vdash K \rrbracket$.  
  \item  Define the relation $ \Delta \dashv \mc{C}_1[\,\,] \mc{F}_1 \dashv K \equiv \Delta \dashv \mc{C}_2 [\,\,]\mc{F}_2 \dashv K$  as (i) \,$\exists K'$ such that $K \Vdash \mc{F}_1 ::K' \equiv K \Vdash \mc{F}_2:: K'$ and (ii) \,$\exists \Delta'$ such that $(\mc{C}_1, \mc{C}_2)\in \m{Forest}(\Delta' \Vdash \Delta)$ and $\forall \mc{T}_1\in \mc{C}_1,\,\m{and}\, \forall \mc{T}_2 \in \mc{C}_2$ with $(\mc{T}_1,\mc{T}_2)\in \m{Tree}(\Delta'_1\Vdash K'')$, we have 
  $\Delta'_1 \Vdash \mc{T}_1 ::K'' \equiv \Delta'_1 \Vdash \mc{T}_2:: K''$.
\end{itemize}
\end{definition}

\begin{theorem}[$\top\top$-closure]\label{apx:thm:biorthogonal}
Consider $(\mc{D}_1, \mc{D}_2) \in \m{Tree}(\Delta \Vdash K)$, we have
    {\small\[\begin{array}{c}
    (\Delta \Vdash \mc{D}_1:: K) \equiv (\Delta \Vdash \mc{D}_2:: K)\\[5pt]
    \mathbf{iff} \\[4pt]
    \forall \mathcal{C}_1, \mathcal{C}_2, \mathcal{F}_1, \mathcal{F}_2.\,\m{s.t.}\,
    (\Delta \dashv \mc{C}_1[\,\,]\mc{F}_1\dashv  K) \equiv (\Delta \dashv \mc{C}_2[\,\,]\mc{F}_2 \dashv K)\,\,\m{we \, have}\\
    \forall m. \, {{\mathcal{C}_1}{\mathcal{D}_1} {\mathcal{F}_1}}\, \approx_a \,\mathcal{C}_2{\mathcal{D}_2} {\mathcal{F}_2}
  \end{array}\]}
\end{theorem}
  \begin{proof}
There are two directions to consider:
\begin{description}
  \item {\bf Left to Right:} The result is straightforward by the compositionality result(\Cref{apx:lem:compositionality}) and the soundness result (\Cref{apx:lem:soundness-completeness}) and the definition of equivalence $\equiv$ for programs and environments (\Cref{apx:def:eq}). 
  
  \item {\bf Right to Left:} Consider $\mc{C}_1=\mc{C}_2=\mc{F}_1=\mc{F}_2 = \cdot$. By \Cref{apx:lem:reflexivity} and \Cref{apx:def:eq}, we know that $(\Delta \dashv \cdot\,[\,\,]\,\cdot\dashv  K) \equiv (\Delta \dashv\, \cdot[\,\,]\,\cdot \dashv K)$. From this we get  $\forall m. \, {{\mathcal{D}_1} }\, \approx_a \,{\mathcal{D}_2}$, and by the completeness lemma (\Cref{apx:lem:soundness-completeness}) we get $\forall m. \, (\mathcal{D}_1 ,\mathcal{D}_2) \in \mathcal{E}\llbracket \Delta \Vdash K \rrbracket^m$ and $\forall m. \, (\mathcal{D}_2 ,\mathcal{D}_1) \in \mathcal{E}\llbracket \Delta \Vdash K \rrbracket^m$.  By \Cref{apx:def:eq}, we have $\Delta \Vdash \mc{D}_1 :: K \equiv \Delta \Vdash \mc{D}_2 :: K $, which completes the proof. 
\end{description}
This result is enough to prove that our equivalence relation is $\top\top$-closed, i.e., following the results presented in the results presented in Pitts \myCite{PittsMSCS2000}, it shows $r=s^\top$, when $r$ and $s$ defined as our logical equivalence.  
  \end{proof}

\subsection{Noninterference}
\begin{definition}[Session-typed environment]
    A session-typed environment with the interface $\Delta \Vdash K$, is of the form $\mc{C}[\,\,]\mc{F}$, such that for some $\Delta'$ and $K'$, we have $\Delta' \Vdash \mc{C}:: \Delta$ and $K \Vdash \mc{F}:: K'$.
\end{definition}

\begin{definition}[Program- and environment- relations]
A security session program-relation is a binary relation between session-typed programs, i.e., open configurations of the form $\Delta \Vdash \mc{D}:: K$. Given a security level $\xi \in \Psi_0$, and two interfaces $\Gamma_1 \Vdash x_\alpha{:}A_1[c_1]$ and $\Gamma_2 \Vdash y_\beta{:}A_2[c_2]$ 
we write $\mathit{PRel}^\xi(\Gamma_1 \Vdash x_\alpha{:}A_1[c_1], \Gamma_2 \Vdash y_\beta{:}A_2[c_2])$ for the set of all security program relations that relate programs $\erasure{\Gamma_1}\Vdash \mc{D}_1:: x_\alpha{:}A_1$ and $\erasure{\Gamma_2}\Vdash \mc{D}_2:: y_\beta{:}A_2$.

A {\bf low-security} session environment-relation is a binary relation between low security session-typed environments. Given two interfaces $\Gamma_1 \Vdash x_\alpha{:}A_1[c_1]$ and $\Gamma_2 \Vdash y_\beta{:}A_2[c_2]$, we write $\mathit{ERel}^\xi(\Gamma_1 \Vdash x_\alpha{:}A_1[c_1], \Gamma_2 \Vdash y_\beta{:}A_2[c_2])$ for the set of all environment relations that relate low security session-typed environments $\mc{C}_1[\,\,]\mc{F}_1$ with the interface $\erasure{\Gamma_1 \Downarrow \xi} \Vdash \erasure{x_\alpha{:}A_1[c_1] \Downarrow \xi}$ and $\mc{C}_2[\,\,]\mc{F}_2$ with the interface $\erasure{\Gamma_2 \Downarrow \xi} \Vdash {\erasure{y_\beta{:}A_2[c_2] \Downarrow \xi}}$. 
\end{definition}

\begin{definition}[The $(\_)^\top$ operation on relations]
  Given two interfaces $\Gamma_1 \Vdash x_\alpha{:}A_1[c_1]$ and  $\Gamma_2 \Vdash y_\beta{:}A_2[c_2]$ and the observer security level $\xi\in \Psi_0$ and $r \in \mathit{PRel}^\xi(\Gamma_1 \Vdash x_\alpha{:}A_1[c_1], \Gamma_2 \Vdash y_\beta{:}A_2[c_2])$, we define $r^\top \in \mathit{ERel}^\xi(\Gamma_1 \Vdash x_\alpha{:}A_1[c_1], \Gamma_2 \Vdash y_\beta{:}A_2[c_2])$ by
  {\small \[\begin{array}{cc}
    (\mc{C}_1 [\,\,]\mc{F}_1, \mc{C}_2 [\,\,] \mc{F}_2) \in r^\top  \,\, \m{iff}\,\, \forall (\mc{D}_1, \mc{D}_2) \in r.\\ 
     \forall \mc{B}_1 \in \m{H\text{-}Provider}^\xi(\Gamma_1), \mc{B}_2\in \m{H\text{-}Provider}^\xi(\Gamma_2), \mc{T}_1 \in \m{H\text{-}CLient}^\xi(x_\alpha{:}A_1[c_1]),\, \mc{T}_2 \in \m{H\text{-}Client}^\xi(y_\beta{:}A_2[c_2]).\\\,{\mc{B}_1{\mathcal{C}_1}{\mathcal{D}_1} {\mathcal{F}_1}\mc{T}_1}\, \approx_a \,\mc{B}_2\mathcal{C}_2{\mathcal{D}_2} {\mathcal{F}_2}\mc{T}_2
  \end{array}\]}
  and given any $s \in \mathit{ERel}(\Delta \Vdash K)$, we define $s^\top \in \mathit{PRel}(\Delta \Vdash K)$ by
 {\small \[\begin{array}{c}
    (\mc{D}_1, \mc{D}_2) \in s^\top  \,\, \m{iff}\,\, \forall (\mc{C}_1[\,\,]\mc{F}_1, \mc{C}_2[\,\,]\mc{F}_2) \in s.\\
    \forall \mc{B}_1 \in \m{H\text{-}Provider}^\xi(\Gamma_1), \mc{B}_2\in \m{H\text{-}Provider}^\xi(\Gamma_2), \mc{T}_1 \in \m{H\text{-}CLient}^\xi(x_\alpha{:}A_1[c_1]),\, \mc{T}_2 \in \m{H\text{-}Client}^\xi(y_\beta{:}A_2[c_2]).\\
    \,{\mc{B}_1{\mathcal{C}_1}{\mathcal{D}_1} {\mathcal{F}_1}\mc{T}_1}\, \approx_a \,\mc{B}_2\mathcal{C}_2{\mathcal{D}_2} {\mathcal{F}_2}\mc{T}_2
  \end{array}
    \]}

  As explained in Pitts \myCite{PittsMSCS2000}, just by virtue of how these definitions are defined, we get that the operation is a Galois connection, which is inflationary, and idempotent.
\end{definition}

\begin{definition}[Equivalence of forests by the logical relation upto the observer level]\label{def:eq-forests-sec}
We write $\Gamma'_1 \Vdash \mc{C}_1:: \Gamma_1 \equiv^\xi_{\Psi_0} \Gamma'_2 \Vdash \mc{C}_2::\Gamma_2$ iff
\begin{itemize}
  \item[(i)]{\bf (Both forests are well-typed.)}  $\mc{C}_1 \in \m{Forest}(\erasure{\Gamma'_1}\Vdash \erasure{\Gamma_1})$ and  $\mc{C}_2 \in \m{Forest}(\erasure{\Gamma'_2}\Vdash \erasure{\Gamma_2})$
  \item[(ii)] {\bf (Their observable interface is the same.)}
  $\Gamma'_1 \Downarrow \xi = \Gamma'_2 \Downarrow \xi$ and  $\Gamma_1 \Downarrow \xi = \Gamma_2 \Downarrow \xi$, and
  \item[(iii)] {\bf (Each tree in the 1st forest with an observable interface has a counterpart in the 2nd forest which is equivalent to it.)}  for every $\mc{T}_1 \in \mc{C}_1$ such that $ \mc{T}_1\in \m{Tree}(\erasure{\Gamma''_1} \Vdash  \erasure{x_\alpha{:}A_1[c_1]})$, and $\Gamma''_1 \subseteq \Gamma'_1$ and $x_\alpha{:}A_1[c_1] \in \Gamma_1$ and $(\Gamma''_1, x_\alpha{:}A_1[c_1]) \Downarrow \xi \not = \cdot$ there exists a $\mc{T}_2 \in \mc{C}_2$ such that $\ \mc{T}_2 \in \m{Tree}(\erasure{\Gamma''_2}\Vdash \erasure{y_\beta{:}A_2[c_2]})$, and $\Gamma''_2 \subseteq \Gamma'_2$ and $y_\beta{:}A_2[c_2] \in \Gamma_2$ for $\Gamma''_1 \Downarrow \xi =\Gamma''_2 \Downarrow \xi$ and $y_\beta{:}A_2[c_2] \Downarrow \xi  = x_\alpha{:}A_1[c_1] \Downarrow \xi $, and
  \item[(iv)]{\bf (Each tree in the 2nd forest with an observable interface has a counterpart in the 1st forest which is equivalent to it.)}  vice versa. 
\end{itemize}
\end{definition}

\begin{definition}\label{apx:def:eq-sc} $\,$
 Define the relation \[ \Gamma_1 \dashv \mc{C}_1[\,\,] \mc{F}_1 \dashv x_\alpha{:}A_1[c_1] \equiv^\xi_{\Psi_0} \Gamma_2 \dashv \mc{C}_2 [\,\,]\mc{F}_2 \dashv y_\beta{:}A_2[c_2]\,\,\, \m{as}\]
 \begin{enumerate}
  \item[(i)] {\bf (The observable interface of the environments are the same.)}  We have $\Gamma=\Gamma_1\Downarrow \xi= \Gamma_2\Downarrow \xi$ and $K^s=x_\alpha{:}A_1[c_1]\Downarrow \xi= y_\beta{:}A_2[c_2]\Downarrow \xi,$ and 
  \item[(ii)]{\bf (The configurations $\mc{C}_1$ and $\mc{C}_2$ offer channels annotated as observable and use the same channels and are equivalent for any observable (low-confidentiality) annotation of their resource channel.)} for some $\Delta$, we have $(\mc{C}_1, \mc{C}_2)\in \m{Forest}(\Delta \Vdash \erasure{\Gamma})$. In particular, if $\Gamma= \cdot$, we have $\mc{C}_1=\mc{C}_2=\cdot$.
  
  Moreover,  for every $\Gamma'_1$ and $\Gamma'_2$ such that $\erasure{\Gamma'_1}=\erasure{\Gamma'_2}= \Delta$, and $\forall w_\eta{:}A[c]\in \Gamma'_1. \,\,c \sqsubseteq \xi$ and $\forall w_\eta{:}A[c]\in \Gamma'_2. \,\,c \sqsubseteq \xi$  we have $ \Gamma'_1 \Vdash \mc{C}_1:: \Gamma \equiv^\xi_{\Psi_0} \Gamma'_2 \Vdash \mc{C}_2 :: \Gamma$.
  \item[(iii)] {\bf(The configurations $\mc{F}_1$ and $\mc{F}_2$ use channels annotated as observable, offer the same channels, and are equivalent for any observable (low-confidentiality) annotation of their offering channels.)} If $K^s= \_{:}1[\top]$, then $\mc{F}_1=\mc{F}_2=\cdot$. Otherwise, for some $w_\eta{:}A$, we have $(\mc{F}_1,\mc{F}_2) \in \m{Tree}(K^s \Vdash w_{\eta}{:}A)$. Moreover, \,$\forall c,c' \sqsubseteq \xi $, we have $K^s \Vdash \mc{F}_1 ::w_{\eta}{:}A[c] \equiv^\xi_{\Psi_0} K^s \Vdash \mc{F}_2::w_{\eta}{:}A[c']$ and.
  
 \end{enumerate} 
\end{definition}

\begin{lemma}[Compositionality of $\equiv^\xi_{\Psi_0}$] \label{apx:lem:compositionality-eq} 
 If \[\begin{array}{lll}
  (i) & (\Gamma'_1 \Vdash \mc{C}_1:: \Gamma) \equiv^\xi_{\Psi_0} (\Gamma'_2 \Vdash \mc{C}_2:: \Gamma)\,\, \m{with}\,\, \Gamma=\Gamma_1 \Downarrow_\xi= \Gamma_2\Downarrow_\xi\,\,\m{and}\\
  (ii) &(\Gamma_1 \Vdash \mc{D}_1:: x_\alpha{:}A_1[c_1]) \equiv^\xi_{\Psi_0} (\Gamma_2 \Vdash \mc{D}_2:: y_\beta{:}A_2[c_2])\,\m{and}\\
  (iii) &(K^s \Vdash \mc{F}_1::K^s_1) \equiv^\xi_{\Psi_0} (K^s \Vdash \mc{F}_2:: K^s_2)\,\,\m{with}\, K^s=  x_\alpha{:}A_1[c_1] \Downarrow \xi = y_\beta{:}A_2[c_2] \Downarrow \xi\,\m{then}\\
 \end{array}\]
 \[(\Gamma'_1, \Gamma^h_1 \Vdash \mc{C}_1\mc{D}_1\mc{F}_1:: K^s_1) \equiv^\xi_{\Psi_0} (\Gamma'_2, \Gamma^h_2 \Vdash \mc{C}_2\mc{D}_2\mc{F}_2:: K^s_2),\]
 where $\Gamma^h_1$ is the set of all channels $w_\eta{:}C[d] \in \Gamma_1$ with $d \not \sqsubseteq \xi$ and $\Gamma^h_2$ is the set of all channels $w_\eta{:}C[d] \in \Gamma_2$ with $d \not \sqsubseteq \xi$.
\end{lemma}
\begin{proof}
By the configuration typing and definition of equivalence (\Cref{apx:def:noninterference}) and (i), (ii), (iii), we know that $\mc{C}_1\mc{D}_1\mc{F}_1\in \m{Tree}(\erasure{\Gamma'_1} \Vdash \erasure{K^s_1})$ and $\mc{C}_2\mc{D}_2\mc{F}_2\in \m{Tree}(\erasure{\Gamma'_2} \Vdash \erasure{K^s_2})$. Moreover, by definition of equivalence (\Cref{apx:def:noninterference}) and (i), (iii) we get $\Gamma'=(\Gamma'_1, \Gamma^h_1) \Downarrow \xi= (\Gamma'_2, \Gamma^h_1) \Downarrow \xi$ and $K^{s'}=K^s_1\Downarrow \xi = K^s_2 \Downarrow \xi$.

Assume arbitrary configurations {\small$\mc{B}_1 \in \mb{H\text{-}Provider}^\xi(\Gamma'_1)$, $\mc{B}^h_1 \in \mb{H\text{-}Provider}^\xi(\Gamma^h_1)$,  $\mc{B}_2 \in \mb{H\text{-}Provider}^\xi(\Gamma'_2)$, $\mc{B}^h_2 \in \mb{H\text{-}Provider}^\xi(\Gamma^h_2)$,  $\mc{T}_1 \in \mb{H\text{-}Client}^\xi(K^s_1)$, and $\mc{T}_2 \in \mb{H\text{-}Client}^\xi(K^s_2)$.} Our goal is to prove 
\[\begin{array}{l}  \forall\,m.\,(\mc{B}_1\mc{B}^h_1\mc{C}_1\mc{D}_1\mc{F}_1\mc{T}_1, \mc{B}_2\mc{B}^h_2\mc{C}_2\mc{D}_2\mc{F}_2\mc{T}_2) \in \mc{E}\llbracket \erasure{\Gamma'} \Vdash \erasure{K^{s'}} \rrbracket^{m},\,
  \m{and} \\[1pt]
  \forall\,m.\,(\mc{B}_2\mc{B}^h_2\mc{C}_2\mc{D}_2\mc{F}_2\mc{T}_2, \mc{B}_1\mc{B}^h_1\mc{C}_1\mc{D}_1\mc{F}_1\mc{T}_1) \in \mc{E}\llbracket \erasure{\Gamma'} \Vdash \erasure{K^{s'}} \rrbracket^{m}.
  \end{array}
  \]

By assumptions (i)-(iii) we have:
\begin{itemize}
\item[(i')] 
Note that for any trees $\mc{A}''_1 \in \mc{C}_1$ and $\mc{A}''_2 \in \mc{C}_2$ they have an observable offering channel occuring in $\Gamma$.
Assume that $\mc{C}_1$ as a forest includes $n$ separate trees. By the previous observation, we know that $\mc{C}_2$ also consists of $n$ separate trees. For any $j\le n$ consider the tree $\mc{A}^j_1 \in \mc{C}_1$ consider the corresponding tree $\mc{A}^j_2 \in \mc{C}_2$ that exists by \Cref{def:eq-forests-sec} and for which we have $(\dagger_j)\,\,\Gamma^{j'}_1\Vdash \mc{A}^j_1:: w^j_\eta{:}A^j[c^j] \equiv^\xi_{\Psi_0}\Gamma^{j'}_2 \Vdash \mc{A}^j_2:: w^j_\eta{:}A^j[c^j]$. 

\,\,From $\dagger_j$, we get $\Gamma^{j'}=\Gamma^{j'}_1=\Gamma^{j'}_2 $ and for arbitrary configurations $\mc{T}^j_1 \in \mb{H\text{-}Client}^\xi(\Gamma^{j'}_1)$, and $\mc{T}^j_2 \in \mb{H\text{-}Client}^\xi(\Gamma^{j'}_2)$, we have 

\[\begin{array}{l}  \forall\,m.\,(\mc{B}^j_1\mc{C}_1, \mc{B}^j_2\mc{C}_2) \in \mc{E}\llbracket \erasure{\Gamma^{j'} \Downarrow \xi} \Vdash w^j_\eta{:}A^j\rrbracket^{m},\,
  \m{and} \\[1pt]
  \forall\,m.\,(\mc{B}^j_2\mc{C}_2, \mc{B}^j_1\mc{C}_1) \in \mc{E}\llbracket  \erasure{\Gamma^{j'} \Downarrow \xi} \Vdash w^j_\eta{:}A^j \rrbracket^{m}.
  \end{array}
  \]

Observe that $\Gamma'_1=\{\Gamma^{j'}_1\}_{j\le n}$ and  $\Gamma'_2=\{\Gamma^{j'}_2\}_{j\le n}$, and $\Gamma=\{w^j_\eta{:}A^j[c^j]\}_{j \le n}$. 
 \item[(ii')]By \Cref{apx:def:noninterference}, we know that for arbitrary configurations $\mc{B}^h_1 \in \mb{H\text{-}Provider}^\xi(\Gamma_1)$, and  $\mc{B}^d_2 \in \mb{H\text{-}Provider}^\xi(\Gamma_2)$, and $\mc{T}^d_1 \in \mb{H\text{-}Client}^\xi(x_\alpha{:}{A}_1[c_1])$, and $\mc{T}^d_2 \in \mb{H\text{-}Client}^\xi(y_\beta{:}A_2[c_2])$.

  \[\begin{array}{l}  \forall\,m.\,(\mc{B}^d_1\mc{D}_1\mc{T}^d_1, \mc{B}^d_2\mc{D}_2\mc{T}^d_2) \in \mc{E}\llbracket \erasure{\Gamma} \Vdash \erasure{K^s} \rrbracket^{m},\,
    \m{and} \\[1pt]
    \forall\,m.\,(\mc{B}^d_2\mc{D}_2\mc{T}^d_2, \mc{B}^d_1\mc{D}_1\mc{T}^d_1) \in \mc{E}\llbracket \erasure{\Gamma} \Vdash \erasure{K^s} \rrbracket^{m}.
    \end{array}
    \] 

  \item[(iii')] By \Cref{apx:def:noninterference}, if $c_1\sqsubseteq \xi$ and $c_2 \sqsubseteq \xi$, we have  for any $\mc{T}_1 \in \mb{H\text{-}Client}^\xi(K^s_1)$, and $\mc{T}_2 \in \mb{H\text{-}Client}^\xi(K^s_2)$
  
  \[\begin{array}{l}  \forall\,m.\,(\mc{F}_1\mc{T}^h_1, \mc{F}_2\mc{T}^h_2) \in \mc{E}\llbracket \erasure{K^s} \Vdash \erasure{K^{s'}} \rrbracket^{m},\,
    \m{and} \\[1pt]
    \forall\,m.\,(\mc{F}_2\mc{T}^h_2, \mc{F}_1\mc{T}^h_1) \in \mc{E}\llbracket \erasure{K^s} \Vdash \erasure{K^{s'}} \rrbracket^{m},\,
    \end{array}
    \] 

\end{itemize}

It is straightforward to show that given (i'-iii'), and several applications of the compositionality lemma \Cref{apx:lem:compositionality} we get the goal.
\end{proof}

\begin{theorem}[$\top\top$-closure]\label{apx:thm:biorthogonal-sc}
Consider $\mc{D}_1 \in \m{Tree}(\erasure{\Gamma_1} \Vdash \erasure{x_\alpha{:}A_1[c_1]})$ and $\mc{D}_2 \in \m{Tree}(\erasure{\Gamma_2} \Vdash \erasure{y_\beta{:}A_2[c_2]})$ and a given observer level $\xi \in \Psi_0$. We have
  {\small \[\begin{array}{c}
  (\Gamma_1 \Vdash \mc{D}_1:: x_\alpha{:}A_1[c_1]) \equiv^\xi_{\Psi_0} (\Gamma_2 \Vdash \mc{D}_2:: y_\beta{:}A_2[c_2])\\[5pt]
  \mathbf{iff} \\[4pt]
  \forall \mathcal{C}_1, \mathcal{C}_2, \mathcal{F}_1, \mathcal{F}_2.\,\m{s.t.}\,
  (\Gamma_1 \dashv \mc{C}_1[\,\,]\mc{F}_1\dashv  x_\alpha{:}A_1[c_1]) \equiv^\xi_{\Psi_0} (\Gamma_2 \dashv \mc{C}_2[\,\,]\mc{F}_2 \dashv y_\beta{:}A_2[c_2])\,\,\m{we \, have}\\
  \forall \mc{B}_1 \in \m{H\text{-}Provider}^\xi(\Gamma_1), \mc{B}_2\in \m{H\text{-}Provider}^\xi(\Gamma_2), \mc{T}_1 \in \m{H\text{-}CLient}^\xi(x_\alpha{:}A_1[c_1]),\, \mc{T}_2 \in \m{H\text{-}Client}^\xi(y_\beta{:}A_2[c_2]).\\{\mc{B}_1{\mathcal{C}_1}{\mathcal{D}_1} {\mathcal{F}_1}\mc{T}_1}\, \approx_a \,\mc{B}_2\mathcal{C}_2{\mathcal{D}_2} {\mathcal{F}_2}\mc{T}_2
\end{array}\]}
\end{theorem}
\begin{proof}
There are two directions to consider:
\begin{description}
\item {\bf Left to Right:} Consider $(\Gamma_1 \Vdash \mc{D}_1:: x_\alpha{:}A_1[c_1]) \equiv^\xi_{\Psi_0} (\Gamma_2 \Vdash \mc{D}_2:: y_\beta{:}A_2[c_2])$ and arbitrary $\mc{C}_1, \mc{C}_2, \mc{F}_1, \mc{F}_2$ such that we have 
\[(\Gamma_1 \dashv \mc{C}_1[\,\,]\mc{F}_1\dashv  x_\alpha{:}A_1[c_1]) \equiv^\xi_{\Psi_0} (\Gamma_2 \dashv \mc{C}_2[\,\,]\mc{F}_2 \dashv y_\beta{:}A_2[c_2]).\] 

By \Cref{apx:def:eq-sc}, we get that $\Gamma=\Gamma_1\Downarrow \xi= \Gamma_2\Downarrow \xi$ and $K^s=x_\alpha{:}A_1[c_1]\Downarrow \xi= y_\beta{:}A_2[c_2]\Downarrow \xi$. Moreover for some $\Delta$, we have $(\mc{C}_1,\mc{C}_2) \in \m{Forest}(\Delta \Vdash \erasure{\Gamma})$ and for any low secrecy annotations of $\Delta$ as $\Gamma'_1$ and $\Gamma'_2$ we get $\Gamma'_1 \Vdash \mc{C}_1:: \Gamma\equiv^\xi_{\Psi_0} \Gamma'_2 \Vdash \mc{C}_2:: \Gamma$. Also if $K^s=\_{:}1[\top]$, we have $\mc{F}_1=\mc{F}_2=\cdot$, and otherwise $(\mc{F}_1, \mc{F}_2)\in \m{Tree}(K^s\Vdash w_\eta{:}A)$ for some $w_\eta{:}A$, for any low secrecy annotations (with $c_3,c_4 \sqsubseteq \xi$) of $w_\eta{:}A$ as $w_\eta{:}A[c_3]$  and $w_\eta{:}A[c_4]$ we get $K^s \Vdash \mc{C}_1:: w_\eta{:}A[c_3]\equiv^\xi_{\Psi_0} K^s \Vdash \mc{C}_2:: w_\eta{:}A[c_4]$.

 Our goal is to prove:
{\small\[\begin{array}{c}
  \forall \mc{B}_1 \in \m{H\text{-}Provider}^\xi(\Gamma_1), \mc{B}_2\in \m{H\text{-}Provider}^\xi(\Gamma_2), \mc{T}_1 \in \m{H\text{-}CLient}^\xi(x_\alpha{:}A_1[c_1]),\, \mc{T}_2 \in \m{H\text{-}Client}^\xi(y_\beta{:}A_2[c_2]).\\{{\mathcal{C}_1}{\mathcal{D}_1} {\mathcal{F}_1}}\, \approx_a \,\mathcal{C}_2{\mathcal{D}_2} {\mathcal{F}_2}.
\end{array}\]
}

By compositionality of equivalence upto relation (\Cref{apx:lem:compositionality-eq}), we get\\
 $ \Gamma'_1, \Gamma^h_1  \Vdash \mc{C}_1\mc{D}_1 \mc{F}_1:: x_\alpha{:}A_1[c_1] \equiv^\xi_{\Psi_0} \Gamma'_2, \Gamma^h_2 \Vdash \mc{C}_2\mc{D}_2\mc{F}_2 :: y_\beta{:}A_2[c_2]$ in the case where $K^s:= \_{:}1[\top]$, i.e., $c_1, c_2 \not \sqsubseteq \xi$, and otherwise we get $\Gamma'_1, \Gamma^h_1  \Vdash \mc{C}_1\mc{D}_1 \mc{F}_1:: w_\eta{:}A[c_3] \equiv^\xi_{\Psi_0} \Gamma'_2, \Gamma^h_2 \Vdash \mc{C}_2\mc{D}_2\mc{F}_2 ::  w_\eta{:}A[c_4]$ where $\Gamma^h_1$ is the set of all channels $z_\gamma{:}C[d] \in \Gamma_1$ with $d \not \sqsubseteq \xi$ and $\Gamma^h_2$ is the set of all channels $z_\gamma{:}C[d] \in \Gamma_2$ with $d \not \sqsubseteq \xi$.

By the soundness theorem for equivalence relation (\Cref{apx:cor:eq-soundness-completeness}), we get $\, \Gamma'_1, \Gamma^h_1 \Vdash \mathcal{C}_1{\mathcal{D}_1} {\mathcal{F}_1} :: x_\alpha{:}A_1[c_1] \, \approx^\xi_a \,\Gamma'_1, \Gamma^h_1 \Vdash \mathcal{C}_2{\mathcal{D}_2} {\mathcal{F}_2} :: y_\beta{:}A_2[c_2]$ in the first case where $K^s=\_{:}1[\top]$ and otherwise $ \,\Gamma'_1, \Gamma^h_1 \Vdash \mathcal{C}_1{\mathcal{D}_1} {\mathcal{F}_1} :: w_\eta{:}A[c_3] \, \approx^\xi_a \,\Gamma'_1, \Gamma^h_1 \Vdash \mathcal{C}_2{\mathcal{D}_2} {\mathcal{F}_2} ::w_\eta{:}A[c_4]$. 
In the first case, the proof is complete by \Cref{apx:def:simequiv} and the observation that all channels in $\Gamma'_1$ and $\Gamma'_2$ are annotated with channels of secrecy $d' \sqsubseteq \xi$. In the second case, the proof is complete by \Cref{apx:def:simequiv} and observing that all channels in $\Gamma'_1$ and $\Gamma'_2$ are annotated with channels of secrecy $d' \sqsubseteq \xi$ and that $c_1,c_2 \sqsubseteq \xi$ and $c_3,c_4 \sqsubseteq \xi$.

\item {\bf Right to Left:} Assume 
$\mc{C}_1=\mc{C}_2=\mc{F}_1=\mc{F}_2 = \cdot$. By \Cref{apx:lem:reflexivity} and \Cref{apx:def:eq}, we know that $(\Gamma_1 \dashv\cdot\,[\,\,]\,\cdot\dashv  x_\alpha{:}A_1[c_1]) \equiv (\Gamma_2 \dashv\, \cdot[\,\,]\,\cdot \dashv y_\beta{:}A_1[c_2])$. From this we get $ {\Gamma_1 \Vdash {\mathcal{D}_1}:: x_{\alpha}{:}A_1[c_1] }\, \approx^\xi_a \,{\Gamma_2 \Vdash \mathcal{D}_2}:: y_{\beta}{:}A_2[c_2]$. Now, we can apply the completeness lemma (\Cref{apx:cor:eq-soundness-completeness}) to get $(\Gamma_1 \Vdash \mc{D}_1:: x_\alpha{:}A_1[c_1]) \equiv^\xi_{\Psi_0} (\Gamma_2 \Vdash \mc{D}_2:: y_\beta{:}A_2[c_2])$.  
\end{description}

This result is enough to prove that our equivalence relation is $\top\top$-closed, i.e., following the results presented in the results presented in Pitts \myCite{PittsMSCS2000}, it shows $r=s^\top$, when $r$ and $s$ defined as our logical equivalence.  

\end{proof}

\end{document}